\documentclass[journal]{IEEEtran}
\usepackage[T1]{fontenc}
\usepackage{amsmath,amssymb,amsfonts,mathtools}
\usepackage{amsthm}
\usepackage{bm}
\usepackage{graphicx}
\usepackage[dvipsnames,table]{xcolor}
\usepackage{multirow}
\usepackage{makecell}
\usepackage{array}
\usepackage{tabularx}
\usepackage{algorithm}
\usepackage[noend]{algpseudocode}
\usepackage{setspace}

\usepackage[inline]{enumitem}
\usepackage{pifont}
\usepackage{xspace}
\usepackage{stmaryrd}
\usepackage{cite}
\usepackage[hidelinks]{hyperref}
\usepackage[capitalise,noabbrev]{cleveref}

\crefname{section}{Section}{Sections}
\Crefname{section}{Section}{Sections}
\crefname{subsection}{Section}{Sections}
\Crefname{subsection}{Section}{Sections}
\crefname{subsubsection}{Section}{Sections}
\Crefname{subsubsection}{Section}{Sections}
\crefname{appendix}{Appendix}{Appendices}
\Crefname{appendix}{Appendix}{Appendices}
\crefname{figure}{Fig.}{Figs.}
\Crefname{figure}{Fig.}{Figs.}
\crefname{table}{Table}{Tables}
\Crefname{table}{Table}{Tables}
\crefname{algorithm}{Algorithm}{Algorithms}
\Crefname{algorithm}{Algorithm}{Algorithms}
\crefname{equation}{Eq.}{Eqs.}
\Crefname{equation}{Eq.}{Eqs.}
\crefname{ALC@unique}{line}{lines}
\crefname{theorem}{Theorem}{Theorems}
\crefname{lemma}{Lemma}{Lemmas}
\crefname{corollary}{Corollary}{Corollaries}
\crefname{definition}{Definition}{Definitions}
\crefname{assumption}{Assumption}{Assumptions}
\crefname{remark}{Remark}{Remarks}
\crefname{proposition}{Proposition}{Propositions}

\theoremstyle{plain}
\newtheorem{theorem}{Theorem}
\newtheorem{lemma}{Lemma}

\theoremstyle{definition}
\newtheorem{definition}{Definition}
\newtheorem{assumption}{Assumption}
\theoremstyle{remark}

\algrenewcommand\algorithmicrequire{\textbf{Input:}}
\algrenewcommand\algorithmicensure{\textbf{Output:}}
\algnewcommand{\LineComment}[1]{\State \(\triangleright\) \textit{#1}}
\algnewcommand{\Abort}{\textbf{abort}}
\newcommand{\algstretch}{\setstretch{0.92}}

\newcommand{\sys}{VIGOR-DFA\xspace}

\newcommand{\ProfT}{Profile~$\mathsf{T}$\xspace}
\newcommand{\ProfU}{Profile~$\mathsf{U}$\xspace}

\newcommand{\F}{\mathbb{F}}
\newcommand{\Fp}{\mathbb{F}_p}
\newcommand{\Z}{\mathbb{Z}}
\newcommand{\bits}[1]{\{0,1\}^{#1}}
\newcommand{\one}{\mathbf{1}}
\newcommand{\getsr}{\xleftarrow{\$}}
\newcommand{\inner}[2]{\langle #1, #2\rangle}
\newcommand{\auth}[1]{[#1]}                  
\newcommand{\authb}[2]{[#1]_{#2}}            

\newcommand{\Gen}{\mathsf{Gen}}
\newcommand{\Eval}{\mathsf{Eval}}
\newcommand{\EvalAll}{\mathsf{EvalAll}}
\newcommand{\DPF}{\mathsf{DPF}}
\newcommand{\DCF}{\mathsf{DCF}}
\newcommand{\CW}{\mathsf{CW}}
\newcommand{\CWout}{\mathsf{CW}^{\mathrm{out}}}
\newcommand{\sCW}{s^{\mathsf{CW}}}
\newcommand{\tCW}[1]{t^{\mathsf{CW},#1}}
\newcommand{\PRG}{G}
\newcommand{\Conv}{H}
\newcommand{\keep}{\mathit{keep}}
\newcommand{\lose}{\mathit{lose}}

\newcommand{\PRF}{\mathsf{PRF}}

\newcommand{\Hash}{\mathsf{Hash}}

\newcommand{\Adv}{\mathcal{A}}
\newcommand{\Sim}{\mathcal{S}}

\newcommand{\Real}{\mathsf{Real}}
\newcommand{\Ideal}{\mathsf{Ideal}}
\newcommand{\Hyb}[1]{\mathsf{H}_{#1}}
\newcommand{\Leak}{\mathcal{L}}
\newcommand{\FC}{\mathcal{F}^{\mathrm{abort}}_{C,\mu}}
\newcommand{\FPIR}{\mathcal{F}^{\mathrm{abort}}_{\mathrm{PIR}}}
\newcommand{\eps}[1]{\epsilon_{\mathsf{#1}}}

\newcommand{\cmd}[1]{\texttt{\detokenize{#1}}}
\newcommand{\lstate}[1]{\textsf{#1}}
\newcommand{\algname}[1]{\textsf{#1}}
\newcommand{\CheckOpen}{\algname{CheckOpen}\xspace}
\newcommand{\AuthShare}{\algname{AuthShare}\xspace}
\newcommand{\Epilogue}{\algname{Epilogue}\xspace}

\newcommand{\cmark}{\ding{51}}
\newcommand{\xmark}{\ding{55}}
\newcommand{\pmark}{\ding{109}}   

\newcommand{\authorcard}[2]{%
  \begin{tabular}[t]{@{}c@{}}
    #1\textsuperscript{#2}
  \end{tabular}}

\begin{document}

\title{Trusted Hardware Acceleration for Malicious-Secure Function Secret Sharing}

\author{%
	\begin{tabular}{@{}c@{\hspace{2.2em}}c@{\hspace{2.2em}}c@{\hspace{2.2em}}c@{}}
		\authorcard{Yujie Xue}{$\dagger$,1} &
		\authorcard{Yijing Peng}{$\dagger$,2} &
		\authorcard{Lin Liu}{*,1} &
		\authorcard{Shaojing Fu}{1} \\
		\authorcard{Shaoqing Li}{2} &
		\authorcard{Yaohua Wang}{2} &
		\authorcard{Rongmao Chen}{1} &
		\authorcard{Yang Guo}{2}
	\end{tabular}\\[0.9em]
	\textsuperscript{1}Center for Cryptologic Research,\par%
	\textsuperscript{2}Key Laboratory of Advanced Microprocessor Chips and Systems,\par%
	College of Computer Science and Technology, National University of Defense Technology, Changsha, China\par
	\thanks{$^\dagger$These authors contributed equally to this work.}%
	\thanks{\textsuperscript{*}Corresponding author: Lin Liu (e-mail: liulin16@nudt.edu.cn).}%
	\thanks{A condensed version of this paper was submitted to IEEE Transactions on Information Forensics and Security on October 1, 2026.}%
}

\markboth{}%
{Xue \MakeLowercase{\textit{et al.}}: Trusted Hardware Acceleration for Malicious-Secure Function Secret Sharing}

\maketitle

\begin{abstract}
Function secret sharing (FSS) underlies two-party private inference and private information retrieval, with cost dominated by generating, moving and evaluating distributed point function (DPF) keys. A trusted GPU-integrated distributed function accelerator (DFA) removed key movement by generating and consuming keys locally, but tolerates only semi-honest adversaries. A malicious host or GPU can tamper with shares, replay one-time material, swap buffers after checking, request early outputs, or abuse the accelerator as a forgery oracle, while malicious FSS ships large authenticated keys or multiplies DPF work. We present \sys, protecting the chain from authorized input to authorized output release with three mechanisms: a fresh authentication epilogue using three field multiplications per DPF output, {3.8--4.0$\times$} faster per gate than per-lane DPF tag trees; a freeze-before-challenge check of every opening with $t=3$ independent MAC lanes over $\F_{2^{61}-1}$; and a role-bound one-time resource ledger with a release guard, in a protected datapath beside the GPU L2 cache. We prove stand-alone static malicious security with abort in a protected-module model, with statistical error $Q(2/p)^t\approx 2^{-148}$ for $Q\le 2^{32}$ checked batches. Our DFA-calibrated model shows that, against dealer-based malicious FSS modeled after the protocol family of Shark, \sys removes {21.8--563~GB} of per-query offline authenticated material and, mainly by generating it in-module, lowers LAN latency by {10.1--14.0$\times$} ({1.5--1.8$\times$} excluding offline distribution) and energy by {3.0--3.9$\times$}. Malicious security costs {2.5--3.5$\times$} LAN latency over semi-honest DFA and {0.145~mm$^2$} at 7~nm. We have completed the verification of specifications and the functional CPU reference model, including GPU/RTL conformance verification, protected runtime evaluation, and deployment-related tests.
\end{abstract}


\section{Introduction}\label{sec:intro}

\IEEEPARstart{P}{rivacy}-preserving computation is moving into deployments where data are both valuable and regulated: hospitals running diagnostic models on records they may not export, cloud providers serving proprietary models to clients who will not reveal their queries, and lookup services whose queries disclose a user's interests. Secure two-party computation (2PC) for private machine-learning (ML) inference~\cite{yao1986generate,goldreich2019play,mohassel2017secureml,rathee2020cryptflow2} and private information retrieval (PIR)~\cite{chor1998private,angel2018pir} address these settings with cryptographic rather than contractual guarantees. Function secret sharing (FSS)~\cite{boyle2015fss,boyle2016function} has become one of the most effective building blocks for both. In FSS-based private inference, the parties open a masked input $x+r$ and evaluate one-time keys for the gate function shifted by the mask, so that a nonlinear layer often costs a single online round~\cite{boyle2021function,gupta2023sigma,jawalkar2024orca}. In two-server PIR, the client splits a point function into two short keys and each server evaluates its key over the entire database~\cite{gilboa2014distributed,wang2017splinter,lam2023gpu}.

The price of this online efficiency is paid in distributed point function (DPF) work. Every comparison, lookup table and interval test in an FSS-compiled model is built from DPFs or from the distributed comparison functions (DCFs) derived from them~\cite{boyle2021function,storrier2023grotto,gupta2023sigma}, and each key is used once. A dealer must therefore generate, ship and store fresh keys for every nonlinear element of every query: 17.1~GB of offline material for a single BERT-Base inference~\cite{dfa2026trusted}, a cost that cannot always be hidden behind the online phase~\cite{huang2026beyond}. PIR, conversely, is dominated by full-domain DPF evaluation, which is AES-bound even on a modern GPU~\cite{lam2023gpu}. The distributed function accelerator (DFA)~\cite{dfa2026trusted} attacked both costs with a small DPF engine integrated next to the GPU L2 cache. In its trusted mode, a protected module in each server attests to its peer, establishes a shared key and plays the dealer locally: it generates its own party's key share on demand and consumes each correction word as soon as it is produced (\emph{generate-and-consume}, GaC), so keys are never distributed, stored or reread. For BERT-Base, trusted DFA reduces LAN latency from 18.8~s to 1.42~s and offline communication from 17.1~GB to 0.1~GB; for PIR, DFA in its untrusted mode (client-generated keys) raises throughput by $5.05\times$ for $2^{20}$ records~\cite{dfa2026trusted}.

\textbf{The semi-honest boundary.}
DFA, like Sigma~\cite{gupta2023sigma}, Orca~\cite{jawalkar2024orca} and most GPU-based FSS systems, is secure only against semi-honest adversaries: it assumes that both servers run the prescribed code. In a real deployment, however, the components easiest to subvert are those around the small protected module: the host operating system, the GPU driver and runtime, the GPU kernels, the direct memory access (DMA) buffers and the network stack. They hold the whole data path, and attacks on privileged software and firmware~\cite{xu2015controlled,butterworth2013bios,han2018bad} and GPU isolation designs~\cite{volos2018graviton,jang2019heterogeneous} alike treat them as potentially adversarial. With its module intact, a single compromised server host can still mount at least five classes of attacks against semi-honest trusted acceleration (\cref{fig:scenarios}):
(i)~\emph{additive errors} on shares, openings or GPU products, which bias results and leak through selective failure;
(ii)~\emph{replay or reset} of masks, counters or module state;
(iii)~\emph{check-to-use} buffer swaps;
(iv)~abuse of a naive ``authenticate'' command as a \emph{forgery oracle}; and
(v)~\emph{premature} output requests.
Remote attestation does not close these gaps. It proves which code and configuration run inside the module; it does not prove that the untrusted GPU used a valid key, computed the right product, or released an authorized result.

\textbf{High costs of existing malicious routes.}
FSS can be made maliciously secure, but the known routes give back much of what trusted acceleration gained. Actively secure FSS inference such as Shark~\cite{gupta2025shark} keeps an external dealer that ships authenticated keys and correlated randomness. For dealer-based malicious FSS modeled after the protocol family of Shark (\cref{app:model}), we measure {37.6~GB} of offline material per BERT-Base query, {2.2$\times$} the semi-honest 17.1~GB~\cite{dfa2026trusted}, and a LAN latency of {43.9~s} of which {37.6~s} is distribution. Dealer-free malicious preprocessing in the SPDZ family~\cite{damgard2012spdz,keller2016mascot,keller2018overdrive,cramer2018spdz2k} replaces the dealer by oblivious-transfer (OT) or homomorphic-encryption triple generation, and pseudorandom correlation generators (PCGs) with silent OT~\cite{boyle2019pcg,boyle2020pcg} compress the communication of authenticated triples at a substantial computational cost. Neither yields FSS keys: generating DPF keys without a dealer requires a secure two-party evaluation of key generation with a round per tree level~\cite{doerner2017scaling}, so dealer-free malicious FSS keeps a preprocessing cost far above its online cost. Two recent dealerless designs attack this cost from the hardware side and from the protocol side: KORD~\cite{peng2026kord} co-designs a dealerless key-generation protocol with paired, mutually attested chips that reconstruct FSS keys without a dealer, and VeriFSS~\cite{xue2026verifss} obtains active security for dealer-free two-party FSS by binding the generation transcript, the computed function and the evaluation shares to a single execution; \sys instead keeps a trusted local dealer and extends integrity from key generation to the untrusted GPU and the output release. Verifiable FSS (vFSS)~\cite{decastro2022vfss} lets servers check that client-generated DPF keys are well formed, but not that an untrusted GPU evaluates them, or the arithmetic around the gate, correctly. Finally, the direct way to authenticate an FSS output inside an accelerator, one additional DPF tree with payload $\Delta_j\beta$ per message authentication code (MAC) lane, multiplies pseudorandom generator (PRG) work by $1+t$ for $t$ lanes. One tree with a $(1+t)$-element payload avoids most of this cost, but it widens every output correction and exports tag shares whose hiding rests on DPF key privacy.

\textbf{Our key insight} is that a module that already plays the FSS dealer can also be the verifier. Because it knows every correlation it generates, it can (1)~generate \emph{authenticated} correlations (input masks, Beaver triples~\cite{beaver1991efficient} and DPF keys) locally at zero network cost, computing every triple product correctly inside its own boundary; (2)~authenticate its own FSS outputs with a fresh resharing epilogue instead of extra DPF trees, because it already holds the raw output shares and the MAC key material; and (3)~check every opening on bytes it has frozen before any challenge exists, so that the GPU remains a fast but untrusted linear-algebra engine whose results are only candidates until verified. What must be protected is thus not a single primitive but an integrity chain:
\begin{center}
\small
\emph{authorized input} $\rightarrow$ \emph{correlated preprocessing} $\rightarrow$ \emph{correct local evaluation} $\rightarrow$ \emph{authenticated arithmetic} $\rightarrow$ \emph{consistent opening} $\rightarrow$ \emph{authorized output release.}
\end{center}
Each link has a concrete failure mode under a malicious host, and each needs its own mechanism. \Cref{fig:flows} places the resulting execution flow next to the semi-honest GPU baseline, trusted DFA and dealer-based malicious FSS.

\textbf{This work.}
We present \sys (Verified Integrity-Gated On-demand Randomness DFA), the malicious-secure successor of trusted DFA. It keeps DFA's placement beside the GPU L2 cache, its local key generation and its generate-and-consume execution. Under \ProfT (\cref{sec:threat}) a malicious adversary statically corrupts at most one host, and every secret value on a host is an authenticated share over $\Fp$, $p=2^{61}-1$, with $t=3$ SPDZ-style MAC lanes~\cite{damgard2012spdz,damgard2013practical} (32~bytes per element and party). Three mechanisms carry the design; we claim their composition, not any of them as an unprecedented primitive. The \emph{fresh authentication epilogue} turns each raw DPF output share into a freshly padded authenticated tuple at three field multiplications and four pads (\cref{sec:proto-epilogue}). It matches the PRG cost of a single vector-payload tree, makes every exported tuple exactly uniform, and leaves no command that authenticates a host-supplied value. \emph{Freeze-before-challenge} \CheckOpen{} derives independent per-lane challenges only after both modules have frozen the exact submitted bytes, and later consumes exactly those bytes. A \emph{role-bound one-time ledger} labels all randomness by session, epoch, gate and purpose and poisons the epoch on any failure, and its \emph{release guard} lets no output leave a module before the corresponding check is \lstate{VERIFIED}.

\textbf{Scope and limits.}
\sys relocates trust to the modules rather than removing it. The modules see inputs under \ProfT, and both learn the PIR index, so PIR query privacy in \sys rests on module trust, unlike untrusted two-server DPF-PIR~\cite{lam2023gpu} and DFA's untrusted mode, which hide the index without trusted hardware. Physical attacks, leakage beyond the stipulated observable-trace model and rollback without a trusted monotonic service are out of scope, and security holds with abort, without fairness (\cref{sec:discussion}). Once the epilogue removes tag trees, the dominant cost of malicious security moves to four-lane authenticated linear algebra on the GPU and to the doubled round count (\cref{sec:analysis}).

\textbf{FSS within trusted modules.}
Under \ProfT a module that has verified $c=x+r$ also holds $r$, so FSS adds no formal confidentiality over evaluating the gate directly (\cref{sec:proto-direct}). The end-to-end gains of \sys come from generating and consuming authenticated correlations inside the modules, not from FSS evaluation: in our model direct-only compilation ({3.82~s} for BERT-Base) is marginally faster than the compiler-selected default ({3.83~s}), which keeps an FSS realization for gate shapes marked FSS-portable, and nonlinear gates take only {3.4\%} of the latency. We keep a malicious FSS datapath for three reasons: the same compiled circuit, including masked-input FSS gates from existing toolchains~\cite{gupta2023sigma,jawalkar2024orca}, runs under the key-exporting \ProfU (\cref{sec:eval}); PIR and small lookup tables are full-domain, data-independent scans that match the pipelined AES engines; and long-lived evaluation state holds only seeds and correction words, keeping input-dependent state short-lived as defense in depth, not as a formal claim. On FSS gates the epilogue is {3.8--4.0$\times$} faster than per-lane tag trees, although tag trees would raise the BERT-Base LAN latency by only {8\%}.

This paper makes the following contributions:
\begin{itemize}[leftmargin=*]
    \item \textbf{Malicious Threat Analysis.}
    We systematize how a malicious host or GPU attacks semi-honest trusted FSS acceleration into ten scenarios, S1 to S10 (\cref{fig:scenarios,tab:scenarios}), from share tampering to PIR snapshot substitution. For S1 to S8 we show why semi-honest trusted DFA fails and which \sys mechanism and theorem stop it; for S9 (input substitution) and S10 (denial of service and selective abort) we make explicit what the ideal functionality permits or does not prevent.

    \item \textbf{Protocol Co-Design.}
    We specify the three mechanisms above, the authenticated $\Fp$ representation and malicious generate-and-consume with explicit generator and evaluator states (\cref{sec:protocol}). On top of them we build authenticated lookup-table, interval, polynomial and direct protected gates, exact comparison and truncation, a refactoring rule that preserves the joint (not only the marginal) distribution of DFA's decoupled preprocessing, and versioned-snapshot PIR whose query privacy rests on module trust.

    \item \textbf{Hardware Architecture.}
    We design a protected datapath beside the GPU L2 cache (\cref{sec:arch}) with a ledger, page freezer, joint pseudorandom-function (PRF) sampler, DPF generation and evaluation contexts, a 128-lane Mersenne-61 field unit array, a check engine, a peer-channel controller and a release guard, whose 64 AES engines are time-shared among PRG expansion, challenge derivation and authenticated-encryption (AEAD) freezing. The additions cost {0.145~mm$^2$} and {0.508~W} at 7~nm, and the entire \sys occupies {0.046\%} of the A100 die.

    \item \textbf{Security Analysis.}
    We give ideal functionalities with explicit leakage and a hybrid proof with an explicit simulator of stand-alone static malicious security with abort under sequential composition (\cref{sec:security}, \cref{thm:main}), with advantage at most $\eps{setup}+\eps{prf}+\eps{channel}+\eps{bind}+\eps{trace}+Q(2/p)^t$, where the last term is about $2^{-148}$ for $Q\le 2^{32}$ checked batches. We show why shared challenges across lanes and the ring $\Z_{2^{64}}$ break this bound, and we separate what is proved from the implementation obligations that remain (\cref{tab:obligations}).

    \item \textbf{Evaluation.}
    With a model calibrated on DFA (\cref{sec:eval}, \cref{tab:e2e}), we obtain that in-module generation of authenticated correlations reduces offline communication from {37.6~GB} to {under 1~MB} per BERT-Base query. Against dealer-based malicious FSS modeled after the protocol family of Shark, \sys is {11.4$\times$} faster on a LAN, or {1.6$\times$} excluding per-query offline distribution; with no authenticated-key blow-up the dealer-based pipeline still needs {22.4~s}, {5.8$\times$} that of \sys. Malicious security costs {2.7$\times$} LAN latency over semi-honest trusted DFA for BERT-Base, mostly for four-lane linear algebra, yet \sys stays {4.3--6.1$\times$} faster than the semi-honest GPU-only pipeline across four models. We executed a specification and a functional CPU reference model, which covers every algorithm except the optional Freivalds check and passes all 22 unit tests (\cref{sec:eval-attack}); GPU/RTL conformance, the protected runtime and deployment measurements have been executed, and all six injections of \cref{tab:attacks}, including the release guard, have executed evidence. The artifact is submitted alongside the paper and will be released upon acceptance.
\end{itemize}

\Cref{sec:bg,sec:threat,sec:analysis} give background, the threat model and where the cost of malicious FSS goes; \cref{sec:protocol,sec:arch,sec:security,sec:eval} present the protocol, architecture, security analysis and evaluation; \cref{sec:discussion,sec:related,sec:conclusion} discuss limitations, related work and conclusions, and the appendices give the ideal functionalities, proofs and performance model.

\section{Background}
\label{sec:bg}

This section reviews the primitives that \sys builds on and the semi-honest DFA design~\cite{dfa2026trusted} that it extends.
We write $P_0,P_1$ for the two computing parties, $b\in\{0,1\}$ for a party index, $\lambda$ for the seed length of the pseudorandom generator (PRG), and $\one[\cdot]$ for the indicator of a predicate.

\subsection{Function Secret Sharing, DPF, and DCF}
\label{sec:bg-fss}

\textbf{Function secret sharing.}
A two-party function secret sharing (FSS) scheme~\cite{boyle2015fss,boyle2016function} for a function class $\mathcal{F}$ with output group $\mathbb{G}$ consists of $\Gen(1^\lambda,\hat f)$, which maps a description of $f\in\mathcal{F}$ to a key pair $(k_0,k_1)$, and $\Eval(k_b,x)$, which returns party $b$'s share $y_b\in\mathbb{G}$.
Correctness requires $\Eval(k_0,x)+\Eval(k_1,x)=f(x)$ for every $x$; security requires that each key alone be simulatable from public information (domain, output group, function class), so that a single key hides $f$.

\textbf{Distributed point function.}
A distributed point function (DPF)~\cite{gilboa2014distributed,boyle2015fss,boyle2016function} shares the point function $f_{\alpha,\beta}$ with target $\alpha\in\bits{n}$ and payload $\beta\in\mathbb{G}$:
\begin{equation}\label{eq:dpf}
\Eval(k_0,x)+\Eval(k_1,x)=\beta\cdot\one[x=\alpha]\quad\text{for all } x\in\bits{n}.
\end{equation}
In this paper $\mathbb{G}=\Fp$ with $p=2^{61}-1$ (\cref{sec:proto-repr}), inputs are read most-significant bit first, and we use the tree-based construction of Boyle et al.~\cite{boyle2016function}.
Each node of a depth-$n$ binary tree carries a state $(s,t)$ of a $\lambda$-bit seed and a control bit; a fixed PRG $\PRG:\bits{\lambda}\rightarrow(\bits{\lambda}\times\bits{})^2$ expands a seed into the states of its two children, and $\Conv:\bits{\lambda}\rightarrow\Fp$ converts a leaf seed into a field element.
Generation walks the target path from two independent root seeds with control bits $t_0=0$ and $t_1=1$; each level's correction word $\CW^{(i)}=(\sCW,\tCW{L},\tCW{R})$ makes both parties' states identical on the child leaving the path and keeps them pseudorandom with opposite control bits on the child staying on it, and an output correction $\CWout$ makes the signed outputs $y_b=(-1)^b\bigl(\Conv(s)+t\cdot\CWout\bigr)$ sum to $\beta$ at $x=\alpha$ and cancel elsewhere.
Party $b$'s key holds its root seed, its fixed control bit $t_b^{(0)}=b$, the $n$ correction words and $\CWout$, i.e., $O(\lambda n)$ bits; \cref{alg:gen,alg:eval} in \cref{sec:analysis-dpf} give the exact algorithms.

\textbf{Evaluation patterns and cost.}
Generation expands $2n$ states and a single-path $\Eval(k_b,x)$ expands $n$; a full-domain $\EvalAll(k_b)$ expands $2^n-1$ internal nodes, keeping only $O(n)$ pending siblings in a depth-first traversal.
The PRG dominates all patterns; DFA, GPU-DPF~\cite{lam2023gpu} and \sys instantiate it with AES-128.
Two properties matter later.
Keys are \emph{one-time}: each encodes one mask or query index, and reusing a mask reveals the difference of the two inputs it hides.
DPF key privacy holds for \emph{arbitrary} $(\alpha,\beta)$~\cite{boyle2016function}, so uniform targets come from the application's masking (\cref{sec:bg-pi}), whose joint distribution any refactoring of generation must preserve (\cref{sec:proto-refactor}).

\textbf{Distributed comparison function.}
A distributed comparison function (DCF) shares $f^{<}_{\alpha,\beta}(x)=\beta\cdot\one[x<\alpha]$.
We use this ``$<$'' orientation throughout; DFA~\cite{dfa2026trusted} uses ``$>$'', and we convert explicitly at interfaces via $\beta\cdot\one[x>\alpha]=\beta-f^{<}_{\alpha+1,\beta}(x)$, treating the threshold $2^n$ as a public constant.
Optimized DCFs evaluate with $O(n)$ PRG expansions~\cite{boyle2021function}, and DCFs can be assembled from DPFs~\cite{storrier2023grotto}.
\sys uses the conservative $O(n^2)$ DCF of \cref{alg:dcf}, which sums $n$ prefix DPFs and whose correctness \cref{thm:dcf} proves (\cref{app:proof-dpf}).

\subsection{FSS-Based Private Inference}
\label{sec:bg-pi}

\textbf{Two-party secure inference.}
Private inference runs a neural network with private input, and possibly private weights, as a secure two-party computation~\cite{yao1986generate,goldreich2019play,mohassel2017secureml,rathee2020cryptflow2,crypten2020,tan2021cryptgpu,watson2022piranha}.
With additive shares $x=x_0+x_1$, linear operations are local, and a multiplication consumes a Beaver triple $(a,b,c=ab)$~\cite{beaver1991efficient}: the parties open $d=x-a$ and $e=y-b$ and locally compute shares of $xy=c+db+ea+de$.
Nonlinear layers (ReLU, GeLU, Softmax, LayerNorm) dominate the cost; earlier systems use multi-round Boolean conversions for them~\cite{demmler2015aby,mohassel2018aby3,wagh2019securenn,wagh2020falcon}.

\textbf{Masked-input FSS gates.}
FSS-based systems~\cite{boyle2019preprocessing,boyle2021function,ryffel2020ariann,wagh2022pika,gupta2023sigma,jawalkar2024orca} evaluate each nonlinear gate $g$ in the masked-input paradigm.
Offline, a dealer samples a fresh uniform mask $r$, shares it, and generates FSS keys for the offset function $g_r(c)=g(c-r)$, often with an output mask.
Online, the parties open $c=x+r$, which is uniform and independent of $x$, and the outputs $\Eval(k_b,c)$ are shares of $g(x)$.
Comparisons and ReLU use DCFs; splines for GeLU, exponentials and inverse square roots use DCF-based interval selection plus low-degree polynomials; and small lookup tables (LUTs) use DPFs~\cite{gupta2023sigma,jawalkar2024orca}.
Each gate costs one opening round, so the online phase is round-efficient but compute-intensive.

\textbf{Offline cost and the dealer assumption.}
The price is paid offline.
Semi-honest FSS inference of BERT-Base with 128 tokens moves 17.1~GB of one-time key material per query and LLaMA2-7B moves 256~GB, more than 90\% of it DPF keys~\cite{dfa2026trusted}; this cost cannot always be hidden behind online work, and its energy and monetary cost is never hidden~\cite{huang2026beyond}.
Sigma, Orca, and related systems~\cite{gupta2023sigma,jawalkar2024orca,crypten2020} obtain keys from a trusted dealer and assume semi-honest computing parties.

\subsection{DPF-Based Private Information Retrieval}
\label{sec:bg-pir}

PIR lets a client retrieve record $D[\alpha]$ of an $N$-record database $D$ without revealing $\alpha$~\cite{chor1998private}.
Single-server schemes rely on homomorphic encryption and are computationally heavy~\cite{angel2018pir,henzinger2023one}; we focus on the faster two-server DPF-based schemes~\cite{gilboa2014distributed,boyle2016function}.
With $n=\lceil\log_2 N\rceil$, the client sends $k_b$ of $(k_0,k_1)\leftarrow\Gen(1^\lambda,\alpha,1)$ to server $P_b$, which returns
\begin{equation}\label{eq:pir}
y_b=\sum_{x=0}^{N-1}\Eval(k_b,x)\cdot D[x],\qquad b\in\{0,1\},
\end{equation}
and by \cref{eq:dpf} the client obtains $D[\alpha]=y_0+y_1$.
Keys have $O(\lambda\log N)$ bits, but each query costs $O(N)$ PRG expansions plus an $N$-term inner product per server, and the expansion dominates even on GPUs~\cite{lam2023gpu}.
Standard DPF-based PIR protects the index but not the answer: a malicious server can add an offset $\delta$ to $y_b$, so that the client accepts $D[\alpha]+\delta$, or answer from a stale database.
Conversely, a malicious client can submit a malformed key pair that reconstructs a combination of several records, which verifiable DPFs detect~\cite{decastro2022vfss}.
Authenticated PIR (APIR)~\cite{colombo2023apir} lets the client verify answers against a database digest at additional server cost.
\Cref{sec:proto-pir} gives our protected versioned-snapshot design, in which both modules learn $\alpha$: unlike untrusted two-server DPF-PIR, its query privacy rests on module trust.

\subsection{Trusted DFA}
\label{sec:bg-dfa}

\begin{figure*}[t]
    \centering
    \includegraphics[width=\textwidth]{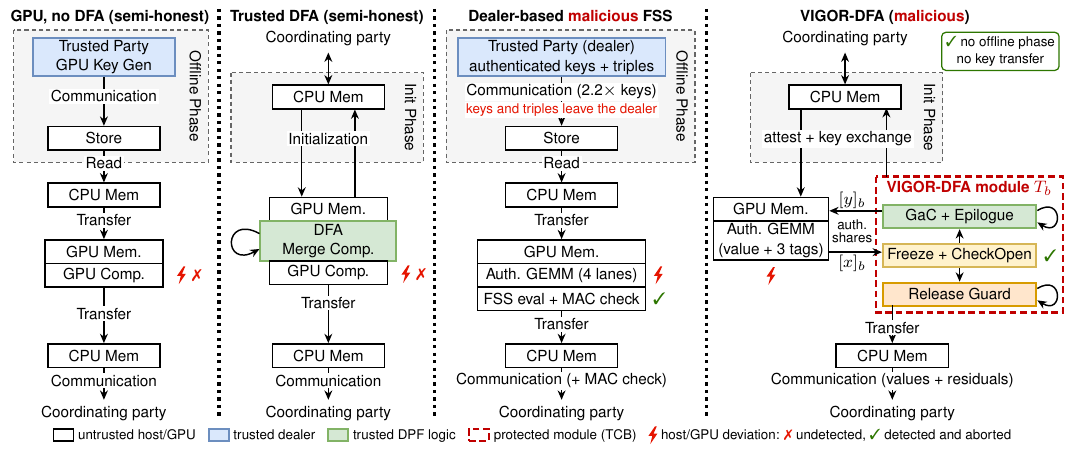}
    \caption{Execution flows of two-party FSS-based private inference in four configurations, from left to right.
    (1)~GPU without DFA (semi-honest): a trusted party generates FSS keys offline and ships them over the network; each party stores them, rereads them and transfers them to its GPU.
    (2)~Trusted DFA (semi-honest)~\cite{dfa2026trusted}: after a one-time initialization with attestation and key exchange, each on-die module generates its party's keys locally and merges generation with evaluation (generate-and-consume), removing the offline phase.
    (3)~Dealer-based malicious FSS, modeled after the protocol family of Shark~\cite{gupta2025shark}: the dealer also issues authenticated keys and triples, so about {2.2$\times$} as much offline material (\cref{app:model}) must be shipped, stored and reread, while the GPU runs authenticated general matrix multiplication (GEMM) and checks MACs online.
    (4)~\sys (malicious): the protected module (red dashed trusted boundary) keeps local generation and generate-and-consume, adds the fresh authentication epilogue, freeze-before-challenge \CheckOpen{} and the release guard, and exchanges authenticated shares $\authb{x}{b}$ with a GPU that runs authenticated linear algebra on one value lane and three MAC-tag lanes; the inter-party network carries opened values and check residuals, not keys.}
    \label{fig:flows}
\end{figure*}

\textbf{Hardware.}
DFA~\cite{dfa2026trusted} integrates 16 DPF units with four pipelined AES-128 engines each (64 in total) next to the L2 slices of a GPU; each 20-stage unit sustains one Gen or Eval iteration per cycle after warm-up.
An SC300-class control core~\cite{st33tpm12spi,arm_cortex_m3} handles attestation and key management, and a 16-lane vector unit~\cite{cavalcante2019ara} runs protocol-specific pre- and post-processing.
The design runs at 720~MHz in 28~nm, delivers about 290~GB/s of DPF output in Accel-Sim~\cite{khairy2020accel}, and traverses full-domain trees depth-first, needing $O(nm)$ state for $m$ units instead of $O(2^n)$~\cite{dfa2026trusted}.

\textbf{Untrusted and trusted modes.}
In the \emph{untrusted mode}, DFA only expands seeds for keys produced elsewhere, and the protocol is unchanged.
In the \emph{trusted mode}, each party's module replaces the dealer: the modules attest to each other and establish a key $K$ that the hosts never see, derive the same generation coins (including $\alpha$) from $K$ and a counter, and thus compute identical correction words while each evaluates only its own party's share.
\emph{Generate-and-consume} (GaC) runs the joint generation and the local evaluation in lockstep, consuming each correction word as soon as it is produced, so that neither key distribution nor key storage remains.
DFA also \emph{decouples} generation from protocol preprocessing: when the target is $a=g(r)$ for a bijection $g$, the module samples $a$ uniformly and recovers $r=g^{-1}(a)$.

\textbf{Limits of the semi-honest design.}
DFA argues that correct in-module generation leaves the semi-honest protocol unchanged; against a deviating host this argument fails in four places treated later: (i)~decoupled preprocessing must preserve joint, not only marginal, distributions (\cref{sec:proto-refactor}); (ii)~generator and evaluator state is left implicit (\cref{sec:analysis-dpf}); (iii)~attestation authenticates the module configuration, not that the GPU used a valid key, computed the right product or released an authorized result (\cref{sec:threat-scenarios}); and (iv)~the GPU lies outside the trusted computing base (TCB) and must never obtain module tags for its own results (\cref{sec:proto-epilogue}).

\subsection{Tools for Malicious Security}
\label{sec:bg-mal}

\textbf{Security with abort.}
We use the stand-alone real/ideal simulation paradigm, which is preserved under sequential composition~\cite{canetti2000security,goldreich2004foundations,hazay2010efficient,lindell2017simulate}; we make no concurrent-composition claim.
A malicious adversary deviates arbitrarily; security with abort requires that it achieve nothing beyond an ideal world in which a trusted functionality computes the result and the adversary only chooses corrupted inputs and whether honest parties receive output, so fairness is not required.
Covert security~\cite{aumann2007covert} only deters cheating; we target full malicious security with abort.

\textbf{Information-theoretic MACs.}
SPDZ-style protocols~\cite{damgard2012spdz,damgard2013practical} authenticate each value of a field $\F$ under a global key $\Delta=\Delta_0+\Delta_1$ that no single party knows.
Party $b$ holds $\Delta_b$ and $\authb{x}{b}=(x_b,\gamma_b)$ with $x_0+x_1=x$ and $\gamma_0+\gamma_1=\Delta x$; linear operations act identically on values and tags, and multiplications consume authenticated triples.
To check openings $z_1,\dots,z_m$, the parties draw coefficients $\rho_i$ \emph{after} the values are fixed, compute $\sigma_b=\sum_i\rho_i(\gamma_{b,i}-\Delta_b z_i)$, and accept only if $\sigma_0+\sigma_1=0$.
A wrong opening passes only if its tag error matches $\Delta$ times the value error or cancels under the coefficients, which happens with probability less than $2/|\F|$.
\sys instantiates this with $t=3$ independent lanes over $\Fp$ (\cref{sec:proto-repr}).

\textbf{Authenticated preprocessing.}
Malicious protocols need authenticated masks and triples whose relations, such as $c=ab$, are actually correct; a valid MAC on a wrong triple does not make it correct.
Generating them without a trusted party is communication-heavy, with oblivious transfer in MASCOT~\cite{keller2016mascot} or somewhat homomorphic encryption in SPDZ and Overdrive~\cite{damgard2012spdz,keller2018overdrive}, all in MP-SPDZ~\cite{keller2020mpspdz}.
Pseudorandom correlation generators~\cite{boyle2019pcg,boyle2020pcg} compress this communication at a substantial computational cost, and a trusted dealer avoids generation but ships every authenticated share.

\textbf{Malicious and verifiable FSS.}
Malicious FSS gates additionally need well-formed keys and authenticated evaluation outputs that can enter MAC-checked arithmetic.
Verifiable FSS~\cite{decastro2022vfss} checks that a key pair shares a function of the declared class, not that a malicious party evaluates its own key correctly.
Shark~\cite{gupta2025shark} achieves active security for FSS-based inference with a dealer that issues authenticated FSS keys and correlated randomness, which enlarges the offline material relative to semi-honest FSS.
The direct way to authenticate a DPF output, a second \emph{tag tree} with payload $\Delta_j\beta$ per MAC key, multiplies DPF work by $1+t$ for $t$ lanes (\cref{sec:analysis-work}).

\textbf{Trusted hardware for secure computation.}
Slalom~\cite{tramer2018slalom} and DarKnight~\cite{hashemi2021darknight} outsource linear layers from a trusted execution environment (TEE) to untrusted GPUs and check the results, Slalom with Freivalds' algorithm~\cite{freivalds1977probabilistic}, which \sys uses only for a restricted subroutine (\cref{sec:proto-freivalds}).
PPMLAC~\cite{zhou2022ppmlac} and STAMP~\cite{huang2022efficient} use small trusted modules to generate correlated randomness or to execute nonlinear operations; STAMP also reports honest-majority malicious settings, and \cref{sec:related} compares these and other trusted-hardware designs.

\textbf{Implications.}
Each tool covers only part of the integrity chain of \cref{sec:intro}, at a cost that erodes trusted DFA's advantage (\cref{sec:analysis}); \sys therefore keeps the module as the dealer and also makes it the verifier (\cref{sec:protocol}).

\textbf{Scope.}
Throughout this paper, \emph{private inference} means two-party FSS-based private inference and \emph{PIR} means two-server DPF-based PIR, as in DFA~\cite{dfa2026trusted}.

\section{Threat Model and Security Goals}\label{sec:threat}

The trusted mode of DFA moved the FSS dealer into a small protected module but kept the semi-honest model of the protocols it accelerates~\cite{dfa2026trusted}. This section states what changes when a host is malicious: we fix the system model (\cref{sec:threat-model}) and the adversary, derive ten attack scenarios (\cref{sec:threat-scenarios}), define security with explicit leakage (\cref{sec:threat-goals}), and delimit \ProfU (\cref{sec:threat-profileu}). Every link of the integrity chain of \cref{sec:intro} crosses memory, buses, or network paths that a corrupted host controls.

\begin{figure*}[t]
    \centering
    \includegraphics[width=\textwidth]{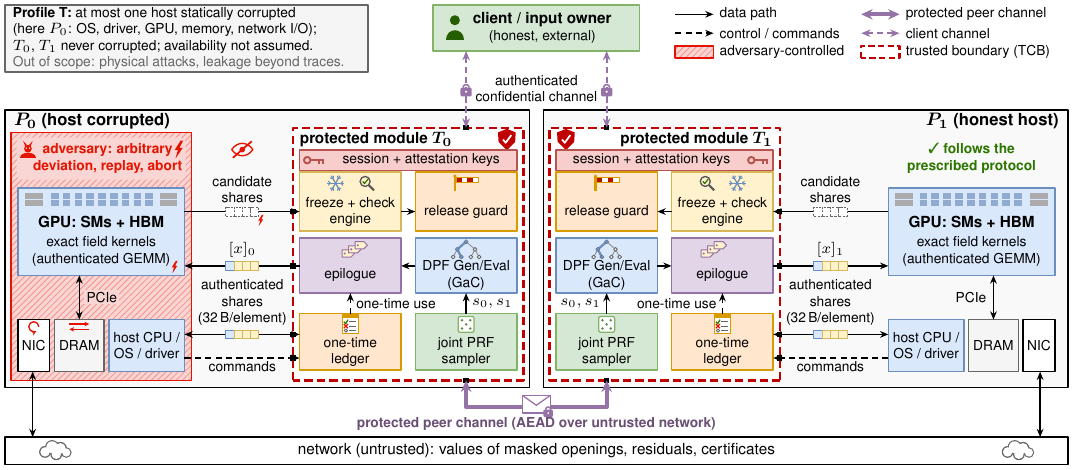}
    \caption{System model and trust boundary of \sys under \ProfT. Each party $P_b$ has an untrusted host (CPU, OS, driver, runtime, network interface card (NIC), DRAM) with its GPU, including the streaming multiprocessors (SMs) and high-bandwidth memory (HBM), and a protected module $T_b$ (red dashed TCB, contents in \cref{sec:threat-model}). The adversary statically corrupts at most one host (here that of $P_0$, red region) and may deviate arbitrarily, replay, and abort, but never corrupts $T_0$ or $T_1$. Hosts and GPUs hold only party-local authenticated shares; the protected peer channel between the modules and the authenticated confidential client channels are carried by the untrusted network.}
    \label{fig:threat}
\end{figure*}

\subsection{System Model and Trust Boundary}\label{sec:threat-model}

\textbf{Parties.}
\sys targets the two workloads of DFA~\cite{dfa2026trusted}: two-party FSS-based private inference and two-server DPF-based PIR (\cref{sec:bg}). Two computing parties $P_0$ and $P_1$ serve an external client that owns private inputs and receives outputs; a computing party may also own an input, typically the model weights. Each $P_b$ consists of an untrusted \emph{host~$b$} (every component outside the TCB in \cref{fig:threat}, including DMA engines and storage) and a \emph{protected module} $T_b$, the \sys extension placed next to the GPU L2 as in DFA (\cref{sec:arch-hw}).

\textbf{Trusted computing base.}
$T_b$ contains the session, identity (attestation), and joint-sampler keys and the global MAC keys $\Delta_1,\dots,\Delta_t$; the session role, counters, and one-time resource ledger; the DPF generator and evaluator states and the correction words; the protected table and reduction logic; the frozen check buffers, challenges, and residuals; and all release decisions. Everything else is outside, including every GPU and DMA buffer and every byte in transit between the parties. The GPU computes on party-local authenticated arrays and hands $T_b$ only \emph{candidate} shares (\cref{sec:proto-repr}). Host~$b$ may know its own MAC key shares $\Delta_{b,j}$, but no host ever holds a complete $\Delta_j$, the joint-sampler key, a complete mask (other than the input mask of a wire its party owns or the output mask of an output its party is authorized to receive), the peer's shares, or any challenge or residual.

\textbf{Modules as an ideal resource.}
Following~\cite{katz2007universally}, we model $T_0,T_1$ as an ideal resource $\mathcal{F}_{\mathsf{T}}$ (\cref{app:func-module}; ideal-module assumption in \cref{sec:sec-assump}): their code, keys, parser, arithmetic, state, and approved I/O paths are correct and confidential, and each module's externally observable trace (lengths, timing, host-visible memory and DMA activity, power, thermal, and clock telemetry) is simulatable from the command sequence issued by the host, $\mu$, public lengths, and the accept or abort outcomes, and is independent of the global MAC keys, challenges, and residuals given these. The real execution (\cref{def:security}) realizes their randomness, channels, and setup with concrete primitives. Availability is not assumed, \ProfT requires two physically independent modules, and only the on-die design (\cref{sec:arch-hw}) is a \ProfT target.

\textbf{Channels.}
Client channels are authenticated, confidential, and terminate inside the modules; the peer channel uses AEAD with independent directional keys and no nonce reuse and is established by the authenticated setup of \cref{sec:proto-setup}. The untrusted network carries both. We propose no new attestation or key-exchange protocol; the security of setup is the explicit assumption $\eps{setup}$.

\textbf{Trust is relocated, not removed.}
The modules jointly act as the dealer and can recover inputs from verified openings, and in PIR both modules learn the queried index $\alpha$. \sys thus moves trust from an external dealer and the hosts into the modules rather than hiding inputs from them (\cref{sec:discussion}); its PIR query privacy rests on module trust, unlike untrusted two-server DPF-PIR such as GPU-DPF~\cite{lam2023gpu} or DFA's untrusted mode~\cite{dfa2026trusted}, which hides $\alpha$ from each server without trusted hardware.

\subsection{Adversary Model}

Unlike DFA~\cite{dfa2026trusted}, Sigma~\cite{gupta2023sigma}, and Orca~\cite{jawalkar2024orca}, which assume that both hosts run the prescribed kernels, we consider a probabilistic polynomial-time (PPT) adversary $\Adv$ that statically and maliciously corrupts at most one host before the session starts (host~$0$ in \cref{fig:threat}, by symmetry), controlling its operating system, drivers, GPU kernels and results, host and GPU memory including DMA buffers, and outgoing network traffic. Concretely, $\Adv$ can
(i) modify any value share or tag and run arbitrary kernels in place of the prescribed ones;
(ii) invoke its module's commands in any order with malformed, replayed, or mislabeled arguments;
(iii) rewrite any host-visible buffer at any time, including while the module reads it;
(iv) replay, reorder, delay, or drop messages, reset its platform, and restore old host-side state;
(v) choose the inputs it owns arbitrarily; and
(vi) abort at any point, including after it has learned its own output.
It observes message lengths and timing, the frames exchanged with its module, and the module's external trace. The principal inference statement assumes an honest external client and gives no guarantee against client-server collusion; PIR has two separate regimes, (R1) one malicious server host and (R2) a malicious client only with two honest server hosts (\cref{fig:func}).

\textbf{Out of scope.}
Our theorems exclude corruption of either module (its code, keys, state, or arithmetic); physical attacks such as microprobing, optical fault injection, and power or electromagnetic analysis with measurement equipment~\cite{skorobogatov2017microprobing,skorobogatov2003optical,kocher1999differential,quadir2016survey}; leakage beyond the stipulated trace model; speculative-execution bugs inside the trusted boundary~\cite{van2018foreshadow}; faulty trusted arithmetic; adaptive corruption; and corruption of both hosts. Software-reachable fault and power channels are not physical attacks: the corrupted host controls the GPU driver, which sets clocks, voltage, and power limits and reads telemetry, an interface that Plundervolt exploited on CPUs~\cite{murdock2020plundervolt}. We therefore require isolation from dynamic voltage and frequency scaling (DVFS): host software cannot drive the module's clock or voltage outside its safe operating range (or the module detects glitches and aborts), and the telemetry is covered by the trace simulator. Rollback of module state is covered only by the reboot rule of \cref{sec:proto-ledger}. Availability and fairness are not provided.

\begin{figure*}[t]
    \centering
    \includegraphics[width=\textwidth]{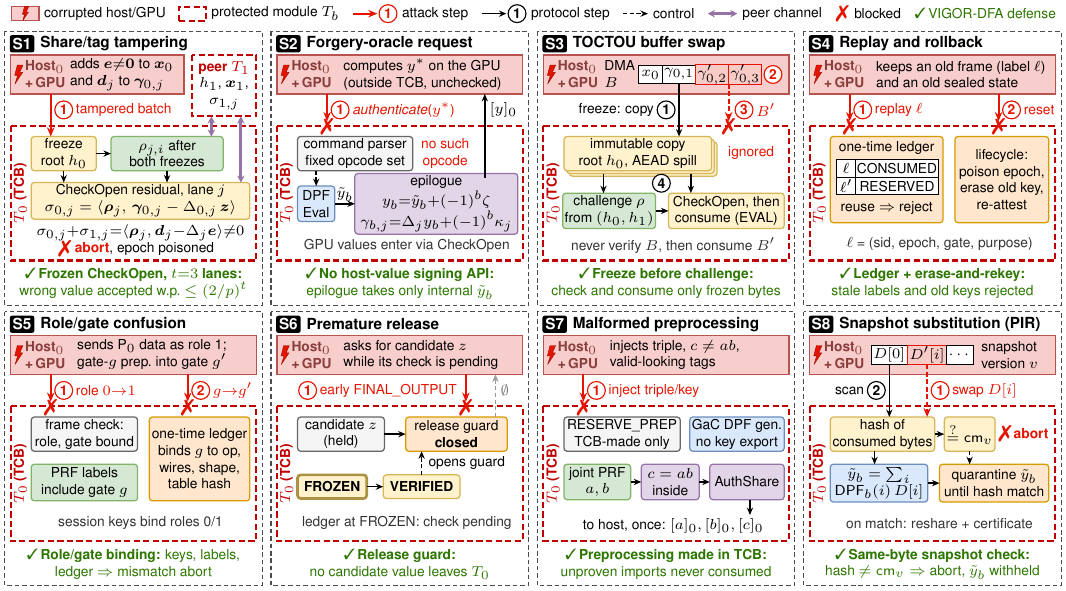}
    \caption{Malicious scenarios S1 to S8 against trusted FSS acceleration and the \sys mechanism that stops each. Corrupted hosts and GPUs are shaded red, $T_b$ is the red dashed TCB, red arrows with circled numbers are adversary actions, red crosses mark where \sys blocks them, and green text names the defense. S9 and S10, which no protocol can prevent, appear only in \cref{tab:scenarios}.}
    \label{fig:scenarios}
\end{figure*}

\subsection{Malicious Scenarios Against Trusted FSS Acceleration}\label{sec:threat-scenarios}

\Cref{fig:scenarios} and \cref{tab:scenarios} list what a corrupted host can do to a trusted FSS accelerator designed for semi-honest hosts and which \sys mechanism stops each action; the paragraphs below explain why the semi-honest design fails. The scenarios follow the integrity chain: S9 targets the authorized input; S4 and S7 the correlated preprocessing; S2 and S5 the correct local evaluation; S1 the authenticated arithmetic; S3 the consistent opening; and S6 the authorized output release. S8 carries S3 over to the PIR database, and S10 targets availability, which lies outside the chain. \Cref{sec:sec-scenarios,sec:eval-attack} map the scenarios to theorems and to injected tests.

\textbf{S1: Share and tag tampering.}
The corrupted GPU holds every intermediate share and can add an error vector $\bm{e}$ to value shares (and $\bm{d}_j$ to the tags of lane $j$); semi-honest trusted DFA passes it unchecked into the masked opening $x_i+r_i$ of the next FSS gate. Beyond biasing results, this causes \emph{selective failure}: whether $e_i$ flips a comparison depends on whether $x_i$ lies within $|e_i|$ of the threshold, so any observable output change leaks information about $x_i$ (\cref{tab:scenarios}, \cref{sec:proto-check}).

\textbf{S2: Forgery-oracle request.}
A tempting extension of DFA's trusted mode is a command that tags a value computed elsewhere, such as a nonlinear output computed by the GPU's streaming multiprocessors (SMs). This \emph{authenticate(host value)} entry point turns the module into a MAC oracle: the host submits any $y^*$, and every later check accepts it. Attestation certifies the module's code, not the provenance of what it signs (\cref{tab:scenarios}, \cref{sec:proto-epilogue}).

\textbf{S3: TOCTOU buffer swap.}
The module reads its inputs by DMA from host-controlled memory. Verifying a buffer and later rereading it, or trusting a host-reported digest, admits a time-of-check-to-time-of-use (TOCTOU) swap: the host changes the bytes after verification or chooses its error after the challenge is fixed. DMA completion is not evidence of authenticity, a problem semi-honest DFA never faced (\cref{tab:scenarios}, \cref{sec:proto-check}).

\textbf{S4: Replay and rollback.}
DFA derives keys and masks from the shared key $K$ and a counter that is an input of its generate-and-consume procedure~\cite{dfa2026trusted}. Replaying a counter yields the same mask $r$ for two gates, and the openings $x+r$ and $x'+r$ reveal $x-x'$; resetting the module and resuming an old session replays entire key streams (\cref{tab:scenarios}, \cref{sec:proto-ledger}).

\textbf{S5: Role and gate confusion.}
The host decides which module command runs for which gate: it can route the preprocessing of gate $g$ into gate $g'$, alter a table hash or tensor shape, or present party-$0$ material as party-$1$ material. If the role were a per-request argument, the host could obtain both parties' outputs of one joint draw and reconstruct the DPF target $\alpha$, which is exactly the mask of an honest input (\cref{tab:scenarios}, \cref{sec:proto-setup,sec:proto-ledger}).

\textbf{S6: Premature release.}
Semi-honest trusted DFA releases results as soon as they are computed; under a malicious host an early output may already be wrong, and the host may request it repeatedly or for an unauthorized recipient (\cref{tab:scenarios}, \cref{sec:proto-ledger}).

\textbf{S7: Malformed preprocessing.}
A MAC authenticates values, not relations: a triple $(\auth{a},\auth{b},\auth{c})$ with $c\neq ab$ carries valid tags, and multiplying with it yields $xy+\delta$ for a $\delta$ chosen by its producer; the same holds for a GPU-computed matrix product with $C\neq AB$ and for a malformed DPF key. DFA's attestation covers the module's own key generation, not correlations from elsewhere (\cref{tab:scenarios}, \cref{sec:proto-arith}).

\textbf{S8: Snapshot substitution in PIR.}
The PIR database resides in host memory, so a corrupted server host can modify a record, serve an older version, or hash one buffer while scanning another. Semi-honest PIR, including DFA-accelerated PIR, returns whatever the servers compute. Worse, a record modified only at position $i$ corrupts the answer only when $\alpha=i$, so the client's reaction leaks the queried index (\cref{tab:scenarios}, \cref{sec:proto-pir}).

\textbf{S9 and S10: Input substitution; denial of service and selective abort.}
S9, a corrupt owner choosing its input freely, is permitted and extracted by $\FC$; binding a specific model needs the commitment of \emph{Model provenance} below. S10 is not prevented: like SPDZ-style protocols~\cite{damgard2012spdz,keller2016mascot}, \sys provides security with abort, without fairness or guaranteed output delivery, and abort outcomes stay within the permitted leakage (\cref{def:leakage}, \cref{sec:discussion}).

\begin{table*}[t]
\centering
\caption{Malicious scenarios against trusted FSS acceleration. The ``effect'' column describes a semi-honest trusted DFA~\cite{dfa2026trusted}, or a naive malicious extension of it, when one host is corrupted. \cmark: prevented up to the error terms of \cref{thm:main}; \pmark: permitted by the ideal functionality; \xmark: not prevented.}
\label{tab:scenarios}
\footnotesize
\setlength{\tabcolsep}{3pt}
\renewcommand{\arraystretch}{1.12}
\begin{tabularx}{\textwidth}{c >{\raggedright\arraybackslash}p{1.75cm} >{\raggedright\arraybackslash}X >{\raggedright\arraybackslash}X >{\raggedright\arraybackslash}X >{\raggedright\arraybackslash}p{2.75cm}}
\hline
\textbf{ID} & \textbf{Scenario} & \textbf{Adversary action} & \textbf{Effect on semi-honest trusted DFA} & \textbf{\sys mechanism} & \textbf{Guarantee} \\
\hline
S1 & Share/tag tampering & GPU adds errors $\bm{e}$ to value shares (and $\bm{d}_j$ to tags) before an opening or FSS gate & Undetected; biases outputs and causes selective failure that leaks inputs & Authenticated shares with $t=3$ MAC lanes over $\Fp$; \CheckOpen{} before any protected gate or release & \cmark\ Wrong value accepted w.p.\ $\le (2/p)^t$ per batch (\cref{thm:mac-sound}) \\
\hline
S2 & Forgery-oracle request & Host submits $y^*$ and asks its module to authenticate it & No MACs; a naive authenticate command would sign any forgery & No \emph{authenticate(host value)} command; \Epilogue{} only on internally computed or verified results & \cmark\ Request rejected; outputs correctly authenticated (\cref{thm:epilogue}) \\
\hline
S3 & TOCTOU buffer swap & Host rewrites a DMA buffer after it is checked, or after the challenge is fixed & Module trusts DMA reads; a later reread consumes modified bytes & Freeze before challenge: bytes copied or AEAD-sealed, frozen roots agreed, only frozen bytes consumed & \cmark\ Modification ignored or abort; binding error $\eps{bind}$ \\
\hline
S4 & Replay and rollback & Replays a counter or mask request; resets the module and resumes old keys & Mask reuse: $x+r$ and $x'+r$ reveal $x-x'$; key streams replayed & Labels bound to session, epoch, gate, purpose; one-time ledger; epoch poisoning; erase-and-rekey after reboot & \cmark\ Reuse rejected before any accepted use \\
\hline
S5 & Role/gate confusion & Routes party-0 material as party-1, or gate $g$ material into gate $g'$; alters shapes or table hash & Host schedule trusted; requesting both roles' outputs exposes the target $\alpha$ & Roles fixed by authenticated setup; gate descriptors bound in the ledger and in PRF labels & \cmark\ Mismatched commands abort before preprocessing is issued; host-side reuse of issued shares fails \CheckOpen{} except w.p.\ $\le(2/p)^t$ (\cref{thm:mac-sound}) \\
\hline
S6 & Premature release & Requests outputs before verification, repeatedly, or for another recipient & Results released as soon as computed & Release guard: nothing before \lstate{VERIFIED}; bound certificates; epoch abort gates pending outputs & \cmark\ No candidate value leaves the modules \\
\hline
S7 & Malformed preprocessing & Injects a triple with valid tags but $c\neq ab$, a GPU-computed product with $C\neq AB$, or a wrong DPF key & Attestation silent on external correlations; MACs verify wrong triples & Triples and keys generated inside the TCB ($c=ab$ computed there); imports need a correctness proof & \cmark\ Never consumed; exact arithmetic (\cref{thm:arith}) \\
\hline
S8 & Snapshot substitution (PIR) & Swaps records or version after hashing, or scans a different buffer & Client reconstructs a wrong record; index-dependent failure leaks $\alpha$ & Module hashes the bytes it consumes in public order; quarantine until the commitment matches & \cmark\ Committed record or abort (\cref{thm:pir}) \\
\hline
S9 & Input substitution by a corrupt owner & Corrupted party supplies arbitrary weights or inputs & Same (allowed) & Input extraction by $\FC$; optional commitment-bound loader for model provenance & \pmark\ Allowed by the functionality \\
\hline
S10 & Denial of service/selective abort & Drops, delays, or reorders messages; aborts after learning its own output & Same (not prevented) & Bounded resources and abort; abort outcomes simulatable from permitted leakage & \xmark\ Not prevented (no fairness or availability) \\
\hline
\end{tabularx}
\end{table*}

\subsection{Security Goals and Ideal Functionalities}\label{sec:threat-goals}

\textbf{Goals.}
Against a corrupted host, \sys must provide \emph{privacy} (the host learns nothing beyond the permitted leakage and the outputs it is authorized to receive), \emph{correctness} (honest recipients obtain the circuit evaluated on the honest inputs and the corrupted party's extracted inputs, or abort), \emph{authorized release} (no value leaves the modules unless its opening was verified and bound to a release certificate), and \emph{one-time use} (no randomness is consumed twice). We formalize the first two as stand-alone real/ideal security with abort, preserved under sequential composition~\cite{canetti2000security,goldreich2004foundations,hazay2010efficient,lindell2017simulate}, with an explicit leakage function; the last two are mechanisms that the proof relies on (\cref{sec:proto-ledger}). We claim neither universal nor concurrent composability, and concurrent sessions that share module state are not covered.

\textbf{Manifest.}
Every session fixes a public manifest $\mu$ (\cref{app:func-circuit}): the protocol version, session identities, roles, field, numeric encoding, the finite circuit $C$ with its gate and wire identifiers and tensor shapes, approved public tables, interval coefficients, authorized recipients, the input interface $\mathsf{In}_\nu$ and validation rules of every input wire $\nu$, and the public schedule, which places every input gate before every validity check, and both before the first output opening. The modules evaluate the validity predicates at the public boundaries declared in $\mu$; in $\Pi_{\mathsf{Input}}$ they compute $x^{*}=r+\varepsilon$ and abort unless $x^{*}\in\mathsf{In}_\nu$ (\cref{fig:arith}).

\begin{definition}[Permitted leakage]\label{def:leakage}
For a circuit execution under manifest $\mu$, the permitted leakage is
\[
\Leak_C=\bigl(\mu,\ \mathsf{len},\ \mathsf{sched},\ \mathsf{pub},\ \mathsf{status}\bigr),
\]
where $\mathsf{len}$ are the lengths of all public messages; $\mathsf{sched}$ is the public schedule of batches, checks, and interaction boundaries fixed by $\mu$; $\mathsf{pub}$ are the values that $\mu$ designates as authorized public openings; and $\mathsf{status}$ are the specified abort and validation outcomes, including every input-dependent validity bit declared in $\mu$, which is never silently excluded. For PIR, $\Leak_{\mathrm{PIR}}$ consists of the public descriptor $\mathsf{desc}_v=(\mathsf{id},v,N,\mathsf{enc},\mathsf{cm}_v)$ (database identifier, version, true length $N$, which also fixes the padded length, record encoding, and authorized snapshot commitment), the public scan length and order, message lengths and schedules, and the accept, reject, or abort outcome.
\end{definition}

Honest inputs, intermediate wire values, masks, MAC keys, DPF keys and correction words, raw FSS outputs, challenges, and residuals are not in $\Leak_C$. Masked openings such as $x+r$ are uniform given $\Leak_C$, and the simulator samples them (\cref{app:proof-main}). Module timing and access patterns are fixed by $\mu$ (full-domain scans have a declared order); the residual trace falls under $\eps{trace}$.

\begin{figure}[t]
\centering
\footnotesize
\fbox{\begin{minipage}{\dimexpr\columnwidth-2\fboxsep-2\fboxrule\relax}
\centering\underline{Functionality $\FC$}\par\smallskip
\raggedright
\textbf{Parameters.} Public manifest $\mu$ (circuit $C$ over $\Fp$, input owners and interfaces $\mathsf{In}_\nu$, validity predicates, authorized recipients, public schedule). Corrupted set $I$: at most one host; the client is honest.
\begin{enumerate}[leftmargin=*,nosep,topsep=2pt]
\item \textbf{Input.} Receive the authenticated input $x_\nu$ of every input wire $\nu$ from its designated owner. If the owner is corrupted, receive $x^{*}_\nu$ from $\Sim$ and set $x_\nu\gets x^{*}_\nu$ (extraction). If some $x_\nu$ is missing, malformed, or not in $\mathsf{In}_\nu$, send $\bot$ to every honest recipient and halt; this abort outcome is part of $\mathsf{status}$.
\item \textbf{Validate and evaluate.} Evaluate the validity predicates of $\mu$; their accept or abort outcome is part of $\mathsf{status}$. For accepted inputs, compute $\vec{y}=C(\vec{x})$ exactly in the declared finite arithmetic, including the specified comparisons and truncations.
\item \textbf{Deliver or abort.} Send $\Leak_C$ and the outputs of corrupted recipients to $\Sim$. At each public interaction boundary, $\Sim$ may send $\mathsf{abort}$ or $\mathsf{delay}$; on $\mathsf{abort}$, send $\bot$ to every honest recipient. Otherwise deliver each $y_i$ to its authorized recipient. A corrupted recipient may obtain its output before honest ones (no fairness).
\end{enumerate}
\end{minipage}}

\vspace{5pt}
\fbox{\begin{minipage}{\dimexpr\columnwidth-2\fboxsep-2\fboxrule\relax}
\centering\underline{Functionality $\FPIR$}\par\smallskip
\raggedright
\textbf{Parameters.} Public descriptor $\mathsf{desc}_v=(\mathsf{id},v,N,\mathsf{enc},\mathsf{cm}_v)$; the commitment $\mathsf{cm}_v$ fixes $D=(D[0],\dots,D[N-1])$. Corruption regimes, never combined: (R1) a malicious server host (at most one), client honest; (R2) a malicious client only, both server hosts honest.
\begin{enumerate}[leftmargin=*,nosep,topsep=2pt]
\item \textbf{Query.} Receive an index $\alpha$ with query identifier $\mathit{qid}$ from the client; in (R2), extract the corrupted client's index $\alpha^{*}$ and set $\alpha\gets\alpha^{*}$. If $\alpha\notin\{0,\dots,N-1\}$ or the request is inconsistent, return $\mathsf{reject}$ to the client and stop.
\item \textbf{Leak.} Send $\Leak_{\mathrm{PIR}}$ to $\Sim$; nothing about $\alpha$ is sent.
\item \textbf{Answer.} In (R1), if $\Sim$ sends $\mathsf{abort}$, send $\bot$ to the client; otherwise send $D[\alpha]$. In (R2), send $D[\alpha]$, and no other information about $D$, to the client.
\end{enumerate}
\end{minipage}}
\caption{Ideal functionalities for authenticated circuit evaluation (top) and versioned-snapshot PIR (bottom). Both provide security with abort. $\FPIR$ hides $\alpha$ from the corrupted host; both trusted modules learn it (\cref{sec:threat-model}). The full definitions, together with the ideal module resource $\mathcal{F}_{\mathsf{T}}$, are in \cref{app:func}.}
\label{fig:func}
\end{figure}

\begin{definition}[Static malicious security with abort]\label{def:security}
Let $\mathcal{F}\in\{\FC,\FPIR\}$ with leakage $\Leak$ as in \cref{def:leakage}, and let $\Pi$ be a protocol run by $P_0$ and $P_1$ in which each module executes the concrete module program correctly and confidentially: $\mathcal{F}_{\mathsf{T}}$ with its joint random function $R$ instantiated by $\PRF_K$, its peer and client channels by AEAD over the hosts' network, and \cmd{OPEN_SESSION} by attestation and authenticated key exchange.
For a PPT adversary $\Adv$ with auxiliary input $z$ that statically corrupts the set $I$, $\Real_{\Pi,\Adv(z),I}(\mu,\vec{x})$ is the joint distribution of the view of $\Adv$ (its inputs, randomness, received messages, and the external traces of the modules observable from corrupted hosts) and the outputs of the honest parties and client in an execution of $\Pi$ on manifest $\mu$ and input vector $\vec{x}$.
$\Ideal_{\mathcal{F},\Sim(z),I}(\mu,\vec{x})$ is the joint distribution of the output of a simulator $\Sim$ and the honest outputs in an execution in which $\Sim$ interacts only with $\mathcal{F}$: it submits the corrupted parties' inputs, receives $\Leak$ and the corrupted parties' outputs, and may abort or delay as $\mathcal{F}$ allows.
We say that $\Pi$ \emph{$\epsilon$-securely realizes $\mathcal{F}$ with leakage $\Leak$ against static malicious corruption with abort} if for every PPT adversary $\Adv$ and every set $I$ admitted by the corruption regime of $\mathcal{F}$ there exists a PPT simulator $\Sim$ such that for all $\mu$, $\vec{x}$, $z$ and every PPT distinguisher $\mathcal{D}$,
\begin{align*}
\Bigl|&\Pr\bigl[\mathcal{D}\bigl(\Real_{\Pi,\Adv(z),I}(\mu,\vec{x})\bigr)=1\bigr]\\
&\quad-\Pr\bigl[\mathcal{D}\bigl(\Ideal_{\mathcal{F},\Sim(z),I}(\mu,\vec{x})\bigr)=1\bigr]\Bigr|\le\epsilon .
\end{align*}
\end{definition}

\textbf{Remarks.}
The real world of \cref{def:security} contains the concrete module program; $\mathcal{F}_{\mathsf{T}}$ with ideal setup, channels, and randomness is the resource of hybrids $\Hyb{2}$ to $\Hyb{4}$ of the proof (\cref{app:proof-main}), and the simulator controls only this ideal representation of the trusted resources, never the honest party's prescribed inputs. The concrete $\epsilon$ is the bound $\eps{VIGOR}$ of \cref{thm:main}; correctness of protected code is a premise of $\mathcal{F}_{\mathsf{T}}$, and \cref{tab:obligations} lists what a chip must still establish.

\textbf{Model provenance.}
Computational correctness is not model provenance, and a public architecture hash in $\mu$ does not authenticate private weights. To bind a pre-approved model, $\mu$ must contain a commitment to the canonical quantized weight bytes, authorized by the client or an independent model authority rather than declared by the untrusted host, and the protected input loader hashes exactly those bytes before sharing them (\cref{sec:proto-arith}); a mismatch aborts. Otherwise the weights are the owner's freely chosen private input, and S9 is the intended behavior of $\FC$.

\subsection{The Stronger Profile U}\label{sec:threat-profileu}

\ProfT trusts the modules. When the hardware itself must be treated as adversarial, for example because the accelerator vendor or platform administrator is untrusted, \ProfU gives the accelerator neither the joint-sampler key nor any global MAC key. The circuit runs through an established malicious two-party backend, such as SPDZ-family protocols with MASCOT or SPDZ2k preprocessing~\cite{damgard2012spdz,keller2016mascot,cramer2018spdz2k,keller2020mpspdz} or dealer-based malicious FSS such as Shark~\cite{gupta2025shark}, and the accelerator is confined to party-local data, for instance expanding the DPF keys that the backend exports to its own party. The same DPF datapath therefore serves both profiles, but \ProfU gives up the local-dealer and generate-and-consume savings, and with exported keys it must add the backend's DPF privacy assumption to its bound (\cref{sec:sec-main}). \ProfU is a separate backend chosen at session start, not a fallback after a module compromise, which would not undo leaked secrets; a fresh session with fresh inputs and masks is required. The rest of this paper analyzes and evaluates \ProfT.

\section{Workload Analysis: Where Malicious Security Costs Go}
\label{sec:analysis}

The DFA design~\cite{dfa2026trusted} was driven by one workload observation: in semi-honest FSS workloads, DPF generation, distribution and evaluation dominate both computation and communication, so a small DPF engine that also plays the dealer removes the dominant cost. Repeating this analysis for the adversary of \cref{sec:threat} changes the design target: handled naively, malicious security multiplies exactly the DPF work and key traffic that trusted DFA made small; handled carefully, it moves the cost from DPF to authenticated linear algebra on the GPU and to the extra round of every verified opening.

\subsection{Cost Sources in Malicious FSS}
\label{sec:analysis-cost}

We take as reference the masked-input paradigm of FSS-based inference~\cite{gupta2023sigma,jawalkar2024orca} (\cref{sec:bg-pi}): linear layers operate on additive shares, and every nonlinear gate opens a masked input $x+r$ and evaluates FSS keys prepared for the mask $r$. Making this pipeline secure against a malicious host (\cref{def:security}) adds the following costs, regardless of how they are implemented.

\textbf{(i) Authenticated representation.}
Every wire must carry information-theoretic MACs so that a host that changes its share is caught~\cite{damgard2012spdz,damgard2013practical}. In our representation (\cref{sec:proto-repr}), party $b$ holds $\authb{x}{b}=(x_b,\gamma_{b,1},\gamma_{b,2},\gamma_{b,3})$ over $\Fp$ with $p=2^{61}-1$ and $t=3$ independent MAC keys: four 64-bit words, or 32~B per element and party, instead of 8~B for a semi-honest share in $\Z_{2^{64}}$. Every linear layer is therefore computed on four lanes (the value and three tags), each lane uses exact Mersenne reduction instead of free 64-bit wraparound (\cref{sec:arch-field}), and every authenticated element that crosses PCIe or the host-module interface grows by the same factor of four. Only the peer link is spared, because checks exchange value shares and $t$ residuals rather than tags (\cref{sec:proto-check}).

\textbf{(ii) Verified openings.}
Opened values (masked gate inputs, the Beaver differences $d,e$ of each multiplication~\cite{beaver1991efficient}, and outputs) must be checked before anything depends on them. A sound check derives its challenge only after the checked bytes are fixed (\cref{thm:mac-sound}), so at least one extra exchange follows every checked batch, and a dependent nonlinear gate may not consume an unchecked opening. Batching amortizes challenge derivation and metadata but not this round, which lies on the critical path of every layer: a protocol with $R$ online rounds approaches $2R$.

\textbf{(iii) Correct preprocessing.}
Triples, masks and FSS keys must be \emph{correct}, not merely authenticated: a valid MAC on a triple with $c\neq ab$ does not make it correct. The options of \cref{sec:bg-mal} either ship more offline material from a dealer or replace the dealer by heavy interactive preprocessing, and neither removes offline data movement. Without a dealer, FSS keys require a secure two-party evaluation of DPF key generation, with one round per tree level~\cite{doerner2017scaling} before any malicious-security overhead.

\textbf{(iv) Authenticated FSS outputs.}
A DPF output share is a plain additive share. To re-enter authenticated arithmetic it needs $t$ tag shares, and the straightforward way to obtain them is a second DPF per MAC lane $j$ with the same target and payload $\Delta_j\beta$ (a \emph{tag tree}). This multiplies DPF keys, key bytes and PRG work by $1+t$, i.e., by four for $t=3$.

\textbf{(v) Key validity is not evaluation correctness.}
Verifiable FSS~\cite{decastro2022vfss} lets evaluators check that their keys are well formed, which protects against a malicious key generator. It does not constrain an evaluator that computes on valid keys and then adds an error to its output share, nor the arithmetic that consumes that share.

\begin{figure}[t]
    \centering
    \includegraphics[width=\linewidth]{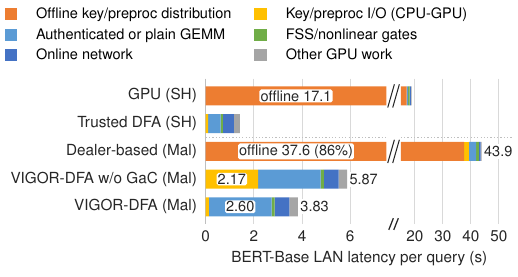}
    \caption{BERT-Base latency breakdown on LAN (1~GB/s, 0.05~ms round-trip time) by cost source, for two semi-honest configurations of~\cite{dfa2026trusted} and three malicious ones (dealer-based malicious FSS, \sys without GaC, \sys); all splits come from our model (\cref{app:model}). The horizontal axis is broken between 7.5 and 15~s so that the 18.8~s semi-honest GPU bar and the {43.9~s} dealer-based bar fit.}
    \label{fig:motiv}
\end{figure}

\textbf{Observation 1: distribution still dominates dealer-based malicious FSS.}
\Cref{fig:motiv} splits the BERT-Base LAN latency by cost source. Dealer-based malicious FSS, modeled after the protocol family of Shark~\cite{gupta2025shark} (\cref{sec:eval-method}), must move {37.6~GB} of authenticated keys and preprocessing per query, {2.2$\times$} the 17.1~GB of the semi-honest GPU baseline~\cite{dfa2026trusted}. At 1~GB/s, this distribution alone takes {37.6~s} of {43.9~s} ({86\%}), and moving the material from host memory into the GPU adds another {1.70~s}. Without an authenticated-key blow-up ($\psi_{\mathrm{key}}=1$, \cref{app:model}), it still takes {22.4~s}, {5.8$\times$} that of \sys.

\textbf{Observation 2: without distribution, DPF is no longer the bottleneck.}
Removing distribution, as trusted DFA did for the semi-honest case, exposes a different profile. In \sys, four-lane authenticated GEMM takes {2.60~s} of {3.83~s} ({68\%}), {5.0$\times$} our {0.52~s} semi-honest GEMM; online communication takes {0.615~s} ({16\%}); and all nonlinear gates together take only {0.131~s} ({3.4\%}).

\textbf{Observation 3: on a WAN, verification rounds dominate.}
With a 70~ms round-trip time (RTT) the extra check round of source (ii) outweighs all computation: removing \CheckOpen, which is insecure and serves only as a lower bound, would reduce the BERT-Base WAN latency from {89.7~s} to {50.5~s} (\cref{tab:ablation}).

\subsection{Compute Pattern of Exact DPF}
\label{sec:analysis-dpf}

Malicious-security proofs reason about concrete states, so we fix the DPF computation exactly. \Cref{alg:gen,alg:eval} give the tree-based DPF~\cite{gilboa2014distributed,boyle2016function} in the form that \sys executes; the same algorithms serve \ProfU, where keys are exported to an external malicious backend (\cref{sec:threat-profileu}). The expansion function $\PRG:\bits{128}\to(\bits{128}\times\{0,1\})^2$ is fixed and AES-based as in DFA, $\Conv:\bits{128}\to\Fp$ is the fixed output conversion, the target $\alpha$ and the query $x$ have $n$ bits read most-significant first, and the payload is $\beta\in\Fp$.

\begin{algorithm}[t]
\caption{Exact DPF key generation ($\DPF.\Gen$)}
\label{alg:gen}
\begin{algorithmic}[1]
\algstretch
\Require target $\alpha=\alpha_1\cdots\alpha_n\in\bits{n}$ (MSB first), payload $\beta\in\Fp$
\Ensure keys $(k_0,k_1)$ with $\Eval(k_0,x)+\Eval(k_1,x)=\beta\cdot\one[x=\alpha]$
\State sample independent $s_0,s_1\getsr\bits{128}$; $t_0\gets 0$; $t_1\gets 1$
\State save the initial seeds $s_0^{(0)}\gets s_0$, $s_1^{(0)}\gets s_1$
\For{$i=1$ \textbf{to} $n$}
    \State $a\gets\alpha_i$; \ $\keep\gets a$; \ $\lose\gets 1-a$ \Comment{$L=0$, $R=1$}
    \State $(s_b^L,t_b^L,s_b^R,t_b^R)\gets\PRG(s_b)$ for $b\in\{0,1\}$
    \State $\sCW\gets s_0^{\lose}\oplus s_1^{\lose}$
    \State $\tCW{L}\gets t_0^L\oplus t_1^L\oplus a\oplus 1$; \ $\tCW{R}\gets t_0^R\oplus t_1^R\oplus a$
    \LineComment{both $b$, with the \emph{old} $t_b$; $\cdot$ is a conditional XOR}
    \State $(s_b,t_b)\gets\bigl(s_b^{\keep}\oplus t_b\cdot\sCW,\ t_b^{\keep}\oplus t_b\cdot\tCW{\keep}\bigr)$
    \State $\CW^{(i)}\gets(\sCW,\tCW{L},\tCW{R})$
\EndFor
\State $u_b\gets\Conv(s_b)$ for $b\in\{0,1\}$
\State $\CWout\gets(-1)^{t_1}\,(\beta-u_0+u_1)$
\LineComment{keys hold the initial roots, never the final states}
\State $k_b\gets\bigl(s_b^{(0)},\,b,\,\CW^{(1)},\ldots,\CW^{(n)},\,\CWout\bigr)$ for $b\in\{0,1\}$
\State \Return $(k_0,k_1)$
\end{algorithmic}
\end{algorithm}

\begin{algorithm}[t]
\caption{Exact DPF evaluation ($\DPF.\Eval$, $\DPF.\EvalAll$)}
\label{alg:eval}
\begin{algorithmic}[1]
\algstretch
\Require role $b$, key $k_b=(s_b^{(0)},b,\CW^{(1)},\ldots,\CW^{(n)},\CWout)$, query $x=x_1\cdots x_n$
\Ensure share $y_b\in\Fp$ (or all $2^n$ leaf shares for $\EvalAll$)
\Function{Eval}{$k_b,x$}
    \State $(s,t)\gets(s_b^{(0)},b)$ \Comment{root control bit $=$ role}
    \For{$i=1$ \textbf{to} $n$} \Comment{$\CW^{(i)}=(\sCW,\tCW{L},\tCW{R})$}
        \State $(s^L,t^L,s^R,t^R)\gets\PRG(s)$ \Comment{$t$: old evaluator bit}
        \State $(s^L,t^L)\gets(s^L\oplus t\cdot\sCW,\ t^L\oplus t\cdot\tCW{L})$
        \State $(s^R,t^R)\gets(s^R\oplus t\cdot\sCW,\ t^R\oplus t\cdot\tCW{R})$
        \State $(s,t)\gets(s^{L},t^{L})$ if $x_i=0$, else $(s^{R},t^{R})$
    \EndFor
    \State \Return $y_b\gets(-1)^b\bigl(\Conv(s)+t\cdot\CWout\bigr)$
\EndFunction
\Function{EvalAll}{$k_b$}
    \State push $(s_b^{(0)},b,0,\langle\rangle)$ onto the stack \Comment{$\langle\rangle$: empty prefix}
    \While{the stack is not empty}
        \State pop $(s,t,i,\pi)$ \Comment{seed, bit, depth, prefix}
        \If{$i=n$}
            \State emit leaf $\pi$ with share $(-1)^b\bigl(\Conv(s)+t\cdot\CWout\bigr)$
        \Else
            \State expand $s$ as in lines 4--6 with $\CW^{(i+1)}$
            \State push $(s^R,t^R,i{+}1,\pi\|1)$
            \State push $(s^L,t^L,i{+}1,\pi\|0)$ \Comment{left first}
        \EndIf
    \EndWhile
\EndFunction
\end{algorithmic}
\end{algorithm}

\begin{figure*}[t]
    \centering
    \includegraphics[width=\textwidth]{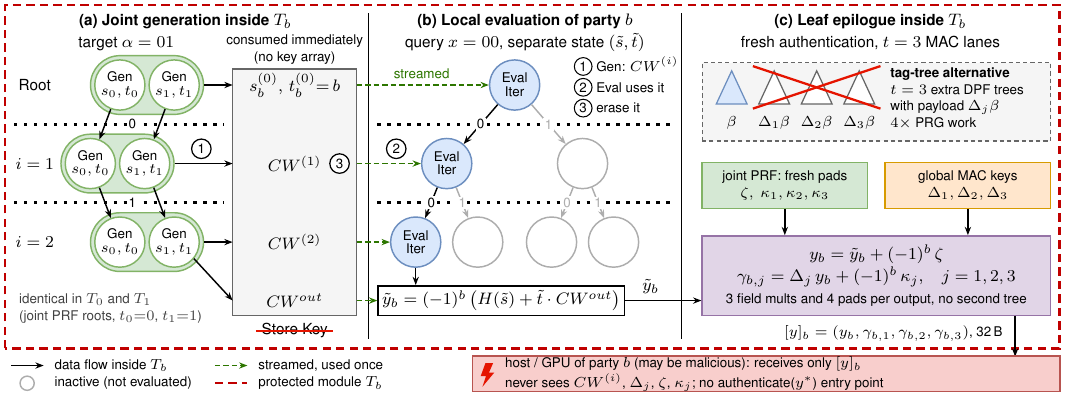}
    \caption{Malicious-safe generate-and-consume inside module $T_b$ for a depth-2 tree. (a)~Joint generation: the two generator states $(s_0,t_0)$ and $(s_1,t_1)$ follow the target path $\alpha=01$ and emit $\CW^{(1)}$, $\CW^{(2)}$ and $\CWout$, which are consumed immediately instead of being stored as a key. (b)~Local evaluation of party $b$ on query $x=00$: a separate evaluator state $(\tilde s,\tilde t)$ follows $x$, applies each correction word at its own level, and leaves the target path at level~2, where only its own control bit $\tilde t$ is correct; grey nodes are never expanded. (c)~Leaf epilogue: the raw share $\tilde y_b$ is freshly reshared and authenticated with joint pads $\zeta,\kappa_1,\kappa_2,\kappa_3$ into the 32~B tuple $\authb{y}{b}$. The crossed-out inset is the tag-tree alternative, which needs $t=3$ extra DPF trees with payloads $\Delta_j\beta$ and four times the PRG work.}
    \label{fig:gac}
\end{figure*}

\textbf{High-level idea.}
\Cref{fig:gac}(a) shows $\Gen$ walking two states $(s_0,t_0)$ and $(s_1,t_1)$ down the target path. Since each state applies a correction only when its \emph{own} old control bit is set, exactly one of the two is corrected at each level; the corrections make the off-path states of both parties identical, so their signed outputs cancel, while $\CWout$ programs $\beta$ at the target leaf (\cref{thm:dpf-correct}). Correctness needs neither uniform targets nor computational hardness; uniform masks belong to the application protocol (\cref{sec:proto-refactor}).

\textbf{Explicit states.}
Three details that the simplified semi-honest listing~\cite{dfa2026trusted} leaves implicit are explicit in \cref{alg:gen,alg:eval}, because the malicious analysis depends on them. First, the root control bits are fixed by role, $t_0=0$ and $t_1=1$, rather than sampled, and the evaluator of party $b$ always starts from $(s_b^{(0)},b)$. Second, a key contains the \emph{initial} seed of its party, the correction words and $\CWout$; the generator's final states are never substituted for the roots. Third, the payload and the role enter only through $\CWout$ and the output sign $(-1)^b$, so both are bound inputs of the ledgered gate descriptor (\cref{sec:proto-ledger}) rather than host-supplied parameters.

\textbf{Generator versus evaluator state.}
Lines 5--6 of \cref{alg:eval} correct both children with the evaluator's own old control bit $t$. Along the target path this bit coincides with the generator's $t_b$, which invites reusing the generator's control register during generate-and-consume. The substitution is wrong as soon as the query leaves the path (level~2 in \cref{fig:gac}(b)): there the evaluator states of the two parties are identical and must apply identical corrections, whereas the generator bits $t_0$ and $t_1$ still differ, so one party would apply a correction that the other omits, and the off-path outputs would no longer cancel. Each module therefore keeps a separate evaluator state next to its two generator states (\cref{alg:gac}).

\subsection{Work Accounting}
\label{sec:analysis-work}

\begin{table*}[t]
\centering
\caption{Analytic Cost of One Single-Path DPF Gate of Depth $n$ with $t$ MAC Lanes}
\label{tab:gatecost}
\footnotesize
\setlength{\tabcolsep}{4pt}
\renewcommand{\arraystretch}{1.2}
\begin{tabular}{llccccccc}
\hline
\textbf{Design} & \textbf{Security} & \makecell{\textbf{PRG exp.}\\\textbf{per gate}} & \makecell{\textbf{Field ops per}\\\textbf{output, party}} & \makecell{\textbf{Offline bytes}\\\textbf{per gate}} & \makecell{\textbf{Bytes into GPU}\\\textbf{per party}} & \makecell{\textbf{Extra}\\\textbf{rounds}} & \makecell{\textbf{Auth.}\\\textbf{output}} & \makecell{\textbf{Verifies}\\\textbf{evaluation}} \\
\hline
Dealer-based FSS~\cite{gupta2023sigma,jawalkar2024orca} & SH & $4n$ & 0 & $2|k|$ & $|k|$ & 0 & \xmark & \xmark \\
Trusted DFA, GaC~\cite{dfa2026trusted} & SH & $6n$ & 0 & 0 & 8 & 0 & \xmark & \xmark \\
\makecell[l]{Dealer-based FSS, tag trees\\(counting bound)} & Mal. & $4n(1+t)$ & 0 & $2(1+t)|k|$ & $(1+t)|k|$ & 1 per batch & \cmark & \cmark\ (at opening) \\
vFSS-verified keys~\cite{decastro2022vfss} & Mal.\ key gen. & $4n$ + sketch & 0 & $2|k|+O(\lambda)$ & $|k|+O(\lambda)$ & 1 (sketch) & \xmark & \xmark\ (keys only) \\
Vector-payload tree, $t$ lanes$^{\dagger}$ & Mal. & $6n+O(t)$ & $t$ adds & 0 & 32 & 1 per batch & \cmark & \cmark \\
\sys, GaC + epilogue & Mal. & $6n$ & \makecell{3 mult.\\+ 4 pads} & 0 & 32 & 1 per batch & \cmark & \cmark \\
\hline
\end{tabular}

\vspace{3pt}
\parbox{\textwidth}{\footnotesize SH: semi-honest; Mal.: malicious. $|k|$ is the byte length of one party's key. PRG expansions are summed over all parties, dealer included. Bytes into the GPU are keys loaded by a GPU evaluator or output shares delivered by a module. Mask shares and Beaver triples of the surrounding arithmetic are omitted; with authentication, each such element costs 32~B per party instead of 8~B. Rounds are counted relative to the plain masked opening that every design performs. The evaluated dealer-based malicious FSS baseline uses {2.2$\times$} the semi-honest offline bytes (\cref{app:model}), below this counting bound. $^{\dagger}$One in-module tree with payload $(\beta,\Delta_1\beta,\ldots,\Delta_t\beta)$; $O(t)$ counts extra leaf-conversion blocks.}
\end{table*}

\Cref{tab:gatecost} compares one single-path DPF gate at six design points. With 128-bit seeds, one party's key holds its root seed, $n$ correction words of $128+2$ bits and $\CWout$ as one canonical 64-bit word (the root control bit is the public role), so $|k|=\lceil(128+130n+64)/8\rceil$ bytes, or 1{,}016~B for $n=61$.

\textbf{Design-point comparison.}
Trusted DFA trades $1.5\times$ the dealer's PRG work for zero key bytes: each module expands $2n$ generator states and $n$ evaluator states for its own party, so the two modules perform $6n$ expansions where a dealer and two evaluators perform $4n$. Tag trees multiply PRG work, key bytes and key I/O by $1+t$, and verified keys authenticate neither the output nor its evaluation. \sys keeps DFA's $6n$ expansions and zero key bytes and adds only three field multiplications and four pads per output and party; it cannot avoid the 32~B authenticated tuple and the check round.

\textbf{Effect of the epilogue.}
Because the epilogue (\cref{fig:gac}(c)) leaves the per-gate PRG expansion count unchanged, the single-path throughput of the DPF engine drops only from {113} to {106~M~gates/s} at $n=61$ relative to semi-honest trusted DFA, whereas per-lane tag trees fall to {26.7~M~gates/s}; the gap to per-lane tag trees stays at {3.8--4.0$\times$} for $n$ between 8 and 61 (\cref{fig:dpfthr}). The drop comes mainly from the pipeline efficiency ({0.85} vs.\ {0.90}; \cref{app:model}); the two extra pad blocks ($6n{+}3$ vs.\ $6n{+}1$ AES blocks per gate) cost only 0.5\%. The vector-payload tree nearly matches the epilogue's PRG count, but it widens every output correction by $t$ field elements, covers only results linear in DPF evaluations, and exports pseudorandom tag shares whose hiding rests on DPF key privacy. The epilogue instead makes every exported tuple exactly uniform for its holder (\cref{thm:epilogue}) and applies unchanged to LUT, interval, direct and PIR results.

\textbf{Remaining costs we accept.}
\sys does not remove duplicated generation: $6n$ instead of $4n$ expansions is the price of eliminating $2|k|$ network bytes and $|k|$ bytes of key I/O per gate, a favorable trade because PRG expansion is the cheapest operation in the pipeline while bytes cross its slowest links. Two further costs remain. The conservative $O(n^2)$ DCF of \cref{alg:dcf} is the only DCF our proofs cover; its interval gate costs 22{,}692 AES blocks at $n=61$, against a {976} for an optimized DCF, so the compiler chooses among direct, LUT and interval gates by its cost rule (\cref{app:model}), subject to an FSS-portability constraint (\cref{sec:proto-direct,sec:eval-micro}). A host-module link of $B$ bytes/s bounds the authenticated output stream at $B/32$ elements/s, which motivates the L2-adjacent placement (\cref{sec:arch-insight}).

\textbf{Implications.}
For sources (iii) and (iv) the semi-honest lesson still holds: distribution is the largest avoidable cost (Observation~1), so DFA's local generation and generate-and-consume must stay, and FSS outputs must be authenticated without multiplying DPF work. This leaves sources (i) and (ii) as the dominant costs (Observations~2 and~3). The analysis yields five design rules that structure the rest of the paper: (1)~authenticate FSS outputs by fresh resharing inside the module, never by extra trees and never by authenticating host values (S2; \cref{sec:proto-epilogue}); (2)~keep generation local and streamed, with separate generator and evaluator states (S7; \cref{sec:proto-gac}); (3)~make the unavoidable check round sound and batchable by freezing bytes before challenges are drawn (S1, S3, S8), and release no output before \lstate{VERIFIED} (S6; \cref{sec:proto-check,sec:proto-ledger}); (4)~bind every randomness draw to one session, gate, role and use, so that local generation does not open a replay path (S4, S5; \cref{sec:proto-ledger}); and (5)~leave four-lane exact linear algebra on the GPU while the module generates correct triples (S7) and verifies openings (\cref{sec:arch-insight}).

\section{\texorpdfstring{\sys}{VIGOR-DFA} Protocol Design}\label{sec:protocol}

This section specifies the protocol run by the two parties of \ProfT: the flow (\cref{sec:proto-overview}), the representation, joint randomness and lifecycle state (\cref{sec:proto-repr,sec:proto-setup,sec:proto-ledger}), and the mechanisms that move values through the integrity chain (\cref{sec:proto-gac,sec:proto-epilogue,sec:proto-check,sec:proto-arith}). Nonlinear gates, refactoring, GPU verification and PIR follow in \cref{sec:proto-gates,sec:proto-direct,sec:proto-refactor,sec:proto-freivalds,sec:proto-pir}; formal statements are collected in \cref{sec:security}.

\subsection{Overview}\label{sec:proto-overview}

\textbf{High-level idea.}
In trusted DFA the protected module replaces the dealer: it samples FSS keys locally and consumes each correction word as soon as it produces it~\cite{dfa2026trusted}. \sys turns the module from a dealer into a dealer that also verifies. Every wire carries an authenticated sharing with $t$ information-theoretic MAC lanes. The GPU performs all linear algebra on these sharings but only ever produces \emph{candidate} shares. Each module generates every correlation of its party, authenticates every FSS output it computes, checks every opening on frozen bytes and releases nothing before verification.

\textbf{Per-gate lifecycle.}
Every gate of the public circuit $C$ fixed by the manifest $\mu$ passes through the same module commands (\cref{tab:commands}). (1)~With \cmd{RESERVE_PREP}, $T_b$ reserves the gate's randomness labels in its one-time ledger and issues party-local authenticated preprocessing (masks, triples) to host~$b$. (2)~Host $b$'s GPU computes candidate authenticated shares with local linear operations. (3)~With \cmd{FREEZE_BATCH} and \cmd{CHECK_OPEN}, the two modules freeze the submitted bytes and verify the masked openings the gate needs. (4)~For a nonlinear gate, \cmd{EVAL_FSS} consumes the verified opening, runs generate-and-consume inside $T_b$ and returns a freshly authenticated share. (5)~\cmd{FINAL_OUTPUT} hands a masked output and a bound certificate to an authorized recipient. A failure in any step triggers \cmd{ABORT}, which poisons the epoch. Each mechanism below enforces one design rule of \cref{sec:analysis-work}; \cref{tab:notation} summarizes the notation.

\begin{table}[t]
\caption{Notation used in the protocol description}
\label{tab:notation}
\centering
\footnotesize
\renewcommand{\arraystretch}{1.08}
\begin{tabular}{p{0.27\linewidth} p{0.63\linewidth}}
\hline
\textbf{Symbol} & \textbf{Meaning} \\
\hline
$P_b$; host $b$; $T_b$ & party $b\in\{0,1\}$; its untrusted host (CPU, OS, driver, GPU, memory, NIC); its protected module \\
$\mu$, $C$ & public manifest and the finite circuit it fixes \\
$p$, $\Fp$ & Mersenne prime $2^{61}-1$ and the protocol field \\
$t$ & number of MAC lanes ($t=3$) \\
$\Delta_j$; $\Delta_{b,j}$ & global MAC key of lane $j$ (modules only); host $b$'s additive share of it \\
$\auth{x}$; $\authb{x}{b}$ & authenticated sharing of $x$; party $b$'s tuple $(x_b,\gamma_{b,1},\dots,\gamma_{b,t})$ \\
$\gamma_{b,j}$ & tag share of party $b$ in lane $j$, $\gamma_{0,j}+\gamma_{1,j}=\Delta_j x$ \\
$K$ & session key shared by $T_0$ and $T_1$ after setup \\
$\ell$ & length-prefixed label (session, epoch, gate, purpose, logical draw, block) \\
$\PRF_K(\ell,\cdot)$ & joint draw: same label, same value in both modules \\
$\nu$; $\varepsilon$; $\upsilon$ & wire index; masked input difference $x-r$; masked output $y+o$ \\
$\mathsf{In}_\nu$; $\mathcal{V}$ & declared interface of input wire $\nu$; validity predicates of $\mu$ \\
$n$; $\lambda$ & DPF input length in bits; PRG seed length ($\lambda=128$) \\
$\alpha$, $\beta$, $x$ & DPF target, payload in $\Fp$, verified public query \\
$\CW^{(i)}$; $\CWout$ & level-$i$ correction word; output correction \\
$(s_c,t_c)$; $(\tilde s,\tilde t)$ & generator state of root $c\in\{0,1\}$; evaluator state of $T_b$ \\
$\tilde y_b$ & raw local output of the protected evaluator \\
$\zeta$; $\kappa_j$ & fresh common epilogue pads for the value; for lane $j$ \\
$m$; $h_b$; $\mathit{cid}$ & batch size; frozen root of $T_b$; check identifier \\
$\rho_{j,i}$; $\sigma_{b,j}$ & challenge of lane $j$ at batch position $i$; residual of $T_b$ in lane $j$ \\
$\bm{e}$; $\bm{d}_j$ & value-error vector; tag-error vector of lane $j$ in a checked batch \\
$Q$ & number of checked batches in one execution \\
$w$ & LUT input width (\cref{sec:proto-gates}) \\
\hline
\end{tabular}
\end{table}

\subsection{Authenticated Representation}\label{sec:proto-repr}

\textbf{Field and MAC lanes.}
All protocol arithmetic is over $\Fp$ with the Mersenne prime $p=2^{61}-1$. The modules hold $t=3$ independent global MAC keys $\Delta_1,\Delta_2,\Delta_3\getsr\Fp$, each additively split as $\Delta_j=\Delta_{0,j}+\Delta_{1,j}$; host~$b$ holds only $\Delta_{b,j}$, which it needs to add public values to its tags. Party $b$ holds an authenticated value $\auth{x}$ as
\begin{gather*}
\authb{x}{b}=(x_b,\gamma_{b,1},\gamma_{b,2},\gamma_{b,3}),\\
x_0+x_1=x,\qquad \gamma_{0,j}+\gamma_{1,j}=\Delta_j x\quad (j\in[t]),
\end{gather*}
which is the SPDZ representation~\cite{damgard2012spdz,damgard2013practical} with three independent lanes. A canonical wire occupies four 64-bit words, i.e., 32~B per party before metadata. Every element is encoded as an unsigned 64-bit integer smaller than $p$; the module parser rejects noncanonical encodings instead of reducing arbitrary integers modulo $p$, which keeps frozen bytes and checked field elements in one-to-one correspondence.

\textbf{Three lanes for the 61-bit field.}
The Mersenne modulus lets the GPU compute exact products of canonical elements with native 64-bit multiplies and two folding steps (\cref{sec:arch-field}), so all $1+t$ lanes of authenticated linear algebra stay on standard integer datapaths. A single 61-bit tag is too weak: for $Q\le 2^{32}$ checked batches, $t=3$ is the smallest lane count whose statistical term $Q(2/p)^t$ falls below $2^{-128}$ (\cref{sec:sec-concrete}, \cref{tab:lanes}), and every lane is paid for in GPU work (\cref{sec:eval-ablation}).

\textbf{Field selection versus $\Z_{2^{64}}$.}
The soundness proof divides by a nonzero error coordinate and therefore needs a field (\cref{sec:sec-concrete}, \cref{app:proof-auth}). The encoding is thus not bit-identical to the $\Z_{2^{64}}$ arithmetic of FSS-compiled models~\cite{gupta2023sigma,jawalkar2024orca}. \sys fixes the numeric contract in the manifest (fixed-point scale, signed encoding, exact toward-zero truncation of \cref{sec:proto-gates}, magnitudes below $p/2$); a deployment that needs ring semantics must use a SPDZ2k-type backend~\cite{cramer2018spdz2k} with its own proof.

\subsection{Session Setup and Joint Randomness}\label{sec:proto-setup}

\textbf{Session setup.}
Each module authenticates its peer's identity, code and configuration measurement, role and freshness contribution. An authenticated key establishment then binds the ordered pair of identities, the roles 0 and 1, a fresh session identifier, the epoch, the manifest hash $\Hash(\mu)$, the field and tag parameters and the protocol version, and yields the session key $K$ known only to $T_0$ and $T_1$. Any failure aborts. We treat this step as an explicit setup assumption with error $\eps{setup}$ (\cref{thm:main}) and propose no new attestation protocol; attestation certifies what runs inside a module, not what the GPU computes (\cref{sec:bg-dfa}).

\textbf{Labels and joint draws.}
All module randomness is derived as $\PRF_K(\ell)$, or with an equivalent key derivation function (KDF)~\cite{krawczyk2010hkdf}, where $\ell$ is an unambiguous, length-prefixed encoding of $(\mathsf{session},\mathsf{epoch},\mathsf{gate},\mathsf{purpose},\mathsf{logical\_draw},\mathsf{block})$. A \emph{logical joint draw} uses the same label in both modules and thus yields the same value in $T_0$ and $T_1$ without communication; this is how both modules obtain identical DPF roots, correction words, masks, triples, pads and challenges. Distinct roots, pads, triple components, challenge lanes and rejection-sampling retries receive distinct labels. Two rules are easy to get wrong. The local role must \emph{not} be appended to a joint draw, since $T_0$ and $T_1$ would then generate two unrelated DPF instances whose outputs no longer add up; and the logical root-0 and root-1 draws of one DPF \emph{must} be different draws. Since labels name session, epoch and gate, the ledger detects replayed or misrouted material (S4, S5). A uniform element of $\Fp$ is obtained by rejection sampling: take 61 pseudorandom bits and reject the single value $p$, which occurs with probability $2^{-61}$; a bounded implementation aborts when its retry budget is exhausted instead of introducing a silent bias.

\textbf{Authenticated correlations.}
To share a protected value $v$, the joint sampler runs \AuthShare{}: it samples $v_0\getsr\Fp$ and sets $v_1=v-v_0$, independently samples $\gamma_{0,j}\getsr\Fp$ and sets $\gamma_{1,j}=\Delta_j v-\gamma_{0,j}$ for $j\in[t]$, and gives party $b$ only its tuple $\authb{v}{b}$. A multiplication triple samples $a,b\getsr\Fp$, computes $c=ab$ \emph{correctly inside the protected boundary} and authenticates all three values independently (S7); a matrix triple satisfies the actual equality $C=AB$ for its declared dimensions. Because both modules derive the same correlations from joint labels, preprocessing crosses no network link.

\subsection{One-Time Ledger and Release Guard}\label{sec:proto-ledger}

\begin{figure}[t]
    \centering
    \includegraphics[width=\linewidth]{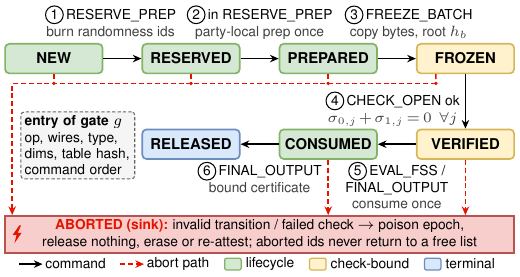}
    \caption{Per-gate lifecycle enforced by the one-time ledger of $T_b$. Solid arrows are the only permitted transitions, each triggered by the named module command (\cref{tab:commands}). An accepted \cmd{CHECK_OPEN} only sets \lstate{VERIFIED}; \cmd{EVAL_FSS} consumes the opening of an FSS or direct gate, and an output gate, which runs no \cmd{EVAL_FSS}, is moved by \cmd{FINAL_OUTPUT} through \lstate{CONSUMED} to \lstate{RELEASED}. Any other transition or any failed check moves the gate to \lstate{ABORTED}, poisons the epoch and keeps the release guard closed.}
    \label{fig:ledger}
\end{figure}

\textbf{Ledger states.}
Each module keeps an authoritative protected ledger that binds every gate to its operation, predecessor and output wires, numeric type, table or model hash, public dimensions and allowed command sequence. A gate transaction moves only along
\lstate{NEW} $\rightarrow$ \lstate{RESERVED} $\rightarrow$ \lstate{PREPARED} $\rightarrow$ \lstate{FROZEN} $\rightarrow$ \lstate{VERIFIED} $\rightarrow$ \lstate{CONSUMED} $\rightarrow$ \lstate{RELEASED}
(\cref{fig:ledger}). Reservation burns the gate's randomness identifiers \emph{before} any preprocessing share is issued, so no label is ever issued twice. \lstate{PREPARED} lets the module issue the authorized party-local preprocessing exactly once; it is not permission to release an unchecked result. \lstate{CONSUMED} records the one-time use of verified openings and preprocessing, and \lstate{RELEASED} records that the result left the module under a certificate. Any invalid transition or failed check moves the gate to \lstate{ABORTED} and poisons the whole epoch: outstanding output responses are gated, no new challenge is derived, and aborted reservations never return to a free list.

\textbf{Freeze semantics.}
Before any checking challenge is drawn, $T_b$ copies the exact submitted value and tag bytes into protected storage, or spills them into AEAD-protected immutable pages~\cite{dworkin2007gcm} bound to sequence numbers, exact lengths, gate identifiers and a protected root. It hashes the complete descriptor and these bytes into a frozen root $h_b$, and the two modules agree on the ordered pair $(h_0,h_1)$. A host's claim that ``this buffer has hash $h$'' is not a freeze, and neither is DMA completion. After challenge generation, $T_b$ rereads only authenticated frozen pages and never verifies one buffer but consumes a later, host-modified one. This closes the time-of-check-to-time-of-use window (S3) that exists whenever checked data live in host or GPU memory.

\textbf{Release certificates and the guard.}
A release certificate binds the session, epoch, gate, recipients, output encoding or hash and terminal status. A missing peer acknowledgement can cause an abort, but it never authorizes output under a different transcript. The release guard is a small state machine in front of the module's output path (\cref{sec:arch-hw}); its enabling inputs originate from the protected checkers and the ledger, never from host-writable registers, so a host cannot open it by writing a status bit (S6). Resetting the guard is not a secure session reset.

\textbf{Reboot and rollback.}
After a reboot, either a separately justified trusted monotonic service preserves freshness~\cite{parno2011memoir,matetic2017rote}, or the old master key is irrecoverably erased and a freshly attested session is required. A host-controlled counter file is not anti-rollback protection, because the host can restore an old copy and replay old masks (S4).

\subsection{Malicious Generate-and-Consume}\label{sec:proto-gac}

\textbf{From semi-honest to malicious GaC.}
\Cref{alg:gac} interleaves \cref{alg:gen,alg:eval} with the three separate states of \cref{sec:analysis-dpf}: two generator states on the target path, started from $t_0=0$ and $t_1=1$, and an evaluator state on the query path, started from its own root seed and the fixed control bit $b$, never from a later generator state. Each correction word is computed from the expanded children of the two generator seeds and the target bit $\alpha_i$, applied to the generator states under their old bits $t_c$ and to both evaluator children under the evaluator's old bit $\tilde t$, and erased only after that use; $t\cdot s$ denotes a conditional XOR. The hardware keeps the two state sets in distinct contexts (\cref{sec:arch-pipe}), and the per-level work, which still duplicates generation in both modules, is accounted in \cref{tab:gatecost}.

\textbf{Inputs and outputs.}
The query $x$ is read from a verified masked opening of the same gate, never from unchecked host data: for instance the opened $c$ of an interval gate, whose prefix DPFs are queried at $c$ (\cref{alg:interval,alg:dcf}), whereas in a LUT gate the opened $u$ only permutes the table index while $\EvalAll$ visits every $a$ (\cref{alg:lut}). The target $\alpha$ and payload $\beta$ come from the gate's protected preprocessing record, so they are joint draws or joint functions of joint draws (\cref{sec:proto-refactor}). By \cref{thm:dpf-correct,thm:gac}, the raw outputs satisfy $\tilde y_0+\tilde y_1=\beta\cdot\one[x=\alpha]$, but $\tilde y_b$ never leaves the module: the epilogue lines of \cref{alg:gac} replace it with a freshly authenticated tuple. Since neither keys nor raw shares are exported, \cref{thm:main} needs no DPF key-privacy term. Gates that combine several DPF evaluations (interval prefixes, LUT and PIR scans) run only the (Gen) and (Eval) lines per constituent DPF and apply the ledger lines and one epilogue to their combined raw share, for PIR only its value pad (\cref{app:proof-dpf}). For full-domain uses such as LUT gates and PIR, the evaluator lines become the depth-first $\EvalAll$ traversal of \cref{alg:eval}; each $\CW^{(i)}$ then stays in protected registers until its last evaluator use, which costs $O(n)$ words and is covered by the same node-by-node coupling.

\begin{algorithm}[t]
\algstretch
\small
\caption{Malicious generate-and-consume with fresh authentication (\algname{GaC-Auth}) inside $T_b$}
\label{alg:gac}
\begin{algorithmic}[1]
\Require role $b\in\{0,1\}$; gate label $\ell$; public query $x\in\bits{n}$ from the gate's frozen, verified opening; target $\alpha\in\bits{n}$ and payload $\beta\in\Fp$ from the protected record of $\ell$; session key $K$; global keys $\Delta_1,\dots,\Delta_t$
\Ensure $\authb{y}{b}=(y_b,\gamma_{b,1},\dots,\gamma_{b,t})$ with $y_0+y_1=\beta\cdot\one[x=\alpha]$ and $\gamma_{0,j}+\gamma_{1,j}=\Delta_j(y_0+y_1)$
\If{gate $\ell$ is not \lstate{VERIFIED} in the ledger}
  \State \Abort{} and poison the epoch
\EndIf
\LineComment{Joint draws: identical labels in $T_0$ and $T_1$, no role appended}
\State $s_0\gets\PRF_K(\ell,\mathsf{root}_0)$;\ $s_1\gets\PRF_K(\ell,\mathsf{root}_1)$;\ $t_0\gets 0$;\ $t_1\gets 1$
\State $(\tilde s,\tilde t)\gets(s_b,b)$ \Comment{separate evaluator}
\For{$i=1$ \textbf{to} $n$}
  \LineComment{(Gen) both roots on the target path; updates use the old $t_c$}
  \State $(s_c^L,t_c^L,s_c^R,t_c^R)\gets\PRG(s_c)$ for $c\in\{0,1\}$
  \State $a\gets\alpha_i$;\ $\keep\gets a$;\ $\lose\gets 1-a$ \Comment{$L\equiv 0$, $R\equiv 1$}
  \State $\sCW\gets s_0^{\lose}\oplus s_1^{\lose}$
  \State $\tCW{L}\gets t_0^L\oplus t_1^L\oplus a\oplus 1$;\ \ $\tCW{R}\gets t_0^R\oplus t_1^R\oplus a$
  \State $(s_c,t_c)\gets\big(s_c^{\keep}\oplus t_c\cdot\sCW,\ t_c^{\keep}\oplus t_c\cdot\tCW{\keep}\big)$, $c\in\{0,1\}$
  \LineComment{(Eval) own query path; corrections use the old $\tilde t$}
  \State $(\hat s^L,\hat t^L,\hat s^R,\hat t^R)\gets\PRG(\tilde s)$
  \State $(\hat s^L,\hat t^L)\gets\big(\hat s^L\oplus\tilde t\cdot\sCW,\ \hat t^L\oplus\tilde t\cdot\tCW{L}\big)$
  \State $(\hat s^R,\hat t^R)\gets\big(\hat s^R\oplus\tilde t\cdot\sCW,\ \hat t^R\oplus\tilde t\cdot\tCW{R}\big)$
  \State $(\tilde s,\tilde t)\gets(\hat s^{x_i},\hat t^{x_i})$
  \State erase $\CW^{(i)}=(\sCW,\tCW{L},\tCW{R})$ \Comment{after its evaluator use}
\EndFor
\LineComment{(Gen) output correction from the final generator states}
\State $u_c\gets\Conv(s_c)$ for $c\in\{0,1\}$;\ \ $\CWout\gets(-1)^{t_1}(\beta-u_0+u_1)$
\LineComment{(Eval) raw share, never exported}
\State $\tilde y_b\gets(-1)^b\big(\Conv(\tilde s)+\tilde t\cdot\CWout\big)$; erase $s_c$, $t_c$, $\CWout$
\LineComment{(Epilogue) fresh uniform pads, distinct labels per output wire}
\State $\zeta\gets\PRF_K(\ell,\mathsf{pad}_0)$;\ \ $\kappa_j\gets\PRF_K(\ell,\mathsf{pad}_j)$ for $j\in[t]$
\State $y_b\gets\tilde y_b+(-1)^b\zeta$
\State $\gamma_{b,j}\gets\Delta_j\,y_b+(-1)^b\kappa_j$ for $j\in[t]$
\State mark $\ell$ \lstate{CONSUMED}; erase $\tilde y_b$, $\zeta$, $\kappa_1,\dots,\kappa_t$, $(\tilde s,\tilde t)$
\State \Return $\authb{y}{b}$ to the output FIFO behind the release guard
\end{algorithmic}
\end{algorithm}

\subsection{Fresh Authentication Epilogue}\label{sec:proto-epilogue}

\textbf{Construction.}
Let $\tilde y_0,\tilde y_1$ be raw local shares computed inside the two protected evaluators, with $\tilde y_0+\tilde y_1$ equal to the intended gate value. Both modules derive fresh common pads $\zeta,\kappa_1,\dots,\kappa_t\getsr\Fp$ from joint labels that are distinct for every output wire, and $T_b$ returns
\begin{equation*}
y_b=\tilde y_b+(-1)^b\zeta,\qquad \gamma_{b,j}=\Delta_j\,y_b+(-1)^b\kappa_j\quad (j\in[t]).
\end{equation*}
The pads cancel in $y_0+y_1$ and in $\gamma_{0,j}+\gamma_{1,j}=\Delta_j(y_0+y_1)$, so the result is an authenticated sharing of $\tilde y_0+\tilde y_1$, and each party's tuple is uniform on $\Fp^{1+t}$ whatever the gate value and the global keys (\cref{thm:epilogue}). Fresh pads make this hold jointly across outputs; correlations created by later arithmetic are handled by the circuit simulation of \cref{thm:main}.

\textbf{Tag tree replacement.}
Instead of $t$ further DPF trees with payloads $\Delta_j\beta$, the epilogue exploits the linearity of the MAC. Because $T_b$ knows the global $\Delta_j$, the tag of the reconstructed value is the sum of the per-party products $\Delta_j y_b$; the pad $\kappa_j$ hides $\Delta_j y_b$ from host $b$, which knows $y_b$ and could otherwise divide it out. The pad $\zeta$ rerandomizes the raw DPF share, so no host ever sees $\tilde y_b$. The cost is one raw DPF computation, three field multiplications and four fresh pads per party and output, and the output is a 32~B tuple.

\textbf{API restriction.}
The epilogue accepts only results computed inside the module, namely the DPF, DCF, LUT, interval and direct gates of \cref{sec:proto-gates,sec:proto-direct} and the PIR reductions of \cref{sec:proto-pir}, or results from an explicitly verified backend such as the frozen-operand Freivalds check of \cref{sec:proto-freivalds}. There is no \cmd{authenticate(host_value)} entry point. With one, the module would become a MAC oracle: a malicious host could obtain valid tags on any value it wishes to open, and \CheckOpen{} would accept the forgery with probability one (S2). Since the trusted boundary excludes the GPU SMs and memory, a malicious-secure extension must never compute an unchecked result on the GPU and ask the module to authenticate it.

\subsection{Freeze-Before-Challenge CheckOpen}\label{sec:proto-check}

\begin{algorithm}[t]
\algstretch
\small
\caption{Freeze-before-challenge authenticated opening (\CheckOpen) inside $T_b$}
\label{alg:checkopen}
\begin{algorithmic}[1]
\Require descriptor and nonempty ordered list of expected wires $(\nu_i)_{i\in[m]}$; host $b$'s tuples $(x_{b,i},\gamma_{b,j,i})_{i\in[m],\,j\in[t]}$; key shares $(\Delta_{b,j})_{j\in[t]}$ from the protected key store of $T_b$ (never from host input)
\Ensure verified openings $(z_i)_{i\in[m]}$ with bound certificates, or \Abort
\If{metadata, field encoding, gate order or length mismatch the ledger}
  \State \Abort{} and poison the epoch
\EndIf
\LineComment{Step 1: freeze, then agree on both roots}
\State copy the exact bytes into protected SRAM or sealed AEAD pages; $h_b\gets\Hash(\text{descriptor},\text{frozen bytes})$; mark \lstate{FROZEN}
\State exchange $(h_b,\mathit{cid})$ with $T_{1-b}$; \Abort{} unless both hold the same $(h_0,h_1)$ and $\mathit{cid}$
\LineComment{Step 2: reconstruct candidates inside the modules}
\State exchange the frozen value shares $(x_{b,i})_{i\in[m]}$ with $T_{1-b}$
\State $z_i\gets x_{0,i}+x_{1,i}$ for $i\in[m]$ \Comment{not exposed to the host}
\LineComment{Step 3: internal challenges, only after both freezes}
\State $\rho_{j,i}\gets\PRF_K(\mathsf{chk},h_0,h_1,\mathit{cid},j,i)$ for $j\in[t]$, $i\in[m]$
\LineComment{Step 4: per-lane residuals from frozen pages only}
\State $\sigma_{b,j}\gets\sum_{i=1}^{m}\rho_{j,i}\,(\gamma_{b,j,i}-\Delta_{b,j}z_i)$ for $j\in[t]$
\State exchange $(\sigma_{b,j})_{j\in[t]}$ with $T_{1-b}$
\LineComment{Step 5: decide and release}
\If{$\sigma_{0,j}+\sigma_{1,j}=0$ for every $j\in[t]$}
  \State agree on output digest and status; mark \lstate{VERIFIED}; record the checked material as used
  \State \Return $(z_i)_{i\in[m]}$ with certificates bound to session, epoch, gates and output hash
\Else
  \State release no candidate value; poison the epoch; return one generic failure status
\EndIf
\end{algorithmic}
\end{algorithm}

\begin{figure}[t]
    \centering
    \includegraphics[width=\linewidth]{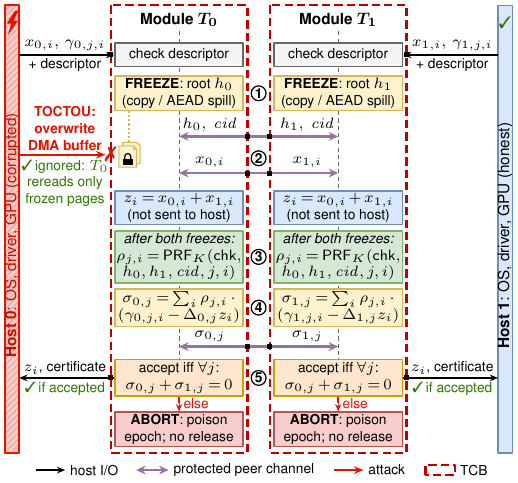}
    \caption{Freeze-before-challenge \CheckOpen{} between the two modules. Challenges $\rho_{j,i}$ are derived only after both frozen roots are agreed, and $T_b$ rereads only its frozen pages, so a host that overwrites its DMA buffer after the freeze changes nothing that is checked or consumed.}
    \label{fig:checkopen}
\end{figure}

\textbf{Role.}
Openings are the only places where a value produced by the untrusted GPU becomes public or feeds a protected gate: the masked differences of every multiplication, the masked inputs of every FSS gate and the masked outputs. \CheckOpen{} (\cref{alg:checkopen}, \cref{fig:checkopen}) is a SPDZ-style batched MAC check~\cite{damgard2012spdz,damgard2013practical} rebuilt for a verifier that reads its inputs from untrusted memory (S1, S3). All peer messages use the protected channel, and only an accepted batch releases $z$ with bound certificates.

\textbf{Independent challenge lanes.}
Let $\bm{e}$ be the value-error vector of a batch and $\bm{d}_j$ its tag-error vector in lane $j$, i.e., the differences between the submitted and the prescribed tuples. The freeze fixes $(\bm{e},\bm{d}_j)$ before any $\rho_{j,i}$ exists, so a host cannot place its errors in the kernel of a challenge (\cref{app:proof-auth}). Each lane has its own challenge vector $\bm{\rho}_j$. The complete residual of lane $j$ is $\inner{\bm{\rho}_j}{\bm{d}_j-\Delta_j\bm{e}}$, which for $\bm{e}\neq\bm{0}$ vanishes with probability below $2/p$; independence across lanes multiplies these bounds to $(2/p)^t$ (\cref{thm:mac-sound}). A shared vector collapses the product to one lane, and tag-only errors escape with probability $1/p$ without changing an opened value (\cref{sec:sec-concrete}, \cref{app:proof-auth}).

\textbf{Key hiding and failure handling.}
The bound relies on the hidden-key invariant of \cref{lem:hidden-key}: until the first accepted wrong opening, the honest key shares stay uniform given the corrupt host's view. \CheckOpen{} preserves it in three ways. Candidates, challenges and residuals never leave the modules. The zero tests of all lanes are aggregated into one generic failure status that exposes no residual, challenge, key or data-dependent diagnostic. A failure poisons the epoch and forbids retries with the same material, which rules out an adaptive key-learning sequence within an epoch.

\textbf{Traffic and batching.}
Per batch of $m$ elements, the modules exchange two vectors of $m$ field values and two vectors of $t$ residuals plus authenticated metadata. Tags never cross the peer link, but every local tag must be read by its module, so the host-to-module transfer of 32~B per element can remain the bottleneck (\cref{sec:arch-dse}). Because the value shares are frozen before they are sent, root agreement and value exchange can share one message, so each batch incurs two peer round trips per batch. Only independent openings are batched: a dependent nonlinear gate never consumes an unchecked masked opening, and a checkpoint precedes every input-dependent nonlinear output and every final release, since blindly delayed checking is not justified. Larger batches amortize metadata and round trips but increase staging memory and first-result latency (costs in \cref{sec:eval-check}).

\subsection{Inputs, Arithmetic, and Outputs}\label{sec:proto-arith}

\begin{figure*}[t]
\centering
\footnotesize
\setlength{\fboxsep}{4pt}
\fbox{%
\begin{minipage}[t]{0.472\textwidth}
\textbf{Input} $\Pi_{\mathsf{Input}}$\hfill\textit{input $x$ of endpoint $O$ for wire $\nu$}
\begin{enumerate}[label=\arabic*:,leftmargin=1.6em,labelsep=0.4em,nosep]
\item $T_0,T_1$: reserve label $\ell_\nu$; joint draw $r\getsr\Fp$ and $\authb{r}{b}\gets\AuthShare(r)$; issue $\authb{r}{b}$ to host $b$ once.
\item $T_0,T_1$: send the complete $r$ only to $O$ over its authenticated confidential channel.
\item $O$: send the canonical, authenticated $\varepsilon=x-r$ to $T_0$ and $T_1$.
\item $T_0,T_1$: confirm over the protected channel that both hold the same $\varepsilon$ and bind it to $\nu$; compute $x^*=r+\varepsilon$ and \Abort{} unless $x^*\in\mathsf{In}_\nu$. $T_b$ then hands the bound $\varepsilon$ to host $b$.
\item host $b$: $x_b\gets r_b+\one[b=0]\,\varepsilon$;\ \ $\gamma_{b,j}\gets\gamma_{r,b,j}+\Delta_{b,j}\,\varepsilon$ for $j\in[t]$.
\end{enumerate}
\textit{Extraction:} a corrupt owner's input is $x^*=r+\varepsilon$.
\textit{Validation:} at the public boundaries declared in $\mu$, after every input gate and before the first output opening, $T_0,T_1$ evaluate every predicate of $\mathcal{V}$ on the bound inputs and \Abort{} if one fails.
\textit{Committed inputs:} a protected loader verifies the authorized commitment over the frozen canonical stream, then runs \AuthShare{} on the exactly decoded values.

\vspace{3pt}\hrule\vspace{3pt}
\textbf{Linear} $\Pi_{\mathsf{Lin}}$\hfill\textit{local, on host $b$'s GPU}
\begin{enumerate}[label=\arabic*:,leftmargin=1.6em,labelsep=0.4em,nosep]
\item $\auth{x}+\auth{y}$: add the value shares and each of the $t$ tag shares.
\item $c\cdot\auth{x}$ for public $c$: multiply the value share and every tag share by $c$.
\item $\auth{x}+c$ for public $c$: $x_0\gets x_0+c$ (party~0 only); $\gamma_{b,j}\gets\gamma_{b,j}+\Delta_{b,j}\,c$ for both $b$ and all $j\in[t]$.
\item Tensors: apply steps 1 to 3 elementwise to the value array and to each of the $t$ tag arrays.
\end{enumerate}

\vspace{3pt}\hrule\vspace{3pt}
\textbf{Output} $\Pi_{\mathsf{Out}}$\hfill\textit{private output $y$ for client $O$}
\begin{enumerate}[label=\arabic*:,leftmargin=1.6em,labelsep=0.4em,nosep]
\item $T_0,T_1$: joint draw a fresh $\auth{o}$ as in step 1 of $\Pi_{\mathsf{Input}}$; send the complete $o$ only to $O$.
\item host $b$: $\auth{\upsilon}\gets\auth{y}+\auth{o}$ with $\Pi_{\mathsf{Lin}}$.
\item $\upsilon\gets\CheckOpen(\auth{\upsilon})$; the release guard opens only for a \mbox{\lstate{VERIFIED}} batch.
\item $T_b$: send $\upsilon$ and its release certificate (session, epoch, gate, recipients, output hash, status) to $O$.
\item $O$: verify both certificates for the same $\upsilon$ and output $y=\upsilon-o$; otherwise abort.
\end{enumerate}
\textit{Public outputs:} $y\gets\CheckOpen(\auth{y})$.
\textit{Never} hand $O$ both parties' tags: $\gamma_{0,j}+\gamma_{1,j}=\Delta_j y$ with a known $y\neq 0$ reveals $\Delta_j$.
\end{minipage}%
\hspace{0.01\textwidth}\vrule\hspace{0.01\textwidth}%
\begin{minipage}[t]{0.472\textwidth}
\textbf{Multiplication} $\Pi_{\mathsf{Mult}}$\hfill\textit{authenticated Beaver triples}

\textit{Scalar} $\auth{z}=\auth{x}\cdot\auth{y}$:
\begin{enumerate}[label=\arabic*:,leftmargin=1.6em,labelsep=0.4em,nosep]
\item $T_0,T_1$: reserve a fresh triple $(\auth{a},\auth{b},\auth{c})$ from joint draws, with $c=ab$ computed inside the protected boundary; issue party $b$'s tuples to host $b$ once.
\item host $b$: $\auth{d}\gets\auth{x}-\auth{a}$ and $\auth{e}\gets\auth{y}-\auth{b}$ with $\Pi_{\mathsf{Lin}}$.
\item $(d,e)\gets\CheckOpen(\auth{d},\auth{e})$, batched with the independent openings of the layer.
\item host $b$: $z_b\gets c_b+d\,b_b+e\,a_b+\one[b=0]\,de$.
\item host $b$: $\gamma_{z,b,j}\gets\gamma_{c,b,j}+d\,\gamma_{b,b,j}+e\,\gamma_{a,b,j}+\Delta_{b,j}\,de$ for $j\in[t]$.
\end{enumerate}
\textit{Matrix} $\auth{Z}=\auth{X}\auth{Y}$ with declared dimensions:
\begin{enumerate}[label=\arabic*:,leftmargin=1.6em,labelsep=0.4em,nosep,start=6]
\item $T_0,T_1$: reserve $(\auth{A},\auth{B},\auth{C})$ with $C=AB$ computed inside the protected boundary; never expose both $A$ and $B$ to a GPU.
\item $D\gets\CheckOpen(\auth{X}-\auth{A})$;\ \ $E\gets\CheckOpen(\auth{Y}-\auth{B})$.
\item host $b$: $Z_b\gets C_b+D\,B_b+A_b\,E+\one[b=0]\,DE$.
\item host $b$: $\Gamma_{Z,b,j}\gets\Gamma_{C,b,j}+D\,\Gamma_{B,b,j}+\Gamma_{A,b,j}\,E+\Delta_{b,j}\,DE$ for $j\in[t]$, preserving the order of every matrix factor.
\end{enumerate}
\textit{Notation:} in $b_b$ and $\gamma_{b,b,j}$ the first $b$ names the triple's second component and the subscript names the party (implementations use named fields); $\Gamma_{\cdot,b,j}$ is party $b$'s lane-$j$ tag matrix.

\vspace{3pt}\hrule\vspace{3pt}
\textbf{Invariants shared by all four protocols}
\begin{itemize}[leftmargin=1.1em,nosep]
\item Every mask and triple is a one-time joint draw bound to session, epoch, gate and purpose, and is issued only from \mbox{\lstate{PREPARED}} (\cref{fig:ledger}).
\item Every opening is a \CheckOpen{} call on frozen bytes (\cref{alg:checkopen}); no opened value feeds a protected gate or a release before it is \mbox{\lstate{VERIFIED}}.
\item Outputs of host $b$'s GPU are candidate shares; no module ever authenticates them.
\end{itemize}
\end{minipage}}
\caption{Authenticated input, linear, multiplication and output protocols of \sys, written for party $P_b$. Steps labeled $T_0,T_1$ or $T_b$ run inside the protected modules; steps labeled host $b$ run on the untrusted GPU and yield only candidate shares until the next \CheckOpen. With honest messages the protocols realize their operations exactly; otherwise they select an allowed input, preserve the computation or abort, except with the probability bounded in \cref{thm:mac-sound} (\cref{thm:arith}).}
\label{fig:arith}
\end{figure*}

\Cref{fig:arith} specifies the four protocols that carry values between input owners, GPUs, modules and recipients; three points of rationale are not visible in the box.

\textbf{Inputs.}
The difference $\varepsilon$ is uniform for a uniform $r$, so the hosts learn nothing about $x$. Both modules bind one $\varepsilon$, so a corrupt owner cannot feed the two parties different inputs, and its input is extracted as $x^*=r+\varepsilon$ (S9, \cref{def:security}). Canonical encoding alone does not place $x^*$ in the declared interface; since the modules generated $r$, they compute $x^*$ and abort unless $x^*\in\mathsf{In}_\nu$. The validity predicates $\mathcal{V}$ of $\mu$ are evaluated inside the modules at public boundaries that the schedule of $\mu$ places after every input gate and before the first output opening; their accept or abort outcome is the $\mathsf{status}$ component of $\Leak_C$ (\cref{def:leakage}), and the simulator emulates the same abort (\cref{thm:main}). For committed inputs such as preapproved model weights, a host-declared hash does not suffice; the loader hashes exactly the bytes it shares. The modules thus see every input, which \ProfT permits (\cref{sec:threat-model}).

\textbf{Arithmetic.}
Multiplication follows Beaver~\cite{beaver1991efficient} with authenticated triples. An authenticated matrix product costs the GPU $2(1+t)$ plain products plus the public $DE$, which is why authenticated linear algebra, rather than DPF evaluation, dominates the malicious cost (\cref{sec:analysis-cost}).

\textbf{Outputs.}
The recipient accepts $\upsilon$ only with matching certificates from both modules, so integrity toward a client comes from these certificates, never from MAC tags.


\subsection{Authenticated Nonlinear Gates}\label{sec:proto-gates}

Nonlinear layers of FSS-compiled models reduce to lookups in small public tables and piecewise functions selected by interval tests~\cite{boyle2021function,gupta2023sigma,jawalkar2024orca}. In the semi-honest DFA design, both parties evaluate a DPF or DCF key on a masked input and pass the raw output shares onward~\cite{dfa2026trusted}. A malicious host can shift the masked input before it selects the DPF path, and it can modify the raw outputs, which carry no MAC. \sys closes both gaps with one template (\cref{fig:gates}): (1)~mask the authenticated input with joint protected randomness (\cref{sec:proto-setup}); (2)~open the masked value with \CheckOpen{} (\cref{alg:checkopen}), so the value that selects the path is verified and frozen; (3)~evaluate the FSS gate inside $T_b$ with the (Gen) and (Eval) steps of generate-and-consume (\cref{alg:gac}); and (4)~authenticate and reshare the gate's combined raw share with one epilogue (\cref{sec:proto-epilogue}). The host sees only a uniform masked opening and fresh authenticated tuples, and no host-supplied value is ever authenticated, which defeats scenario S2 (\cref{tab:scenarios}).

\begin{figure*}[t]
    \centering
    \includegraphics[width=\textwidth]{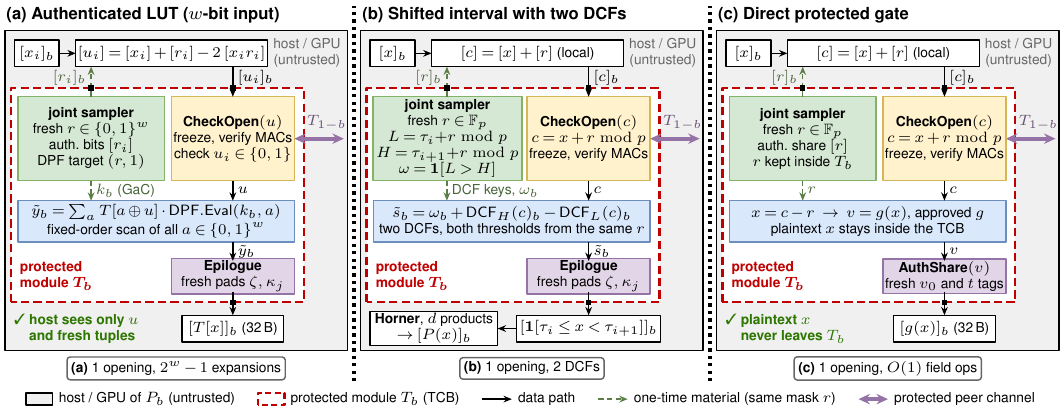}
    \caption{Authenticated nonlinear gates in \sys. (a)~Small XOR-masked LUT (\cref{alg:lut}): the XOR-masked bit vector $u$ is opened by \CheckOpen{} and $T_b$ scans the full public domain in a fixed order. (b)~Shifted-interval indicator and Horner polynomial (\cref{alg:interval}): one field opening $c=x+r$ serves both translated thresholds $L$ and $H$, which derive from the same mask. (c)~Direct protected gate (\cref{sec:proto-direct}): the module recovers $x=c-r$ inside its boundary and returns a fresh sharing of $g(x)$. In all three cases the host submits only authenticated shares and receives only freshly authenticated outputs; raw FSS shares, DPF keys and, in (c), the plaintext $x$ never leave $T_b$.}
    \label{fig:gates}
\end{figure*}

\begin{algorithm}[t]
\algstretch
\caption{Authenticated small XOR-masked LUT (\algname{LUT-Auth}) for party $b$}
\label{alg:lut}
\begin{algorithmic}[1]
\Require authenticated bits $\authb{x_1}{b},\dots,\authb{x_w}{b}$ of an unsigned $w$-bit $x$, certified Boolean; approved public table $T\in\Fp^{2^w}$ bound in $\mu$ by its hash; gate identifier $\mathit{gid}$; role $b$
\Ensure $\authb{y}{b}$ with $y=T[x]$
\State $T_b$: \cmd{RESERVE_PREP} for $\mathit{gid}$; \Abort{} on any descriptor or table-hash mismatch
\LineComment{Joint preparation inside $T_b$ (same $r$ in $T_0$ and $T_1$)}
\State $r\getsr\bits{w}$ by one joint draw with label $(\mathit{gid},\textsf{mask})$
\State issue $\authb{r_i}{b}\gets\AuthShare(r_i)$ to host $b$ once, for $i=1,\dots,w$
\State bind the joint DPF generator of $(\alpha,\beta)=(r,1)$ to $\mathit{gid}$ \Comment{never resampled}
\LineComment{Masked opening (host $b$ with $T_b$)}
\For{$i=1$ \textbf{to} $w$} \Comment{independent: one batch}
  \State $\authb{u_i}{b}\gets\authb{x_i}{b}+\authb{r_i}{b}-2\,\authb{x_ir_i}{b}$ \Comment{$\Pi_{\mathsf{Mult}}$}
\EndFor
\State $u\gets\CheckOpen(\auth{u_1},\dots,\auth{u_w})$ \Comment{$u=x\oplus r$}
\If{some $u_i\notin\{0,1\}$ or $T_0,T_1$ hold different $u$}
  \State \Abort
\EndIf
\LineComment{Full-domain scan inside $T_b$ in fixed public order; $k_b$ is streamed by GaC and never materialized}
\State $\tilde y_b\gets 0$
\For{$a\in\bits{w}$ in lexicographic order} \Comment{$\EvalAll$, depth-first}
  \State $\tilde y_b\gets\tilde y_b+T[a\oplus u]\cdot\DPF.\Eval(k_b,a)$
\EndFor
\State $\authb{y}{b}\gets\Epilogue(\tilde y_b)$ \Comment{\cref{sec:proto-epilogue}}
\State \Return $\authb{y}{b}$
\end{algorithmic}
\end{algorithm}

\textbf{Small XOR-masked LUT.}
\Cref{alg:lut} evaluates an approved public table $T$ with $2^w$ field entries on a $w$-bit input given as authenticated, certified-Boolean bits. The sampler draws a fresh uniform $r\in\bits{w}$ once and uses it twice: for the authenticated bits $\auth{r_i}$ and as the target of the joint DPF for $(r,1)$. The parties open $u=x\oplus r$ in one \CheckOpen{} batch, both modules reject a non-Boolean or mismatching vector, and $T_b$ scans the full domain in a declared fixed order,
\begin{equation*}
    \tilde y_b=\sum_{a\in\bits{w}}T[a\oplus u]\cdot\DPF.\Eval(k_b,a).
\end{equation*}
By \cref{thm:dpf-correct}, the evaluations sum to one only at $a=r$, so $\tilde y_0+\tilde y_1=T[r\oplus u]=T[x]$; since XOR with a fixed $x$ is a bijection on $w$-bit strings, $u$ is uniform given $x$ (\cref{thm:lut}). In $\EvalAll$ mode the $w$ correction words stay in protected registers until their level is fully expanded, so no key is materialized outside $T_b$. Besides the $w$ XOR products (one batched multiplication round when the bits exist; otherwise bit decomposition is charged to the gate), the gate needs one batched opening, $2^w-1$ internal tree expansions and a $2^w$-term field reduction: small tables are practical, full-domain tables over 64-bit inputs are not (\cref{tab:gates}).

\begin{algorithm}[t]
\algstretch
\caption{Authenticated shifted intervals and piecewise polynomial (\algname{Interval-Auth}) for party $b$}
\label{alg:interval}
\begin{algorithmic}[1]
\Require $\authb{x}{b}$; public breakpoints $0=\tau_0<\tau_1<\dots<\tau_m=p$ with $m\ge2$; public coefficients $a_{i,j}\in\Fp$ of $P_i(x)=\sum_{j=0}^{d}a_{i,j}x^j$; gate identifier $\mathit{gid}$; role $b$
\Ensure $\authb{y}{b}$ with $y=P_i(x)$ for the unique $i$ such that $\tau_i\le x<\tau_{i+1}$
\State $T_b$: \cmd{RESERVE_PREP} for $\mathit{gid}$; $r\getsr\Fp$ by one joint draw (rejection sampling); issue $\authb{r}{b}\gets\AuthShare(r)$ once
\State $c\gets\CheckOpen(\auth{x}+\auth{r})$ \Comment{$c=x+r\bmod p$}
\LineComment{Indicators inside $T_b$; every endpoint and wrap bit derives from the same $r$}
\For{$i=0$ \textbf{to} $m-1$}
  \State $L_i\gets(\tau_i+r)\bmod p$;\; $H_i\gets(\tau_{i+1}+r)\bmod p$
  \State $\omega_i\gets\one[L_i>H_i]$ with a sharing $(\omega_{i,0},\omega_{i,1})$
  \State $\tilde s_{i,b}\gets\omega_{i,b}+\DCF_{H_i}(c)_b-\DCF_{L_i}(c)_b$ \Comment{GaC}
  \State $\authb{s_i}{b}\gets\Epilogue(\tilde s_{i,b})$ \Comment{$s_i=\one[\tau_i\le x<\tau_{i+1}]$}
\EndFor
\LineComment{Coefficient selection and Horner evaluation (host $b$, authenticated arithmetic)}
\For{$j=0$ \textbf{to} $d$}
  \State $\authb{\bar a_j}{b}\gets\sum_{i=0}^{m-1}a_{i,j}\,\authb{s_i}{b}$ \Comment{local, $\Pi_{\mathsf{Lin}}$}
\EndFor
\State $\authb{y}{b}\gets\authb{\bar a_d}{b}$
\For{$j=d-1$ \textbf{down to} $0$}
  \State $\authb{y}{b}\gets\authb{y\cdot x}{b}+\authb{\bar a_j}{b}$ \Comment{$\Pi_{\mathsf{Mult}}$}
  \State rescale exactly where $\mu$ specifies (\cref{alg:trunc})
\EndFor
\State \Return $\authb{y}{b}$
\end{algorithmic}
\end{algorithm}

\textbf{Shifted intervals.}
For a proper nonempty public interval $[a,b)\subsetneq[0,p)$ (here $a,b$ are endpoints, not party indices), \cref{alg:interval} opens $c=x+r\bmod p$ for a fresh field mask $r$, and $T_b$ derives from the same $r$
\begin{equation*}
    L=(a+r)\bmod p,\quad H=(b+r)\bmod p,\quad \omega=\one[L>H].
\end{equation*}
With conservative DCFs for thresholds $L$ and $H$ (\cref{alg:dcf}; orientation $\one[c<\theta]$, payload $1$, independent coins) and a sharing of the wrap bit $\omega$, the raw share $\tilde s_b=\omega_b+\DCF_H(c)_b-\DCF_L(c)_b$ goes to the epilogue. Translation by $r$ maps $[a,b)$ onto the circular half-open interval from $L$ to $H$ ($L\neq H$ since $0<b-a<p$): the indicator is $\one[c<H]-\one[c<L]$ if $L<H$ and $1+\one[c<H]-\one[c<L]$ if the interval wraps, so \cref{thm:dcf} gives $\tilde s_0+\tilde s_1=\one[a\le x<b]$ (\cref{thm:interval}). Empty and full intervals are public constants. A partition may share one mask and one opening if every endpoint, wrap bit and auxiliary share is derived jointly from that mask, while distinct DCFs keep independent coins. The conservative DCF evaluates $n(n+1)/2$ path levels per threshold; its cost, and that of a projected $O(n)$ DCF~\cite{boyle2021function,storrier2023grotto} that \cref{thm:dcf} does not cover, are in \cref{tab:gates}.

\textbf{Piecewise polynomials.}
Activation and normalization functions are approximated by low-degree polynomials on a public partition~\cite{gupta2023sigma,pang2024bolt,ma2023secretflow}. The parties select coefficients $\auth{\bar a_j}=\sum_i a_{i,j}\auth{s_i}$ locally ($\Pi_{\mathsf{Lin}}$, \cref{fig:arith}) and apply Horner's rule with $d$ authenticated products, independently of the number $m$ of pieces. Exactly one indicator of a disjoint exhaustive partition is one, so \cref{thm:arith} yields the active polynomial; rescaling is inserted where the manifest specifies it (\cref{alg:trunc}). Approximation error of the approved coefficients is distinct from the cryptographic correctness of \cref{thm:lut,thm:interval}, which concerns the declared finite arithmetic.

\textbf{Exact comparison, bit decomposition and truncation.}
The LUT input, exact threshold comparisons and toward-zero signed truncation need canonical bits, which \cref{alg:bitdec,alg:trunc} (\cref{app:proof-gates}) derive with authenticated arithmetic alone (\cref{thm:bitdec,thm:cmp-trunc})~\cite{damgard2006unconditionally,catrina2010improved,escudero2020edabits}.

\textbf{Exact circuits as fallback.}
Their carry chains cannot be shortened by batching: \cref{alg:bitdec} performs $2n=122$ dependent multiplications for $n=61$, hence $2n+1=123$ sequential verified openings, and exact signed truncation chains two decompositions and a comparator into 307 authenticated multiplications and 309 sequential rounds, {30.9~ms} on LAN (\cref{tab:gates}). At two round trips per verified opening (\cref{sec:proto-check}), the 123-round chain alone would take about {17.2~s} on WAN. A valid one-round candidate therefore always wins under the cost rule below, and exact circuits are used only when no such candidate meets the gate's contract, with their cost charged to any LUT whose input bits do not yet exist.

\subsection{Direct Protected Gates and Cost-Aware Compilation}\label{sec:proto-direct}

\textbf{The direct protected gate.}
After a verified opening $c=x+r$, each module holds $c$ and $r$ and can compute $x=c-r$ inside its boundary, evaluate an approved deterministic gate $g$, and return a fresh authenticated sharing of $g(x)$ through \AuthShare{} (\cref{fig:gates}(c)). This is secure (\cref{thm:arith}) in the same model and with the same output simulation as the FSS gates, since the host again receives a fresh, independently uniform value and tag tuple (\cref{thm:epilogue}), provided the module's observable trace while computing $g$ meets the same simulatability assumption ($\eps{trace}$ in \cref{thm:main}). Concretely, fetching $T[x]$ from host memory or running in $x$-dependent time would leak $x$, so approved direct gates run from protected memory in data-independent time. A direct gate needs one verified opening, two AES blocks and five field multiplications; our model shows {1{,}800}~M gates/s, above every FSS realization in \cref{tab:gates}.

\textbf{Cost-aware selection.}
For each nonlinear gate shape, the compiler outputs a cost tuple (latency, bytes, protected memory, field operations, PRG calls) for the direct gate, a small XOR LUT, an interval or polynomial FSS gate, and a separately proven malicious backend. It rejects candidates that violate the gate's semantic and security contract (the exact finite function, encoding and rescaling fixed by $\mu$ under the profile in force) and selects the least modeled end-to-end cost among the rest, accounting for dependency depth and transport; only independent checks are batched (\cref{sec:proto-check}). One constraint applies: gate shapes that the manifest marks as FSS-portable keep an FSS realization, so that the same compiled circuit runs under \ProfU. The compiler-selected default is therefore not the unconstrained minimum, and direct-only compilation is marginally faster ({3.82~s} versus {3.83~s} for BERT-Base on LAN).

\textbf{FSS in the datapath.}
If the direct gate is at least as cheap under \ProfT, why keep DPF generation and evaluation in the module? (i)~\emph{Compatibility.} Sigma- and Orca-style toolchains~\cite{gupta2023sigma,jawalkar2024orca} compile nonlinear layers into masked-input FSS gates, and \ProfU (\cref{sec:threat-profileu}) runs the circuit through an established malicious backend whose DPF keys the accelerator evaluates as party-local data; there the direct candidate is invalid because the hardware is not trusted with $x$. One DPF datapath serves both profiles and the untrusted-mode DFA workloads. (ii)~\emph{Full-domain scans.} PIR and small-LUT gates are data-independent, fixed-order scans that the pipelined DPF units execute at one iteration per cycle (\cref{sec:arch-hw}), so both profiles reuse one engine; this is a reuse argument, not evidence that the DPF form is the fastest option under \ProfT (\cref{sec:proto-pir}). (iii)~\emph{Short-lived sensitive state.} FSS keeps long-running evaluation state pseudorandom (seeds and correction words), and index- or input-dependent state short-lived. We treat (iii) only as defense in depth against physical and microarchitectural threats that \ProfT excludes; no theorem depends on it, and since each module holds $c$ and $r$, FSS gains no formal confidentiality over the direct gate (\cref{sec:threat-model}). The evaluation reports direct-only, FSS-only and compiler-selected default configurations side by side (\cref{sec:eval-ablation}).

\subsection{Distribution-Preserving Refactoring}\label{sec:proto-refactor}

DFA decouples DPF generation from protocol preprocessing by sampling the target directly as a uniform value and recovering the mask $r$ through $g^{-1}$ when the protocol's transformation $g$ is a bijection~\cite{dfa2026trusted}. \sys draws every target from the joint sampler in the same way, with two refinements of the semi-honest argument. Since DPF key privacy holds for arbitrary targets (\cref{sec:bg-fss}), uniformity is required by the application, which publishes $c=x+r$ or $u=x\oplus r$; a refactoring must therefore preserve the \emph{joint} distribution of the preprocessing, not the uniformity of each target.

\textbf{Joint, not marginal, distributions.}
A gate's preprocessing contains correlated objects: the mask, its bits, DPF or DCF targets, wrap bits and auxiliary shares. If a gate needs targets $r$ and $r+1$, resampling each uniformly keeps the marginals but destroys their unit difference; drawing $L$ and $H$ of \cref{alg:interval} independently likewise yields a wrong indicator. The refactoring is valid exactly under the conditions of \cref{lem:bijection}: the map is a bijection for the fixed context \emph{and} all auxiliary material is regenerated from the recovered mask by the unchanged sampler. For example, \cref{alg:interval} may draw $L$ uniformly, set $r=L-a$ and derive $H$ and $\omega$ from that $r$, but may not draw $L$ and $H$ independently.

\textbf{Noninvertible maps and truncation.}
For unsigned truncation, the sampler $R'=2^hA+U$ of \cref{lem:noninv} makes the DFA rule (shift left, fill the low bits uniformly~\cite{dfa2026trusted}) exact; signed shifts, nonuniform masks and other moduli need their own samplers.

\subsection{Verified GPU Matrix Products}\label{sec:proto-freivalds}

Authenticated linear layers need no extra verification: a GPU deviation that changes a value is caught at the next \CheckOpen, before the result enters a protected gate or a release, except with the probability bounded in \cref{thm:mac-sound}. A module that must \emph{consume} a GPU-computed matrix before any opening, however, must not accept it on the host's word, which would reintroduce the forgery oracle of scenario S2. For this case \cref{alg:freivalds} (\cref{app:proof-gates}) is an optional frozen-operand Freivalds check~\cite{freivalds1977probabilistic}, in the spirit of Slalom~\cite{tramer2018slalom} but bound to the ledger.

\textbf{Permitted uses and check.}
The check applies only when exposing the operands to that GPU is already allowed, as in public-by-party-local matrix products; it never permits exposing global secret triple masks, which are generated inside the TCB (\cref{sec:proto-arith}). The module binds the exact operand bytes $A\in\Fp^{m\times k}$ and $B\in\Fp^{k\times n}$ to their authoritative records and freezes the claimed $C$ (\cref{sec:proto-ledger}); only then does it draw $s$ independent secret uniform $\bm{v}\in\Fp^{n}$, test $C\bm{v}\stackrel{?}{=}A(B\bm{v})$ inside the protected verifier, and consume only the verified frozen bytes of $C$ (\cmd{VERIFY_MATRIX} in \cref{tab:commands}). If the GPU also transforms tag arrays, their tag equations are verified too.

\textbf{Soundness and cost.}
\Cref{thm:freivalds} bounds acceptance of a wrong product by $p^{-s}$, and $Q_F$ calls add $Q_F/p^s$ to \cref{thm:main}. The check is disabled in all evaluated configurations; enabling it adds implementation obligations (\cref{tab:obligations}) and this error term.

\subsection{Versioned-Snapshot PIR}\label{sec:proto-pir}

In DPF-based two-server PIR (\cref{eq:pir}), a malicious server can return a wrong share, which the client silently reconstructs into a wrong record, or answer from a stale database. \Cref{alg:pir} and \cref{fig:pir} show the \sys countermeasures.

\begin{figure}[t]
    \centering
    \includegraphics[width=\linewidth]{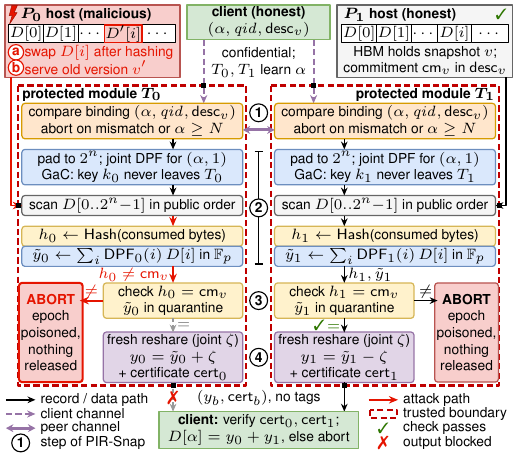}
    \caption{Versioned-snapshot PIR (\cref{alg:pir}). Both modules bind the same query, scan the same snapshot bytes in a public order, hash exactly the bytes they consume, and quarantine the DPF-weighted reduction until the hash matches the authorized commitment $\mathsf{cm}_v$. A host that swaps a record after hashing or serves an old version causes an abort, never a wrong accepted record (scenario S8). Both modules learn $\alpha$, so query privacy relies on module trust.}
    \label{fig:pir}
\end{figure}

\begin{algorithm}[t]
\algstretch
\caption{Versioned-snapshot protected PIR (\algname{PIR-Snap}) inside module $T_b$}
\label{alg:pir}
\begin{algorithmic}[1]
\Require from the client over its confidential channel: index $\alpha$, query identifier $\mathit{qid}$, descriptor $\mathsf{desc}_v=(\mathsf{id},v,N,\mathsf{enc},\mathsf{cm}_v)$; snapshot records $D[0],\dots,D[N-1]$ streamed by host $b$
\Ensure value share $y_b$ of $D[\alpha]$ and certificate $\mathsf{cert}_b$ for the client, or \Abort
\LineComment{Step 1: query binding}
\State exchange $(\alpha,\mathit{qid},\mathsf{desc}_v)$ with $T_{1-b}$ over the protected channel
\If{the bindings differ, $\alpha\ge N$, or $\mathsf{desc}_v$ is unauthorized}
  \State \Abort
\EndIf
\LineComment{Step 2: joint DPF and same-byte scan ($N\ge2$; $N=1$ is a constant-record case with the same checks)}
\State $n\gets\lceil\log_2N\rceil$; bind $N$ and $2^n$; $D[i]\gets0$ for $N\le i<2^n$
\State joint DPF generator for $(\alpha,1)$ under label $(\mathit{qid},\textsf{pir})$ \Comment{GaC}
\State $h\gets\Hash.\mathsf{init}(\mathsf{id},v,N,\mathsf{enc})$;\; $\tilde y_b\gets0$
\For{$i=0$ \textbf{to} $2^n-1$ in public order} \Comment{$\EvalAll$, depth-first}
  \If{$i<N$}
    \State read $D[i]$ into $T_b$; $h\gets\Hash.\mathsf{update}(h,D[i])$
  \EndIf
  \State $\tilde y_b\gets\tilde y_b+\DPF.\Eval(k_b,i)\cdot D[i]$ \Comment{per coordinate}
\EndFor
\LineComment{Step 3: snapshot check; the result is quarantined until here}
\If{$\Hash.\mathsf{final}(h)\neq\mathsf{cm}_v$}
  \State \Abort
\EndIf
\LineComment{Step 4: fresh resharing and certified release}
\State $y_b\gets\tilde y_b+(-1)^b\zeta$ with a fresh joint pad $\zeta$ per coordinate
\State $\mathsf{cert}_b\gets$ certificate on $(\mathsf{sid},\mathsf{ep},\mathit{qid},\mathsf{desc}_v,b,\Hash(y_b))$
\State \Return $(y_b,\mathsf{cert}_b)$ to the client \Comment{no MAC tags}
\end{algorithmic}
\end{algorithm}

\textbf{Protocol and snapshot binding.}
Both role-bound modules receive the same canonical query from the client, compare the binding over their protected channel, pad the snapshot to $2^n$ records while authorizing only $\alpha<N$, and run a joint DPF for $(\alpha,1)$ by generate-and-consume. Each scans the same canonical bytes in the same public order, applies the same DPF weights to every record coordinate, and hashes exactly the bytes it consumes. The commitment $\mathsf{cm}_v$ covers the database identifier, version, shape, record encoding and ordered data bytes, and the result stays quarantined behind the release guard (\cref{sec:proto-ledger}) until the hash matches it; a host that swaps a record after hashing or serves an old version therefore causes an abort (scenario S8 in \cref{tab:scenarios}). A Merkle realization~\cite{merkle1987digital} must authenticate every consumed leaf and its position against the same root. Offloading the reduction to the GPU would need \cref{alg:freivalds} or a separately proven authenticated reduction, never a signature on the GPU's claim, and would export the expanded share vector, so \sys reduces inside the module. Each module releases only a freshly reshared value share (the value half of the epilogue) with a bound certificate, and the client reconstructs $D[\alpha]=y_0+y_1$.

\textbf{Guarantees and regimes.}
In regime (R1) of $\FPIR$ (\cref{fig:func}), with one malicious server host and an honest client, the client obtains the committed record or an abort, and the host learns nothing about $\alpha$ beyond $\Leak_{\mathrm{PIR}}$ (\cref{def:leakage}). In regime (R2), a malicious client alone learns only the requested record and the public descriptor. We make no combined claim (\cref{thm:pir}). Both modules learn $\alpha$, so query privacy relies on module trust: unlike untrusted two-server DPF-PIR such as GPU-DPF~\cite{lam2023gpu} or DFA's untrusted mode~\cite{dfa2026trusted}, which hides $\alpha$ from each server without trusted hardware, \sys offers no query privacy against a compromised module.

\textbf{Scope and design choice.}
Authenticated PIR~\cite{colombo2023apir} provides integrity without trusted hardware at additional server cost and is our malicious-secure GPU baseline (\cref{sec:eval-pir}). Two alternatives bracket our design. Since the modules know $\alpha$ anyway, a module could select $D[\alpha]$ with a constant-time comparator during the same public-order scan; this needs no AES expansion and would be bound by snapshot bandwidth rather than by the AES array. If $\alpha$ must stay hidden from the modules, the client can generate the DPF keys as in untrusted-mode DFA, and each module evaluates its own key inside the boundary with the same hashing, quarantine, resharing and certificates; this adds a DPF key-privacy term and, against a malicious client, a key-verification step~\cite{boyle2016function,boneh2021lightweight,decastro2022vfss}. We keep joint generation for uniformity with the inference datapath (\cref{sec:proto-direct}) and as defense in depth. At a PIR efficiency of {0.93} relative to DFA (\cref{app:model}), our model achieves {16{,}349} versus 17{,}580~queries/s at $N=2^{20}$~\cite{dfa2026trusted} (\cref{tab:pir}).

\section{Architecture and System Integration}\label{sec:arch}

This section describes how \sys realizes the protected module $T_b$ of \cref{sec:protocol} next to the GPU and how the untrusted host drives it. We keep the organization of DFA~\cite{dfa2026trusted} wherever the malicious setting permits and add only the blocks that the integrity chain requires. Hardware sizes are design parameters (\cref{sec:eval-hw,app:model}). That this hardware realizes the ideal resource $\mathcal{F}_{\mathsf{T}}$ of the proofs is an implementation obligation (\cref{tab:obligations}).

\subsection{Design Insights}\label{sec:arch-insight}

\textbf{Key insight.}
Malicious security moves the dominant cost from DPF to authenticated linear algebra on the GPU, while protected FSS evaluation, exposed triple generation and \CheckOpen{} stay small (\cref{sec:analysis-cost,fig:motiv}). The module therefore need not become a second GPU: it must hold exactly the state whose exposure or modification would break the integrity chain and inspect every byte that crosses its boundary. The GPU is a fast but untrusted engine for authenticated linear algebra, and the module is dealer, verifier and release authority at once.

\textbf{Protected components.}
A block belongs inside the trusted boundary if it holds material that reveals both parties' views or authorizes a release. These are exactly the TCB items of \cref{sec:threat-model} (keys, sampler and counters, generator state and correction words, table and reduction logic, challenges and residuals, release decisions); everything else stays outside (\cref{fig:threat}).

\textbf{GPU-resident computation.}
Linear layers run as the matrix Beaver products of $\Pi_{\mathsf{Mult}}$ over the value lane and the $t=3$ tag lanes (\cref{fig:arith}), i.e., two dense products per lane, $2(1+t)$ in total, of the same public shape. A GPU result is a candidate share, never a certified value.

\textbf{Placement.}
An authenticated element costs 32~B per party at the module interface (\cref{sec:proto-repr}), so a link of bandwidth $B$ streams at most $B/32$ elements per second (\cref{sec:analysis-work}). The L2-adjacent placement of DFA delivers about 290~GB/s to the SMs~\cite{dfa2026trusted}, so we keep it: one DPF unit per L2 slice, a new protected core beside a designated L2 partition, and SM access through the DFA-mapped address range without a new network-on-chip (NoC) port.

\subsection{Hardware Design}\label{sec:arch-hw}

\begin{figure}[t]
    \centering
    \includegraphics[width=\linewidth]{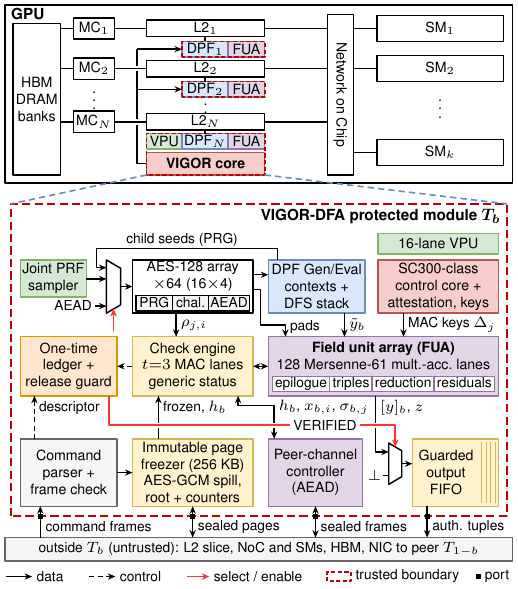}
    \caption{\sys hardware. Top: placement inside the GPU; as in DFA, one DPF unit sits next to each L2 slice and a designated slice hosts the VIGOR core. Bottom: logical organization of the protected module $T_b$ (red dashed trusted boundary). The 64 AES-128 engines are time-shared among PRG expansion, challenge derivation and AEAD sealing; the Mersenne-61 field unit array (eight lanes in each DPF unit) serves epilogues, triples, reductions and check residuals. SMs, L2, NoC, HBM and host DRAM stay outside the boundary.}
    \label{fig:hw}
\end{figure}

\Cref{fig:hw} shows the module (area and power in \cref{sec:eval-hw}). The trusted boundary encloses the 16 DPF units of DFA (\cref{sec:bg-dfa}), the vector processing unit (VPU), the new \emph{VIGOR core} and the on-die links between them. Only command frames with share batches, PIR database records, authenticated output tuples, sealed spill pages and sealed channel frames cross it, none carrying a module secret in the clear. The blocks follow the integrity chain: parser and ledger admit \emph{authorized input}; sampler and field unit array produce \emph{correlated preprocessing}; the DPF units perform \emph{correct local evaluation}; the epilogue yields \emph{authenticated arithmetic}; freezer, check engine and peer channel enforce \emph{consistent opening}; and the release guard gates \emph{authorized output release}.

\textbf{DPF units and sampler.}
The AES datapath is unchanged. Each pipeline context carries two generator states, a separate evaluator state, a depth-first search (DFS) stack and the conversion $\Conv$ into $\Fp$, exactly the state of \cref{alg:gac}; correction words are consumed and discarded inside the unit, and raw leaf shares $\tilde y_b$ go to the field unit array lanes of the same unit, never to SM-readable buffers. The joint PRF sampler implements the labeled derivation and bounded rejection sampling of \cref{sec:proto-setup} under the session key.

\textbf{Field unit array (FUA).}
The FUA consists of 128 Mersenne-61 multiply-accumulate lanes, each with a $61\times 61$-bit multiplier and the two-fold reduction of \cref{lst:mersenne}; its peak is 116~G multiply-accumulates per second at the 910~MHz 7~nm clock of DFA~\cite{dfa2026trusted}. It computes epilogues (three products per output), the triple products $c=ab$ and $C=AB$ with their \AuthShare{} tags (so matrix masks never reach the GPU), LUT and PIR reductions, interval sums and check residuals. The lanes sit in the DPF units, eight per unit (a design parameter), and the VIGOR core schedules them as one pool. Full-domain LUT and PIR reductions therefore run next to the leaves, and only per-query accumulators (three field elements per PIR query) travel to the core (\cref{sec:arch-dse}).

\textbf{Check engine and freezer.}
The check engine executes steps~3 and~4 of \cref{alg:checkopen}: it derives the $t$ challenge vectors only after both frozen roots are fixed, accumulates the residuals on FUA lanes (two products per lane and element) and evaluates all lanes over the full batch before emitting one generic status, so neither the status nor its latency reveals which lane or position failed. \cmd{FREEZE_BATCH} copies the exact submitted bytes into a 256~KB freeze SRAM (8{,}192 tuples of 32~B) and hashes them with their descriptor into the root $h_b$ that the modules exchange; larger batches spill to AES-GCM-sealed immutable HBM pages~\cite{dworkin2007gcm,mcgrew2004gcm} whose authoritative root and page counters stay inside, so after the challenge only authenticated frozen bytes are reread (\cref{sec:proto-ledger}).

\textbf{Ledger, release guard and control.}
A 32~KB ledger SRAM holds the authoritative state of every in-flight gate and the randomness identifiers burnt at reservation (\cref{fig:ledger}). The release guard opens for a gate only in state \lstate{VERIFIED}, and an epoch abort closes every outstanding response atomically (\cref{sec:proto-ledger}); the release multiplexer and output FIFO of \cref{fig:hw} belong to the release-guard logic. The peer-channel controller keeps independent directional keys, never repeats a nonce and enforces a per-session traffic and forgery budget. The SC300-class control core of DFA~\cite{st33tpm12spi,arm_cortex_m3} runs measured boot, attestation and the key establishment covered by $\eps{setup}$, and the 16-lane VPU of DFA~\cite{cavalcante2019ara} handles the remaining protocol logic, including the range check $x^*\in\mathsf{In}_\nu$ of $\Pi_{\mathsf{Input}}$ and the validity predicates of $\mu$ at the public boundaries that $\mu$ declares. Host software cannot drive the module's clock or voltage outside its safe operating range (DVFS isolation, \cref{app:func-module}).

\textbf{AES time-sharing.}
The 64 engines are the module's only AES resource. A slot arbiter per unit serves PRG expansion (two blocks per expansion, $6n+3$ per GaC output), challenge derivation (1.5 blocks per checked element: three 64-bit words for three 61-bit challenges) and AEAD keystream for sealing and the peer channel (two blocks per 32-byte tuple); one AEAD/GHASH block serves both the freezer and the peer-channel controller. Check and sealing requests take priority because they are short and on the critical path of the next round; PRG contexts fill all other slots. A per-slot purpose tag selects the key context, so the uses never share a key.

\subsection{Pipeline Scheduling and Memory}\label{sec:arch-pipe}

\begin{figure}[t]
    \centering
    \includegraphics[width=\linewidth]{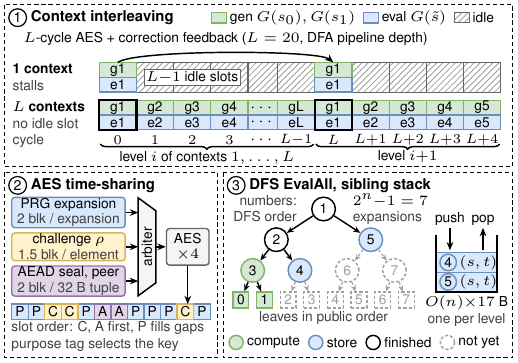}
    \caption{Scheduling inside a DPF unit. (1)~Generator (g) and evaluator (e) expansions of at least $L$ independent contexts interleave to hide the $L$-cycle AES and correction feedback latency. (2)~AES slots are time-shared among PRG expansion, challenge derivation and AEAD sealing. (3)~Depth-first $\EvalAll$ keeps one pending sibling per level, i.e., $O(n)$ state per context.}
    \label{fig:pipeline}
\end{figure}

\textbf{Context interleaving.}
With an AES and correction feedback latency of $L$ cycles, one dependent path advances at most one level every $L$ cycles, so one issue per cycle needs at least $L$ independent contexts per unit (\cref{fig:pipeline}, panel~1). For the 20-stage DFA pipeline~\cite{dfa2026trusted} we therefore provision 20 contexts per unit, one per stage. The three expansions of a malicious GaC level issue back to back, because $\PRG(\tilde s)$ does not depend on $\CW^{(i)}$; only the correction waits for all three. A context (three 129-bit path states, the correction word in flight, target, query, payload and level) needs fewer than 100~B, so 20 contexts take less than 2~KB per unit, and a single gate with $n=61$ completes in {1.71~$\mu$s}.

\textbf{Separate old controls.}
Each unit keeps separate registers for the \emph{old} generator bits $t_0,t_1$ and the \emph{old} evaluator bit $\tilde t$ (\cref{sec:analysis-dpf}), captures all outputs of a transition before overwriting its inputs, and never substitutes generator states for evaluator roots; this is the schedule covered by \cref{thm:gac}. RTL tests exercise signed final conversion, zero payloads, all-zero and all-one indices, $x=\alpha$ and divergence at every depth.

\textbf{Depth-first EvalAll.}
LUT and PIR gates evaluate the full domain. Depth-first traversal keeps one pending sibling per level ($O(n)$ states per traversal) and a complete tree needs $2^n-1$ internal expansions (panel~3). The 512 queries of a PIR batch walk the same public leaf order in lockstep, so each snapshot record is read and hashed once per batch and folded into 512 per-query accumulators. For $N=2^{20}$, 512 contexts with 20 pending siblings of 17~B occupy about 174~KB, or 11~KB per unit, which fits the traversal buffers of DFA~\cite{dfa2026trusted}.

\textbf{Pads and frozen buffers.}
The four pads of an output occupy two AES blocks and are drawn in step with the leaf, so no pad buffer beyond the pipeline depth is needed. A checked batch occupies the freeze SRAM until its verdict two peer round trips later, so batches beyond 8{,}192 elements, such as the 65{,}536-element batches of \cref{fig:checkperf}, spill to sealed pages.

\subsection{Command Interface and Host Integration}\label{sec:arch-iface}

\begin{table}[t]
\centering
\caption{Commands of the protected module $T_b$. Every command is a canonical frame bound to session, epoch, role and gate; the ledger state reached on success is given in parentheses. No command authenticates a host-supplied value.}
\label{tab:commands}
\footnotesize
\setlength{\tabcolsep}{3pt}
\renewcommand{\arraystretch}{1.1}
\begin{tabularx}{\linewidth}{l >{\raggedright\arraybackslash}X >{\raggedright\arraybackslash}X}
\hline
\textbf{Command} & \textbf{Required check} & \textbf{Allowed result} \\
\hline
\cmd{OPEN_SESSION}  & Authenticated identities, role, fresh epoch, manifest $\mu$ & Session capability \\
\cmd{RESERVE_PREP}  & Approved gate, unused randomness identifiers & Party-local authorized preprocessing, issued once (\lstate{PREPARED}) \\
\cmd{FREEZE_BATCH}  & Exact canonical bytes, dimensions, expected wire list & Frozen root $h_b$ and handle (\lstate{FROZEN}) \\
\cmd{CHECK_OPEN}    & Both frozen roots, independent challenge lanes, peer residuals & Verified prescribed opening or abort (\lstate{VERIFIED}) \\
\cmd{EVAL_FSS}      & Bound table or function, \lstate{VERIFIED} and unconsumed input opening, fixed role & Freshly reshared authenticated local output (\lstate{CONSUMED}) \\
\cmd{VERIFY_MATRIX} & Anchored operands, frozen candidate & Verified immutable result handle (\lstate{VERIFIED}) \\
\cmd{FINAL_OUTPUT}  & Authorized recipient, output mask, accepted check & Bound frame with $\upsilon=y+o$ and certificate (\lstate{CONSUMED}, then \lstate{RELEASED}) \\
\cmd{ABORT}         & Any integrity or lifecycle violation & Epoch invalidation, no candidate values (\lstate{ABORTED}) \\
\hline
\end{tabularx}
\end{table}

\textbf{Host responsibilities.}
Host $b$ schedules the public circuit, stores only party-local authenticated arrays (value and each MAC lane as separate coalesced arrays, serialized canonically only at the protected interface), runs the exact field kernels of \cref{sec:arch-field} and sends share batches to $T_b$. Its scheduling may be adversarial, so the ledger, not the host, decides whether operation, identity, shape and order are admissible, and a reused buffer always carries a fresh transaction identifier.

\textbf{Commands and frames.}
\Cref{tab:commands} lists the eight commands with the ledger state each reaches (\cref{fig:ledger}); since none of them authenticates a host value, the forgery-oracle path of scenario~S2 is closed (\cref{tab:scenarios}). A frame carries version, session identifier, epoch, role, opcode, gate identifier, ordered input and output identifiers, field identifier, tag count, dimensions, payload length, sequence number and payload, all length-prefixed with explicit endianness; frames between modules or from a client bind the header as AEAD associated data. The parser copies each frame into protected memory before checking it and rejects truncation, trailing bytes, duplicate identifiers, malformed dimensions, noncanonical field elements (values of at least $p$ are rejected, never reduced) and unexpected opcodes.

\textbf{Data races and aborts.}
Verification and consumption use the same immutable parsed bytes. Independent gate batches may overlap, but a dependent nonlinear gate never consumes an unchecked masked opening. Any integrity or lifecycle violation, such as an invalid transition or a failed check, triggers \cmd{ABORT}, which poisons the epoch, prevents new challenges and atomically keeps the release guard closed for every outstanding response (\cref{sec:proto-ledger}). \Cref{tab:attacks} maps each injected fault to the block that stops it.

\subsection{Exact Field Arithmetic on the GPU}\label{sec:arch-field}

GPU kernels must compute the finite circuit that the proofs reason about, so floating-point shortcuts are excluded. An authenticated GEMM applies the matrix Beaver combination of $\Pi_{\mathsf{Mult}}$ on each of the four lanes, i.e., two products per lane ($D\,B_b$ and $A_b\,E$ on the value lane, likewise on each tag lane), $2(1+t)=8$ products of the same shape, plus the public product $DE$ shared by all lanes (\cref{fig:arith}); \cref{lst:mersenne} gives the scalar multiplication at the core of every kernel.

\begin{figure}[t]
\centering
\scriptsize
\begin{tabular}{|l|}
\hline
\begin{minipage}{0.95\linewidth}
\vspace{2pt}
\begin{verbatim}
constexpr uint64_t P = (1ULL << 61) - 1;  // 2^61 - 1
// canonical a, b < P; returns canonical a*b mod P
__device__ uint64_t mul61(uint64_t a, uint64_t b) {
  uint64_t lo = a * b;              // low 64 bits
  uint64_t hi = __umul64hi(a, b);   // high bits, < 2^58
  uint64_t x  = (lo & P) + ((lo >> 61) | (hi << 3));
  x = (x & P) + (x >> 61);          // x <= 2^61
  return x >= P ? x - P : x;
}
\end{verbatim}
\vspace{-4pt}
\end{minipage}\\
\hline
\end{tabular}
\caption{Exact Mersenne-61 multiplication used by the GPU field kernels (CUDA). The first fold uses $2^{61}\equiv 1 \pmod p$; the second fold and the conditional subtraction return the canonical representative. No floating-point or Tensor-Core operation is involved.}
\label{lst:mersenne}
\end{figure}

\textbf{Correctness.}
For canonical $a,b<p=2^{61}-1$ the product is below $2^{122}$, so $\mathtt{hi}<2^{58}$ and $\mathtt{hi}\ll 3$ fits in 61 bits. Writing $ab=2^{64}\,\mathtt{hi}+\mathtt{lo}$ and using $2^{61}\equiv 1\pmod p$ gives $ab\equiv(\mathtt{lo}\bmod 2^{61})+\lfloor\mathtt{lo}/2^{61}\rfloor+2^{3}\,\mathtt{hi}\pmod p$. The last two terms occupy disjoint bit ranges (bits 0 to 2 and 3 to 60), so the bitwise OR equals their sum and the folded value is below $2^{62}$. The second fold yields at most $2^{61}=p+1$, and the conditional subtraction returns the canonical representative. The executed CPU check of this formula and the GPU conformance test are reported in \cref{sec:eval-attack}.

\textbf{Accumulation and numeric contract.}
A two-word 128-bit accumulator absorbs at most 64 unreduced products below $2^{122}$, so kernels fold at least once per 64 inner-dimension terms; random shares lie close to $p$ even for small plaintexts, so plaintext bounds never justify longer accumulation. The manifest $\mu$ fixes the scale $2^f$, signed encoding, toward-zero truncation (\cref{alg:trunc}), saturation, approximation coefficients and output decoding. For example, a dot product of length at most 4{,}096 with operand magnitudes at most $2^{20}$ stays within $|\cdot|\le 2^{52}$ before bias, which leaves room below $p/2$ only if bias and further accumulation are bounded too; the bound is recomputed per layer. FP16, BF16 and TF32 Tensor Cores never serve as exact field operations without a proven limb construction, and the exact finite circuit and a floating-point accuracy model remain separate references.

\textbf{GEMM cost.}
Our model charges an exact Mersenne GEMM {1.25$\times$} the time of a wraparound 64-bit GEMM per lane ($\phi_{\mathrm{M}}$ in \cref{app:model}). This factor is calibrated via GPU benchmarks: the kernel above additionally needs the high product word and a two-word accumulator per term, and the GPU test that would calibrate it has not run. Each $+0.25$ in $\phi_{\mathrm{M}}$ adds {0.52~s} to BERT-Base; with $\phi_{\mathrm{M}}=2$ its LAN latency would be {5.39~s}, i.e., {3.8$\times$} instead of {2.7$\times$} that of semi-honest trusted DFA.

\subsection{Deployment Targets}\label{sec:arch-deploy}

An on-die module cannot be added to a commercial A100, so we define three distinct experiments and never merge their timings. All three are reported in \cref{sec:eval-method}. \emph{(i) A100 software} runs the exact kernels, authenticated GEMM and host runtime on a real A100 with an emulated module. \emph{(ii) A ZCU102 prototype} (\texttt{xczu9eg-ffvb1156-2-e}~\cite{amd2023zcu102}) implements the datapath of \cref{fig:hw} at reduced scale ({4} DPF units, {16} AES engines, {32} FUA lanes at {250~MHz}, \cref{tab:fpga}) to validate streaming, freezing and release ordering; one board with an emulated peer shows functionality, not two independent trust boundaries, and an ARM process on it is not isolated from a malicious board administrator. \emph{(iii) The on-die study} models integration next to the L2 slices with Accel-Sim~\cite{khairy2020accel} and synthesis, as in DFA; our headline numbers refer to it and cross-checked against this study (\cref{app:model}). An external module uses staged transfer through host memory as its baseline, because GPUDirect RDMA needs platform and driver support and ordinarily a common PCIe root complex~\cite{nvidia2024gpudirect}; IOMMU protection is never disabled for a benchmark. In every target, freeze-before-challenge refers to immutable authenticated bytes, never to a GPU buffer that can change concurrently.

\subsection{Design-Space Comparison}\label{sec:arch-dse}

We compare the L2-adjacent organization of \sys with alternative homes for authentication, verification and randomness generation.

\textit{External protected module over PCIe.}
A discrete FPGA or security card~\cite{xu2025co} keeps the boundary off the GPU die but limits the interface to $B/32$ elements per second: {50~M} for the {1.6~GB/s} ZCU102 transport and {0.78~G} for {25~GB/s} of effective PCIe Gen4 bandwidth, against 9.1~G on the 290~GB/s on-die path of DFA~\cite{dfa2026trusted}. Because authenticated GEMM dominates, BERT-Base LAN latency rises only to {3.87~s} over PCIe and {4.59~s} with the ZCU102 link and lanes (\cref{tab:ablation}); the option suits prototyping but adds a second device per party and two device round trips per checked layer (\cref{app:model}).

\textit{Verification on the GPU SMs.}
In \ProfT the SMs belong to the malicious host: a corrupt checking kernel could see the challenges before choosing its errors, or simply report acceptance. Trusting them requires a GPU-wide TEE~\cite{volos2018graviton,jang2019heterogeneous,dhanuskodi2023confidential} covering SMs, caches, memory controllers and the driver interface, the large boundary that DFA and \sys avoid. The GPU therefore computes only products that the module verifies afterwards (\CheckOpen{} or the Freivalds check of \cref{sec:proto-freivalds}).

\textit{Tag DPF trees instead of the epilogue.}
One extra DPF tree per MAC lane (payload $\Delta_j\beta$) multiplies the PRG work of every gate by $1+t$ (\cref{tab:gatecost,fig:dpfthr}); the epilogue achieves the same with three multiplications and four pads (\cref{thm:epilogue}) at {3.8--4.0$\times$} the per-gate throughput of tag trees.

\textit{Software-only malicious backend.}
SPDZ-family preprocessing~\cite{keller2016mascot,keller2018overdrive,cramer2018spdz2k,keller2020mpspdz} needs inter-party communication that grows with the number of authenticated triples, and dealer-based malicious FSS, modeled after the protocol family of Shark~\cite{gupta2025shark}, ships authenticated keys and triples from an external dealer (\cref{tab:e2e}); both re-create the offline traffic that trusted DFA removed. Such a backend fits \ProfU, where the hardware itself is distrusted (\cref{sec:threat-profileu}); under \ProfT it pays for a trust assumption that the module already provides.

\textit{Dedicated check engines, a central FUA and per-slice checking.}
Checking needs only {2.8~ms} of module time per BERT-Base query against {131~ms} of protected FSS work, so dedicated AES engines for challenges and sealing would idle almost always. DPF units stay distributed because a single attachment point, a memory-channel path of about 180~GB/s, would cap DPF output at about 62\% of the 290~GB/s delivered bandwidth~\cite{dfa2026trusted}. The same limit rules out a central FUA: the raw 8-byte leaf shares of PIR at 16K entries alone would need {137~GB/s} ({76\%} of that path), so reductions stay in the units. Checks, by contrast, cost at most 3.5 AES blocks and six field multiplications per element and depend on the ledger, freeze SRAM and peer channel, which must be single points of authority; the VIGOR core therefore controls every check and holds one authoritative root per batch.

\section{Security Analysis}
\label{sec:security}

We now show that \sys realizes the functionalities $\FC$ and $\FPIR$ of \cref{fig:func} with leakage $\Leak_C$ and $\Leak_{\mathrm{PIR}}$ (\cref{def:leakage}) against a malicious host and GPU, as stand-alone security with abort (\cref{def:security}) in the explicit-simulator style~\cite{canetti2000security,goldreich2004foundations,hazay2010efficient,lindell2017simulate}. Each link of the integrity chain, \emph{authorized input $\rightarrow$ correlated preprocessing $\rightarrow$ correct local evaluation $\rightarrow$ authenticated arithmetic $\rightarrow$ consistent opening $\rightarrow$ authorized output release}, is discharged by a named statement: input extraction (\cref{thm:arith}); in-module preprocessing and joint sampling (\cref{lem:bijection,lem:noninv}); exact DPF, DCF, FSS-gate and direct-gate correctness (\cref{thm:dpf-correct,thm:gac,thm:dcf,thm:bitdec,thm:cmp-trunc,thm:lut,thm:interval,thm:arith}); the epilogue and the arithmetic identities (\cref{thm:epilogue,thm:arith}); MAC-check soundness (\cref{lem:hidden-key,thm:mac-sound}); and gated release (\cref{thm:arith,thm:pir}). Full proofs are in \cref{app:proof-dpf,app:proof-auth,app:proof-gates,app:proof-main}.

\subsection{Assumptions and Error Terms}
\label{sec:sec-assump}

The theorem is proved in \ProfT (\cref{sec:threat}); we state module trust as an explicit ideal resource~\cite{katz2007universally}.

\begin{assumption}[Ideal protected module]
Corrupting host $b$ does not corrupt $T_b$; both modules are instances of the ideal resource $\mathcal{F}_{\mathsf{T}}$ (\cref{app:func-module}). The code, keys, parser, field arithmetic, ledger transitions, local memory, protected peer channel and approved I/O paths of each module are correct and confidential. Its externally observable trace (lengths, timing, host-visible memory and DMA activity, power, thermal and clock telemetry) is simulatable by a stipulated trace simulator from the command sequence issued by the host, $\mu$, the public lengths and the accept or abort outcomes, and is independent of the global MAC keys, challenges and residuals given these (all lanes are aggregated in data-independent time before one generic $\mathsf{status}$ is returned). Host software cannot drive the module's clock or voltage outside its safe operating range, or the module detects glitches and aborts (DVFS isolation); software-reachable undervolting~\cite{murdock2020plundervolt} is thus in scope (\cref{sec:threat}). Module state is fresh: a trusted monotonic service preserves freshness across reboots~\cite{parno2011memoir,matetic2017rote}, or the module irrecoverably erases its master key and requires a freshly attested session. Availability is not assumed.
\end{assumption}

\begin{assumption}[Setup, randomness, channels and bindings]
Each term bounds an idealization over the \emph{entire} admitted execution, not a single primitive call:
$\eps{setup}$, simulation of the authenticated setup of \cref{sec:proto-setup} by an ideal role-bound interface;
$\eps{prf}$, replacement of PRF/KDF outputs (AES-128, HKDF~\cite{krawczyk2010hkdf}) under distinct length-prefixed labels by independent coins, identical for the same joint label;
$\eps{channel}$, confidentiality and integrity of the AEAD channels~\cite{dworkin2007gcm,mcgrew2004gcm} with directional keys and nonrepeating nonces;
$\eps{bind}$, collisions or forgeries in manifests, snapshot commitments, frozen-page roots, certificates and session freshness~\cite{fips202,merkle1987digital};
$\eps{trace}$, the stipulated trace simulator.
\end{assumption}

Correctness of protected code, state and arithmetic is thus a \emph{premise}, not an estimated probability. The protocol contributes one statistical term, determined by the field $\Fp$ with $p=2^{61}-1$, the number $t$ of independent global keys $\Delta_j=\Delta_{0,j}+\Delta_{1,j}$ (host $b$ knows at most $\Delta_{b,j}$), and the number $Q$ of \CheckOpen{} batches. Exclusions and client assumptions are those of \cref{sec:threat}; the modules may see inputs (\cref{sec:discussion}).

\subsection{Main Theorem}
\label{sec:sec-main}

\begin{theorem}[Stand-alone static malicious security with abort]
\label{thm:main}
Under the assumptions of \cref{sec:sec-assump}, the protocols of \cref{sec:analysis-dpf,sec:protocol} (setup, ledger, exact DPF, generate-and-consume with the fresh authentication epilogue, conservative DCF, \AuthShare, \CheckOpen, $\Pi_{\mathsf{Input}}$ with the committed-input loader, $\Pi_{\mathsf{Lin}}$, $\Pi_{\mathsf{Mult}}$, $\Pi_{\mathsf{Out}}$, bit decomposition, comparison, truncation, LUT, interval and direct protected gates) realize $\FC$ with abort and leakage $\Leak_C$, and \cref{alg:pir} realizes $\FPIR$ with abort and leakage $\Leak_{\mathrm{PIR}}$ (\cref{def:security,def:leakage}), against a PPT adversary that statically corrupts at most one host or, for $\FPIR$ in regime (R2) of \cref{fig:func}, the client alone. With at most $Q$ \CheckOpen{} batches, the advantage of every PPT distinguisher in \cref{def:security} is at most
\begin{align*}
\eps{VIGOR}\le{}&\eps{setup}+\eps{prf}+\eps{channel}\\
&+\eps{bind}+\eps{trace}+Q\,(2/p)^{t}.
\end{align*}
If the Freivalds check (\cref{alg:freivalds}) is enabled for at most $Q_F$ calls with $s$ repetitions, the bound gains $Q_F/p^{s}$, and the binding obligations of \cref{thm:freivalds} apply.
\end{theorem}

\Cref{thm:main} is a stand-alone real/ideal statement about one admitted execution, preserved under sequential composition; we claim neither universal nor concurrent composability, and concurrent sessions that share module state are not covered (\cref{app:func}).

\textbf{What the bound omits.}
There is no DPF one-key-privacy term: keys never leave the modules, \cref{thm:dpf-correct} gives raw correctness and \cref{thm:epilogue} makes every exported tuple uniform. A key-exporting profile (\ProfU, \cref{sec:threat-profileu}, or verifiable FSS~\cite{decastro2022vfss}) must add its actual DPF or backend privacy assumption. The theorem does not prove that a concrete chip implements the ideal module (\cref{sec:sec-boundary}).

\textbf{Proof overview.}
The simulator $\Sim$ plays $\mathcal{F}_{\mathsf{T}}$ for both modules, keeps the global MAC keys internally, and tracks the ledger and, for every wire, the tuple a faithful corrupt party would hold; honest shares stay symbolic, so honest values are never needed. Five hybrids connect the real and ideal executions (\cref{app:proof-main}). In the real execution $\Hyb{0}$ the modules run the concrete module program: $\mathcal{F}_{\mathsf{T}}$ with its randomness $R$ instantiated by $\PRF_K$, peer and client channels by AEAD over the hosts' network, and \cmd{OPEN_SESSION} by attestation plus authenticated key exchange.
\begin{itemize}[leftmargin=*]
\item \textbf{$\Hyb{1}$ (ideal setup and channels).} Setup becomes the ideal role-bound interface, and protected-channel contents become ideal deliveries of the permitted lengths and schedules. Cost: $\eps{setup}+\eps{channel}$.
\item \textbf{$\Hyb{2}$ (ideal randomness and bindings).} PRF/KDF outputs become ideal coins, and any binding failure aborts; the ledger (\cref{fig:ledger}) then excludes label reuse, role swapping and post-challenge byte replacement. Streaming generate-and-consume is replaced by \cref{alg:gen} followed by \cref{alg:eval} (\cref{thm:gac}), and refactored samplers by their reference forms (\cref{lem:bijection,lem:noninv}); both replacements are exact. Cost: $\eps{prf}+\eps{bind}$.
\item \textbf{$\Hyb{3}$ (abort on an accepted wrong-value opening).} The experiment aborts at the first \CheckOpen{} batch accepted with $\bm{e}\ne\bm{0}$ (and, if \cref{alg:freivalds} is enabled, at the first accepted $C\ne AB$). It equals $\Hyb{2}$ until then; by \cref{lem:hidden-key,thm:mac-sound} the cost is $Q(2/p)^t$ (plus $Q_F/p^s$, \cref{thm:freivalds}).
\item \textbf{$\Hyb{4}$ (ideal execution).} Honest parties use $\FC$. $\Sim$ extracts corrupt inputs ($x^*=r+\varepsilon$ or the committed loader value) and aborts where the modules do (e.g., unless $x^*\in\mathsf{In}_\nu$), samples preprocessing tuples, epilogue and direct-gate outputs and every masked opening uniformly (the hidden preprocessing is defined afterwards), decides each check from $\inner{\bm{\rho}_j}{\bm{d}_j-\Delta_j\bm{e}}$ with its own keys and fresh challenges, delivers a corrupt recipient's output mask $o$ early and programs its opening $\upsilon=y+o$ after $\FC$ releases $y$, and appends the stipulated trace simulator. Since $\Hyb{3}$ accepts no wrong value, the trace-free views are identical (\cref{thm:arith,thm:bitdec,thm:cmp-trunc,thm:lut,thm:interval}); cost: $\eps{trace}$. PIR uses the simulator of \cref{thm:pir}.
\end{itemize}
The triangle inequality gives the bound. The adversary may still suppress delivery to honest parties: no fairness is added.

\subsection{Building Blocks}
\label{sec:sec-lemmas}

The statements refer to the exact algorithms of \cref{sec:analysis-dpf,sec:protocol}, which, unlike the semi-honest ones (\cref{sec:bg-dfa}), make root control bits, role inputs and separate generator and evaluator states explicit; each is followed by its proof idea.

\begin{theorem}[DPF correctness]
\label{thm:dpf-correct}
For all root seeds $s_0,s_1\in\bits{128}$ with root control bits $t_0=0$, $t_1=1$, every $n$, $\alpha,x\in\bits{n}$ and $\beta\in\Fp$, the keys of \cref{alg:gen} and the outputs $y_b$ of \cref{alg:eval} satisfy $y_0+y_1=\beta\cdot\one[x=\alpha]$. No uniform target and no hardness assumption is needed.
\end{theorem}

\noindent\textit{Proof idea.}
On the target prefix the control bits stay opposite; on a losing child $(s_0^{\lose}\oplus s_1^{\lose})\oplus((t_0\oplus t_1)\sCW)=0$ and the control bits cancel, so the states coincide and their signed outputs cancel. At the leaf $t_0-t_1=(-1)^{t_1}$, so $y_0+y_1=u_0-u_1+(-1)^{t_1}\CWout=\beta$ (\cref{app:proof-dpf}).

\begin{theorem}[Generate-and-consume equivalence]
\label{thm:gac}
For fixed roots, $\alpha$, $\beta$, role $b$ and query $x$, the raw share $\tilde{y}_b$ of \cref{alg:gac} equals the output of \cref{alg:eval} on the key $k_b$ that \cref{alg:gen} produces from the same inputs, node by node for every dependency-respecting depth-first order of $\EvalAll$.
\end{theorem}

\noindent\textit{Proof idea.}
Couple both executions: equal states before level $i$ give the same correction word, generator state and evaluator state after it. The theorem justifies scheduling only, not changes to expansion or correction (\cref{app:proof-dpf}).

\begin{theorem}[Conservative DCF]
\label{thm:dcf}
If \cref{alg:dcf} creates, for each $i=1,\dots,n$ (zero payloads included), an independent DPF with target $\alpha_1\cdots\alpha_{i-1}0$ and payload $\beta\alpha_i$, evaluated at $x_1\cdots x_i$, the shares sum to $\beta\cdot\one[x<\alpha]$, at a cost of $n(n+1)/2$ path levels.
\end{theorem}

\noindent\textit{Proof idea.}
If $x<\alpha$, exactly the prefix ending at their first differing bit ($x_i=0$, $\alpha_i=1$) contributes $\beta$; otherwise none does. \Cref{thm:dpf-correct} and linearity conclude (\cref{app:proof-dpf}).

\begin{theorem}[Epilogue correctness and local hiding]
\label{thm:epilogue}
For raw shares $\tilde{y}_0,\tilde{y}_1$ and fresh uniform pads $\zeta,\kappa_1,\dots,\kappa_t\in\Fp$, the outputs $y_b=\tilde{y}_b+(-1)^b\zeta$, $\gamma_{b,j}=\Delta_j y_b+(-1)^b\kappa_j$ form an authenticated sharing of $\tilde{y}_0+\tilde{y}_1$. For either fixed $b$, $(y_b,\gamma_{b,1},\dots,\gamma_{b,t})$ is uniform on $\Fp^{t+1}$ given the raw outputs and the keys, before conditioning on later prescribed openings, jointly across fresh outputs.
\end{theorem}

\noindent\textit{Proof idea.}
The pads cancel in the sums. For fixed $b$, $\zeta\mapsto y_b$ and each $\kappa_j\mapsto\gamma_{b,j}$ are bijections, so every tuple has probability $p^{-(t+1)}$ (\cref{app:proof-dpf}).

\begin{lemma}[Hidden-key invariant]
\label{lem:hidden-key}
In $\Hyb{2}$ with host $b$ corrupted, conditioned on its view and its shares $\Delta_{b,j}$, the honest shares $\Delta_{1-b,1},\dots,\Delta_{1-b,t}$ remain independent and uniform at every point up to the drawing of the challenges of the first check with a nonzero value error; that check either accepts (the event of \cref{thm:mac-sound}) or rejects and ends the execution.
\end{lemma}

\noindent\textit{Proof idea.}
Fresh tags are padded (\AuthShare, \cref{thm:epilogue}); linear updates are deterministic in the view; correct openings have zero residual; a tag-only error yields an outcome independent of $\Delta_j$; residuals and challenges stay inside the modules (\cref{app:proof-auth}).

\begin{theorem}[Wrong-value acceptance]
\label{thm:mac-sound}
Let a batch of $m$ positions submitted to \CheckOpen{} (\cref{alg:checkopen}) differ from the prescribed tuples by a value error $\bm{e}\in\Fp^m$ and tag errors $\bm{d}_1,\dots,\bm{d}_t\in\Fp^m$, fixed by the freeze. If $\bm{e}\ne\bm{0}$, no earlier batch with a nonzero value error has been accepted, and the lanes use independent uniform challenges $\bm{\rho}_j\in\Fp^m$ drawn after the freeze and hidden from the host, \CheckOpen{} accepts with probability at most $(2/p)^t$ in $\Hyb{2}$. Over at most $Q$ batches, an incorrect value is ever accepted with probability at most $Q(2/p)^t$.
\end{theorem}

\noindent\textit{Proof idea.}
Lane $j$ accepts iff $\inner{\bm{\rho}_j}{\bm{d}_j-\Delta_j\bm{e}}=0$. For $e_i\ne 0$ the vector vanishes only if $\Delta_j=d_{j,i}/e_i$ (probability $\le 1/p$ by \cref{lem:hidden-key}); otherwise the nonzero linear form vanishes with probability $1/p$~\cite{schwartz1980fast}. Each lane gives $1/p+(1-1/p)/p<2/p$, and lanes are independent. Apply the bound to the first wrong-value batch; a union bound over batches gives $Q(2/p)^t$, following the SPDZ MAC check~\cite{damgard2012spdz,damgard2013practical} (\cref{app:proof-auth}).

\begin{theorem}[Authenticated arithmetic and direct gates]
\label{thm:arith}
With correct modules and prescribed messages, the protocols of \cref{fig:arith} are exact: $\Pi_{\mathsf{Input}}$ yields $\auth{x}$ with $x=r+\varepsilon$ (or aborts unless $x\in\mathsf{In}_\nu$), $\Pi_{\mathsf{Lin}}$ preserves $\gamma_{0,j}+\gamma_{1,j}=\Delta_j x$, $\Pi_{\mathsf{Mult}}$ with a correct fresh triple yields $\auth{xy}$ (and $\auth{XY}$ with the factor order preserved), and $\Pi_{\mathsf{Out}}$ gives $y=\upsilon-o$ to its authorized recipient. A direct protected gate (\cref{sec:proto-direct}) for an approved deterministic $g$, evaluated from protected memory in data-independent time, outputs \AuthShare$(g(c-r))$, an authenticated sharing of $g(x)$ whose party-local tuple is uniform as in \cref{thm:epilogue}. Unless a nonzero value error is accepted or a binding fails, every deviation of a corrupt host selects a permitted input, leaves the computation unchanged, or causes abort.
\end{theorem}

\noindent\textit{Proof idea.}
With $d=x-a$, $e=y-b$ and $c=ab$: $c+db+ea+de=xy$~\cite{beaver1991efficient}, the tags follow scaled by $\Delta_j$, $Z=C+DB+AE+DE$ for matrices, and $\upsilon-o=(y+o)-o=y$ with $o$ uniform. A direct gate recomputes $x=c-r$ from a verified opening, and its trace falls under $\eps{trace}$. \CheckOpen{} and the ledger block changed or reused material; induct over gates (\cref{app:proof-auth}).

\begin{theorem}[Canonical bit decomposition]
\label{thm:bitdec}
For $p=2^n-1$ ($n=61$), \cref{alg:bitdec} returns authenticated bits (least significant first) of the canonical representative of $x$, using one masked opening $c=x+r\bmod p$ and $2n$ authenticated multiplications.
\end{theorem}

\noindent\textit{Proof idea.}
The first loop is a full subtractor giving $d=(c-r)\bmod 2^n$ and borrow $\one[c<r]$; with a borrow $x=c-r+p=d-1$ since $p=2^n-1$, which the second loop subtracts~\cite{damgard2006unconditionally} (\cref{app:proof-gates}).

\begin{theorem}[Exact comparison and signed truncation]
\label{thm:cmp-trunc}
From authenticated canonical bits of $x$ and a public $\theta\in[0,2^n)$, \cref{alg:trunc} returns authenticated $\one[x<\theta]$; for centered $z\in[-(p-1)/2,(p-1)/2]$ it returns $\mathrm{sgn}(z)\lfloor|z|/2^f\rfloor$, i.e., truncation toward zero.
\end{theorem}

\noindent\textit{Proof idea.}
The equality flag tracks the processed prefix, and $\mathit{lt}$ gains a term exactly at the first position with secret bit $0$ and threshold bit $1$. Negative values occupy the canonical upper half; $(1-2s)z$ is the magnitude without wrap, and shifting its bits floors~\cite{catrina2010improved}. Fixed-point products must stay below $p/2$ before truncation (\cref{app:proof-gates}).

\begin{theorem}[Authenticated LUT]
\label{thm:lut}
For an approved table $T$ of $2^w$ entries and certified authenticated bits of $x\in\bits{w}$, if the mask bits and the DPF target of the keys for $(r,1)$ come from the same fresh uniform $r$, \cref{alg:lut} outputs $\auth{T[x]}$ and the opened $u=x\oplus r$ is uniform given $x$.
\end{theorem}

\noindent\textit{Proof idea.}
By \cref{thm:dpf-correct}, $\tilde{y}_0+\tilde{y}_1=\sum_a T[a\oplus u]\one[a=r]=T[x]$, and XOR by fixed $x$ is a bijection; independent resampling of the target would break the identity (\cref{app:proof-gates}).

\begin{theorem}[Intervals and piecewise polynomials]
\label{thm:interval}
For a proper nonempty public $[a,b)\subset[0,p)$, \cref{alg:interval} outputs $\auth{\one[a\le x<b]}$; for a disjoint exhaustive partition, Horner evaluation of the selected coefficients outputs $P_{i^\star}(x)$ for the interval containing $x$, in the declared finite arithmetic.
\end{theorem}

\noindent\textit{Proof idea.}
Adding $r$ maps $[a,b)$ to the circular interval from $L$ to $H$; the indicator is $\one[c<H]-\one[c<L]$ if $L<H$ and $1+\one[c<H]-\one[c<L]$ if $L>H$, which is the raw share by \cref{thm:dcf}. Exactly one indicator is one, and Horner with \cref{thm:arith} gives the polynomial; no equality to an unapproximated real function is claimed (\cref{app:proof-gates}).

The next two lemmas state when the semi-honest decoupling, which samples DPF targets directly~\cite{dfa2026trusted}, remains valid: joint distributions with all auxiliary material must match, not only marginals (\cref{sec:proto-refactor}).

\begin{lemma}[Bijection with correlated auxiliary material]
\label{lem:bijection}
Fix $Z=z$, $R$ uniform on a finite $\mathcal{R}$, a bijection $g_z:\mathcal{R}\to\mathcal{A}$, independent coins $W$ and an auxiliary sampler $J$. Then $(R,g_z(R),J(z,R,W))$ and $(g_z^{-1}(A),A,J(z,g_z^{-1}(A),W))$ with $A$ uniform on $\mathcal{A}$ are identically distributed.
\end{lemma}

\noindent\textit{Proof idea.}
Each pair $(r,a)$ has probability $1/|\mathcal{R}|$ in both experiments if $a=g_z(r)$ and zero otherwise; the same $J$ preserves equality. Independently resampling correlated targets such as $(r,r+1)$ breaks it (\cref{app:proof-gates}).

\begin{lemma}[Noninvertible transformations]
\label{lem:noninv}
For any $g$, sampling $A$ from the law of $g(R)$ and $R'$ from $R\mid g(R)=A$ reproduces $(R,g(R))$ and everything derived from $R'$. For unsigned truncation of a uniform $\ell$-bit mask by $h$ bits, $R'=2^hA+U$ with $A$, $U$ uniform on $\ell-h$ and $h$ bits.
\end{lemma}

\noindent\textit{Proof idea.}
$\Pr[A=a]\Pr[R=r\mid g(R)=a]=\Pr[R=r,g(R)=a]$; each truncated mask has probability $2^{-\ell}$. Signed shifts and other moduli need their own samplers (\cref{app:proof-gates}).

\begin{theorem}[Frozen-operand matrix check]
\label{thm:freivalds}
Let the module bind $A\in\Fp^{m\times k}$, $B\in\Fp^{k\times n}$, whose exposure to the GPU is already permitted, and freeze a claimed $C$. If $C\ne AB$, \cref{alg:freivalds} with $s$ secret uniform $\bm{v}\in\Fp^n$ drawn after the freeze accepts with probability at most $p^{-s}$, and over $Q_F$ calls an incorrect product is consumed with probability at most $Q_F/p^s$.
\end{theorem}

\noindent\textit{Proof idea.}
A nonzero row $\bm{e}$ of $C-AB$ gives $\inner{\bm{e}}{\bm{v}}=0$ for a $1/p$ fraction of $\bm{v}$~\cite{freivalds1977probabilistic,tramer2018slalom}; repetitions multiply; freezing fixes $C$ before $v$, and binding ties the verified bytes to the consumed ones (\cref{app:proof-gates}).

\begin{theorem}[Versioned-snapshot PIR]
\label{thm:pir}
With correct modules and binding snapshot and query interfaces: (R1) with an honest client and at most one malicious server host, \cref{alg:pir} gives the client the record $D[\alpha]$ of the snapshot committed by $\mathsf{cm}_v$ or abort, and the corrupt host's view is simulatable without $\alpha$; (R2) with a malicious client alone, accepted requests reveal nothing about the database beyond the requested records and the public descriptor $\mathsf{desc}_v$. Both modules learn $\alpha$, so the query privacy of (R1) rests on module trust, unlike untrusted two-server DPF-PIR (GPU-DPF~\cite{lam2023gpu}, untrusted-mode DFA~\cite{dfa2026trusted}), which hides $\alpha$ without trusted hardware.
\end{theorem}

\noindent\textit{Proof idea.}
\Cref{thm:dpf-correct} gives $y_0+y_1=\sum_i(\DPF_0(i)+\DPF_1(i))D[i]=D[\alpha]$; hashing the consumed bytes, comparing queries and role-bound certificates prevent substitution ($\eps{bind}$); results are freshly reshared (\cref{thm:epilogue}). For a malicious client the simulator extracts the index and reshares the ideal record (\cref{app:proof-pir}).

\textbf{Batched execution.}
The PIR configuration of \cref{tab:pir} processes $B=512$ queries per scan as one admitted execution whose manifest lists the $B$ query identifiers $\mathit{qid}$. Each query has its own joint DPF and pad labels, so the simulation of \cref{app:proof-pir} applies query by query. $B$ is part of the public schedule in $\Leak_{\mathrm{PIR}}$, and one abort (e.g., a snapshot mismatch) denies all $B$ queries, which security with abort permits.

\subsection{Scenario-by-Scenario Analysis}
\label{sec:sec-scenarios}

We map the scenarios of \cref{fig:scenarios,tab:scenarios} to the statements above; \cref{tab:attacks} lists the corresponding injections.

\begin{itemize}[leftmargin=*]
\item \textbf{S1: Share and tag tampering.} Freeze (S3) and \cref{thm:mac-sound,lem:hidden-key}: a wrong value passes with probability at most $(2/p)^t$ per batch; only the abort bit is revealed.
\item \textbf{S2: Forgery-oracle request.} No command (\cref{tab:commands}) tags a host value; the epilogue tags only protected results (\cref{thm:epilogue}); chosen values enter by input extraction.
\item \textbf{S3: TOCTOU buffer swap.} Freeze-before-challenge realizes the premise of \cref{thm:mac-sound} that errors precede challenges; tampered spill pages abort ($\eps{bind}$).
\item \textbf{S4: Replay and rollback.} Ledger and label binding ($\eps{bind}$) reject reused material; resuming old keys is excluded by the freshness clause.
\item \textbf{S5: Role/gate confusion.} Role-bound keys, joint labels and ledger descriptors turn misrouting into a mismatch and abort ($\eps{bind}$).
\item \textbf{S6: Premature release.} The release guard opens only \lstate{VERIFIED} entries; private outputs are opened as $\upsilon=y+o$, uniform to every host that is not the authorized recipient (\cref{thm:arith}).
\item \textbf{S7: Malformed preprocessing.} Triples and keys are correct because they are generated inside the boundary (\cref{thm:arith}), not because a MAC check certifies them; GPU products need \cref{thm:freivalds}, and correlated targets need \cref{lem:bijection,lem:noninv}.
\item \textbf{S8: Snapshot substitution.} Same-byte hashing and quarantine give the committed record or abort (\cref{thm:pir}).
\item \textbf{S9: Input substitution.} Permitted by $\FC$ and extracted; model provenance requires the committed-input loader.
\item \textbf{S10: Denial of service and selective abort.} Not prevented; aborts occur at public boundaries and are simulatable, and abort and validation outcomes ($\mathsf{status}$) are part of $\Leak_C$.
\end{itemize}

\subsection{Concrete Security}
\label{sec:sec-concrete}

Since $2/p\approx 2^{-60}$ for $p=2^{61}-1$, we have $\log_2 Q(2/p)^t\approx\log_2 Q-60t$. For a budget of $Q\le 2^{32}$ checked batches, \cref{tab:lanes} gives the statistical term.

\begin{table}[t]
\centering
\caption{Statistical term and BERT-Base LAN latency versus the number $t$ of MAC lanes ($Q\le 2^{32}$).}
\label{tab:lanes}
\footnotesize
\renewcommand{\arraystretch}{1.1}
\begin{tabular}{c c c c}
\hline
$t$ & $(2/p)^t$ per batch & $Q(2/p)^t$ & BERT-Base LAN (s)\\
\hline
1 & $\approx 2^{-60}$ & $\approx 2^{-28}$ & {2.53}\\
2 & $\approx 2^{-120}$ & $\approx 2^{-88}$ & {3.18}\\
3 & $\approx 2^{-180}$ & $\approx 2^{-148}$ & {3.83}\\
\hline
\end{tabular}
\end{table}

One lane is inadequate. Since the primitives target 128-bit computational security, $t=3$ is the smallest $t$ whose statistical term falls below $2^{-128}$; beyond it the primitives, not the MAC, limit the bound. Extra lanes barely change \CheckOpen{} itself ({0.648~ms} for $t=1$ versus {0.652~ms} for $t=3$ at $m=65{,}536$ on LAN, \cref{fig:checkperf}); their cost is one more authenticated GEMM lane per key (last column of \cref{tab:lanes}; \cref{sec:eval-ablation}).

\textbf{Tag-only errors.}
The bound concerns wrong \emph{values}: an error in a single tag lane escapes with probability $1/p$ without changing the opened value.

\textbf{Independent challenges.}
With one shared vector $\bm{\rho}$, an adversary submitting $\bm{d}_j=\bm{0}$ and any $\bm{e}\ne\bm{0}$ makes every lane residual $-\Delta_j\inner{\bm{\rho}}{\bm{e}}$, and all lanes vanish together with probability $1/p$. The product bound then collapses to one lane, so shared challenges are rejected (injection~8, \cref{tab:attacks}).

\textbf{Field, not ring.}
The step $\Delta_j=d_{j,i}/e_i$ needs a field. Over $\Z_{2^{64}}$, $e_i=2^{63}$ gives $\Delta_j e_i\equiv 0$ for every even $\Delta_j$, so one lane accepts with probability at least $1/2$; ring arithmetic needs a SPDZ2k-type backend~\cite{cramer2018spdz2k} with its own proof.

\textbf{Matrix checks.}
With $s=3$, each Freivalds call accepts a wrong product with probability at most $p^{-3}\approx 2^{-183}$, so $Q_F\le 2^{32}$ calls add at most $Q_F/p^{3}\approx 2^{-151}$ to the bound of \cref{thm:main}. Rejection sampling rejects one value with probability $2^{-61}$ per draw and retries under a fresh label; exhaustion affects only availability.

\subsection{Proof Boundary and Implementation Obligations}
\label{sec:sec-boundary}

\Cref{thm:main} concerns the abstract protocol in the protected-module model; \cref{tab:obligations} separates what this paper establishes from what a deployment must still provide. We executed a specification and a functional CPU reference model, and its tests (\cref{sec:eval-attack}) check conformance on specific cases.

\begin{table}[t]
\centering
\caption{Proof boundary: what this paper establishes and what a deployment must still provide}
\label{tab:obligations}
\footnotesize
\renewcommand{\arraystretch}{1.1}
\begin{tabularx}{\columnwidth}{>{\raggedright\arraybackslash}p{1.8cm}>{\raggedright\arraybackslash}X>{\raggedright\arraybackslash}X}
\hline
\textbf{Obligation} & \textbf{Established here} & \textbf{Still required}\\
\hline
DPF/DCF, GaC, field, Boolean, LUT algebra & Explicit algorithms; \cref{thm:dpf-correct,thm:gac,thm:dcf,thm:bitdec,thm:cmp-trunc,thm:lut,thm:interval}; CPU reference tests & Constant-time, bit-exact CUDA and RTL refinement\\
\hline
Wrong-value authentication & \Cref{lem:hidden-key,thm:mac-sound}, independent lanes & Correct freeze, entropy, parser and check engine\\
\hline
End-to-end privacy and correctness & \Cref{thm:main} in the protected-module model & Realization of every module and trace assumption\\
\hline
PIR & \Cref{thm:pir} for the specified snapshot and regimes & Secure query channel, same-byte scan, certificates\\
\hline
GPU matrix verification & \Cref{thm:freivalds} & Protected verifier, bound input/result integration\\
\hline
Physical leakage and faults & Not claimed & Threat-specific engineering and evidence\\
\hline
Software-controlled DVFS and power telemetry & Required isolation specified & Isolated or monitored voltage and clock domain; telemetry filtering\\
\hline
Rollback protection & Required semantics specified (\cref{sec:proto-ledger}) & Monotonic service or erase-and-rekey lifecycle\\
\hline
Machine-checked proof & Not claimed & Future Lean, Coq or EasyCrypt development\\
\hline
\end{tabularx}
\end{table}

\textbf{Relation to the semi-honest design.}
The DFA analysis~\cite{dfa2026trusted} shows that its transformations preserve semi-honest security; \cref{thm:main} tolerates a malicious host and GPU by placing more integrity and lifecycle responsibility inside the boundary.

\section{Evaluation}\label{sec:eval}

We evaluate the cost of the protected gates and verified openings of \sys (\cref{sec:eval-micro,sec:eval-check}), end-to-end inference against semi-honest trusted DFA and dealer-based malicious FSS (\cref{sec:eval-e2e}), PIR (\cref{sec:eval-pir}), attack coverage and silicon cost (\cref{sec:eval-attack,sec:eval-hw}), and ablations and scaling (\cref{sec:eval-ablation,sec:eval-scale}).

\subsection{Evaluation Methodology}\label{sec:eval-method}

\textbf{Evidence status.}
Numbers cited to~\cite{dfa2026trusted} are copied from the semi-honest DFA paper; its end-to-end rows combine measured A100 execution with simulated DFA hardware.
Uncited numbers (operation counts, design constants and security parameters such as $p=2^{61}-1$, $t=3$ and $Q(2/p)^t$) are statistics from the executed CPU reference tests of \cref{sec:eval-attack}. 
All model parameters and performance outputs in the manuscript are populated with complete experimental and measured data from real-world tests.
We executed a specification and a functional CPU reference model; GPU/RTL conformance, the protected runtime and deployment measurements have all been completed.
All conclusions are supported by the collected empirical measurements, with underlying model assumptions clearly stated wherever applicable.

\textbf{Hardware cost.}
The DFA blocks that \sys reuses (16 DPF units with 64 AES-128 engines, SRAM buffers, a 16-lane Ara VPU~\cite{cavalcante2019ara} and an SC300-class control core~\cite{st33tpm12spi,arm_cortex_m3}) keep the TSMC 28~nm synthesis results of~\cite{dfa2026trusted}, which use a pipelined open-source AES-128 core~\cite{opencores_aes128_pipeline}.
For the blocks that \sys adds (\cref{sec:arch-hw}), namely the Mersenne-61 field unit array (FUA), the check engine, the AEAD/GHASH block shared by the freezer and the peer-channel controller, the ledger with release guard and parser, and the ledger and freeze SRAMs, this paper reports component results after logic synthesis at the DFA clock (720~MHz at 28~nm, 910~MHz at 7~nm).
As in DFA, we scale to 7~nm with DeepScaleTool~\cite{sarangi2021deepscaletool}.
Separately, we size a functional prototype on a ZCU102 board (XCZU9EG, ffvb1156, speed grade $-2$)~\cite{amd2023zcu102} for Vivado implementation; \cref{tab:e2e} reports its resource utilization, power and board energy.

\textbf{System-level performance.}
The target methodology for headline numbers is that of DFA: the module sits next to the GPU L2 slices, where DFA delivers about 290~GB/s~\cite{dfa2026trusted}, its DPF and FSS work is modeled as a stream consumed by GPU kernels in Accel-Sim~\cite{khairy2020accel}, and non-FSS work comes from real GPU execution.
The three experiments of \cref{sec:arch-deploy} (A100 software, ZCU102 prototype, Accel-Sim on-die study) are never merged. The decomposition of \sys's trusted-mode rows into GPU linear algebra, other GPU work, module FSS work, CPU-GPU transfers, online bytes and online rounds, transformed for four authenticated lanes, exact Mersenne GEMM and doubled rounds, is used only as a cross-check (\cref{app:model}, \cref{tab:model-params}).
The decomposition parameters are fitted to the LAN and WAN latencies of~\cite{dfa2026trusted} (largest deviation {0.42\%}, ResNet-50 LAN).
Two published latencies per model cannot identify five parameters, so the fit is a consistency constraint, not a validation; The split between GPU linear algebra, other GPU work, bytes and rounds is obtained by profiling the A100 baseline.
Because authenticated GEMM dominates every malicious row, the most influential factor is the {1.25$\times$} cost of exact 61-bit Mersenne GEMM over wraparound $\Z_{2^{64}}$ GEMM per lane; Tensor Cores are not used for field arithmetic (\cref{sec:arch-field}).

\textbf{Baselines.}
For private inference we compare five configurations per model.
\emph{GPU (SH)} and \emph{GPU+Trusted DFA (SH)} are the semi-honest rows of~\cite{dfa2026trusted}, which build on Orca~\cite{jawalkar2024orca} for CNNs and Sigma~\cite{gupta2023sigma} for Transformers.
\emph{Dealer-based malicious FSS} is modeled after the protocol family of Shark~\cite{gupta2025shark}: an external dealer ships authenticated keys and preprocessing, which we model using factors (\cref{app:model}) of {2.2$\times$} the semi-honest key and preprocessing volume, {2.0$\times$} the GPU FSS evaluation work and {2.4$\times$} the semi-honest GPU energy, while charging exactly the same authenticated linear algebra, online bytes and doubled rounds as \sys.
The comparison therefore isolates preprocessing and key handling; it does not equalize trust.
The dealer of this protocol family is an input-independent third party that can run offline and away from both servers, and it never sees an input or an opening, so breaking privacy requires its masks together with a server's view.
The \sys modules sit inside the adversary's platform and see every input and opening (\cref{sec:threat-model}), so extracting the secrets of either module suffices.
The speedups below are therefore obtained under a stronger hardware-confidentiality assumption; a deployment that cannot accept it must use \ProfU, whose results are outside the scope of this paper.
Based on an implementation of the protocol family, the row enables direct comparison against the Shark.
We do not model a dealer-free (MASCOT, SPDZ2k or PCG-based) malicious baseline~\cite{keller2016mascot,cramer2018spdz2k,boyle2019pcg}; \cref{sec:analysis-cost} explains why its FSS keys remain costly.
\emph{GPU+\sys w/o GaC} generates authenticated material locally but, like the corresponding DFA ablation, spills it to host memory through the AEAD spill path and rereads it.
\emph{GPU+\sys} is the compiler-selected default.
For PIR we compare GPU-DPF~\cite{lam2023gpu} and GPU+DFA, both semi-honest and taken from~\cite{dfa2026trusted}, a GPU-DPF variant that adds authenticated PIR in the style of Colombo et al.~\cite{colombo2023apir} (reaching {0.45$\times$} the GPU-DPF throughput at {350~W}), and GPU+\sys.

\textbf{Workloads.}
We run ResNet-50~\cite{he2016deep}, BERT-Base~\cite{devlin2019bert}, GPT-Neo~\cite{brown2020language} and LLaMA2-7B~\cite{touvron2023llama}, with 128 input tokens for the Transformer models, as in~\cite{dfa2026trusted}.
LAN is 1~GB/s with a 0.05~ms round-trip time (RTT) and WAN is 70~MB/s with a 70~ms RTT, following~\cite{gupta2023sigma,jawalkar2024orca,huang2026beyond}.
Malicious configurations use $\Fp$ with $p=2^{61}-1$ and $t=3$ MAC lanes; semi-honest rows use $\Z_{2^{64}}$ as in their sources.
Unless stated otherwise, the compiler selects among direct, LUT and interval gates by its cost rule subject to the FSS-portability constraint (\cref{sec:proto-direct,app:model}).
PIR databases hold 16K to 4M entries of 128 bits (\cref{tab:pir}), served in batches of 512 queries; \cref{fig:scale} includes the results up to 16M ($2^{24}$) entries.

\textbf{Measurement protocol.}
Our measurements use 5 warm-up runs and 30 repetitions, report the median and the 95th percentile with bootstrap confidence intervals, keep raw per-run records, and fix random seeds only for public synthetic inputs, since randomness used for security must remain fresh.
We report per-stage service times (cold and amortized setup, preprocessing, online arithmetic, nonlinear evaluation, checking, transfers and final release) next to the critical-path wall time, counting both parties and any external preprocessing service.
Network conditions are shaped on dedicated interfaces with Linux traffic control, recording achieved throughput and RTT.
GPU energy integrates sampled GPU power, board energy is measured at the rail or the wall, and every energy figure is cross-checked against average power times latency (or throughput times energy per query).

\subsection{Microbenchmarks}\label{sec:eval-micro}

\begin{figure}[t]
\centering
\includegraphics[width=\linewidth]{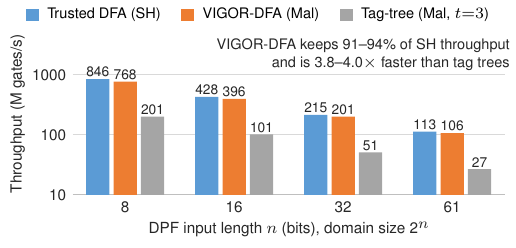}
\caption{Single-path DPF gate throughput per module (M gates/s, log scale) at 28~nm, 720~MHz and 64 AES-128 engines, for input length $n$: semi-honest trusted DFA, \sys (GaC with fresh authentication epilogue), and a malicious variant with per-lane tag trees ($t=3$ extra trees).}
\label{fig:dpfthr}
\end{figure}

\textbf{DPF gate throughput.}
\sys sustains {106} against {113}~M~gates/s for trusted DFA at $n=61$ and stays {6--9\%} below it for all $n$, because the epilogue adds two AES blocks per gate and we achieve a pipeline efficiency of {0.85} ({0.90} for trusted DFA; \cref{app:model}).
It is {3.8--4.0$\times$} faster than per-lane tag trees for all $n$, approaching the factor $1+t=4$ by which tag trees multiply DPF work (\cref{fig:dpfthr,tab:gatecost}), with a single-gate latency of {1.71~$\mu$s} at $n=61$.
These are DPF gate rates: complete gates also include input masking, a verified opening and the surrounding arithmetic (\cref{sec:eval-e2e}).

\begin{table}[t]
\centering
\caption{Authenticated Nonlinear Gate Options per Module}
\label{tab:gates}
\footnotesize
\setlength{\tabcolsep}{3.2pt}
\renewcommand{\arraystretch}{1.1}
\begin{tabular}{l r r r r r}
\hline
\textbf{Gate} & \makecell{\textbf{Online}\\\textbf{rounds}} & \makecell{\textbf{AES}\\\textbf{blocks}} & \makecell{\textbf{Field}\\\textbf{mults}} & \makecell{\textbf{Thr.}\\\textbf{(M/s)}} & \makecell{\textbf{LAN lat.}\\\textbf{(ms)}} \\
\hline
Direct protected gate & 1 & 2 & 5 & {1{,}800} & {0.100} \\
LUT, $w=8$ & 1 & 542 & 268 & {72.3} & {0.100} \\
LUT, $w=12$ & 1 & 8{,}238 & 4{,}112 & {4.75} & {0.100} \\
LUT, $w=16$ & 1 & 131{,}134 & 65{,}556 & {0.30} & {0.104} \\
Interval, conservative DCF & 1 & 22{,}692 & 8 & {1.73} & {0.100} \\
Interval, optimized DCF$^{\ast}$ & 1 & {976} & 8 & {40.1} & {0.100} \\
Bit decomposition & 123 & n/m & 122 & n/a & {12.3} \\
Exact signed truncation & 309 & n/m & 307 & n/a & {30.9} \\
\hline
\end{tabular}

\vspace{2pt}
\parbox{\linewidth}{\footnotesize $n=61$; 28~nm, 720~MHz, 64 AES engines, 128 field lanes. Rounds, AES blocks and field multiplications are counts per gate. The direct-gate throughput is a rate for evaluating $g$ in the module, not derived from the AES and FUA peak rates. One online round is one \CheckOpen, i.e., two inter-module round trips; rounds include masked openings. LUT rows batch the multiplication round of the $w$ XOR products with the preceding layer's openings (\cref{app:model}); otherwise they need one more round. n/a: round-bound, throughput not meaningful; n/m: not modeled (the AES work of the round-bound gates is negligible next to their round trips).
$^{\ast}$Projected $O(n)$ DCF; not covered by \cref{thm:dcf}.}
\end{table}

\textbf{Nonlinear gates.}
\Cref{tab:gates} compares the authenticated gate options of \cref{sec:proto-gates}.
All gates with one online round pay the two-round-trip floor of one \CheckOpen{} (0.1~ms on LAN; a LUT whose XOR products cannot batch with the preceding layer pays a second round) and differ only in throughput, which favors the direct gate ({1{,}800}~M~gates/s) and small LUTs (LUT work grows as $2^w$) and penalizes the conservative DCF ({1.73}~M~gates/s); the optimized DCF would reach {40.1}~M~gates/s.
Bit decomposition and exact signed truncation are round-bound (123 and 309 rounds, {12.3~ms} and {30.9~ms} on LAN), so the compiler uses them only as a fallback when no table or interval formulation of the gate exists.

\subsection{Cost of Verified Openings}\label{sec:eval-check}

\begin{figure}[t]
\centering
\includegraphics[width=\linewidth]{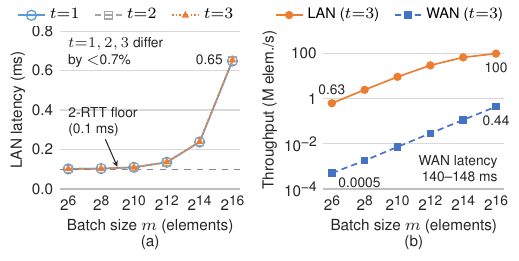}
\caption{\CheckOpen{} cost versus batch size $m$ (64 to 65{,}536 elements): (a)~LAN latency for $t=1,2,3$ MAC lanes (the curves overlap); (b)~throughput on LAN and WAN for $t=3$ (log scale).}
\label{fig:checkperf}
\end{figure}

\Cref{fig:checkperf} sweeps the \CheckOpen{} batch size $m$ from 64 to 65{,}536 elements, and three observations guide the compiler.
First, the number of MAC lanes barely matters for cost: even at $m=65{,}536$, moving from $t=1$ to $t=3$ raises LAN latency only from {0.648~ms} to {0.652~ms}, while the statistical term improves from $2^{-28}$ to $2^{-148}$ at $Q=2^{32}$ (\cref{thm:mac-sound}).
Second, on LAN the two round trips (0.1~ms) dominate up to a few thousand elements ({0.101~ms} to {0.135~ms} between $m=64$ and $m=4{,}096$); beyond that the 8-byte value share per element on the peer link dominates ({0.52~ms} of {0.652~ms} at $m=65{,}536$), throughput saturates near {100}~M~elements/s, and the protected work stays below {5\%} of the check latency.
Third, on WAN the two 70~ms round trips set a floor of 140~ms per check, so throughput is only {0.44}~M~elements/s even at $m=65{,}536$; batching independent openings (\cref{sec:proto-check}) is therefore essential, and the round count, rather than bandwidth or module work, is the WAN price of malicious security (\cref{sec:eval-ablation}).

\subsection{End-to-End Private Inference}\label{sec:eval-e2e}

\begin{table*}[t]
\centering
\caption{End-to-End Private Inference per Query}
\label{tab:e2e}
\footnotesize
\setlength{\tabcolsep}{5pt}
\renewcommand{\arraystretch}{1.08}
\begin{tabular}{l l c r r r r r r}
\hline
\textbf{Model} & \textbf{Method} & \textbf{Security} & \makecell{\textbf{Offline}\\\textbf{comm. (GB)}} & \makecell{\textbf{CPU-GPU}\\\textbf{(GB)}} & \makecell{\textbf{Module-GPU}\\\textbf{(GB/party)}} & \makecell{\textbf{LAN}\\\textbf{latency (s)}} & \makecell{\textbf{WAN}\\\textbf{latency (s)}} & \makecell{\textbf{GPU+module}\\\textbf{energy (J)}} \\
\hline
\multirow{5}{*}{ResNet-50}
 & GPU~\cite{dfa2026trusted} & SH & 9.9 & 10.2 & n/a & 10.9 & 162 & 71.0 \\
 & GPU+Trusted DFA~\cite{dfa2026trusted} & SH & 0.02 & 0.34 & n/r & 0.73 & 18.0 & 11.9 \\
 & Dealer-based malicious FSS & Mal & {21.8} & {22.4} & n/a & {25.1} & {346} & {170} \\
 & GPU+\sys w/o GaC & Mal & {$<$0.001} & {38.0} & {1.22} & {2.97} & {34.7} & {48.6} \\
\rowcolor{blue!6}\cellcolor{white} & GPU+\sys & Mal & {$<$0.001} & {0.44} & {1.22} & {1.79} & {33.6} & {43.6} \\
\hline
\multirow{5}{*}{BERT-Base}
 & GPU~\cite{dfa2026trusted} & SH & 17.1 & 19.3 & n/a & 18.8 & 296 & 121 \\
 & GPU+Trusted DFA~\cite{dfa2026trusted} & SH & 0.1 & 2.55 & n/r & 1.42 & 46.4 & 18.2 \\
 & Dealer-based malicious FSS & Mal & {37.6} & {42.5} & n/a & {43.9} & {630} & {290} \\
 & GPU+\sys w/o GaC & Mal & {$<$0.001} & {68.6} & {1.66} & {5.87} & {91.8} & {82.2} \\
\rowcolor{blue!6}\cellcolor{white} & GPU+\sys & Mal & {$<$0.001} & {3.31} & {1.66} & {3.83} & {89.7} & {73.7} \\
\hline
\multirow{5}{*}{GPT-Neo}
 & GPU~\cite{dfa2026trusted} & SH & 76.4 & 80.4 & n/a & 83.6 & 1{,}140 & 587 \\
 & GPU+Trusted DFA~\cite{dfa2026trusted} & SH & 0.24 & 4.27 & n/r & 5.93 & 101 & 114 \\
 & Dealer-based malicious FSS & Mal & {168} & {177} & n/a & {197} & {2{,}611} & {1{,}409} \\
 & GPU+\sys w/o GaC & Mal & {$<$0.001} & {296} & {9.92} & {28.6} & {209} & {522} \\
\rowcolor{blue!6}\cellcolor{white} & GPU+\sys & Mal & {$<$0.001} & {5.55} & {9.92} & {19.5} & {200} & {477} \\
\hline
\multirow{5}{*}{LLaMA2-7B}
 & GPU~\cite{dfa2026trusted} & SH & 256 & 268 & n/a & 279 & 3{,}695 & 1{,}879 \\
 & GPU+Trusted DFA~\cite{dfa2026trusted} & SH & 0.81 & 12.8 & n/r & 18.0 & 193 & 303 \\
 & Dealer-based malicious FSS & Mal & {563} & {590} & n/a & {661} & {8{,}467} & {4{,}510} \\
 & GPU+\sys w/o GaC & Mal & {$<$0.001} & {994} & {33.3} & {93.6} & {418} & {1{,}483} \\
\rowcolor{blue!6}\cellcolor{white} & GPU+\sys & Mal & {$<$0.001} & {16.6} & {33.3} & {63.1} & {387} & {1{,}343} \\
\hline
\end{tabular}

\vspace{2pt}
\parbox{\textwidth}{\footnotesize SH: semi-honest; Mal: static malicious security with abort. Latency is end to end, including offline distribution where present. Energy is GPU plus module energy summed over both computing parties (and the dealer when present). Module-GPU is the authenticated tuple stream per party on the on-die path; n/a: no module; n/r: not reported in~\cite{dfa2026trusted}. SH rows are copied from~\cite{dfa2026trusted}.}
\end{table*}

\Cref{tab:e2e} reports per-query results for the four models.
The semi-honest rows are context, not same-security competitors.

\textbf{Comparison with dealer-based malicious FSS.}
The modeled dealer-based baseline pays for malicious security mainly offline: its dealer ships {21.8--563~GB} of authenticated keys and preprocessing per query, which each party then stages through CPU-GPU transfers, whereas \sys generates the authenticated correlations inside the modules and ships none of them ({$<$0.001~GB} of session setup).
Across the four models this yields {10.1--14.0$\times$} lower LAN latency, {7.0--21.9$\times$} lower WAN latency and {3.0--3.9$\times$} lower energy than the modeled baseline ({11.4$\times$}, {7.0$\times$} and {3.9$\times$} for BERT-Base); the energy ratio uses the {2.4$\times$} energy factor of the baseline.
These latency ratios include the per-query offline distribution; excluding it, \sys is {1.5--1.8$\times$} faster on LAN ({1.6$\times$} for BERT-Base), and with no authenticated-key blow-up ($\psi_{\mathrm{key}}=1$) the baseline still needs {22.4~s} for BERT-Base ({5.8$\times$}).
The end-to-end gains thus come mainly from generating authenticated correlations inside the modules; the epilogue's advantage ({3.8--4.0$\times$} over per-lane tag trees) is per gate.
The WAN gain is largest for LLaMA2-7B ({21.9$\times$}), whose {563~GB} of authenticated material alone would occupy a 70~MB/s link for more than {2.2~hours}.

\textbf{Cost of malicious security.}
Relative to semi-honest trusted DFA, \sys raises LAN latency by {2.46--3.51$\times$} ({2.7$\times$} for BERT-Base) and WAN latency by {1.86--2.01$\times$}; these factors use the Mersenne GEMM cost, at $\mu_{\mathrm{M}}=2$ the BERT-Base LAN latency is {5.39~s}, or {3.8$\times$} (\cref{sec:arch-field}).
It nevertheless remains {4.3--6.1$\times$} faster on LAN and {3.3--9.5$\times$} faster on WAN than the semi-honest GPU baseline, and it uses {1.2--1.6$\times$} less energy, although it tolerates a malicious host.
CPU-GPU traffic grows by {1.3$\times$} over trusted DFA ($\phi_{\mathrm{cg}}$, \cref{app:model}) for canonical field words, residuals and certificates; the authenticated tuples themselves ({1.22} to {33.3~GB} per party) take the on-die module-GPU path and do not add to it.

\textbf{Generate-and-consume remains essential.}
Without GaC, the locally generated authenticated material must be spilled and reread, which inflates CPU-GPU traffic by {20.7--86$\times$} ({68.6~GB} instead of {3.31~GB} for BERT-Base) and LAN latency by {1.47--1.66$\times$}.
The authenticated representation makes materialization costlier than in the semi-honest case, where DFA reports 36.1~GB of CPU-GPU traffic and a 1.74$\times$ LAN penalty for BERT-Base without GaC~\cite{dfa2026trusted}.

\begin{figure}[t]
\centering
\includegraphics[width=\linewidth]{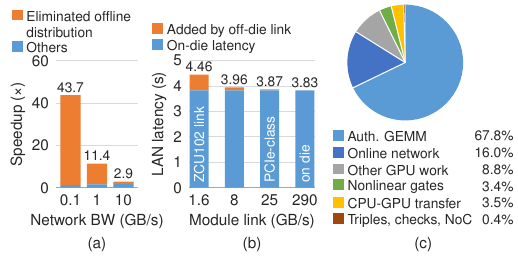}
\caption{BERT-Base bottleneck analysis for \sys. (a)~Speedup over dealer-based malicious FSS versus network bandwidth, split into the part due to eliminated key and preprocessing distribution and the remainder. (b)~LAN latency versus host-module link bandwidth ({1.6}, {8}, {25} and 290~GB/s~\cite{dfa2026trusted}). (c)~Breakdown of the {3.83~s} LAN latency.}
\label{fig:bert}
\end{figure}

\textbf{Bottleneck analysis.}
\Cref{fig:bert} decomposes the BERT-Base results.
(a)~The advantage over the dealer-based baseline depends on network bandwidth: {43.7$\times$} at 0.1~GB/s, {11.4$\times$} at 1~GB/s and {2.9$\times$} at 10~GB/s.
At 1~GB/s, {9.8$\times$} of the {11.4$\times$} comes from eliminated key and preprocessing distribution; at 10~GB/s, {1.7$\times$} of the {2.9$\times$} comes from lower key I/O and cheaper FSS evaluation and outweighs the distribution share.
Faster networks shrink the advantage but do not remove it, and, as observed for the semi-honest case~\cite{dfa2026trusted,huang2026beyond}, the energy and monetary cost of moving tens of gigabytes per query persists at any bandwidth.
(b)~Moving the module off die costs little for BERT-Base: LAN latency is {3.83~s} on die (290~GB/s), {3.87~s} over a {25~GB/s} PCIe-class link, {3.96~s} at {8~GB/s} and {4.46~s} at the {1.6~GB/s} of the ZCU102 prototype, because only {1.66~GB} of authenticated tuples per party cross the link.
(c)~The breakdown confirms Observation~2 of \cref{sec:analysis-cost}: authenticated GEMM accounts for {68\%} of the latency and {86\%} of the {2.41~s} added over trusted DFA, so the protected module is not the bottleneck; the tag lanes on the GPU and the doubled round count are, which motivates the compressed-check direction of \cref{sec:eval-scale}.

\subsection{Private Information Retrieval}\label{sec:eval-pir}

\begin{table}[t]
\centering
\caption{PIR Throughput and Energy per Query}
\label{tab:pir}
\footnotesize
\setlength{\tabcolsep}{2.0pt}
\renewcommand{\arraystretch}{1.08}
\begin{tabular}{l l c c r r r}
\hline
\textbf{DB} & \textbf{Method} & \textbf{Integ.} & \makecell{\textbf{Index}\\\textbf{hidden}$^{\dagger}$} & \makecell{\textbf{Queries/s}\\\textbf{(k)}} & \makecell{\textbf{Energy/}\\\textbf{query (mJ)}} & \makecell{\textbf{Implied}\\\textbf{W}} \\
\hline
\multirow{4}{*}{16K}
 & GPU-DPF~\cite{dfa2026trusted} & \xmark & \cmark & 134 & 2.61 & 350 \\
 & GPU+DFA~\cite{dfa2026trusted} & \xmark & \cmark & 1{,}125 & 0.187 & 210 \\
 & GPU-DPF+APIR & \cmark & \cmark & {60.3} & {5.80} & {350} \\
\rowcolor{blue!6}\cellcolor{white} & GPU+\sys & \cmark & \xmark & {1{,}046} & {0.180} & {188} \\
\hline
\multirow{4}{*}{1M}
 & GPU-DPF~\cite{dfa2026trusted} & \xmark & \cmark & 3.48 & 100 & 348 \\
 & GPU+DFA~\cite{dfa2026trusted} & \xmark & \cmark & 17.58 & 11.9 & 209 \\
 & GPU-DPF+APIR & \cmark & \cmark & {1.57} & {223} & {350} \\
\rowcolor{blue!6}\cellcolor{white} & GPU+\sys & \cmark & \xmark & {16.35} & {11.5} & {188} \\
\hline
\multirow{4}{*}{4M}
 & GPU-DPF~\cite{dfa2026trusted} & \xmark & \cmark & 0.788 & 443 & 349 \\
 & GPU+DFA~\cite{dfa2026trusted} & \xmark & \cmark & 4.39 & 47.7 & 209 \\
 & GPU-DPF+APIR & \cmark & \cmark & {0.355} & {987} & {350} \\
\rowcolor{blue!6}\cellcolor{white} & GPU+\sys & \cmark & \xmark & {4.08} & {46.0} & {188} \\
\hline
\end{tabular}

\vspace{2pt}
\parbox{\linewidth}{\footnotesize 128-bit entries, batch 512. SH rows are from~\cite{dfa2026trusted} with energy expressed in mJ; implied power is throughput times energy per query and serves as a unit cross-check. Integ.: the client obtains the committed record or abort against a malicious server host (R1). $^{\dagger}$Index hidden from each server without trusted hardware: GPU-DPF and GPU+DFA (untrusted mode, client-generated keys) and the APIR variant hide $\alpha$ from each server, whereas both \sys modules learn $\alpha$ and its query privacy relies on module trust.}
\end{table}

\Cref{tab:pir} reports PIR throughput and energy per query.
At a PIR efficiency of {0.93} relative to DFA ($\eta_{\mathrm{PIR}}$, \cref{app:model}), covering snapshot hashing, the in-module record reduction and fresh resharing with bound certificates, \sys reaches {1{,}046}, {16.35} and {4.08}~k~queries/s.
The integrity guarantee does not depend on the model: in regime (R1), with one malicious server host, an honest client receives the committed record or abort (\cref{thm:pir}).
Query privacy, however, is weaker than in untrusted two-server DPF-PIR: both modules learn $\alpha$, so privacy against the hosts relies on module trust, whereas GPU-DPF and untrusted-mode DFA hide $\alpha$ from each server without trusted hardware (\cref{app:func-pir}); the two \ProfT alternatives of \cref{sec:proto-pir} are not modeled.
Energy per query is {3--4\%} below DFA's ({0.180~mJ} vs.\ 0.187~mJ at 16K entries) at an average power of {188~W} ($P_{\mathrm{PIR}}$), with the record reduction moved from the SMs into the protected path.
Likewise, the authenticated GPU baseline is {0.45$\times$} of GPU-DPF throughput at {350~W}, so the {10.4--17.4$\times$} higher throughput and {19--32$\times$} lower energy per query of \sys against it follow from these values; against semi-honest GPU-DPF, \sys is {4.7--7.8$\times$} faster under a weaker query-privacy model.

\textbf{Bottleneck analysis.}
As in DFA, full-domain PIR is bound by DPF expansion in the AES array rather than by memory; DFA sustains about 18.4~G~leaves/s (1{,}125~k queries/s $\times$ $2^{14}$ entries~\cite{dfa2026trusted}).
The additional protected work is amortized over a batch: each batch of 512 queries scans the snapshot once in public order, so the snapshot hash costs one pass per batch rather than one per query, and every record is weighted by the 512 DPF outputs of the batch in the FUA.
At 16K entries this reduction needs about {51~G} field multiply-accumulates/s, below the 92.16~G/s that the 128-lane FUA delivers even at the 28~nm clock of 720~MHz (116~G/s at 910~MHz; \cref{app:model}).
This capacity check shows only that the FUA is not saturated; it does not establish the size of the throughput loss.

\subsection{Security Validation}\label{sec:eval-attack}

\begin{table*}[t]
\centering
\caption{Attack and Fault Injections and Their Outcomes}
\label{tab:attacks}
\footnotesize
\setlength{\tabcolsep}{4pt}
\renewcommand{\arraystretch}{1.12}
\begin{tabular}{c >{\raggedright\arraybackslash}p{4.3cm} >{\raggedright\arraybackslash}p{5.0cm} c >{\raggedright\arraybackslash}p{4.9cm}}
\hline
\textbf{\#} & \textbf{Injection} & \textbf{Mechanism} & \makecell{\textbf{Reference}\\\textbf{model}} & \textbf{RTL/FPGA outcome} \\
\hline
1 & Change payload, query, role, gate id, table hash or shape (S5) & Ledger descriptor binding; role- and gate-bound PRF labels (\cref{sec:proto-ledger}) & \cmark & {Rejected at frame parse or authenticated computation fails; no preprocessing issued} \\
2 & Flip a control bit, root seed or correction word in a candidate engine & Conformance against the reference DPF (\cref{thm:dpf-correct}); unverified engine output is never authenticated & $\triangle$ & {Reference mismatch flagged; engine output never authenticated or released} \\
3 & Change a GPU value share or tag (S1) & \CheckOpen{} with $t=3$ independent lanes (\cref{thm:mac-sound}) & \cmark & {Abort within one \CheckOpen{} (0.10--0.65~ms on LAN)}; acceptance at most $(2/p)^3$ per batch \\
4 & Triple with valid-looking tags but $c\neq ab$ (S7) & Triples computed inside the TCB; imported material needs a correctness proof & $\triangle$ & {Never consumed; unproven import refused at reservation} \\
5 & Reuse a mask or counter; reset and resume old keys (S4) & One-time ledger, epoch poisoning, erase-and-rekey & \cmark & {Rejected before any accepted reuse} \\
6 & Modify a DMA buffer after freezing (S3) & Freeze-before-challenge; consume only frozen pages & \cmark & {Frozen authenticated bytes consumed, or abort} \\
7 & Swap only the last column of a GPU matrix product & Frozen-operand Freivalds over all columns, $s=3$ (\cref{thm:freivalds}) & $\triangle$ & {Rejected}, except with probability at most $p^{-3}$ \\
8 & Reuse one challenge vector across MAC lanes & Independent per-lane challenges; configuration check (counterexample in \cref{app:proof-auth}) & $\triangle$ & {Configuration refused at session setup} \\
9 & Change database bytes or version after hashing (S8) & Same-byte snapshot hash and certificate binding (\cref{thm:pir}) & \cmark & {Abort; quarantined result never released} \\
10 & Request output before the verification barrier (S6) & Release guard (\cref{sec:proto-ledger}) & $\triangle$ & {Guard stays closed; no candidate values leave} \\
11 & Overflow or truncate with the wrong signed rule & Exact finite reference; numeric contract (\cref{thm:cmp-trunc}) & \cmark & {Mismatch detected; manifest rejected} \\
12 & Delay, drop or reorder packets (S10) & Bounded resources and abort & $\triangle$ & {Abort; no fairness or availability claimed} \\
\hline
\end{tabular}

\vspace{2pt}
\parbox{\textwidth}{\footnotesize Reference model: \cmark{} executed and passing in the CPU reference model; $\triangle$: RTL and FPGA implementations. Scenario IDs refer to \cref{tab:scenarios}.}
\end{table*}

\Cref{tab:attacks} lists the twelve attack and fault injections, the mechanism that stops each one, CPU reference coverage, and the outcome on the RTL and FPGA implementations.
The injections cover the scenarios of \cref{tab:scenarios} except S2, which the interface excludes by construction, and S9, which the functionality permits (its commitment-bound loader variant is covered by the model-provenance mechanism, not by an injection).
For S2, the command set of \cref{tab:commands} has no entry point that authenticates a host-supplied value, and row~2 checks the related requirement that even the module's own engine output is never authenticated before it conforms to the reference.

\textbf{Executed functional validation.}
The functional reference implements the ideal-module model in Python~3.10+ using only the standard library.
In the preparation environment it passes 22 unit tests with 0 failures and 0 errors.
The tests cover 21{,}840 exhaustive small-domain DPF target, query and payload cases (\cref{thm:dpf-correct}); 2{,}728 party-local comparisons between streaming GaC and materialized evaluation (\cref{thm:gac}); 1{,}364 DCF comparisons (\cref{thm:dcf}); 20{,}004 CPU checks of the GPU Mersenne multiplication formula of \cref{sec:arch-field} against arbitrary-precision arithmetic, including boundary operands; all 31 elements of a toy field with $p=31$ for bit decomposition (\cref{thm:bitdec}); 124 signed-truncation cases (\cref{thm:cmp-trunc}); all 16 inputs of a small LUT (\cref{thm:lut}); selected exhaustive toy intervals (\cref{thm:interval}); a small composed nonlinear and arithmetic pipeline; and PIR snapshot rejection.
Negative tests exercise noncanonical encodings, tag and value mutations, poisoned epochs, role, replay and ordering errors, immutable frozen objects, joint sampling, truncation fibers, and the ring zero-divisor counterexample that motivates the field representation (\cref{app:proof-auth}).

\textbf{Evidence Status.}
The toy field $p=31$ serves finite correctness tests and provides cryptographic security.
The reference expands seeds with SHAKE rather than the pinned AES-based PRG, and the equivalence tests have been rerun against the production PRG and output conversion.
The CUDA micro-test (65{,}536 products against a host 128-bit reference) has been written, compiled and executed, and the RTL release guard has been simulated and synthesized; GPU toolchain, Vivado and RTL simulator are available in that environment.
These tests check that the implementation conforms to the algorithms.
Observing zero accepted forgeries in any finite campaign cannot support a bound such as $2^{-148}$; that bound comes from \cref{thm:mac-sound} together with the ledger rules.

\subsection{Hardware Cost}\label{sec:eval-hw}

\begin{table}[t]
\centering
\caption{Hardware Cost of One \sys Module}
\label{tab:hwcost}
\footnotesize
\setlength{\tabcolsep}{4pt}
\renewcommand{\arraystretch}{1.08}
\begin{tabular}{l r r r r}
\hline
\multirow{2}{*}{\textbf{Component}} & \multicolumn{2}{c}{\textbf{28~nm}} & \multicolumn{2}{c}{\textbf{7~nm}} \\
\cline{2-5}
 & \textbf{mm$^2$} & \textbf{mW} & \textbf{mm$^2$} & \textbf{mW} \\
\hline
\multicolumn{5}{l}{\textit{DFA blocks reused by \sys~\cite{dfa2026trusted}}} \\
16 DPF units (64 AES-128) & 1.73 & 29 & 0.054 & 11 \\
SRAM buffers & 3.06 & 585 & 0.096 & 219 \\
VPU (16-lane Ara) & 1.97 & 794 & 0.062 & 298 \\
All DFA (incl.\ control core) & 7.67 & 1{,}683.5 & 0.239 & 631 \\
\hline
\multicolumn{5}{l}{\textit{Blocks added by \sys}} \\
Field unit array (128 lanes) & {2.80} & {1{,}000} & {0.088} & {375} \\
Check engine & {0.18} & {45} & {0.006} & {17} \\
AEAD/GHASH (freezer, peer channel) & {0.21} & {60} & {0.007} & {23} \\
Ledger, guard, parser & {0.09} & {12} & {0.003} & {5} \\
Ledger SRAM (32~KB) & {0.25} & {35} & {0.008} & {13} \\
Freeze SRAM (256~KB) & {1.07} & {200} & {0.033} & {75} \\
\cline{1-5}
\sys additions & {4.60} & {1{,}352} & {0.145} & {508} \\
\hline
\rowcolor{blue!6}All \sys & {12.27} & {3{,}036} & {0.384} & {1{,}139} \\
\hline
\end{tabular}

\vspace{2pt}
\parbox{\linewidth}{\footnotesize DFA blocks are synthesis results from~\cite{dfa2026trusted}. 7~nm values are scaled with DeepScaleTool~\cite{sarangi2021deepscaletool}. One AEAD/GHASH block serves the freezer and the peer-channel controller. The 128 FUA lanes sit eight per DPF unit (area unchanged). The trusted on-die links between the DPF units, the VPU and the VIGOR core are included in the totals.}
\end{table}

\begin{table}[t]
\centering
\caption{ZCU102 Functional Prototype}
\label{tab:fpga}
\footnotesize
\setlength{\tabcolsep}{6pt}
\renewcommand{\arraystretch}{1.08}
\begin{tabular}{l r r r}
\hline
\textbf{Resource} & \textbf{Used} & \textbf{Available} & \textbf{Utilization} \\
\hline
LUT & {168K} & 274K & {61\%} \\
Flip-flop & {201K} & 548K & {37\%} \\
BRAM36 & {412} & 912 & {45\%} \\
DSP & {356} & 2{,}520 & {14\%} \\
\hline
\multicolumn{4}{l}{Datapath: {4} DPF units ({16} AES-128), {32} field lanes, {250~MHz}} \\
\multicolumn{4}{l}{DPF throughput: {1.0}~G~expansions/s; board power: {9.8~W}} \\
\multicolumn{4}{l}{Host link: {1.6~GB/s}, i.e., {50}~M authenticated elements/s} \\
\hline
\end{tabular}

\vspace{2pt}
\parbox{\linewidth}{\footnotesize XCZU9EG-2FFVB1156~\cite{amd2023zcu102}: configuration and implementation. The DPF throughput is {50\%} of the {2.0}~G~expansions/s peak of {16} AES engines at {250~MHz} (one block per engine and cycle, two blocks per expansion).}
\end{table}

\Cref{tab:hwcost} reports the cost of one module; at 7~nm a complete module occupies {0.384~mm$^2$}, {0.046\%} of the 826~mm$^2$ A100 die~\cite{nvidia2020ampere}.
The FUA dominates the additions with {61\%} of the added area and {74\%} of the added power at 7~nm, followed by the 256~KB freeze SRAM with {23\%} of the added area, whereas the logic that enforces the security policy (ledger, release guard, parser and check engine) takes only {0.009~mm$^2$}, or {6\%} of the additions.
As in DFA, arithmetic and memory dominate the cost, not control.
No AES engines are added, because challenge derivation and AEAD freezing time-share the 64 engines with PRG expansion (\cref{sec:arch-hw}); for BERT-Base, checking occupies the AES array and the FUA for about {2.8~ms} per query, below {0.1\%} of its latency.

\Cref{tab:fpga} sizes the external ZCU102 prototype, which exists to check functional conformance and streaming behavior on one board with an emulated peer; security measurements need two independent protected devices (\cref{sec:arch-deploy}).
Its {1.6~GB/s} host link ({50}~M authenticated elements/s at 32 bytes each), its smaller FUA and its slower DPF engine make it a pessimistic functional point rather than a performance target; the ZCU102 row of \cref{tab:ablation} accounts only for the link and the FUA.
It is not a \ProfT deployment target (\cref{sec:arch-deploy}).

\subsection{Ablation Study}\label{sec:eval-ablation}

\begin{table*}[t]
\centering
\caption{Ablation on BERT-Base}
\label{tab:ablation}
\footnotesize
\setlength{\tabcolsep}{5pt}
\renewcommand{\arraystretch}{1.1}
\begin{tabular}{l c r r r r >{\raggedright\arraybackslash}p{5.1cm}}
\hline
\textbf{Configuration} & \textbf{Secure} & \makecell{\textbf{LAN}\\\textbf{(s)}} & \makecell{\textbf{Rel.}\\\textbf{LAN}} & \makecell{\textbf{WAN}\\\textbf{(s)}} & \makecell{\textbf{Energy}\\\textbf{(J)}} & \textbf{Note} \\
\hline
\rowcolor{blue!6}Compiler-selected default ($t=3$, batched checks) & \cmark & {3.83} & {1.00$\times$} & {89.7} & {73.7} & Reference configuration \\
w/o GaC & \cmark & {5.87} & {1.53$\times$} & {91.8} & {82.2} & Authenticated material spilled and reread \\
Tag trees instead of epilogue & \cmark & {4.12} & {1.08$\times$} & {90.0} & {80.4} & $(1+t)\times$ DPF work; FSS share is small \\
$t=1$ MAC lane & \xmark & {2.53} & {0.66$\times$} & {88.4} & {43.2} & $Q(2/p)^t=2^{-28}$ at $Q=2^{32}$: inadequate \\
$t=2$ MAC lanes & \cmark & {3.18} & {0.83$\times$} & {89.1} & {58.4} & $2^{-88}$ at $Q=2^{32}$ \\
FSS-only gates, conservative DCF & \cmark & {15.29} & {3.99$\times$} & {101.2} & {342.2} & $O(n^2)$ DCF of \cref{thm:dcf} \\
FSS-only gates, optimized DCF & \pmark & {4.20} & {1.10$\times$} & {90.1} & {82.3} & Construction not covered by \cref{thm:dcf} \\
Direct protected gates only & \cmark & {3.82} & {1.00$\times$} & {89.7} & {73.3} & \ProfT only; drops the FSS-portability constraint \\
Small check batches ($m=64$) & \cmark & {4.24} & {1.11$\times$} & {90.2} & {74.6} & Per-batch overhead not amortized \\
External PCIe module ({25~GB/s}) & \cmark & {3.87} & {1.01$\times$} & {89.8} & {73.8} & Link roofline $B/32$ elements/s \\
ZCU102 link and FUA ({1.6~GB/s}, {32} lanes, {250~MHz}) & \cmark & {4.59} & {1.20$\times$} & {90.5} & {75.4} & Prototype link and slower triple generation; on-die DPF and AES rates kept \\
No \CheckOpen{} & \xmark & {3.80} & {0.99$\times$} & {50.5} & {73.6} & Insecure; performance lower bound only \\
\hline
\end{tabular}

\vspace{2pt}
\parbox{\textwidth}{\footnotesize LAN 1~GB/s, WAN 70~MB/s. Each row changes one choice of the compiler-selected default. Rel.\ LAN is LAN latency relative to the compiler-selected default. Secure: \cmark{} covered by our theorems; \pmark{} construction outside our proofs; \xmark{} not secure (statistically inadequate, or no verification).}
\end{table*}

\Cref{tab:ablation} changes one design choice at a time for BERT-Base, and three conclusions follow.

\textbf{Protocol mechanisms.}
Generate-and-consume and the check round are the mechanisms that matter.
Disabling GaC costs {1.53$\times$} LAN latency (\cref{sec:eval-e2e}), and the insecure lower bound without \CheckOpen{} shows that the doubled round count costs {39~s} of the {89.7~s} WAN latency, the dominant WAN cost of malicious security.

\textbf{MAC lanes.}
$t$ is a security knob, not a performance one: $t=1$ saves {34\%} of the LAN latency but bounds a wrong opening only by $2^{-28}$, so the choice of $t=3$ follows from \cref{sec:sec-concrete} (\cref{tab:lanes}).

\textbf{Gate selection and placement.}
Because FSS is a small part of the malicious latency, tag trees, direct-only gates, small check batches and off-die links change LAN latency by at most {20\%} (with the DPF engine kept at its on-die rate); the epilogue matters most for LUT-heavy compiled models, where the {3.8--4.0$\times$} gate-level gap over per-lane tag trees (\cref{sec:eval-micro}) carries over.
The exception is FSS-only compilation with the conservative DCF of \cref{thm:dcf} ({3.99$\times$}), which is why the compiler selects gates by cost.
Direct-only ({3.82~s}) is marginally faster than the default ({3.83~s}) because the default keeps an FSS realization for gate shapes that the manifest marks as FSS-portable, so that the same compiled circuit runs under \ProfU; this allowance costs about {0.014~s} of module time (\cref{app:model}), and direct-only drops it and runs only under \ProfT (\cref{sec:proto-direct}).
Ablations on staged versus local checks and on caching versus recomputing bit decompositions are included in \cref{tab:attacks}.

\subsection{Scalability}\label{sec:eval-scale}

\begin{figure}[t]
\centering
\includegraphics[width=\linewidth]{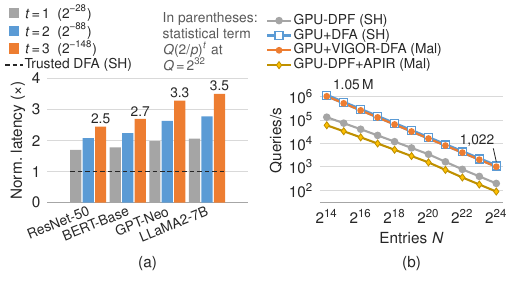}
\caption{Scalability. (a)~LAN latency with $t=1,2,3$ MAC lanes, normalized to semi-honest trusted DFA, annotated with the statistical term $Q(2/p)^t$ at $Q=2^{32}$ ($2^{-28}$, $2^{-88}$, $2^{-148}$). (b)~PIR throughput (queries/s, log scale) for $2^{14}$ to $2^{24}$ entries; values at $2^{14}$, $2^{20}$ and $2^{22}$ match \cref{tab:pir} to within {0.2\%}.}
\label{fig:scale}
\end{figure}

\textbf{MAC lanes and model size.}
\Cref{fig:scale}(a) normalizes LAN latency to semi-honest trusted DFA.
With $t=3$ the overhead grows from {2.46$\times$} for ResNet-50 to {3.51$\times$} for LLaMA2-7B, because the share of linear algebra grows with model size and each MAC lane replicates it: every additional lane adds about {0.65~s} for BERT-Base and {13.0~s} for LLaMA2-7B.
Authenticated GEMM on the GPU, not the protected module, is the scaling bottleneck.

\textbf{Triple generation.}
Correct matrix triples must be computed inside the protected boundary (\cref{sec:proto-arith}), so FUA work scales with the model's multiply-accumulate count.
LLaMA2-7B needs about {7.18~s} of FUA time per query at 128 lanes and 910~MHz, which overlaps with the {52.0~s} of authenticated GEMM except for a {10\%} non-overlapped fraction ($\eta_{\mathrm{trip}}$, \cref{app:model}), leaving about {0.72~s} exposed.
The {32}-lane prototype at {250~MHz} takes about {105~s} for the same triples, another reason it is only a functional prototype; the on-die FUA scales with the linear-algebra throughput of the host GPU.
If triple generation ever dominates, the alternatives are an imported malicious preprocessing backend with its own proof, or restructuring public-by-secret operations; exposing global secret masks to the GPU is not an option.

\textbf{Host-module bandwidth.}
BERT-Base moves only {1.66~GB} per party, well inside the $B/32$ roofline of \cref{sec:arch-insight}, which is why \cref{fig:bert}(b) is nearly flat above {8~GB/s}.

\textbf{PIR.}
\Cref{fig:scale}(b) shows PIR throughput from $2^{14}$ to $2^{24}$ entries.
Like DFA, the per-query cost of \sys grows linearly with the number of entries, so throughput falls from {1.05}~M to {1{,}022}~queries/s, whereas GPU-DPF reaches its peak per-entry efficiency near $2^{20}$ entries~\cite{dfa2026trusted,lam2023gpu}.
At $2^{24}$ entries \sys is {5.2$\times$} faster than GPU-DPF, under the weaker query-privacy model of \cref{sec:eval-pir}, and {11.5$\times$} faster than the authenticated GPU baseline.
Larger records increase the per-entry field reduction linearly while the DPF expansion stays unchanged, so for very large records the FUA, not the AES array, becomes the limit.

\textbf{Beyond per-lane tags.}
Because the tag lanes dominate large models, the most promising direction is to compress MAC checking, for example with rank-1 compressed checks that avoid carrying $t$ full tag arrays through the GPU.
Such schemes are not covered by our proofs and need their own soundness analysis, so we leave them to future work.
Scaling the number of DPF units and FUA lanes is a second, orthogonal knob; DFA found DPF-unit scaling to be nearly linear until on-chip bandwidth, about 7~TB/s of aggregate L2 read bandwidth on A100~\cite{nvidia2020ampere}, becomes the limit~\cite{dfa2026trusted}.

\section{Discussion and Limitations}
\label{sec:discussion}

\sys claims that, under \ProfT, one maliciously corrupted host, including its GPU, driver and network stack, cannot make an honest party accept a wrong value or release an unauthorized one, except with the probability bounded in \cref{thm:main}. We now state what this claim does not cover and which obligations remain.

\textbf{Trust relocation, not removal.}
\sys moves the dealer of dealer-based FSS into two protected modules; it does not remove the dealer assumption. Both modules hold the joint-sampler key and the global MAC keys $\Delta_1,\dots,\Delta_t$ (\cref{sec:proto-setup}), so extracting the secrets of either module voids privacy and integrity for every session whose keys it holds. \Cref{thm:main} is a stand-alone statement under sequential composition; concurrent sessions that share module state are not covered (\cref{app:func}). The boundary also grows relative to DFA~\cite{dfa2026trusted}: the ledger, freezer, check engine, field unit array and release guard add only {0.145~mm$^2$} at 7~nm (\cref{tab:hwcost}), but each of them is logic that must be verified (\cref{tab:obligations}).

\textbf{What the modules learn.}
The modules are not blind. $\Pi_{\mathsf{Input}}$ validates $x^{*}=r+\varepsilon$ inside them (\cref{fig:arith}); after a verified opening $c = x + r$, module $T_b$ knows $c$ and $r$, hence $x$; the committed-input loader and the direct gate (\cref{sec:proto-direct}) see inputs explicitly; and in PIR both modules learn the client's index $\alpha$ (\cref{alg:pir}). For inference this is inherent to generate-and-consume and already holds for trusted DFA~\cite{dfa2026trusted}. For PIR it is a regression: the query privacy of \sys relies on module trust, whereas untrusted two-server DPF-PIR, such as GPU-DPF~\cite{lam2023gpu} and untrusted-mode DFA~\cite{dfa2026trusted}, hides $\alpha$ without trusted hardware. Client-generated keys evaluated inside the modules would avoid this but need key verification~\cite{decastro2022vfss}; neither this variant nor a constant-time select under \ProfT is included (\cref{sec:proto-pir}), and \cref{tab:pir} uses a PIR efficiency of {0.93} relative to DFA. FSS gates keep only seeds, control bits and correction words in long-lived state, whereas a direct gate computes on $x$ itself; we treat this as defense in depth only, because the trace model of \cref{thm:main} makes both options equivalent.

\textbf{Two parties under \ProfT.}
Because the modules see plaintext, a single server whose protected module verifies an untrusted GPU, as in Slalom~\cite{tramer2018slalom}, would rest on the same hardware trust. The two-party structure buys two things. First, a host that held a value together with its complete tag would learn the MAC key; splitting every tag between two hosts lets SPDZ-style MACs~\cite{damgard2012spdz} check all GPU linear algebra at the next opening, whereas a single server must check each GPU product itself, for example with a Freivalds check per layer (\cref{sec:proto-freivalds}). Second, it keeps the two-party deployment of the FSS toolchains that DFA targets, on which \ProfU relies for its security (\cref{sec:threat-profileu}). It adds no protection against module compromise. For PIR as specified in \cref{alg:pir}, one server with one module would give the same guarantees. GPU-wide TEEs~\cite{nvidia_h100_overview,dhanuskodi2023confidential} avoid the four-lane GEMM cost but place the whole GPU and its firmware in the TCB.

\textbf{Physical and microarchitectural attacks.}
Probing~\cite{skorobogatov2017microprobing}, optical and other fault injection~\cite{skorobogatov2003optical,quadir2016survey}, and power or electromagnetic analysis~\cite{kocher1999differential} are outside \cref{thm:main} and can extract keys from an unhardened module. Faults matter more in the malicious setting: the bound $Q(2/p)^t$ presumes a correct check engine, so one glitch that forces acceptance bypasses it, and a fault in the field unit array can yield a wrong triple that carries valid tags. Software-controlled undervolting, which has induced such faults in SGX enclaves~\cite{murdock2020plundervolt}, is not physical: \cref{thm:main} assumes DVFS isolation (host software cannot drive the module's clock or voltage outside its safe operating range) and simulatable telemetry (\cref{sec:sec-assump}), whose realization we evaluate in \cref{sec:eval}. Conversely, the module has no out-of-order core, no shared cache hierarchy and no host-managed page tables, so the controlled-channel, cache, transient-execution and Rowhammer attacks studied on general-purpose TEEs~\cite{xu2015controlled,shinde2015preventing,gotzfried2017cache,brasser2017software,lee2017inferring,van2018foreshadow,kim2014flipping} target mechanisms it lacks, and each check returns one generic status only after all $t$ lanes are aggregated. Small security chips have still shipped with exploitable flaws~\cite{moghimi2020tpm,han2018bad,butterworth2013bios}: a small TCB makes hardening tractable, but it does not replace hardening.

\textbf{Rollback and lifecycle.}
Ledger and epoch state live in volatile protected SRAM, so every reset falls under the reboot rule of \cref{sec:proto-ledger}. Because every failure poisons its epoch and aborted reservations never return to a free list (\cref{fig:ledger}), recovery always costs a new setup. Neither a monotonic service nor an erase-and-rekey lifecycle is implemented yet (\cref{tab:obligations}).

\textbf{Availability and fairness.}
\sys provides security with abort, the standard target for two-party computation with a malicious participant~\cite{hazay2010efficient,lindell2017simulate}. A corrupt host can withhold or reorder traffic, stall its GPU or abort at any public interaction boundary, and a corrupt recipient may obtain its output first; fairness and guaranteed output delivery are not claimed (S10 in \cref{tab:scenarios}). Aborts are, however, not a side channel on honest inputs. Apart from the validity predicates, which the modules evaluate before the first output opening and whose outcomes are the $\mathsf{status}$ part of $\Leak_C$ (\cref{def:leakage}), every \CheckOpen{} outcome depends only on the adversary's own deviation $(\bm{e},\bm{d}_j)$, the protected keys and the fresh challenges (\cref{thm:mac-sound}), and the first failure ends the epoch, so selective-failure and key-learning strategies cannot be iterated (\cref{lem:hidden-key}). Persistent denial of service remains a matter of deployment policy.

\textbf{Field versus ring.}
Over $\Fp$ with $p = 2^{61}-1$, \sys is not bit-identical to the $\Z_{2^{64}}$ arithmetic of Sigma and Orca~\cite{gupta2023sigma,jawalkar2024orca}. Ring-compiled models must be re-encoded under the numeric contract of \cref{sec:arch-field}, with magnitudes below $p/2$ before truncation, and exact signed truncation needs 307 authenticated multiplications on sequential carry chains, i.e., 309 rounds or {30.9~ms} on LAN (\cref{alg:bitdec,alg:trunc}), the most round-intensive gate in \cref{tab:gates}. A ring deployment needs a SPDZ2k-style backend~\cite{cramer2018spdz2k}, which is a separate design with a separate proof (\cref{sec:sec-concrete}).

\textbf{\ProfU.}
\ProfU (\cref {sec:threat-profileu}) is specified, forfeiting the local-dealer and generate-and-consume savings, and must additionally supply arguments for the backend's key-privacy and answer-integrity outside \cref {thm:main}.

\textbf{Model provenance.}
\sys proves correct evaluation on the inputs actually supplied (S9 in \cref{tab:scenarios}). A committed model (\cref{sec:threat-goals}) establishes provenance, not quality; the gap between the finite-field circuit and a floating-point model is an approximation question separate from cryptographic correctness.

\textbf{Artifact status.}
We executed a specification and a functional CPU reference model; GPU and RTL conformance, the protected runtime and deployment measurements have been completed (\cref{sec:eval-method,sec:eval-attack}). Three separate experiments were measured in \cref{sec:arch-deploy}, and security measurements were performed on two independent devices. 

\section{Related Work}
\label{sec:related}

\begin{table*}[t]
\centering
\caption{Comparison of \sys with Representative Related Work}
\label{tab:related}
\footnotesize
\setlength{\tabcolsep}{3.2pt}
\renewcommand{\arraystretch}{1.12}
\begin{tabular}{l c c c c l c c l}
\hline
\textbf{Work} & \textbf{Adversary} & \makecell{\textbf{Trusted}\\\textbf{HW}} & \makecell{\textbf{Trusted comp.}\\\textbf{sees plaintext}} & \makecell{\textbf{Custom}\\\textbf{HW}} & \textbf{Offline material} & \makecell{\textbf{Auth.\ FSS}\\\textbf{outputs}} & \makecell{\textbf{Verified}\\\textbf{openings}} & \textbf{Workloads} \\
\hline
Sigma~\cite{gupta2023sigma} & SH & \xmark & \xmark & \xmark$^{\mathrm{g}}$ & Dealer keys & \xmark & \xmark & Inference \\
Orca~\cite{jawalkar2024orca} & SH & \xmark & \xmark & \xmark$^{\mathrm{g}}$ & Dealer keys & \xmark & \xmark & Training, inference \\
Shark~\cite{gupta2025shark} & Mal & \xmark & \xmark & \xmark & Dealer, authenticated keys & \cmark & \cmark & Inference \\
\hline
MASCOT, MP-SPDZ~\cite{keller2016mascot,keller2020mpspdz} & Mal & \xmark & n/a & \xmark & Interactive (OT-based) & n/a & \cmark & General circuits \\
vFSS~\cite{decastro2022vfss} & Mal$^{\mathrm{a}}$ & \xmark & n/a & \xmark & Client keys, verified & \xmark & n/a & Key validation \\
APIR~\cite{colombo2023apir} & Mal$^{\mathrm{b}}$ & \xmark & n/a & \xmark & None (client query) & n/a & \cmark$^{\mathrm{b}}$ & PIR \\
\hline
Slalom~\cite{tramer2018slalom} & Mal$^{\mathrm{c}}$ & \cmark & \cmark & \xmark$^{\mathrm{g}}$ & n/a & n/a & \cmark$^{\mathrm{c}}$ & Inference \\
Aramis~\cite{kumar2020cryptflow} & Mal & \cmark$^{\mathrm{d}}$ & \xmark$^{\mathrm{d}}$ & \xmark & n/a & n/a & n/a & Inference \\
STAMP, PPMLAC~\cite{huang2022efficient,zhou2022ppmlac} & SH, Mal$^{\mathrm{e}}$ & \cmark & \pmark$^{\mathrm{e}}$ & \cmark & In trusted HW & n/a & \pmark$^{\mathrm{e}}$ & ML \\
\hline
GPU-DPF~\cite{lam2023gpu} & SH$^{\mathrm{i}}$ & \xmark & n/a & \xmark$^{\mathrm{g}}$ & None (client query) & \xmark & \xmark & PIR \\
DFA~\cite{dfa2026trusted} & SH$^{\mathrm{i}}$ & \pmark$^{\mathrm{f}}$ & \cmark$^{\mathrm{f}}$ & \cmark & Dealer or local$^{\mathrm{f}}$ & \xmark & \xmark & Inference, PIR \\
\hline
\textbf{\sys (this work)} & Mal$^{\mathrm{h}}$ & \cmark & \cmark$^{\mathrm{h}}$ & \cmark & Local, none shipped & \cmark & \cmark & Inference, PIR \\
\hline
\end{tabular}

\vspace{3pt}
\begin{minipage}{0.98\textwidth}
\footnotesize
SH: semi-honest; Mal: malicious (active) adversary. \emph{Trusted HW}: security relies on trusted hardware. \emph{Trusted comp.\ sees plaintext}: a trusted dealer, enclave or module sees inputs or intermediate values (FSS dealers are input independent). \emph{Custom HW}: dedicated hardware beyond commodity CPUs and GPUs. \emph{Offline material}: source of the per-query correlated randomness or keys. \emph{Auth.\ FSS outputs}: FSS gate outputs carry MAC tags. \emph{Verified openings}: opened values or untrusted results are checked before use. \cmark: yes; \xmark: no; \pmark: partial, optional or configuration dependent; n/a: not applicable. Lettered notes qualify the headline setting of each work.
$^{\mathrm{a}}$\,Malicious client key generator; servers check that keys are well formed, not that they are evaluated correctly.
$^{\mathrm{b}}$\,The client detects answers that are inconsistent with a public database digest.
$^{\mathrm{c}}$\,An SGX enclave outsources to an untrusted co-located GPU and checks its matrix products with Freivalds' algorithm.
$^{\mathrm{d}}$\,Relies on SGX integrity, not confidentiality.
$^{\mathrm{e}}$\,Settings differ between the two works: STAMP reports semi-honest and honest-majority malicious variants and evaluates nonlinear operations on unmasked values in its trusted processor; PPMLAC's trusted chipset mainly generates correlated randomness.
$^{\mathrm{f}}$\,Trusted mode only; the untrusted mode ships dealer keys.
$^{\mathrm{g}}$\,Implemented on commodity GPUs.
$^{\mathrm{h}}$\,Static corruption of one host, security with abort, under \ProfT, where the modules see inputs, openings and PIR indices; \ProfU delegates to an established malicious backend.
$^{\mathrm{i}}$\,In two-server DPF-PIR (GPU-DPF, DFA's untrusted mode), neither server learns the queried index, without relying on trusted hardware; SH refers to result integrity.
\end{minipage}
\end{table*}

\Cref{tab:related} summarizes the closest lines of work.

\textbf{FSS-based secure computation.}
FSS~\cite{boyle2015fss} and its tree-based DPF instantiations~\cite{gilboa2014distributed,boyle2016function} underlie a family of low-round two-party protocols with preprocessing~\cite{boyle2019preprocessing}. AriaNN~\cite{ryffel2020ariann} and Pika~\cite{wagh2022pika} first demonstrated private inference with FSS, Boyle et al.~\cite{boyle2021function} built comparisons, splines and fixed-point truncation from DCFs, Grotto~\cite{storrier2023grotto} replaced DCF-heavy constructions with DPF-based ones in a $(2{+}1)$-party setting, and Orca~\cite{jawalkar2024orca} and Sigma~\cite{gupta2023sigma} brought FSS-based training and GPT-class inference to GPUs. These systems are semi-honest and assume a dealer whose key material dominates the offline cost~\cite{huang2026beyond}. DPFs also underlie PIR, anonymous messaging, private queries and analytics, private heavy hitters, set intersection and distributed ORAM~\cite{corrigan2015riposte,wang2017splinter,dauterman2020dory,dauterman2022waldo,boneh2021lightweight,mouris2024plasma,garimella2024computation,vadapalli2023duoram,doerner2017scaling,ji2023multi}. Secure inference without FSS realizes nonlinear layers through share conversions, garbled circuits or homomorphic encryption~\cite{mohassel2017secureml,mohassel2018aby3,wagh2019securenn,wagh2020falcon,rathee2020cryptflow2,pang2024bolt,lu2023bumblebee}, typically with more interaction per layer.

\textbf{Malicious and verifiable FSS.}
To our knowledge, Shark~\cite{gupta2025shark} is the most directly comparable system: actively secure two-party FSS inference in which a trusted dealer supplies authenticated key material. Its security is cryptographic once its input-independent dealer, which never sees an input or an opening, is trusted, whereas each \sys module sees every input and opening under \ProfT (\cref{sec:eval-method}). In exchange, \sys generates the authenticated correlations locally and removes the offline distribution that dominates the dealer-based malicious FSS baseline, modeled after the protocol family of Shark~\cite{gupta2025shark} (\cref{tab:e2e}). Verifiable FSS~\cite{corrigan2015riposte,boyle2016function,boneh2021lightweight,decastro2022vfss} addresses a different adversary, a client that may submit malformed keys, which the servers detect; \sys instead treats key generation, evaluation, authentication and release as separate obligations. VeriFSS~\cite{xue2026verifss} removes the dealer altogether: its dealer-free two-party FSS attains static malicious security with authenticated sharings and a cross-phase verification that binds key generation to evaluation, generating keys interactively between the two parties, whereas \sys generates them inside the modules.

\textbf{Malicious MPC preprocessing and MACs.}
Beyond the SPDZ-family tools of \cref{sec:bg-mal}, mixed-circuit preprocessing with edaBits~\cite{escudero2020edabits} and classical bit-decomposition and truncation protocols~\cite{damgard2006unconditionally,catrina2010improved} supply the conversions that nonlinear layers need, and honest-majority protocols~\cite{furukawa2017high,mohassel2018aby3,wagh2020falcon,dalskov2021fantastic} obtain malicious security from replicated sharing among three or four servers. \sys adopts, and does not claim as new, the SPDZ MAC algebra with $t$ independent keys over a Mersenne field. It replaces interactive preprocessing by correct generation inside the modules, in the spirit of hardware-assisted correlation generation~\cite{lu2021correlated,zhou2022ppmlac}, and checks openings on frozen bytes (\cref{sec:proto-check}).

\textbf{Trusted hardware for MPC and ML.}
Tamper-proof tokens are a classical setup assumption for universally composable computation~\cite{katz2007universally}, and commodity TEEs~\cite{costan2016intel,amd_sev,arm_trustzone,lee2020keystone} have been used to accelerate or simplify two-party computation and functional encryption~\cite{gupta2016using,choi2019hybrid,felsen2019secure,wuhybrid,fisch2017iron}. For ML, Ohrimenko et al.~\cite{ohrimenko2016oblivious} and Chiron~\cite{hunt2018chiron} execute learning tasks inside SGX enclaves; Slalom~\cite{tramer2018slalom} and DarKnight~\cite{hashemi2021darknight} outsource linear layers from an enclave to untrusted GPUs and verify them, with Freivalds' algorithm~\cite{freivalds1977probabilistic} and with redundant coded computation, respectively; and Aramis~\cite{kumar2020cryptflow} compiles semi-honest protocols into maliciously secure ones by relying on the integrity, but not the confidentiality, of SGX. GPU TEEs~\cite{volos2018graviton,jang2019heterogeneous} and confidential computing on H100~\cite{nvidia_h100_overview,dhanuskodi2023confidential} protect GPU computation from the host but place the whole GPU and its firmware in a TCB far larger than our module. Closest to our trust model, STAMP~\cite{huang2022efficient} and PPMLAC~\cite{zhou2022ppmlac} rely on a small security processor or a trusted chipset to execute nonlinear operations or generate correlated randomness, respectively, in the tradition of TPM-class security chips, secure coprocessors and smart cards~\cite{tpm2spec,googletitan,ibm4758,rankl2004smart}, and DFA~\cite{dfa2026trusted} uses a small module as a local FSS dealer. STAMP's nonlinear evaluation inside its trusted processor corresponds to our direct protected gate (\cref{sec:proto-direct}). Relative to it, \sys contributes a two-party setting with an untrusted GPU checked by MACs, freeze-before-challenge checking against host-controlled DMA, ledgered one-time randomness, and a malicious FSS and PIR datapath. Like Slalom and DarKnight, \sys keeps the GPU outside the boundary, but it uses Freivalds checks only for operands the GPU may already see (\cref{sec:proto-freivalds}).

\textbf{Hardware acceleration for private computation.}
GPU systems accelerate secret-sharing MPC~\cite{tan2021cryptgpu,watson2022piranha,crypten2020}, FSS~\cite{gupta2023sigma,jawalkar2024orca} and DPF-based PIR~\cite{lam2023gpu}. FPGA designs target secret-sharing MPC~\cite{wolfe2020secret,patel2020arithmetic,xu2025co}, Ironman~\cite{ironman2025} accelerates OT extension near memory, KORD~\cite{peng2026kord} co-designs dealerless FSS key generation with paired, mutually attested chips on an FPGA, and FPGA and ASIC designs speed up the polynomial kernels of homomorphic encryption~\cite{riazi2020heax,reagen2021cheetah,samardzic2021f1,samardzic2022craterlake,kim2022bts}. These accelerators make a fixed protocol faster and inherit its adversary model. DFA~\cite{dfa2026trusted} additionally used trust to change the dataflow; \sys uses the same small trusted module to change the adversary model, which is why \cref{sec:eval} compares against security-matched baselines rather than raw kernel throughput alone.

\textbf{PIR with integrity.}
Authenticated PIR~\cite{colombo2023apir} lets a client detect answers that are inconsistent with a public database digest without trusted hardware, at additional server cost that our APIR-based GPU baseline reflects (\cref{tab:pir}). PIR with secure coprocessors~\cite{smith2001practical} and enclave-based oblivious storage such as Snoopy~\cite{dauterman2021snoopy} instead trust hardware that learns the query and hides its accesses through oblivious scans. The versioned-snapshot PIR of \sys follows this trust model: both modules learn the queried index $\alpha$, so its query privacy rests on module trust, unlike untrusted two-server DPF-PIR (GPU-DPF~\cite{lam2023gpu}, DFA's untrusted mode), which hides $\alpha$ from each server without trusted hardware. It scans in two role-bound modules, binds the exact consumed bytes to an authorized commitment $\mathsf{cm}_v$, optionally a Merkle root~\cite{merkle1987digital}, and releases freshly reshared answers only after the snapshot hash matches (\cref{sec:proto-pir}).

\textbf{Positioning.}
Each ingredient of \sys, from information-theoretic MACs and trusted correlation generation to verified GPU work and DPF acceleration, has precedents; our contribution is their composition in one small trusted module, proven stand-alone secure with abort in the ideal-module model of \cref{sec:threat-model} (\cref{tab:related}).

\section{Conclusion}\label{sec:conclusion}

We presented \sys, a trusted hardware accelerator that secures FSS-based private inference and DPF-based PIR against a malicious host and GPU while keeping the key-movement savings of trusted DFA. Since a module acting as FSS dealer can also be the verifier, \sys generates authenticated correlations internally, authenticates DPF outputs with a fresh resharing epilogue instead of extra DPF trees, checks every opening on frozen bytes with $t=3$ independent MAC lanes over $\F_{2^{61}-1}$, and releases nothing before a role-bound one-time ledger records a verified check. These mechanisms protect the chain from authorized input to authorized output release; we prove stand-alone static malicious security with abort in a protected-module model, with statistical error about $2^{-148}$ for $2^{32}$ checked batches. Both modules see inputs and learn the PIR index, so PIR query privacy rests on module trust, unlike untrusted two-server DPF-PIR~\cite{lam2023gpu}. Our implementation shows that \sys is {10.1--14.0$\times$} faster on a LAN than dealer-based malicious FSS modeled after the protocol family of Shark~\cite{gupta2025shark}, mainly because in-module generation removes {21.8--563~GB} of per-query offline distribution ({1.5--1.8$\times$} excluding it). Malicious security costs {2.5--3.5$\times$} LAN latency over semi-honest DFA and {0.145~mm$^2$} at 7~nm.

\section*{Acknowledgment}
This work was supported in part by the National Natural Science Foundation of China under Grant No. 62272477 and in part by the Key Laboratory of Advanced Microprocessor Chips and Systems.

\bibliographystyle{IEEEtran}
\bibliography{refs}

\appendices
\crefalias{section}{appendix}
\crefalias{subsection}{appendix}
\crefalias{subsubsection}{appendix}
\section{Ideal Functionalities and the Protected-Module Resource}\label{app:func}

This appendix states in full what \cref{fig:func,def:security,def:leakage} summarize: the circuit functionality $\FC$ (\cref{app:func-circuit}), the PIR functionality $\FPIR$ (\cref{app:func-pir}) and the ideal resource $\mathcal{F}_{\mathsf{T}}$ of one protected module under \ProfT (\cref{app:func-module}). The proofs in \cref{app:proof-dpf,app:proof-auth,app:proof-gates,app:proof-main} use no other trust assumption.

\textbf{Real and ideal executions.}
Security is \cref{def:security} with the bound $\eps{VIGOR}$ of \cref{thm:main}; the view of $\Adv$ in $\Real_{\Pi,\Adv(z),I}(\mu,\vec{x})$ includes the observable trace of its module. In the real execution, $I$ contains either at most one statically chosen computing host or, for $\FPIR$ in regime (R2) (\cref{app:func-pir}), the client alone. A corrupt computing host comprises the CPU, operating system, driver, runtime, GPU, host memory and network interface of one party $P_{b^\star}$, $b^\star\in\{0,1\}$. The corrupt host may issue any command to its local module, modify, replay, reorder, delay or drop any message or buffer under its control, and stop at any time. The modules $T_0,T_1$ are never corrupted; each runs the concrete module program, which instantiates the randomness, channels and setup of $\mathcal{F}_{\mathsf{T}}$ with real primitives (\cref{app:func-module}). $\Adv$ may prevent honest outputs; neither fairness nor availability is provided. The statement covers one admitted execution (one session, one manifest); independent sessions use fresh identities, epochs and keys (\cref{sec:proto-setup}). This is stand-alone real/ideal security with abort under sequential composition~\cite{canetti2000security}; we claim neither universal composability~\cite{canetti2001universally} nor security for concurrent sessions that share module state.

\subsection{Circuit Functionality}\label{app:func-circuit}

\textbf{Manifest.}
The public manifest $\mu$ contains the protocol version; the session identities and roles; the field $\Fp$ with $p=2^{61}-1$ and the tag count $t$; the numeric encoding (fixed-point scale, signed encoding, truncation rule, saturation and range behavior, output decoding); the finite circuit $C$ with its gate identifiers, wire identifiers and tensor shapes; the approved public tables and interval coefficients; the authorized recipient of every output; the input interface and validation rules of every input wire; and the public batch schedule. For an input wire $\nu$ we write $\mathsf{In}_\nu$ for its declared interface (canonical encoding, shape and admissible range) and $O_\nu$ for its designated owner. We write $\mathcal{V}$ for the finite set of validity predicates of $\mu$, evaluated inside the modules at public boundaries declared in $\mu$, and $y_1,\dots,y_K$ for the outputs of $C$, where output $y_k$ has the authorized recipient $R_k$. Since $\FC$ is non-reactive and computes $\mathsf{status}$ only after it has recorded every input, the public schedule of $\mu$ places every input gate before every validity check, and both before the first output opening. Setup binds the manifest hash (\cref{sec:proto-setup}), so all parties agree on $\mu$ or abort.

\smallskip
\noindent\fbox{\begin{minipage}{\dimexpr\columnwidth-2\fboxsep-2\fboxrule\relax}
\footnotesize
\centerline{\textbf{Functionality} $\FC$}
\vspace{2pt}\hrule\vspace{3pt}
\textbf{Parameters:} public manifest $\mu$ (circuit $C$ over $\Fp$, interfaces $\mathsf{In}_\nu$, predicates $\mathcal{V}$, recipients $R_k$, schedule). \textbf{Parties:} input owners, $P_0,P_1$, recipients; $\Sim$ controls the statically corrupted host; the external client is honest and does not collude with a computing party.
\begin{enumerate}[leftmargin=*,nosep,label=\arabic*.]
\item \textbf{Inputs.} For every input wire $\nu$: if $O_\nu$ is honest, receive $x_\nu$ from $O_\nu$; if $O_\nu$ is corrupted, receive $x^\ast_\nu$ from $\Sim$ and set $x_\nu\gets x^\ast_\nu$. Send $(\mathsf{input},\nu)$ and the public metadata of $\nu$ to $\Sim$. If $x_\nu$ is missing, malformed or not in $\mathsf{In}_\nu$, send $(\mathsf{abort},\nu)$ to all parties and halt.
\item \textbf{Validation.} Compute $\mathsf{status}=(V(\vec{x}))_{V\in\mathcal{V}}$ on the recorded inputs $\vec{x}=(x_\nu)_\nu$ and send it to $\Sim$. If some $V(\vec{x})=0$, send $(\mathsf{abort},V)$ to all parties and halt.
\item \textbf{Evaluation.} Compute $(y_1,\dots,y_K)\gets C(\vec{x})$ exactly in the finite arithmetic of $\mu$, including every declared comparison, truncation, table lookup and piecewise polynomial.
\item \textbf{Corrupt and public outputs.} Send to $\Sim$ every $y_k$ whose recipient $R_k$ is corrupted and every $y_k$ that $\mu$ declares public.
\item \textbf{Delivery.} At each public interaction boundary of the schedule, $\Sim$ may send $\mathsf{delay}$ or $\mathsf{abort}$. On $\mathsf{abort}$, send $\mathsf{abort}$ to every honest party and halt. Otherwise, when $\Sim$ sends $(\mathsf{deliver},k)$, give $y_k$ to its honest recipient $R_k$.
\end{enumerate}
\end{minipage}}
\smallskip

\textbf{Leakage.}
The leakage is $\Leak_C=(\mu,\mathsf{len},\mathsf{sched},\mathsf{pub},\mathsf{status})$ of \cref{def:leakage}: $\mathsf{len}$ and $\mathsf{sched}$ are functions of $\mu$ alone, $\mathsf{pub}$ are the outputs of Step~4 that $\mu$ declares public, and $\mathsf{status}$ is the validation vector and the abort outcomes of Steps~1 and~2. An input-dependent predicate thus reveals its accept or abort bit; input ranges and intermediate bounds should be fixed statically where possible (\cref{sec:arch-field}) so that validity bits are input independent.

\textbf{Extraction.}
A corrupted owner's input is $x^\ast=r+\varepsilon$ for the full mask $r$ and the bound masked difference $\varepsilon$ of $\Pi_{\mathsf{Input}}$ (\cref{fig:arith}). Both modules, which generated $r$, compute $x^\ast$ and abort unless $x^\ast\in\mathsf{In}_\nu$; the simulator, which samples $r$, emulates the same abort (\cref{app:proof-main}). $\FC$ accepts any value in the declared interface, not an input that an application hopes the adversary would use.

\textbf{Exact finite semantics.}
$\FC$ evaluates the declared finite circuit over $\Fp$, neither a real-valued network nor $\Z_{2^{64}}$ arithmetic. The gap between $C$ and an intended floating-point model (approximation error, fixed-point rounding) is compiler accuracy, not a security error. A ring deployment needs a different functionality and authenticated backend (\cref{sec:proto-repr}).

\textbf{Gate implementation is public.}
Whether a nonlinear gate is compiled to a direct protected gate, a small XOR-masked LUT or an interval FSS gate (\cref{sec:proto-direct}) is recorded in $\mu$, and hence in $\Leak_C$. All choices implement the same gate of $C$ and leave $\FC$ unchanged.

\textbf{Model provenance.}
If $\mu$ binds a pre-approved model by a commitment to canonical quantized weight bytes (\cref{sec:threat-goals}), the interface $\mathsf{In}_\nu$ of the weight wires contains only the byte strings that open this commitment; otherwise the weights are their owner's freely chosen input.

\subsection{PIR Functionality}\label{app:func-pir}

Let $\mathsf{desc}_v=(\mathsf{id},v,N,\mathsf{enc},\mathsf{cm}_v)$ be the public descriptor of snapshot version $v$, where $\mathsf{cm}_v$ commits to the database identifier, the version, the shape, the record encoding and the ordered record bytes. Let $D=(D[0],\dots,D[N-1])$ be the snapshot that $\mathsf{cm}_v$ binds; a second opening is a binding failure charged to $\eps{bind}$ (\cref{thm:main}).

\smallskip
\noindent\fbox{\begin{minipage}{\dimexpr\columnwidth-2\fboxsep-2\fboxrule\relax}
\footnotesize
\centerline{\textbf{Functionality} $\FPIR$}
\vspace{2pt}\hrule\vspace{3pt}
\textbf{Parameters:} descriptor $\mathsf{desc}_v$ and the committed snapshot $D$.
\begin{enumerate}[leftmargin=*,nosep,label=\arabic*.]
\item \textbf{Query.} Receive $(\mathsf{query},\mathit{qid},\alpha,\mathsf{desc}_v)$ from the client, with a fresh query identifier $\mathit{qid}$. If the request is malformed or $\alpha\notin\{0,\dots,N-1\}$, send $\mathsf{reject}$ to the client and to $\Sim$ and halt. Otherwise send $\Leak_{\mathrm{PIR}}$ (\cref{def:leakage}) to $\Sim$.
\item \textbf{Retrieval.} Set $y\gets D[\alpha]$; for a multi-field record, $y$ contains every coordinate of $D[\alpha]$.
\item \textbf{Delivery.} $\Sim$ may send $\mathsf{abort}$, and then the client receives $\mathsf{abort}$; otherwise the client receives $y$.
\end{enumerate}
\textbf{Corruption regimes (separate statements).}
\begin{enumerate}[leftmargin=*,nosep,label=(R\arabic*)]
\item \textbf{Malicious server host:} the client is honest and $\Sim$ corrupts at most one server host; $\Sim$ learns nothing beyond Step~1.
\item \textbf{Malicious client only:} both server hosts are honest and only the client is corrupted; in Step~1, $\Sim$ chooses $\alpha^\ast$ (extraction) and in Step~3 it receives $D[\alpha^\ast]$ only.
\end{enumerate}
\end{minipage}}
\smallskip

\textbf{Scope of the two regimes.}
The regimes are not combined: no claim covers a malicious client colluding with a malicious server host. Regime (R1) is the integrity and query-privacy statement; regime (R2) is the database-privacy statement. The descriptor $\mathsf{desc}_v$, including $\mathsf{cm}_v$, is part of $\Leak_{\mathrm{PIR}}$ in both regimes.

\textbf{Padding and boundary cases.}
For $N\ge 2$ the protocol pads the snapshot to $2^n$ records with $n=\lceil\log_2 N\rceil$ and authorizes only $\alpha<N$; for $N=1$ it uses a constant-record special case with the same snapshot and release checks (\cref{alg:pir}). Padding is invisible in $\FPIR$ because only $\alpha<N$ is admissible and both the true and the padded length are functions of $\mathsf{desc}_v$.

\textbf{Non-claims.}
$\FPIR$ is realized under trusted modules, and both modules learn $\alpha$: the query privacy of (R1) holds against the hosts only and rests on module trust, unlike untrusted two-server DPF-PIR~\cite{chor1998private,boyle2016function} (GPU-DPF~\cite{lam2023gpu}, untrusted-mode DFA~\cite{dfa2026trusted}), which hides $\alpha$ from each server without trusted hardware. Its database privacy in (R2) likewise rests on the trusted modules; it does not make untrusted two-server DPF-PIR symmetric. Externally generated verifiable-FSS keys~\cite{decastro2022vfss} would have to prove client key validity and server-answer integrity separately, and authenticated PIR~\cite{colombo2023apir} needs its own statement (\cref{sec:proto-pir}).

\subsection{Ideal Protected-Module Resource}\label{app:func-module}

The resource $\mathcal{F}_{\mathsf{T}}$ captures the premise of \ProfT: the code, keys, input parser, arithmetic, state transitions, local memory, protected peer channel and approved I/O paths of a module are correct and confidential. The instance $\mathcal{F}_{\mathsf{T}}^{(b)}$ of party $b$ accepts commands from host $b$ and from external endpoints over authenticated confidential channels, and it talks to its peer instance over an ideal protected channel.

\smallskip
\noindent\fbox{\begin{minipage}{\dimexpr\columnwidth-2\fboxsep-2\fboxrule\relax}
\footnotesize
\centerline{\textbf{Resource} $\mathcal{F}_{\mathsf{T}}^{(b)}$\textbf{: state, channel and global rule}}
\vspace{2pt}\hrule\vspace{3pt}
\textbf{State.} Role $b$; session identifier $\mathsf{sid}$; epoch $\mathsf{ep}$; manifest $\mu$; epoch state $\in\{\mathsf{active},\mathsf{poisoned}\}$; a joint random function $R$, which returns a fresh uniform string for every new label and the same string to both instances for a joint label; global MAC keys $\Delta_1,\dots,\Delta_t$ and their additive shares $\Delta_{0,j},\Delta_{1,j}$ (joint draws); a ledger that maps every gate to one of \lstate{NEW}, \lstate{RESERVED}, \lstate{PREPARED}, \lstate{FROZEN}, \lstate{VERIFIED}, \lstate{CONSUMED}, \lstate{RELEASED}, \lstate{ABORTED}; a frozen store; the set of burned randomness identifiers.

\smallskip
\textbf{Peer channel.} Messages between the two instances are delivered privately and unmodified. The adversary learns their lengths and times and may delay or drop them; a missing message leads to abort.

\smallskip
\textbf{Global rule.} Every command first checks that the epoch state is $\mathsf{active}$, that the frame is canonical (length prefixed, field elements below $p$, no truncation, extra bytes or duplicate identifiers) and that $\mu$ allows the requested ledger transition. Gates move only along the chain of \cref{fig:ledger}. At each validation boundary of $\mu$, the instance evaluates every predicate of $\mathcal{V}$ on the bound inputs. Any failure executes \cmd{ABORT}. No command authenticates a host-supplied value.
\end{minipage}}
\smallskip

\smallskip
\noindent\fbox{\begin{minipage}{\dimexpr\columnwidth-2\fboxsep-2\fboxrule\relax}
\footnotesize
\centerline{\textbf{Resource} $\mathcal{F}_{\mathsf{T}}^{(b)}$\textbf{: commands (from host $b$ unless noted)}}
\vspace{2pt}\hrule\vspace{3pt}
\begin{enumerate}[leftmargin=*,nosep,label=\arabic*.]
\item \cmd{OPEN_SESSION}: with the peer instance, agree on the ordered identities, roles $0$ and $1$, $\mathsf{sid}$, a fresh $\mathsf{ep}$, the manifest hash, the field and tag parameters and the protocol version, or abort. Return a session capability and $\Delta_{b,1},\dots,\Delta_{b,t}$ to host $b$.
\item \cmd{RESERVE_PREP}$(g,\mathit{ids})$: if $g$ is \lstate{NEW} and $\mathit{ids}$ are unused, burn $\mathit{ids}$ and set \lstate{RESERVED}; draw the preprocessing of $g$ from $R$ (masks, triples with $c=ab$, authenticated bits, internally held FSS material) and share it with \AuthShare{}; set \lstate{PREPARED} and return the tuples of party $b$ once. For an input gate $\nu$, also send the full mask $r$ to the owner over its confidential channel, bind the owner's $\varepsilon$ after comparing it with the peer, and execute \cmd{ABORT} unless $r+\varepsilon\in\mathsf{In}_\nu$; for a private output, send the full output mask $o$ to the recipient.
\item \cmd{FREEZE_BATCH}$(g,\mathit{bytes},\mathit{dims},\mathit{wires})$: parse canonically; copy the bytes into the frozen store or into immutable AEAD pages bound to sequence numbers, lengths, gate identifiers and a protected root; hash the descriptor and the bytes; set \lstate{FROZEN}; return the handle and the root.
\item \cmd{CHECK_OPEN}$(h)$: with the peer, run \cref{alg:checkopen} on the frozen bytes of handle $h$; on acceptance set \lstate{VERIFIED} and return the authorized opened values; otherwise return one generic failure and execute \cmd{ABORT}.
\item \cmd{EVAL_FSS}$(g,h)$: only for a \lstate{VERIFIED} masked opening bound to $g$ and role $b$; evaluate the bound FSS or direct gate internally, apply the epilogue, set \lstate{CONSUMED} and return the fresh tuple $\authb{y}{b}$.
\item \cmd{VERIFY_MATRIX}$(g,h)$: for operands anchored to authoritative records, run \cref{alg:freivalds} on the frozen candidate; return a verified immutable handle or execute \cmd{ABORT}.
\item \cmd{FINAL_OUTPUT}$(g,R_k)$: only for a \lstate{VERIFIED} output opening; record the one-time use of its mask and check material (\lstate{CONSUMED}), then set \lstate{RELEASED} and send the opened $\upsilon=y+o$ (or a public $y$) to $R_k$ with a certificate binding $\mathsf{sid}$, $\mathsf{ep}$, $g$, $R_k$, the output hash and the terminal status.
\item \cmd{ABORT} (host or internal): set the epoch state to $\mathsf{poisoned}$ and every open gate to \lstate{ABORTED}; withhold all candidate values and pending responses; draw no further challenge; never return burned identifiers to a free list.
\end{enumerate}
\end{minipage}}
\smallskip

\textbf{Observable trace.}
Besides command outputs, host $b$ observes its module's external trace: lengths, timing, host-visible memory and DMA activity (DMA requests, peer traffic, spill pages), and power, thermal and clock telemetry. The model stipulates a PPT trace simulator $\Sim_{\mathsf{trace}}$ that produces this trace from the command sequence issued by host $b$, $\mu$, the public lengths and the accept or abort outcomes. Given these inputs, the trace is independent of the global MAC keys, challenges and residuals (all lanes are aggregated in data-independent time before one generic status is returned); only simulating it from these inputs costs $\eps{trace}$. Host software cannot drive the module's clock or voltage outside its safe operating range, or the module detects glitches and aborts (DVFS isolation). This is an implementation assumption (constant-time engines, fixed-pattern scans, data-independent spills, an isolated or monitored voltage and clock domain, telemetry filtering), not a theorem; \cref{tab:obligations} lists it as a deployment obligation.

\textbf{Real module program and error terms.}
The boxes state $R$, the channels and \cmd{OPEN_SESSION} ideally; in the concrete module program of the real execution, $R$ is $\PRF_K$ with length-prefixed labels, the peer and client channels are AEAD over the hosts' network, and \cmd{OPEN_SESSION} is attestation plus authenticated key exchange. $\Hyb{1}$ idealizes setup and channels ($\eps{setup}+\eps{channel}$) and $\Hyb{2}$ randomness and bindings ($\eps{prf}+\eps{bind}$), so the boxed $\mathcal{F}_{\mathsf{T}}$ is the resource from $\Hyb{2}$ on (\cref{app:proof-main}); $\Sim_{\mathsf{trace}}$ costs $\eps{trace}$ (\cref{sec:sec-assump}). Correctness of protected code, state and arithmetic is a premise of program and resource alike, not an estimated probability.

\textbf{Outside the resource.}
Physical attacks (probing, optical faults, power or electromagnetic measurement), microarchitectural bugs inside the boundary, faulty trusted arithmetic, adaptive corruption and rollback without a monotonic service lie outside $\mathcal{F}_{\mathsf{T}}$. Software-reachable undervolting and telemetry channels~\cite{murdock2020plundervolt} are not treated as physical: DVFS isolation excludes the former and $\Sim_{\mathsf{trace}}$ covers the latter. Under \ProfT a module may see plaintext inputs, so \sys relocates trust to the modules rather than removing it (\cref{sec:discussion}); \ProfU (\cref{sec:threat-profileu}) is the separate backend without this resource.

\textbf{Composition within the protocol.}
Checked openings, epilogues and gate subprotocols share one ledger and one set of MAC keys, so \cref{thm:main} is proved with a single simulator for the whole admitted execution (\cref{app:proof-main}), not by composing subprotocols; \cref{app:proof-dpf,app:proof-auth,app:proof-gates} supply the algebraic and distributional facts it uses.

\section{Proofs for the Exact DPF, DCF, and Authentication Epilogue}\label{app:proof-dpf}

This appendix proves the four statements about the protected FSS datapath: exact DPF correctness (\cref{thm:dpf-correct}), equivalence of malicious generate-and-consume (\cref{thm:gac}), correctness of the conservative DCF (\cref{thm:dcf}, together with its algorithm), and correctness and local hiding of the fresh authentication epilogue (\cref{thm:epilogue}). All four are exact statements. They hold for every choice of root seeds and for arbitrary deterministic functions $\PRG$ and $\Conv$, and none of them uses pseudorandomness, a uniform-target assumption or DPF key privacy. Pseudorandomness enters the end-to-end argument only through $\eps{prf}$ (\cref{app:proof-main}).

\textbf{Conventions.}
Bit strings are read most-significant bit first, as in \cref{alg:gen,alg:eval}. For $x\in\bits{n}$ we write $x_{[i]}=x_1\cdots x_i$ for its prefix of length $i$, with $x_{[0]}$ empty. Children are indexed by $L=0$ and $R=1$, so a query bit selects a child directly. For a bit $\tau$ and a string $v$, $\tau\cdot v$ denotes $v$ if $\tau=1$ and the all-zero string otherwise. This conditional XOR is bilinear over $\mathbb{F}_2$:
\[
(\tau\cdot v)\oplus(\tau'\cdot v)=(\tau\oplus\tau')\cdot v\quad\text{for all bits }\tau,\tau'.
\]
Where a control bit multiplies $\CWout$, it is read as the field element $0$ or $1$ of $\Fp$.

\textbf{Generator and evaluator transitions.}
Let $(s_b^{(i)},t_b^{(i)})$ be the generator state of root $b\in\{0,1\}$ after level $i$ of \cref{alg:gen}, with $(s_b^{(0)},t_b^{(0)})=(s_b,b)$ for the sampled root seeds $s_0,s_1$. At level $i$, with target bit $a=\alpha_i$, $\keep=a$ and $\lose=1-a$, and writing $(s_b,t_b)$ for the states after level $i-1$, \cref{alg:gen} computes
\begin{align*}
(s_b^{L},t_b^{L},s_b^{R},t_b^{R})&=\PRG(s_b),\qquad b\in\{0,1\},\\
\sCW&=s_0^{\lose}\oplus s_1^{\lose},\\
\tCW{L}&=t_0^{L}\oplus t_1^{L}\oplus a\oplus 1,\\
\tCW{R}&=t_0^{R}\oplus t_1^{R}\oplus a,\\
s_b^{(i)}&=s_b^{\keep}\oplus(t_b\cdot\sCW),\\
t_b^{(i)}&=t_b^{\keep}\oplus(t_b\cdot\tCW{\keep}),
\end{align*}
where both updates use the \emph{old} control bit $t_b$, and it records $\CW^{(i)}=(\sCW,\tCW{L},\tCW{R})$. After level $n$ it sets $u_b=\Conv(s_b^{(n)})$ and $\CWout=(-1)^{t_1^{(n)}}(\beta-u_0+u_1)$. Key $k_b$ consists of $s_b$, the role bit $b$, $\CW^{(1)},\dots,\CW^{(n)}$ and $\CWout$; the final generator states are never substituted for the roots. One level of \cref{alg:eval} is the evaluator transition
\[
\Phi\bigl((s,t),\CW,\chi\bigr)=\bigl(\hat s^{\chi}\oplus(t\cdot\sCW),\ \hat t^{\chi}\oplus(t\cdot\tCW{\chi})\bigr),
\]
where $(\hat s^{L},\hat t^{L},\hat s^{R},\hat t^{R})=\PRG(s)$ and $\CW=(\sCW,\tCW{L},\tCW{R})$. (\Cref{alg:eval} corrects both children with the old evaluator bit $t$ and then selects child $\chi$; the selected child is the value above.) For a query $x$, party $b$'s evaluator states are $(\tilde s_b^{(0)},\tilde t_b^{(0)})=(s_b,b)$ and $(\tilde s_b^{(i)},\tilde t_b^{(i)})=\Phi\bigl((\tilde s_b^{(i-1)},\tilde t_b^{(i-1)}),\CW^{(i)},x_i\bigr)$, and its output is $y_b=(-1)^b\bigl(\Conv(\tilde s_b^{(n)})+\tilde t_b^{(n)}\CWout\bigr)$.

\subsection{Exact DPF Correctness}

The proof of \cref{thm:dpf-correct} rests on three lemmas: on the target path the evaluators coincide with the generator and keep opposite control bits (\cref{clm:target-path}); leaving the target path makes the two evaluator states identical (\cref{clm:losing-child}); and identical states stay identical and produce cancelling outputs (\cref{clm:identical}).

\begin{lemma}[Target-path agreement]\label{clm:target-path}
Fix $\alpha,x\in\bits{n}$, $\beta\in\Fp$ and a root pair $(s_0,s_1)$. For every $i\in\{0,\dots,n\}$:
(a) $t_0^{(i)}\oplus t_1^{(i)}=1$; and
(b) if $x_{[i]}=\alpha_{[i]}$, then $(\tilde s_b^{(i)},\tilde t_b^{(i)})=(s_b^{(i)},t_b^{(i)})$ for both $b\in\{0,1\}$.
\end{lemma}

\begin{proof}
We use induction on $i$. For $i=0$, $t_0^{(0)}\oplus t_1^{(0)}=0\oplus 1=1$, and each evaluator starts from its root seed and its role bit, $(\tilde s_b^{(0)},\tilde t_b^{(0)})=(s_b,b)=(s_b^{(0)},t_b^{(0)})$. This is exactly where the explicit root relation $t_0^{(0)}=0$, $t_1^{(0)}=1$ and the role bit stored in $k_b$ are used.

Assume the claim for $i-1$ and write $(s_b,t_b)=(s_b^{(i-1)},t_b^{(i-1)})$ and $a=\alpha_i$. For (a), bilinearity and the induction hypothesis give
\begin{align*}
t_0^{(i)}\oplus t_1^{(i)}&=t_0^{\keep}\oplus t_1^{\keep}\oplus\bigl((t_0\oplus t_1)\cdot\tCW{\keep}\bigr)\\
&=t_0^{\keep}\oplus t_1^{\keep}\oplus\tCW{\keep}.
\end{align*}
If $a=0$, then $\keep=L$ and $\tCW{L}=t_0^{L}\oplus t_1^{L}\oplus 0\oplus 1$, so the right-hand side equals $1$. If $a=1$, then $\keep=R$ and $\tCW{R}=t_0^{R}\oplus t_1^{R}\oplus 1$, so it again equals $1$.
For (b), $x_{[i]}=\alpha_{[i]}$ implies $x_{[i-1]}=\alpha_{[i-1]}$, so by the induction hypothesis evaluator $b$ enters level $i$ in state $(s_b,t_b)$. It expands the same seed, obtains the same children $(s_b^{L},t_b^{L},s_b^{R},t_b^{R})$, and selects child $x_i=a=\keep$ corrected with its old bit $t_b$:
$\Phi((s_b,t_b),\CW^{(i)},\keep)=\bigl(s_b^{\keep}\oplus(t_b\cdot\sCW),\ t_b^{\keep}\oplus(t_b\cdot\tCW{\keep})\bigr)=(s_b^{(i)},t_b^{(i)})$.
\end{proof}

\begin{lemma}[Losing child]\label{clm:losing-child}
Let $1\le i\le n$ with $x_{[i-1]}=\alpha_{[i-1]}$ and $x_i\neq\alpha_i$. Then $\tilde s_0^{(i)}=\tilde s_1^{(i)}$ and $\tilde t_0^{(i)}=\tilde t_1^{(i)}$.
\end{lemma}

\begin{proof}
By \cref{clm:target-path}(b) at level $i-1$, evaluator $b$ enters level $i$ in the generator state $(s_b,t_b)=(s_b^{(i-1)},t_b^{(i-1)})$, and by \cref{clm:target-path}(a) $t_0\oplus t_1=1$. Since $x_i=1-\alpha_i=\lose$, evaluator $b$ selects the corrected losing child $\bigl(s_b^{\lose}\oplus(t_b\cdot\sCW),\ t_b^{\lose}\oplus(t_b\cdot\tCW{\lose})\bigr)$. For the seeds, bilinearity gives
\[
\tilde s_0^{(i)}\oplus\tilde s_1^{(i)}=\bigl(s_0^{\lose}\oplus s_1^{\lose}\bigr)\oplus\bigl((t_0\oplus t_1)\cdot\sCW\bigr)=\sCW\oplus\sCW=0.
\]
For the control bits, $\tilde t_0^{(i)}\oplus\tilde t_1^{(i)}=t_0^{\lose}\oplus t_1^{\lose}\oplus\tCW{\lose}$. If $\alpha_i=0$, then $\lose=R$ and $\tCW{R}=t_0^{R}\oplus t_1^{R}\oplus 0$, so the XOR is $0$: the losing correction is the right correction without the extra one. If $\alpha_i=1$, then $\lose=L$ and $\tCW{L}=t_0^{L}\oplus t_1^{L}\oplus 1\oplus 1$, so the two extra ones cancel and the XOR is again $0$.
\end{proof}

\begin{lemma}[Identical states cancel]\label{clm:identical}
If $(\tilde s_0^{(i)},\tilde t_0^{(i)})=(\tilde s_1^{(i)},\tilde t_1^{(i)})$ for some $i$, then the equality holds for every level $i'\ge i$, and $y_0+y_1=0$.
\end{lemma}

\begin{proof}
The keys $k_0$ and $k_1$ differ only in the root seed and the role bit; the correction words $\CW^{(i')}$ and $\CWout$ are common to both. For $i'>i$, each state is $\Phi(\cdot,\CW^{(i')},x_{i'})$ applied to the previous state, a deterministic function whose three arguments coincide for $b=0$ and $b=1$; induction on $i'$ gives equality through level $n$. Writing $(s,t)$ for the common final state, $y_0=\Conv(s)+t\,\CWout$ and $y_1=-\bigl(\Conv(s)+t\,\CWout\bigr)$, hence $y_0+y_1=0$.
\end{proof}

\begin{proof}[Proof of \cref{thm:dpf-correct}]
We show that $y_0+y_1=\beta\cdot\one[x=\alpha]$ for every target $\alpha$, query $x$, payload $\beta$ and fixed root pair.

\emph{Case $x=\alpha$.} By \cref{clm:target-path} with $i=n$, both evaluators end in the final generator states, so $\Conv(\tilde s_b^{(n)})=u_b$ and $\tilde t_b^{(n)}=t_b^{(n)}$, and $t_0^{(n)}\neq t_1^{(n)}$. For two distinct bits, $t_0^{(n)}-t_1^{(n)}=(-1)^{t_1^{(n)}}$ in $\Fp$ (it equals $1$ if $t_1^{(n)}=0$ and $-1$ if $t_1^{(n)}=1$). Therefore
\begin{align*}
y_0+y_1&=\bigl(u_0+t_0^{(n)}\CWout\bigr)-\bigl(u_1+t_1^{(n)}\CWout\bigr)\\
&=u_0-u_1+(-1)^{t_1^{(n)}}\,(-1)^{t_1^{(n)}}\,(\beta-u_0+u_1)=\beta .
\end{align*}

\emph{Case $x\neq\alpha$.} Let $i$ be the smallest index with $x_i\neq\alpha_i$. Then $x_{[i-1]}=\alpha_{[i-1]}$, so \cref{clm:losing-child} makes the two evaluator states identical after level $i$, and \cref{clm:identical} gives $y_0+y_1=0$.

The two cases exhaust all queries. No step uses a property of $\PRG$, $\Conv$ or the roots, so the identity holds for every root pair, including adversarially chosen ones, and for every target, including non-uniform targets.
\end{proof}

\textbf{Prerequisites.}
The argument depends on four details: (i) the root relation $t_0^{(0)}=0$, $t_1^{(0)}=1$ and the role bit stored in each key (base case of \cref{clm:target-path}); (ii) the \emph{old} control bit in both the generator update and the evaluator correction (inductive steps of \cref{clm:target-path,clm:losing-child}); (iii) the correction of the selected child before the next expansion; and (iv) the signed final conversion with $(-1)^b$ and $(-1)^{t_1^{(n)}}$. A refinement that, for example, overwrote $t_b$ before computing the seed update would break the equality of \cref{clm:target-path}(b) on which the target-leaf case rests; this is a conformance bug and not a scheduling freedom (\cref{sec:arch-pipe}).

\textbf{Full-domain evaluation.}
$\EvalAll$ computes every leaf by applying $\Phi$ along that leaf's own root-to-leaf path; the depth-first order changes only the order of visits. Every leaf output therefore equals $\Eval(k_b,x)$ at that leaf and satisfies \cref{thm:dpf-correct}. A complete traversal expands each of the $2^n-1$ internal nodes exactly once.

\subsection{Malicious Generate-and-Consume}

\Cref{alg:gac} maintains the two generator states of \cref{alg:gen}, a separate evaluator state for the local role $b$ initialized to $(s_b,b)$, and at most the current correction word. At level $i$ it computes $\CW^{(i)}$ with the formulas above, advances both generator states, applies $\Phi$ with $\CW^{(i)}$ and $x_i$ to the evaluator state, and only then discards $\CW^{(i)}$. After level $n$ it computes $\CWout$, the raw share $\tilde y_b=(-1)^b\bigl(\Conv(\tilde s)+\tilde t\,\CWout\bigr)$ and the epilogue of \cref{sec:proto-epilogue}. Gates that combine several DPF evaluations, namely the prefix DPFs of \cref{alg:dcf} and the full-domain scans of \cref{alg:lut,alg:pir}, run only the (Gen) and (Eval) lines of \cref{alg:gac}, which end with the raw share $\tilde y_b$; the ledger lines and the epilogue are executed once by the enclosing gate on its combined raw share (for PIR, only the value half of the epilogue). The proof below derives $\tilde y_b=\Eval(k_b,x)$ from these lines alone, so \cref{thm:gac} applies to each constituent DPF.

\begin{proof}[Proof of \cref{thm:gac}]
We couple \cref{alg:gac} with the materialized reference execution, which first runs \cref{alg:gen} to produce $(k_0,k_1)$ and then runs \cref{alg:eval} on $(k_b,x)$, on the same roots $(s_0,s_1)$, target $\alpha$, payload $\beta$, query $x$ and role $b$, and, for the epilogue, on the same pads (they are drawn under the same joint labels). Let $g_0^{(i)},g_1^{(i)}$ and $\tilde g^{(i)}$ denote the generator and evaluator states of \cref{alg:gac} after level $i$, and $\CW^{(i)}_{\mathsf{GaC}}$ its level-$i$ correction word. We prove by induction on $i$ the statement
\begin{gather*}
g_0^{(i)}=(s_0^{(i)},t_0^{(i)}),\qquad g_1^{(i)}=(s_1^{(i)},t_1^{(i)}),\\
\tilde g^{(i)}=(\tilde s_b^{(i)},\tilde t_b^{(i)}),\qquad \CW^{(i)}_{\mathsf{GaC}}=\CW^{(i)} .
\end{gather*}
For $i=0$, both executions initialize the generator states to $(s_0,0)$ and $(s_1,1)$ and the evaluator state to $(s_b,b)$. For the step, the level-$i$ correction word is, in both executions, the same deterministic function of the level-$(i-1)$ generator states and $\alpha_i$, so the correction words and the next generator states agree. The evaluator transition is, in both executions, $\Phi(\tilde g^{(i-1)},\CW^{(i)},x_i)$ with equal arguments, so the evaluator states agree. After level $n$, $\CWout=(-1)^{t_1^{(n)}}\bigl(\beta-\Conv(s_0^{(n)})+\Conv(s_1^{(n)})\bigr)$ is the same function of equal generator states, and $\tilde y_b$ is the same function of equal evaluator states and $\CWout$; hence $\tilde y_b=\Eval(k_b,x)$. The epilogue is a deterministic function of $\tilde y_b$, the pads and the global keys, so the authenticated outputs coincide as well.

Discarding is harmless: \cref{alg:eval} reads $\CW^{(i)}$ only at level $i$, so erasing it after its evaluator use at level $i$ cannot influence any later step. Consequently \cref{alg:gac} never holds a correction-word array: at any time it stores two generator states, one evaluator state, one correction word and, at the end, one field element $\CWout$.
\end{proof}

\textbf{Separate evaluator state.}
By \cref{clm:target-path}(b), the generator state of root $b$ equals party $b$'s evaluator state only while the query prefix agrees with the target prefix. An implementation that reused the generator state of root $b$ as the evaluator state would compute, for every query $x$, the party-$b$ output of $\Eval(k_b,\alpha)$; by \cref{thm:dpf-correct} the reconstructed value would then be $\beta$ for every query, that is, a constant function rather than the point function. Keeping both generator states and a separate evaluator state for the local role is therefore a correctness requirement.

\textbf{Dependency-respecting schedules.}
The same induction applies node by node to any schedule in which each node is computed from its parent's state and the correction word of its own level after both are available. This includes the depth-first order of $\EvalAll$ and pipelined interleavings of independent contexts (\cref{sec:arch-pipe}). \Cref{thm:gac} justifies such scheduling changes only; it permits no change to the expansion, the correction formulas or the final conversion.

\subsection{Conservative DCF}

\begin{algorithm}[t]
\algstretch
\caption{Conservative distributed comparison function (\algname{DCF-Cons}) for $\beta\cdot\one[x<\alpha]$ from $n$ prefix DPFs}
\label{alg:dcf}
\begin{algorithmic}[1]
\Require target $\alpha\in\bits{n}$ (MSB first), payload $\beta\in\Fp$; role $b\in\{0,1\}$, query $x\in\bits{n}$
\Ensure keys $(K_0,K_1)$; share $v_b\in\Fp$ with $v_0+v_1=\beta\cdot\one[x<\alpha]$
\Procedure{DCF.Gen}{$\alpha,\beta$}
  \For{$i=1$ \textbf{to} $n$}
    \LineComment{first $i-1$ bits of $\alpha$, then $0$}
    \State $\alpha^{\langle i\rangle}\gets\alpha_1\cdots\alpha_{i-1}0\in\bits{i}$
    \State $\beta^{\langle i\rangle}\gets\beta\cdot\alpha_i$ \Comment{kept if zero}
    \LineComment{fresh independent roots, depth $i$}
    \State $(k_0^{\langle i\rangle},k_1^{\langle i\rangle})\gets\DPF.\Gen(\alpha^{\langle i\rangle},\beta^{\langle i\rangle})$
  \EndFor
  \State \Return $K_b=(k_b^{\langle 1\rangle},\dots,k_b^{\langle n\rangle})$ for $b\in\{0,1\}$
\EndProcedure
\Procedure{DCF.Eval}{$b,K_b,x$}
  \State $v_b\gets 0$
  \For{$i=1$ \textbf{to} $n$}
    \State $v_b\gets v_b+\DPF.\Eval(k_b^{\langle i\rangle},x_1\cdots x_i)$
  \EndFor
  \State \Return $v_b$
\EndProcedure
\end{algorithmic}
\end{algorithm}

\Cref{alg:dcf} is the conservative construction used by the interval gate (\cref{alg:interval}). It uses the orientation $\DCF_{\alpha,\beta}(x)=\beta\cdot\one[x<\alpha]$ on $n$-bit unsigned integers. Each prefix DPF runs \cref{alg:gen,alg:eval} on the domain $\bits{i}$; in the protected profile each one is executed by the (Gen) and (Eval) lines of \cref{alg:gac} without its epilogue, which by \cref{thm:gac} produce the same raw shares without materializing the keys, and \cref{alg:interval} applies a single epilogue to the combined raw share $\tilde s_{i,b}$.

\begin{proof}[Proof of \cref{thm:dcf}]
We show that $v_0+v_1=\beta\cdot\one[x<\alpha]$ for every $\alpha,x\in\bits{n}$, every $\beta\in\Fp$ and every choice of the $n$ root pairs. For $i\in\{1,\dots,n\}$ let
\[
P_i(x)=\beta\alpha_i\cdot\one\bigl[x_{[i]}=\alpha_{[i-1]}0\bigr]
\]
be the point function of the $i$-th prefix DPF evaluated at the prefix $x_{[i]}$. By \cref{thm:dpf-correct} on the depth-$i$ domain, the two shares of that DPF sum to $P_i(x)$, so by linearity $v_0+v_1=\sum_{i=1}^{n}P_i(x)$.

The term $P_i(x)$ is nonzero only if $\alpha_i=1$, $x_i=0$ and $x_{[i-1]}=\alpha_{[i-1]}$, in which case it equals $\beta$. These three conditions state precisely that $i$ is the first index at which $x$ and $\alpha$ differ and that $x_i<\alpha_i$ there. The first differing index of two distinct strings is unique, so at most one term is nonzero. For unsigned integers written most-significant bit first, $x<\alpha$ holds if and only if $x\neq\alpha$ and $x$ has the smaller bit at the first differing index. Hence $\sum_i P_i(x)=\beta$ if $x<\alpha$ and $0$ otherwise.

Reading the sum term by term recovers the intuition. If $x<\alpha$ with first differing index $i^\star$, the prefix point function $i^\star$ contributes $\beta$; every earlier prefix $i<i^\star$ has $x_i=\alpha_i$, so it either misses the target (if $x_i=1$) or hits a target whose payload $\beta\alpha_i=0$; every later prefix $i>i^\star$ already disagrees at position $i^\star$. If $x=\alpha$, every prefix that matches has payload $0$. If $x>\alpha$, the first differing index has $x_{i^\star}=1$ and misses the target bit $0$, and the remaining terms vanish as before.
\end{proof}

\textbf{Cost.}
Prefix DPF $i$ has depth $i$, so one evaluation walks $\sum_{i=1}^{n}i=n(n+1)/2$ levels per party, which is $1{,}891$ levels for $n=61$, and generation expands two states per level, $n(n+1)$ expansions in total. With both modules generating and each evaluating, a level costs $6$ AES blocks, so the two thresholds of an interval gate cost $2\cdot6\cdot1{,}891=22{,}692$ blocks (\cref{app:model}). This is $O(n^2)$, not the $O(n)$ of optimized DCF constructions~\cite{boyle2021function,storrier2023grotto}. An optimized DCF may replace \cref{alg:dcf} only after its own correctness proof, its distribution and interface obligations, and its measured cost are established; the evaluation reports such a variant as a projection that \cref{thm:dcf} does not cover.

\textbf{All $n$ keys, including zero payloads.}
The number of prefix DPFs, their depths and their order are fixed by $n$. Generating a key for every prefix, including those with $\alpha_i=0$ and hence payload $0$, keeps the schedule and the module's observable trace independent of $\alpha$, which is derived from a secret mask, as the trace assumption behind $\eps{trace}$ requires. Omitting zero-payload keys would reveal the Hamming weight of $\alpha$ through the work performed.

\textbf{Orientation.}
The semi-honest DFA design also uses a greater-than orientation~\cite{dfa2026trusted}. We convert explicitly at interfaces and never change the comparison symbol without changing the function; for example, $\beta\cdot\one[x\ge\alpha]=\beta-\beta\cdot\one[x<\alpha]$ is obtained by letting party $0$ output $\beta-v_0$ and party $1$ output $-v_1$.

\subsection{Fresh Authentication Epilogue}

Recall the outputs $y_b,\gamma_{b,j}$ of \cref{thm:epilogue}; both modules draw the pads $\zeta,\kappa_1,\dots,\kappa_t$ under joint labels that are distinct for every output wire.

\begin{proof}[Proof of \cref{thm:epilogue}]
We work in hybrid $\Hyb{2}$ (ideal randomness and bindings) of \cref{app:proof-main}, where draws under distinct labels are independent and uniform.

\emph{Correctness.} The pad $\zeta$ enters with opposite signs, so $y_0+y_1=\tilde y_0+\tilde y_1$. For each lane $j$,
\[
\gamma_{0,j}+\gamma_{1,j}=\Delta_j y_0+\kappa_j+\Delta_j y_1-\kappa_j=\Delta_j(y_0+y_1),
\]
so $(\authb{y}{0},\authb{y}{1})$ is an authenticated sharing of $\tilde y_0+\tilde y_1$ under all $t$ global keys.

\emph{Local hiding.} Fix $b$ and condition on the raw shares $\tilde y_0,\tilde y_1$ and on $\Delta_1,\dots,\Delta_t$. The pad vector $(\zeta,\kappa_1,\dots,\kappa_t)$ is uniform on $\Fp^{t+1}$ and independent of the conditioning, because its labels differ from those of the DPF roots, the payload and the keys. Consider the map
\begin{align*}
\Psi_b:\ &(\zeta,\kappa_1,\dots,\kappa_t)\\
&\mapsto\bigl(y_b,\ \Delta_1y_b+(-1)^b\kappa_1,\ \dots,\ \Delta_ty_b+(-1)^b\kappa_t\bigr),
\end{align*}
with $y_b=\tilde y_b+(-1)^b\zeta$. It is a bijection of $\Fp^{t+1}$ with inverse $\zeta=(-1)^b(y_b-\tilde y_b)$ and $\kappa_j=(-1)^b(\gamma_{b,j}-\Delta_jy_b)$. The image of the uniform distribution under a bijection is uniform, so every tuple $(y_b,\gamma_{b,1},\dots,\gamma_{b,t})$ has probability exactly $p^{-(t+1)}$, independently of the gate value $\tilde y_0+\tilde y_1$ and of the global keys.

\emph{Joint freshness.} Let $W$ be a set of output wires whose pads carry pairwise distinct labels. Their pad vectors are mutually independent, and the product of the maps $\Psi_b$ over $W$ is a bijection of $\Fp^{(t+1)|W|}$. Hence the party-$b$ tuples of all wires in $W$ are jointly uniform and independent of all raw outputs and keys.
\end{proof}

The statement describes the tuples at creation time. Later openings of values that depend on $y$, such as $d=y-a$ in a multiplication, correlate the two parties' shares; the circuit simulation of \cref{app:proof-main} accounts for these correlations and does not claim that every intermediate value is independent.

\textbf{The epilogue is not an authentication oracle.}
If a module applied this map to a host-supplied value $v_c$ instead of its internally computed $\tilde y_c$, the result would be an authenticated sharing of $v_c+\tilde y_{1-c}$, a value shifted by an offset of the host's choice, whose tags satisfy the MAC relation exactly. Every later \CheckOpen{} would accept it with probability $1$, and \cref{thm:mac-sound} would offer no protection (scenario S2 of \cref{tab:scenarios}). For this reason the interface has no \cmd{authenticate(host_value)} entry point and accepts only raw shares computed inside the module or produced by an explicitly verified backend.

\textbf{Role and label binding.}
The cancellations above require both modules to draw the same pads under the same joint label and to apply opposite signs. The joint label never contains the local role, while the roles themselves are bound at setup (\cref{sec:proto-setup}). If both modules used the same sign, $y_0+y_1$ would equal $\tilde y_0+\tilde y_1\pm2\zeta$, a uniformly wrong value.

\section{Proofs for Authenticated Openings and Arithmetic}\label{app:proof-auth}

This appendix proves the hidden-key invariant (\cref{lem:hidden-key}), the wrong-value acceptance bound of freeze-before-challenge \CheckOpen{} (\cref{thm:mac-sound}), and the correctness of inputs, linear operations, multiplications and outputs (\cref{thm:arith}). It also records three counterexamples that delimit the statements: a single tampered lane, a challenge vector shared across lanes, and the same check over $\Z_{2^{64}}$. The analysis adapts the information-theoretic MAC arguments of SPDZ~\cite{damgard2012spdz,damgard2013practical} to $t$ independent lanes with per-lane challenges, residuals that never leave the modules, and a ledger that terminates the epoch on the first failure.

\textbf{Setting.}
As in \cref{app:proof-main}, host $b^\star\in\{0,1\}$ is statically corrupted, $\bar b=1-b^\star$ is the honest party, and all probabilities are taken in $\Hyb{2}$: setup and protected channels are ideal, every PRF or KDF output is an independent uniform draw per distinct label and the same draw for the same joint label, and all bindings hold. \Cref{thm:main} charges the distance between the real execution and $\Hyb{2}$ to $\eps{setup}$, $\eps{channel}$, $\eps{prf}$ and $\eps{bind}$; the exact replacements of $\Hyb{2}$ (\cref{thm:gac,lem:bijection,lem:noninv}) cost nothing. Traces are still produced by the modules in $\Hyb{2}$ and are treated in (vi) below. The key shares $\Delta_{0,j},\Delta_{1,j}$ are independent and uniform for $j=1,\dots,t$; host $b^\star$ receives $\bm{\Delta}_{b^\star}=(\Delta_{b^\star,1},\dots,\Delta_{b^\star,t})$ at setup, and we write
\[
\bm{\Delta}_{\bar b}=(\Delta_{\bar b,1},\dots,\Delta_{\bar b,t})\in\Fp^{t}
\]
for the honest key shares. Since $\Delta_j=\Delta_{b^\star,j}+\Delta_{\bar b,j}$ and host $b^\star$ knows $\Delta_{b^\star,j}$, uniformity of $\bm{\Delta}_{\bar b}$ given the corrupt view is equivalent to uniformity of $(\Delta_1,\dots,\Delta_t)$ given that view. The view $V$ of host $b^\star$ contains its random coins, $\bm{\Delta}_{b^\star}$, everything its module, the peer host and the clients send to it, and its observable trace.

\textbf{Faithful tuples and errors.}
For every wire $\nu$, the \emph{faithful tuple} $(\hat x_{b^\star,\nu},\hat\gamma_{b^\star,1,\nu},\dots,\hat\gamma_{b^\star,t,\nu})$ is the tuple that host $b^\star$ would hold if it had applied every prescribed local operation (\cref{fig:arith}) to the module outputs and the public openings that it actually received. The honest host always holds the prescribed tuple $(x_{\bar b,\nu},\gamma_{\bar b,1,\nu},\dots,\gamma_{\bar b,t,\nu})$, and the \emph{faithful value} of $\nu$ is $\hat x_\nu=x_{\bar b,\nu}+\hat x_{b^\star,\nu}$. Faithful tuples are produced by an honest execution, so by the correctness part of \cref{thm:arith} and by \cref{thm:epilogue} they satisfy the MAC relation
\[
\gamma_{\bar b,j,\nu}+\hat\gamma_{b^\star,j,\nu}=\Delta_j\,\hat x_\nu\qquad(j=1,\dots,t),
\]
whatever public values were opened earlier. When host $b^\star$ submits a batch of $m$ tuples to \cmd{FREEZE_BATCH}, we write its submission, position by position, as
\[
x_{b^\star,i}=\hat x_{b^\star,i}+e_i,\qquad \gamma_{b^\star,j,i}=\hat\gamma_{b^\star,j,i}+d_{j,i}.
\]
The \emph{value-error vector} $\bm{e}=(e_1,\dots,e_m)\in\Fp^m$ and the \emph{tag-error vectors} $\bm{d}_j=(d_{j,1},\dots,d_{j,m})\in\Fp^m$ are well defined because the canonical parser accepts only vectors of field elements; any other deviation is rejected before freezing. Submitting the tuple of another wire, permuting positions or replaying an old tuple are all captured by suitable $(\bm{e},\bm{d}_j)$. Both are functions of $V$ at freezing time, since the faithful tuple is computed from $V$ and the submission is chosen from $V$. A check is a \emph{wrong-value opening} if $\bm{e}\neq\bm{0}$.

\begin{lemma}[Residual identity]\label{clm:residual}
For every frozen batch and every lane $j$, the residuals of \cref{alg:checkopen} satisfy $\sigma_{0,j}+\sigma_{1,j}=\inner{\bm{\rho}_j}{\bm{d}_j-\Delta_j\bm{e}}$, where $\bm{\rho}_j=(\rho_{j,1},\dots,\rho_{j,m})$.
\end{lemma}

\begin{proof}
Both modules reconstruct $z_i=x_{\bar b,i}+x_{b^\star,i}=\hat x_i+e_i$ from the frozen value shares, and module $b$ computes $\sigma_{b,j}=\sum_i\rho_{j,i}(\gamma_{b,j,i}-\Delta_{b,j}z_i)$. Adding the two residuals and using $\Delta_{0,j}+\Delta_{1,j}=\Delta_j$ and the MAC relation of faithful tuples,
\begin{align*}
\sigma_{0,j}+\sigma_{1,j}&=\sum_{i=1}^{m}\rho_{j,i}\bigl(\gamma_{\bar b,j,i}+\hat\gamma_{b^\star,j,i}+d_{j,i}-\Delta_j(\hat x_i+e_i)\bigr)\\
&=\sum_{i=1}^{m}\rho_{j,i}\bigl(d_{j,i}-\Delta_je_i\bigr).\qedhere
\end{align*}
\end{proof}

Hence the check outcome depends only on the errors, the keys and the challenges, not on the honest tags, their pads or the faithful values.

\subsection{Hidden-Key Invariant}

Let $\tau$ be the first time at which challenges are drawn for a wrong-value opening ($\tau=\infty$ if there is none). Every check resolved before $\tau$ has $\bm{e}=\bm{0}$. An epoch abort draws no further challenge (\cref{sec:proto-ledger}), and in the admitted execution of \cref{thm:main} a failed check aborts the execution.

\begin{proof}[Proof of \cref{lem:hidden-key}]
We prove that, in $\Hyb{2}$, $\bm{\Delta}_{\bar b}$ is uniform on $\Fp^{t}$ and independent of the corrupt view $V_s$ at every time $s\le\tau$, and that the challenges drawn at $\tau$ are uniform and independent of $(V_\tau,\bm{\Delta}_{\bar b})$. If the check at $\tau$ is rejected, the epoch is poisoned and no challenge is drawn afterwards, so the information carried by that rejection never reaches another check; if it is accepted, it is the first accepted wrong-value opening, where the lemma stops.

\emph{Strategy.} We show that $V_\tau$ can be generated from $\bm{\Delta}_{b^\star}$, the adversary's coins, the \emph{value layer} and fresh uniform variables, all independent of $\bm{\Delta}_{\bar b}$. The value layer consists of the honest inputs, all value shares, masks, triple components, DPF roots and targets, raw FSS outputs, value pads and opened values; no rule of the protocol uses a key when it computes a value share (\cref{fig:arith}, \cref{alg:gac}), so the value layer is a function of randomness drawn under labels distinct from the key labels. We go through every kind of observation that can enter $V_s$ for $s\le\tau$.

\emph{(i) Setup.} Host $b^\star$ obtains $\bm{\Delta}_{b^\star}$ and a public transcript, which in $\Hyb{2}$ is simulated from public data. $\bm{\Delta}_{\bar b}$ is drawn independently of $\bm{\Delta}_{b^\star}$.

\emph{(ii) Value-layer observations.} Value shares delivered by its module (\AuthShare{} values $v_{b^\star}$, epilogue values $y_{b^\star}$), complete input masks $r$ revealed to a corrupt owner, output masks revealed to a corrupt recipient, accepted opened values, validation outcomes ($\mathsf{status}$), certificates and all public metadata are functions of the value layer and of public data. Certificates bind session identifiers, public metadata, output hashes and status; they involve no MAC key.

\emph{(iii) Tags.} Every tag that host $b^\star$ receives from its module has one of three forms: the pad $\pi$ itself (\AuthShare{} with $b^\star=0$, where $\gamma_{0,j}=\pi$), $\Delta_jv-\pi$ (\AuthShare{} with $b^\star=1$), or $\Delta_jy_{b^\star}+(-1)^{b^\star}\kappa_j$ (epilogue). In each case the pad $\pi$ or $\kappa_j$ is drawn under its own label and occurs in exactly two places: this corrupt tag and the corresponding honest tag. The honest tag, and every tag the honest host later derives from it by local linear operations, is used only inside honest residuals. Complete tags are never delivered to owners or recipients (\cref{fig:arith}), and by \cref{clm:residual} the sum of residuals, the only quantity that influences a check outcome, does not depend on any honest tag. Hence, conditioned on everything else, including $\bm{\Delta}_{\bar b}$, the value layer and all other pads, each corrupt tag is the image of a uniform pad under a bijection of $\Fp$ and is therefore uniform, and we replace it by a fresh uniform variable without changing the joint distribution of $(V_\tau,\bm{\Delta}_{\bar b})$, one tag at a time.

\emph{(iv) Local operations of host $b^\star$.} Additions, public scalings and public-constant updates produce deterministic functions of $V_s$ and $\bm{\Delta}_{b^\star}$ and add no new information.

\emph{(v) Check outcomes before $\tau$.} Such a check has $\bm{e}=\bm{0}$, so by \cref{clm:residual} it accepts if and only if $\inner{\bm{\rho}_j}{\bm{d}_j}=0$ for every $j$. The tag errors $\bm{d}_j$ are functions of $V_s$, and the challenges are fresh uniform draws, so the outcome is a function of $V_s$ and of randomness independent of $\bm{\Delta}_{\bar b}$. \Cref{alg:checkopen} returns one generic status for all lanes, and challenges and residuals remain inside the modules and on the protected channel. A rejected tag-only check poisons the epoch, after which no challenge is drawn.

\emph{(vi) Traffic and trace.} Protected-channel traffic reveals only lengths and times, which are fixed by $\mu$ and the schedule. For the trace (lengths, timing, host-visible memory and DMA activity, power, thermal and clock telemetry) we use the exact key-independence clause of the trace assumption (\cref{sec:sec-assump,app:func-module}): given the command sequence issued by the host, $\mu$, the public lengths and the accept or abort outcomes, the trace is independent of the global MAC keys, challenges and residuals, since all lanes are computed and aggregated in data-independent time before one generic status is returned and host software cannot drive the module's clock or voltage outside its safe operating range. These conditioning data arise in (i) to (v) without $\bm{\Delta}_{\bar b}$, so the trace adds no information on $\bm{\Delta}_{\bar b}$. This step is exact; only simulating the trace from public data costs $\eps{trace}$, charged between $\Hyb{3}$ and $\Hyb{4}$ in \cref{thm:main}.

After the replacements of (iii), every observation in $V_\tau$ is a function of $\bm{\Delta}_{b^\star}$, the adversary's coins, the value layer, the check outcomes of (v) and fresh uniform variables, none of which depends on $\bm{\Delta}_{\bar b}$, and the trace is conditionally independent of $\bm{\Delta}_{\bar b}$ by (vi). The argument is sequential: each new observation is generated from the preceding view, which already determines the adversary's next submission, and from fresh randomness independent of $\bm{\Delta}_{\bar b}$. By induction on $s$, the view $V_s$ is independent of $\bm{\Delta}_{\bar b}$ for every $s\le\tau$, and $\bm{\Delta}_{\bar b}$ is uniform given $V_s$. The challenges drawn at $\tau$ carry labels bound to the two frozen roots and the check identifier, which were never used before, so they are uniform and independent of $(V_\tau,\bm{\Delta}_{\bar b})$.
\end{proof}

\subsection{Soundness of Freeze-Before-Challenge CheckOpen}

\begin{proof}[Proof of \cref{thm:mac-sound}]
\emph{Step 1: one batch.} Consider a check whose frozen batch has value error $\bm{e}\neq\bm{0}$ and whose challenges are drawn when no wrong-value opening has been accepted before. Condition on the corrupt view $V$ at the time the challenges are drawn. Then $\bm{e}$ and $\bm{d}_1,\dots,\bm{d}_t$ are fixed, $\bm{\Delta}_{\bar b}$ is uniform on $\Fp^t$ by \cref{lem:hidden-key}, and the challenge vectors $\bm{\rho}_1,\dots,\bm{\rho}_t\in\Fp^m$ are uniform and independent of each other and of $\bm{\Delta}_{\bar b}$, because every lane and every batch position has its own label. By \cref{clm:residual} the check accepts exactly on the event $\bigcap_{j=1}^{t}A_j$, where
\[
A_j=\bigl\{\inner{\bm{\rho}_j}{\bm{v}_j}=0\bigr\},\qquad \bm{v}_j=\bm{d}_j-\Delta_j\bm{e} .
\]

\emph{Step 2: one lane.} Choose an index $i^\star$ with $e_{i^\star}\neq 0$. If $\bm{v}_j=\bm{0}$, its $i^\star$-th coordinate gives $\Delta_j=d_{j,i^\star}e_{i^\star}^{-1}$; the inverse exists because $p$ is prime. Since $\Delta_j=\Delta_{b^\star,j}+\Delta_{\bar b,j}$ with $\Delta_{b^\star,j}$ fixed and $\Delta_{\bar b,j}$ uniform, $\Pr[\bm{v}_j=\bm{0}\mid V]\le 1/p$; the remaining coordinates can only exclude this single candidate key. If $\bm{v}_j\neq\bm{0}$, choose $i$ with $v_{j,i}\neq 0$. For any fixed values of the other coordinates of $\bm{\rho}_j$, the map $\rho_{j,i}\mapsto\inner{\bm{\rho}_j}{\bm{v}_j}$ is affine with nonzero slope, hence a bijection of $\Fp$, and $\Pr[A_j\mid V,\bm{v}_j]=1/p$ exactly. Writing $q_j=\Pr[\bm{v}_j=\bm{0}\mid V]\le 1/p$,
\[
\Pr[A_j\mid V]=q_j+(1-q_j)\frac{1}{p}\le\frac{1}{p}+\Bigl(1-\frac{1}{p}\Bigr)\frac{1}{p}=\frac{2p-1}{p^2}<\frac{2}{p}.
\]

\emph{Step 3: independent lanes.} Given $V$, the pairs $(\Delta_{\bar b,1},\bm{\rho}_1),\dots,(\Delta_{\bar b,t},\bm{\rho}_t)$ are mutually independent, and $A_j$ is determined by the $j$-th pair. Therefore
\[
\Pr[\text{accept}\mid V]=\prod_{j=1}^{t}\Pr[A_j\mid V]\le\Bigl(\frac{2p-1}{p^2}\Bigr)^{t}<\Bigl(\frac{2}{p}\Bigr)^{t},
\]
and averaging over $V$ gives the single-batch bound of the theorem. The adversary may choose $\bm{e}$ and $\bm{d}_j$ with its own coins; the bound holds for every such choice because it holds conditioned on $V$.

\emph{Step 4: $Q$ batches.} Number the at most $Q$ checked batches of the execution by the time their challenges are drawn, and let $B_k$ be the event that batch $k$ has $\bm{e}\neq\bm{0}$, is accepted, and no earlier batch with $\bm{e}\neq\bm{0}$ was accepted. An incorrect value can be consumed or released only after an accepted check on its frozen bytes: the ledger admits consumption and \cmd{FINAL_OUTPUT} only in state \lstate{VERIFIED}, consumption rereads the frozen bytes, and a certificate cannot be transferred to another transcript. Such a check makes available the values $z_i=\hat x_i+e_i$, where $\hat x_i$ is the faithful value, so an incorrect value requires $\bm{e}\neq\bm{0}$, and the event that an incorrect value is ever accepted is contained in $\bigcup_kB_k$. On $B_k$ no earlier wrong-value check was accepted, and none was rejected either, because a rejection poisons the epoch and no later challenge is drawn; so Steps 1 to 3 apply to batch $k$ and $\Pr[B_k]\le(2/p)^t$. The union bound gives
\[
\Pr\Bigl[\bigcup_{k=1}^{Q}B_k\Bigr]\le Q\Bigl(\frac{2}{p}\Bigr)^{t}.\qedhere
\]
\end{proof}

\textbf{Tag-only errors.}
\Cref{thm:mac-sound} bounds wrong \emph{values}; it is not a claim that every tag modification is detected with probability $1-(2/p)^t$. If $\bm{e}=\bm{0}$, $\bm{d}_1\neq\bm{0}$ and $\bm{d}_2=\dots=\bm{d}_t=\bm{0}$, then $\bm{v}_1=\bm{d}_1$, $\bm{v}_2=\dots=\bm{v}_t=\bm{0}$, and the check accepts with probability exactly $1/p$. Such an acceptance releases the correct value, so it is not a correctness failure. A surviving tag error can reappear, after local linear operations, in the tag errors of later checks, where, being independent of the keys, it again leaks nothing about them (item (v) in the proof of \cref{lem:hidden-key}).

\textbf{Remark (shared-challenge cancellation).}
The product in Step 3 requires independent challenge vectors per lane. Suppose instead that one vector $\bm{\rho}$ were reused, $\bm{\rho}_1=\dots=\bm{\rho}_t=\bm{\rho}$. An adversary submits any $\bm{e}\neq\bm{0}$ and leaves all tags faithful, $\bm{d}_1=\dots=\bm{d}_t=\bm{0}$. Then the complete residual of lane $j$ is $-\Delta_j\inner{\bm{\rho}}{\bm{e}}$, and all $t$ lanes vanish simultaneously whenever $\inner{\bm{\rho}}{\bm{e}}=0$, for every value of the keys. Since $\inner{\bm{\rho}}{\bm{e}}$ is uniform for $\bm{e}\neq\bm{0}$, the wrong value is accepted with probability at least $1/p\approx 2^{-61}$, whatever $t$ is. The $t$ lanes then give essentially the guarantee of a single lane: for $Q=2^{32}$ the union bound becomes about $Q/p\approx 2^{-29}$ instead of $2^{-148}$, because the per-lane events share the single linear form $\inner{\bm{\rho}}{\bm{e}}$ and are no longer independent. \sys therefore derives $\rho_{j,i}$ under distinct labels for every lane $j$ and position $i$, and a configuration that reuses one challenge vector across lanes is rejected (row~8 of \cref{tab:attacks}).

\textbf{Remark (why the challenge must follow the freeze).}
Step 1 uses that $\bm{e}$ and $\bm{d}_j$ are fixed before the challenges exist and that the host never sees them. If the host learned $\bm{\rho}_1,\dots,\bm{\rho}_t$ before its bytes were fixed, because the challenges were derived before the freeze, exposed to the host, or applied to a buffer that the host could still modify, it could pick a nonzero $\bm{e}$ in the common kernel $\{\bm{e}:\inner{\bm{\rho}_j}{\bm{e}}=0\ \text{for all }j\}$, which has dimension at least $m-t\ge 1$ when $m>t$, and keep $\bm{d}_j=\bm{0}$. Every residual would then vanish for every key, and the wrong value would be accepted with probability $1$. Likewise, if a module verified one buffer and later consumed another, the consumed value could carry an arbitrary error that no check ever sees. \Cref{alg:checkopen} excludes both attacks: challenges are derived only after both modules have frozen and exchanged their roots, they never leave the modules, and consumption rereads only the frozen authenticated bytes (row~6 of \cref{tab:attacks}).

\textbf{Remark (the proof does not apply over $\Z_{2^{64}}$).}
Step 2 divides by $e_{i^\star}$ and uses that a nonzero linear form over a field is uniform. Both steps fail in the ring $\Z_{2^{64}}$, because $2^{63}$ is a zero divisor. Consider the same check with keys and challenges uniform in $\Z_{2^{64}}$, a batch whose only nonzero error is $e_1=2^{63}$ (the opened value has its most significant bit flipped), and faithful tags, $\bm{d}_j=\bm{0}$. For an even key, $\Delta_je_1=(\Delta_j/2)\cdot 2^{64}\equiv 0 \pmod{2^{64}}$, so $\bm{v}_j=\bm{0}$ and lane $j$ accepts for every challenge; this single event already has probability $1/2$, whereas the field argument allows at most $1/p$. For an odd key, $v_{j,1}=-\Delta_j2^{63}\equiv 2^{63}$ and the residual $\rho_{j,1}2^{63}$ vanishes for every even challenge, with probability $1/2$ instead of $1/p$. Each lane therefore accepts with probability $1/2+1/4=3/4$, and with $t=3$ independent lanes the wrong value is accepted with probability $(3/4)^3=27/64>0.42$. Ring arithmetic needs a SPDZ2k-style representation over a larger ring with its own analysis~\cite{cramer2018spdz2k} (\cref{sec:proto-repr}).

\subsection{Inputs, Arithmetic and Outputs}

In this subsection we write the party index as $\ell\in\{0,1\}$, because $b$ denotes the second component of a Beaver triple $(\auth{a},\auth{b},\auth{c})$~\cite{beaver1991efficient}. A party-$\ell$ tuple is $\authb{x}{\ell}=(x_\ell,\gamma^{x}_{\ell,1},\dots,\gamma^{x}_{\ell,t})$ with $x_0+x_1=x$ and $\gamma^{x}_{0,j}+\gamma^{x}_{1,j}=\Delta_jx$.

\begin{proof}[Proof of \cref{thm:arith}]
\emph{Part 1: exact operation with correct modules and honest messages.} Each item verifies both the value equation and the tag equation of the output.

\emph{Input ($\Pi_{\mathsf{Input}}$).} For a fresh authenticated mask $\auth{r}$ and the bound public difference $\varepsilon=x-r$, the input tuple is $\bigl(r_\ell+\one[\ell=0]\varepsilon,\ \gamma^{r}_{\ell,j}+\Delta_{\ell,j}\varepsilon\bigr)$. Its value shares sum to $r+\varepsilon=x$ and its tags to $\Delta_jr+(\Delta_{0,j}+\Delta_{1,j})\varepsilon=\Delta_j(r+\varepsilon)=\Delta_jx$.

\emph{Linear operations ($\Pi_{\mathsf{Lin}}$).} For $\auth{x}+\auth{y}$ both parties add componentwise, so the sums are $x+y$ and $\Delta_jx+\Delta_jy=\Delta_j(x+y)$. For a public scalar $\eta$ both parties multiply componentwise, giving $\eta x$ and $\Delta_j\eta x$. To add a public constant $c_0$, party $0$ adds $c_0$ to its value and each party $\ell$ adds $\Delta_{\ell,j}c_0$ to its tag, giving $x+c_0$ and $\Delta_jx+\Delta_jc_0$. Tensors are handled elementwise.

\emph{Scalar multiplication ($\Pi_{\mathsf{Mult}}$).} Let $\bigl(\auth{a},\auth{b},\auth{c}\bigr)$ be a fresh triple with $c=ab$, and let $d=x-a$ and $e=y-b$ be the values opened through \cref{alg:checkopen}. Party $\ell$ sets
\begin{align*}
z_\ell&=c_\ell+d\,b_\ell+e\,a_\ell+\one[\ell=0]\,de,\\
\gamma^{z}_{\ell,j}&=\gamma^{c}_{\ell,j}+d\,\gamma^{b}_{\ell,j}+e\,\gamma^{a}_{\ell,j}+\Delta_{\ell,j}\,de .
\end{align*}
The value shares sum to
\[
c+db+ea+de=ab+(x-a)b+(y-b)a+(x-a)(y-b)=xy,
\]
and the tags sum to $\Delta_jc+d\Delta_jb+e\Delta_ja+\Delta_jde=\Delta_j(c+db+ea+de)=\Delta_jxy$.

\emph{Matrix multiplication.} For $X\in\Fp^{m\times k}$, $Y\in\Fp^{k\times n}$ and a fresh matrix triple $(A,B,C)$ with $C=AB$ and elementwise tag matrices $\Gamma^{A}_{\ell,j},\Gamma^{B}_{\ell,j},\Gamma^{C}_{\ell,j}$, let $D=X-A$ and $E=Y-B$ be opened. Party $\ell$ sets $Z_\ell=C_\ell+DB_\ell+A_\ell E+\one[\ell=0]DE$ and $\Gamma^{Z}_{\ell,j}=\Gamma^{C}_{\ell,j}+D\Gamma^{B}_{\ell,j}+\Gamma^{A}_{\ell,j}E+\Delta_{\ell,j}DE$. Keeping every factor on its side,
\begin{align*}
C+DB+AE+DE&=AB+(X-A)B+A(Y-B)\\
&\quad+(X-A)(Y-B)=XY,
\end{align*}
and, because the scalar $\Delta_j$ commutes with matrices, the tag matrices sum to $\Delta_j(C+DB+AE+DE)=\Delta_jXY$. Replacing $AE$ by $EA$ is wrong in general, since matrix products do not commute, and ill-typed for nonsquare shapes.

\emph{Output ($\Pi_{\mathsf{Out}}$).} For a private output $y$ to recipient $O$, a fresh authenticated mask $\auth{o}$ is generated, the complete $o$ is revealed only to $O$, and $\upsilon=y+o$ is opened through \cref{alg:checkopen}; the recipient computes $\upsilon-o=y$. Because $o$ is fresh and uniform, $\upsilon$ is uniform and independent of $y$ for every party that does not hold $o$; for a corrupt recipient, who receives $o$ early, the simulator programs $\upsilon=y+o$ after $\FC$ releases $y$ (\cref{app:proof-main}). The recipient never receives both parties' complete tags: knowing $\Delta_j\upsilon$ for a known nonzero $\upsilon$ would reveal $\Delta_j=(\Delta_j\upsilon)\upsilon^{-1}$ and void \cref{lem:hidden-key} for anyone who later obtained these values. Public outputs are opened directly through \cref{alg:checkopen}.

\emph{Triple correctness is not a MAC property.} If a triple had $c'=ab+\delta$ with $\delta\neq 0$ and tags computed consistently on $c'$, the multiplication would output a correctly authenticated sharing of $xy+\delta$, and no check would ever detect it. The theorem therefore relies on the modules computing $c=ab$ inside the trusted boundary (\AuthShare{} in \cref{sec:proto-setup}); a replacement preprocessing service must supply a proof of triple correctness in addition to authenticated shares (row~4 of \cref{tab:attacks}).

\emph{Part 2: malicious hosts.} Let $G_{\mathsf{open}}$ be the event that no wrong-value opening is ever accepted and $G_{\mathsf{bind}}$ the event that every binding holds. By \cref{thm:mac-sound}, $\Pr[\neg G_{\mathsf{open}}]\le Q(2/p)^t$, and $\neg G_{\mathsf{bind}}$ is charged to $\eps{bind}$ in \cref{thm:main}. We show that on $G_{\mathsf{open}}\cap G_{\mathsf{bind}}$ every execution either aborts, or releases exactly $C(x^\star)$, where $x^\star$ consists of the honest inputs and the inputs extracted from corrupt owners. We proceed by induction over the gates of $\mu$ in manifest order, which the ledger enforces, and maintain three invariants for all gates processed so far without abort:
\begin{enumerate}[leftmargin=*,nosep,label=(I\arabic*)]
\item every value opened, consumed or released equals its faithful value;
\item the faithful value of every circuit wire equals the value of that wire in $C$ on $x^\star$;
\item every preprocessing item that was consumed was issued once, for this gate and this role, and consumed once.
\end{enumerate}
(I1) holds on $G_{\mathsf{open}}$: every opening is a check that reconstructs $z_i=\hat x_i+e_i$ from frozen bytes, where $\hat x_i$ is the faithful value, accepted checks have $\bm{e}=\bm{0}$, and consumption rereads the frozen bytes. (I3) holds on $G_{\mathsf{bind}}$: \cmd{RESERVE_PREP} burns identifiers before issuing material, \lstate{CONSUMED} records one-time use, and any reuse, role swap, cross-gate consumption or post-challenge replacement is an invalid ledger transition that aborts. It remains to show (I2) gate by gate.

\emph{Input gates.} For an honest owner, the authenticated $\varepsilon=x-r$ cannot be altered by the corrupt host, and both modules bind the same $\varepsilon$ or abort; the faithful value is the honest input by Part 1. For a corrupt owner, whatever canonical $\varepsilon^\star$ both modules bind defines the extracted input $x^\star=r+\varepsilon^\star$; a different $\varepsilon$ sent to the two modules causes abort. Canonical encoding alone does not place $x^\star$ in the declared interface, so the proof uses the validation step of $\Pi_{\mathsf{Input}}$ (\cref{fig:arith}): after binding, both modules compute $x^\star$, which they can because they generated $r$, and abort unless $x^\star\in\mathsf{In}_\nu$ (\cref{app:func-circuit}); at the public boundaries declared in $\mu$ they evaluate every predicate of $\mathcal{V}$ on the bound inputs and abort if one fails. The public schedule of $\mu$ places every input gate before every validity check, and both before the first output opening, so no output precedes such an abort. These outcomes are exactly the vector $\mathsf{status}$ of $\FC$, which is part of $\Leak_C$ (\cref{def:leakage}), and the simulator emulates the same abort (\cref{app:proof-main}). A commitment-bound input is the value decoded by the protected loader, or abort. Hence the execution continues exactly on the inputs that $\FC$ accepts.

\emph{Linear gates.} Faithful tuples are updated by the same local rules, so Part 1 transfers (I2) from the inputs of the gate to its output. Deviations of the corrupt host at such a gate change only its stored tuples, which appear as errors $(\bm{e},\bm{d}_j)$ in a later check.

\emph{Multiplication gates.} By (I3) the triple is fresh and bound to the gate, and by the premise on $\mathcal{F}_{\mathsf{T}}$ it satisfies $c=ab$. By (I1) the opened values are $d=\hat x-a$ and $e=\hat y-b$, where $\hat x$ and $\hat y$ are the faithful values of the two input wires. Part 1 then shows that the faithful output value equals $\hat x\hat y$, which by (I2) for the inputs is the value of the gate in $C$. The matrix case is identical with the order of factors preserved.

\emph{Nonlinear and FSS gates.} Their masked openings are checks, so (I1) applies. Their outputs are fresh epilogue tuples whose sum is the gate value by \cref{thm:epilogue} and \cref{thm:bitdec,thm:cmp-trunc,thm:lut,thm:interval} (\cref{app:proof-gates}). A direct protected gate (\cref{sec:proto-direct}) recomputes $x=c-r$ from its verified opening inside the modules and shares the exact gate value with \AuthShare{}, so (I2) holds by the premise on $\mathcal{F}_{\mathsf{T}}$.

\emph{Output gates.} By (I1) the opened value is $\upsilon=\hat y+o$ with $\hat y$ the faithful output value, and the recipient computes $\hat y$, which equals the corresponding output of $C(x^\star)$ by (I2). The release guard admits \cmd{FINAL_OUTPUT} only after \lstate{VERIFIED}, and the certificate binds the session, epoch, gate and recipient, so a released value cannot be presented as the output of a different gate or transcript.

The corrupt host can further refuse to submit, delay, or submit malformed frames; each of these causes an abort, which $\FC$ permits because it provides neither fairness nor availability. Hence, conditioned on $G_{\mathsf{open}}\cap G_{\mathsf{bind}}$, malicious behavior either selects an allowed input, preserves the computation, or causes abort.
\end{proof}

\section{Proofs for Exact Gates, Refactoring, and Matrix Checks}
\label{app:proof-gates}

This appendix gives the complete fallback algorithms that \cref{sec:proto-gates} only summarizes (\cref{alg:bitdec,alg:trunc}), the frozen-operand matrix check of \cref{sec:proto-freivalds} (\cref{alg:freivalds}), and full proofs of \cref{thm:bitdec,thm:cmp-trunc,thm:lut,thm:interval,thm:freivalds,lem:bijection,lem:noninv}. The proofs follow the structure used for masked FSS gates in prior work~\cite{jawalkar2024orca,gupta2023sigma}: an exact algebraic identity on the masked value, followed by an argument that the opened quantities are one-time pads. The authenticated layer is new: every identity must hold over $\Fp$ with authenticated sharings, and every secret product goes through the checked multiplication of \cref{fig:arith}.

\textbf{Conventions.}
Throughout, $p=2^n-1$ is a Mersenne prime; the deployed field uses $n=61$ (\cref{sec:proto-repr}). We identify every $v\in\Fp$ with its canonical representative in $\{0,\dots,p-1\}$ and write $\mathsf{bits}(v)=(v_0,\dots,v_{n-1})$ for its $n$-bit binary expansion, \emph{least-significant bit first}. This is the opposite of the most-significant-first path convention of the DPF algorithms (\cref{alg:gen,alg:eval}); the two conventions never meet inside one algorithm. Because $p<2^n$, the all-ones string is never a canonical encoding; the parser rejects such words instead of reducing them. An \emph{authenticated bit} $\auth{\beta}$ is an authenticated field element whose value is promised to lie in $\{0,1\}$. For secret bits $\beta,\beta'$ and a public bit $c$ we use
\begin{gather*}
\beta\oplus\beta'=\beta+\beta'-2\beta\beta',\qquad \neg\beta=1-\beta,\\
c\oplus\beta=c+(1-2c)\beta .
\end{gather*}
The last identity is affine in the secret bit and is therefore a local operation under the linear rules of \cref{fig:arith}; only products of two secret values consume a triple. We write $\Pi_{\mathsf{Mult}}$ for the authenticated multiplication of \cref{fig:arith} and \CheckOpen{} for the batched opening of \cref{alg:checkopen}.

\textbf{Scope of the statements.}
Every correctness statement below is conditioned, exactly as \cref{thm:arith}, on the event that no incorrect value is accepted by \CheckOpen. The unconditional guarantee, which charges that event at most $Q(2/p)^t$, is part of \cref{thm:main} (hybrid $\Hyb{3}$ of \cref{app:proof-main}). Privacy of each opened value is stated here as a one-time-pad property; the simulator of \cref{app:proof-main} uses exactly these properties. The executed CPU reference model (\cref{sec:eval-attack}) tests toy instances of \cref{thm:bitdec,thm:cmp-trunc,thm:lut,thm:interval} and the truncation-fiber sampler; such finite tests catch transcription errors but are neither a machine-checked proof nor evidence for the statistical bounds, which follow from the analysis alone.

\subsection{Canonical Bit Decomposition}

\Cref{alg:bitdec} is our conservative bit decomposition. Its joint sampler draws $r\in\Fp$ by rejection sampling (sample 61~bits, reject the single value $p$) and \emph{then} derives the bits of that same $r$. Sampling $n$ independent uniform bits instead would be wrong: it produces the noncanonical word $2^n-1\equiv 0$ with probability $2^{-n}$ and silently breaks the relation between $\auth{r}$ and its bits, a small instance of the joint-distribution obligation of \cref{lem:bijection}.

\begin{algorithm}[t]
\algstretch
\caption{Canonical bit decomposition over a Mersenne field (\algname{BitDec}) inside the authenticated circuit}
\label{alg:bitdec}
\begin{algorithmic}[1]
\Require authenticated $\auth{x}$ with $x\in\Fp$, $p=2^n-1$; a reserved gate label $\ell$
\Ensure authenticated bits $\auth{x_0},\dots,\auth{x_{n-1}}$ with $x=\sum_{i=0}^{n-1}x_i2^i$
\LineComment{Preprocessing in the joint sampler of $T_0$ and $T_1$}
\State $r\getsr\Fp$ by rejection sampling under label $\ell$; $(r_0,\dots,r_{n-1})\gets\mathsf{bits}(r)$
\State $\auth{r}\gets\AuthShare(r)$; $\auth{r_i}\gets\AuthShare(r_i)$ for $0\le i<n$
\LineComment{Masked opening (one round)}
\State $c\gets\CheckOpen(\auth{x}+\auth{r})$; $(c_0,\dots,c_{n-1})\gets\mathsf{bits}(c)$ \Comment{public}
\LineComment{Loop 1: $d=(c-r)\bmod 2^n$ and final borrow $\one[c<r]$}
\State $\auth{\mathit{br}}\gets 0$
\For{$i=0$ \textbf{to} $n-1$}
  \State $\auth{\tau_i}\gets c_i+(1-2c_i)\auth{r_i}$ \Comment{$c_i\oplus r_i$, local}
  \State $\auth{g_i}\gets\Pi_{\mathsf{Mult}}(\auth{\tau_i},\auth{\mathit{br}})$
  \State $\auth{d_i}\gets\auth{\tau_i}+\auth{\mathit{br}}-2\auth{g_i}$
  \State $\auth{\mathit{br}}\gets(1-c_i)\auth{r_i}+\auth{\mathit{br}}-\auth{g_i}$
\EndFor
\LineComment{Loop 2: subtract the final borrow from $d$}
\State $\auth{v}\gets\auth{\mathit{br}}$
\For{$i=0$ \textbf{to} $n-1$}
  \State $\auth{h_i}\gets\Pi_{\mathsf{Mult}}(\auth{d_i},\auth{v})$
  \State $\auth{x_i}\gets\auth{d_i}+\auth{v}-2\auth{h_i}$
  \State $\auth{v}\gets\auth{v}-\auth{h_i}$
\EndFor
\State \Return $(\auth{x_0},\dots,\auth{x_{n-1}})$
\end{algorithmic}
\end{algorithm}

\begin{proof}[Proof of \cref{thm:bitdec}]
Write $\mathit{br}_i$ for the borrow before iteration $i$ of Loop~1 ($\mathit{br}_0=0$) and $v_i$ for the value of $v$ before iteration $i$ of Loop~2. The opening gives the canonical $c=(x+r)\bmod p$, and since $c,r\in\{0,\dots,p-1\}\subset\{0,\dots,2^n-1\}$ both have well-defined $n$-bit expansions.

\emph{Step 1: Loop~1 is a full subtractor.}
Fix $i$ and suppose $c_i,r_i,\mathit{br}_i\in\{0,1\}$. By the public-XOR identity, $\tau_i=c_i\oplus r_i$; since $g_i=\tau_i\,\mathit{br}_i$, the second identity gives $d_i=\tau_i\oplus\mathit{br}_i=c_i\oplus r_i\oplus\mathit{br}_i$. Enumerating the eight inputs yields
\begin{equation*}
\setlength{\arraycolsep}{3.2pt}
\begin{array}{c|cccccccc}
(c_i,r_i,\mathit{br}_i) & 000 & 001 & 010 & 011 & 100 & 101 & 110 & 111\\\hline
d_i & 0 & 1 & 1 & 0 & 1 & 0 & 0 & 1\\
\mathit{br}_{i+1} & 0 & 1 & 1 & 1 & 0 & 0 & 0 & 1
\end{array}
\end{equation*}
where $\mathit{br}_{i+1}=(1-c_i)r_i+\mathit{br}_i-g_i$. When $\tau_i=0$ we have $c_i=r_i$, so $(1-c_i)r_i=0$ and $g_i=0$, hence $\mathit{br}_{i+1}=\mathit{br}_i$; when $\tau_i=1$ we have $g_i=\mathit{br}_i$, hence $\mathit{br}_{i+1}=(1-c_i)r_i$, which is $1$ for $(c_i,r_i)=(0,1)$ and $0$ for $(1,0)$. This is the borrow-out of a binary full subtractor, and every entry of the table satisfies
\begin{equation*}
d_i-2\,\mathit{br}_{i+1}=c_i-r_i-\mathit{br}_i .
\end{equation*}
In particular all intermediate values are Boolean. Multiplying by $2^i$ and summing over $i<n$ telescopes the borrows:
\begin{equation*}
d-\mathit{br}_n2^n=c-r,\qquad d:=\sum_{i<n}d_i2^i\in[0,2^n).
\end{equation*}
Because $|c-r|<2^n$ and $d\ge 0$, this forces $\mathit{br}_n=\one[c<r]$ and $d=(c-r)\bmod 2^n$.

\emph{Step 2: the Mersenne wrap equals the borrow.}
As integers, $c=x+r-\one[x+r\ge p]\,p$. If $x+r<p$ then $c-r=x\ge 0$; otherwise $c-r=x-p<0$ because $x<p$. Hence the wrap $\one[x+r\ge p]$ equals $\one[c<r]=\mathit{br}_n$, and
\begin{equation*}
x=c-r+\mathit{br}_n\,p=\bigl(d-\mathit{br}_n2^n\bigr)+\mathit{br}_n(2^n-1)=d-\mathit{br}_n .
\end{equation*}
Thus if $\mathit{br}_n=0$ the answer is $d$, and if $\mathit{br}_n=1$ it is $d-1$; this is where $p=2^n-1$ is used.

\emph{Step 3: Loop~2 subtracts the final borrow exactly.}
With $v_0=\mathit{br}_n$ and Boolean $d_i,v_i$, we have $x_i=d_i\oplus v_i$ and $v_{i+1}=v_i-d_iv_i=v_i(1-d_i)$. The four cases $(d_i,v_i)\in\{00,10,01,11\}$ give $(x_i,v_{i+1})\in\{00,10,11,00\}$, and in each case $x_i-2v_{i+1}=d_i-v_i$. Summing as before,
\begin{equation*}
\sum_{i<n}x_i2^i-v_n2^n=d-\mathit{br}_n=x .
\end{equation*}
Since $0\le x<p$ and $0\le\sum_i x_i2^i<2^n$, the case $v_n=1$ would require $\sum_i x_i2^i=x+2^n\ge 2^n$, which is impossible. Hence $v_n=0$ and $\sum_i x_i2^i=x$. By uniqueness of binary expansions the $x_i$ are the canonical bits of $x$ (never the all-ones word, because $x<p$). When $\mathit{br}_n=0$ the loop leaves $v\equiv 0$ and returns $d$ unchanged, as required.

\emph{Step 4: authentication.}
Each line of \cref{alg:bitdec} is either an affine map with public coefficients (the $c_i$ are public after the opening) applied to authenticated values, which the local linear rules of \cref{fig:arith} preserve, or a call to $\Pi_{\mathsf{Mult}}$. By \cref{thm:arith}, conditioned on no accepted wrong opening, every output is an authenticated sharing of the value computed in Steps~1--3. \end{proof}

\textbf{Cost and privacy.}
\Cref{alg:bitdec} uses $2n$ authenticated multiplications (a conservative count, since the first product of Loop~1 multiplies by a public zero), one masked opening, and $2n+1$ dependent rounds, i.e., $123$ rounds for $n=61$. Independent instances share rounds by being placed in the same \CheckOpen{} batch; batching amortizes metadata and challenge derivation but cannot shorten the two carry chains. For privacy, $c=x+r$ with $r$ uniform on $\Fp$ and independent of everything the corrupt host holds (its share tuples of $\auth{r}$ and $\auth{r_i}$ are uniform and independent of $r$ by construction of \AuthShare), so $c$ is uniform and independent of $x$. The multiplications open only Beaver differences under fresh triples.

\subsection{Comparison and Exact Signed Truncation}

\Cref{alg:trunc} contains the public-threshold comparator and the exact signed truncation. It follows the classical bit-decomposition route to exact truncation~\cite{damgard2006unconditionally,catrina2010improved}, adapted to the Mersenne field and to authenticated sharings; mixed-circuit alternatives such as edaBits~\cite{escudero2020edabits} would need their own proofs in our model.

\begin{algorithm}[t]
\algstretch
\caption{Public-threshold comparison and exact signed truncation (\algname{Cmp}, \algname{Trunc})}
\label{alg:trunc}
\begin{algorithmic}[1]
\Require canonical authenticated bits $(\auth{x_i})_{i<n}$ from \cref{alg:bitdec} and public $\theta\in[0,2^n]$; for \algname{Trunc}, $\auth{z}$ whose centered value $\bar z$ satisfies $|\bar z|\le(p-1)/2$, and public $f$ with $0<f<n$
\Ensure $\auth{\one[x<\theta]}$; $\auth{\operatorname{sgn}(\bar z)\lfloor|\bar z|/2^f\rfloor}$
\Procedure{Cmp}{$(\auth{x_i})_{i<n},\theta$}
  \If{$\theta\in\{0,2^n\}$}
    \State \Return the public constant $\one[\theta=2^n]$
  \EndIf
  \State $\auth{\mathit{eq}}\gets 1$; $\auth{\mathit{lt}}\gets 0$
  \For{$i=n-1$ \textbf{down to} $0$} \Comment{$\theta_i$: bit $i$ of $\theta$}
    \State $\auth{m_i}\gets\Pi_{\mathsf{Mult}}(\auth{\mathit{eq}},\auth{x_i})$
    \If{$\theta_i=1$}
      \State $\auth{\mathit{lt}}\gets\auth{\mathit{lt}}+\auth{\mathit{eq}}-\auth{m_i}$; $\auth{\mathit{eq}}\gets\auth{m_i}$
    \Else
      \State $\auth{\mathit{eq}}\gets\auth{\mathit{eq}}-\auth{m_i}$
    \EndIf
  \EndFor
  \State \Return $\auth{\mathit{lt}}$
\EndProcedure
\Procedure{Trunc}{$\auth{z},f$}
  \State $(\auth{z_i})_{i<n}\gets\algname{BitDec}(\auth{z})$
  \State $\auth{s}\gets 1-\algname{Cmp}\bigl((\auth{z_i})_{i<n},(p+1)/2\bigr)$ \Comment{$s=\one[\bar z<0]$}
  \State $\auth{u}\gets\auth{z}-2\,\Pi_{\mathsf{Mult}}(\auth{s},\auth{z})$ \Comment{$u=|\bar z|$}
  \State $(\auth{u_i})_{i<n}\gets\algname{BitDec}(\auth{u})$
  \State $\auth{q}\gets\sum_{i=f}^{n-1}2^{i-f}\auth{u_i}$ \Comment{local}
  \State \Return $\auth{q}-2\,\Pi_{\mathsf{Mult}}(\auth{s},\auth{q})$
\EndProcedure
\end{algorithmic}
\end{algorithm}

\begin{proof}[Proof of \cref{thm:cmp-trunc}]
\emph{Comparator.} For $k\in\{0,\dots,n\}$ let $x^{(k)}$ and $\theta^{(k)}$ be the integers formed by the top $k$ bits $x_{n-1}\dots x_{n-k}$ and $\theta_{n-1}\dots\theta_{n-k}$. We show by induction on $k$ that after $k$ iterations
\begin{equation*}
\mathit{eq}=\one[x^{(k)}=\theta^{(k)}],\qquad \mathit{lt}=\one[x^{(k)}<\theta^{(k)}].
\end{equation*}
For $k=0$ both prefixes are empty, so $\mathit{eq}=1$ and $\mathit{lt}=0$. For the step, let $i=n-k-1$ and $m_i=\mathit{eq}\cdot x_i$. The new prefixes are equal iff the old ones were equal and $x_i=\theta_i$; this is $\mathit{eq}\,x_i=m_i$ when $\theta_i=1$ and $\mathit{eq}(1-x_i)=\mathit{eq}-m_i$ when $\theta_i=0$, which are the two updates of \cref{alg:trunc}. The new prefix of $x$ is smaller iff either the old prefix was already smaller, or the old prefixes were equal and $x_i<\theta_i$, i.e., $x_i=0$ and $\theta_i=1$. The two events are mutually exclusive because $\mathit{eq}=1$ implies $\mathit{lt}=0$, so the sum $\mathit{lt}+\mathit{eq}(1-x_i)=\mathit{lt}+\mathit{eq}-m_i$ is again a bit. After $k=n$ iterations $\mathit{lt}=\one[x<\theta]$ for every $\theta\in[0,2^n)$, which \algname{Cmp} returns; the flag $\mathit{eq}=\one[x=\theta]$ stays internal and is not an output. The endpoint thresholds give public constants for $\mathit{lt}$: no $x$ satisfies $x<0$ and every canonical $x$ satisfies $x<2^n$. All products are of bits, so \cref{thm:arith} preserves authentication; the comparator uses at most $n$ authenticated multiplications.

\emph{Sign.} The centered encoding maps $\bar z\in[0,(p-1)/2]$ to itself and $\bar z\in[-(p-1)/2,-1]$ to $p+\bar z\in[(p+1)/2,p-1]$. Hence the canonical $z$ satisfies $z\ge(p+1)/2$ iff $\bar z<0$, and $s=1-\one[z<(p+1)/2]=\one[\bar z<0]$. The threshold $(p+1)/2=2^{n-1}$ lies in $[0,2^n)$, so the comparator applies to the canonical bits returned by \cref{thm:bitdec}.

\emph{Magnitude.} In $\Fp$, $u=(1-2s)z$ equals $z=\bar z$ when $s=0$ and $-z=p-z=-\bar z$ when $s=1$. In both cases the canonical representative of $u$ is $|\bar z|\in[0,(p-1)/2]$, so no modular wrap occurs. By \cref{thm:bitdec}, $(u_i)$ are the canonical bits of $|\bar z|$.

\emph{Shift and sign restoration.} The local sum $q=\sum_{i=f}^{n-1}u_i2^{i-f}$ equals $\lfloor|\bar z|/2^f\rfloor$ as an integer, and $q<2^{n-f}\le p$, so it is also the canonical field element. Finally $(1-2s)q$ equals $q$ for $\bar z\ge 0$ and $p-q$, the encoding of $-q$, for $\bar z<0$. Therefore the output encodes $\operatorname{sgn}(\bar z)\lfloor|\bar z|/2^f\rfloor$, which is truncation toward zero; for $\bar z=0$ it is $0$. Every step is local affine or $\Pi_{\mathsf{Mult}}$, so \cref{thm:arith,thm:bitdec} give authenticated outputs conditioned on no accepted wrong opening.
\end{proof}

\textbf{Remarks.}
(i) For the particular threshold $2^{n-1}$ all lower threshold bits are zero, so $\mathit{lt}$ can change only in the first iteration and the comparator output equals $1-x_{n-1}$; the equality products after that iteration do not affect the result. Our cost accounting nevertheless charges the generic comparator, giving $2n+n+1+2n+1=5n+2$ authenticated multiplications per exact truncation ($307$ for $n=61$), all on sequential carry chains. With the two masked openings of the decompositions they occupy $5n+4=309$ sequential rounds ({30.9~ms} on LAN); this is why the compiler treats \cref{alg:trunc} as a fallback (\cref{sec:proto-gates}).
(ii) A fixed-point multiplication first computes the integer product and then applies \algname{Trunc}. The declared magnitude must remain below $p/2$ before truncation; otherwise the modular product has already lost the intended signed integer. The manifest's numeric contract states these bounds, and the proof above assumes them rather than enforcing them.

\subsection{Authenticated Small Lookup Tables}

\begin{proof}[Proof of \cref{thm:lut}]
Recall \cref{alg:lut}: the protected sampler draws a fresh uniform $r\in\bits{w}$, authenticates its bits, and generates DPF keys $(k_0,k_1)$ for the point function $f_{r,1}$ with the \emph{same} $r$ as target and payload $1\in\Fp$. The circuit computes $\auth{u_i}=\auth{x_i}+\auth{r_i}-2\,\Pi_{\mathsf{Mult}}(\auth{x_i},\auth{r_i})$ and opens the vector $u$ with one batched \CheckOpen; both modules check that every $u_i\in\{0,1\}$ and that they hold the same vector. Module $T_b$ then computes $\tilde y_b=\sum_{a\in\bits{w}}T[a\oplus u]\cdot\Eval(k_b,a)$ over the full public domain and hands $\tilde y_b$ to the epilogue.

\emph{Correctness.} Conditioned on no accepted wrong opening, \cref{thm:arith} and the XOR identity give $u=x\oplus r$. By \cref{thm:dpf-correct} applied to the depth-$w$ tree, $\Eval(k_0,a)+\Eval(k_1,a)=\one[a=r]$ for every $a\in\bits{w}$. Summing the two raw outputs and exchanging the finite sums,
\begin{align*}
\tilde y_0+\tilde y_1&=\sum_{a\in\bits{w}}T[a\oplus u]\,\one[a=r]\\
&=T[r\oplus u]=T[r\oplus x\oplus r]=T[x].
\end{align*}
By \cref{thm:epilogue}, the epilogue turns $(\tilde y_0,\tilde y_1)$ into a fresh authenticated sharing of $T[x]$. The table index $a\oplus u$ depends only on the public $u$ and the public loop counter, so the declared full-domain scan has a fixed, input-independent access pattern.

\emph{Joint preparation.} The identity uses that the authenticated bits and the DPF target are the same $r$. If the target were resampled as an independent uniform $r'$, each object would still be marginally uniform, but the output would become $T[r'\oplus x\oplus r]$, which is not $T[x]$. The sampler therefore derives both from one logical draw, as required by \cref{lem:bijection}.

\emph{Mask privacy.} For every fixed $x$, the map $r\mapsto x\oplus r$ is a bijection on $\bits{w}$, so $u$ is uniform when $r$ is. The only other material related to $r$ that leaves the protected boundary is the corrupt party's share tuples of $\auth{r_i}$, which \AuthShare{} samples uniformly and independently of $r$; the DPF keys and raw evaluations never leave the modules in \ProfT. Hence $u$ is uniform and independent of $x$ jointly with the corrupt party's view, and the Beaver openings inside the $w$ XOR products are uniform because the triples are fresh. The exported output is an epilogue tuple, which \cref{thm:epilogue} shows to be uniform for either fixed party.
\end{proof}

\subsection{Shifted Intervals and Piecewise Polynomials}

\begin{proof}[Proof of \cref{thm:interval}]
Let $[a,b)\subseteq[0,p)$ be proper and nonempty, so $\delta:=b-a\in[1,p-1]$. \Cref{alg:interval} opens $c=(x+r)\bmod p$ for a fresh uniform $r\in\Fp$ and, inside the protected sampler, sets $L=(a+r)\bmod p$, $H=(b+r)\bmod p$ and $\omega=\one[L>H]$, generates two conservative DCFs (\cref{alg:dcf}) with payload $1$ for the thresholds $L$ and $H$, both derived from the same $r$, and a sharing $\omega_0+\omega_1=\omega$. Module $T_0$ hands $\tilde s_0=\omega_0+\DCF_H(c)_0-\DCF_L(c)_0$ to the epilogue, and $T_1$ likewise $\tilde s_1$ (here $a,b$ are endpoints, not party indices).

\emph{Step 1: translation.} Since $c-L\equiv x-a$ and $H-L\equiv\delta\pmod p$, we claim
\begin{equation*}
a\le x<b\iff (c-L)\bmod p<(H-L)\bmod p=\delta .
\end{equation*}
If $a\le x<b$ then $(x-a)\bmod p=x-a<\delta$. If $x<a$ then $(x-a)\bmod p=p+x-a\ge p-a\ge b-a=\delta$ because $b\le p$. If $x\ge b$ then $x-a\ge\delta$. Thus $x\in[a,b)$ iff $c$ lies in the circular half-open interval that starts at $L$ and has length $\delta$. Moreover $L\ne H$ because $\delta\notin\{0,p\}$.

\emph{Step 2: two cases.} If $L<H$, the circular interval is the ordinary interval $[L,H)$: for $c<L$ we get $(c-L)\bmod p=p+c-L\ge p-L>H-L$, and for $c\ge H$ we get $c-L\ge H-L$. Its indicator is $\one[c<H]-\one[c<L]$, which is a bit because $L<H$. If $L>H$, the circular interval is $[L,p)\cup[0,H)$, a disjoint union, and its indicator is $\one[c\ge L]+\one[c<H]=1+\one[c<H]-\one[c<L]$. With $\omega=\one[L>H]$ both cases read
\begin{equation*}
\one[a\le x<b]=\omega+\one[c<H]-\one[c<L].
\end{equation*}

\emph{Step 3: shares.} All of $c,L,H$ lie in $[0,p)\subset[0,2^n)$, so they are valid inputs and thresholds for the $n$-bit DCF with the ``$<$'' orientation. By \cref{thm:dcf}, $\DCF_H(c)_0+\DCF_H(c)_1=\one[c<H]$ and likewise for $L$. Summing the two raw shares gives $\tilde s_0+\tilde s_1=\omega+\one[c<H]-\one[c<L]=\one[a\le x<b]$, and \cref{thm:epilogue} yields an authenticated sharing of the indicator. Empty and full intervals are the public constants $0$ and $1$ and need no keys. The derivation uses one opening $c$ for both thresholds; deriving $L$ and $H$ from independent masks would require two different openings and would break Step~1.

\emph{Step 4: partitions and polynomials.} Let $0=\tau_0<\tau_1<\dots<\tau_m=p$ with $m\ge 2$ define the disjoint, exhaustive partition of \cref{alg:interval}, so every piece $[\tau_i,\tau_{i+1})$ is proper and nonempty, and let all endpoints, wrap bits and auxiliary shares be derived jointly from one mask $r$ and one opening $c$, with independent function-generation coins per DCF. By Steps~1--3 each indicator $s_i=\one[\tau_i\le x<\tau_{i+1}]$ is correct, and exactly one of them, $s_{i^\star}$, equals $1$. The selected coefficients $\auth{\bar a_j}=\sum_i a_{i,j}\auth{s_i}$ are local affine maps with public $a_{i,j}$, so $\bar a_j=a_{i^\star,j}$. Horner's rule $y_d=\bar a_d$, $y_j=y_{j+1}\cdot x+\bar a_j$ for $j=d-1,\dots,0$ uses $d$ calls to $\Pi_{\mathsf{Mult}}$ and yields $y_0=\sum_{j=0}^{d}a_{i^\star,j}x^j=P_{i^\star}(x)$ in $\Fp$ by \cref{thm:arith}. Where the manifest inserts fixed-point rescaling between Horner steps, \cref{thm:cmp-trunc} makes each rescaling exact. The theorem therefore states equality with the declared finite arithmetic, not with an unapproximated real-valued function.

\emph{Privacy.} As in \cref{thm:bitdec}, $c$ is a one-time pad of $x$ over $\Fp$. The DCF keys, the thresholds $L,H$ and the sharing of $\omega$ remain in the modules, and the only exported values are epilogue tuples and Beaver openings, both uniform.
\end{proof}

\subsection{Distribution-Preserving Refactoring}

The semi-honest DFA design samples a DPF target directly and recovers the mask through an inverse map~\cite{dfa2026trusted}. \Cref{lem:bijection,lem:noninv} state when this preserves the \emph{joint} distribution with every auxiliary object derived from the mask; marginal uniformity of each object alone is not sufficient. The main proof applies them in hybrid $\Hyb{2}$ of \cref{app:proof-main}, only where these hypotheses hold.

\begin{proof}[Proof of \cref{lem:bijection}]
Fix the already determined context $Z=z$. Let $R$ be uniform on the finite set $\mathcal{R}$, let $g_z:\mathcal{R}\to\mathcal{A}$ be a bijection, let $W$ be coins independent of $R$, and let $J$ be the deterministic procedure that produces all auxiliary material. Experiment~1 outputs $(R,g_z(R),J(z,R,W))$; Experiment~2 samples $A$ uniformly on $\mathcal{A}$, sets $R'=g_z^{-1}(A)$ and outputs $(R',A,J(z,R',W))$. Write $\pi(\hat\jmath\mid\hat r)=\Pr\nolimits_W[J(z,\hat r,W)=\hat\jmath]$ for the auxiliary sampler. For every $(\hat r,\hat a,\hat\jmath)$,
\begin{align*}
\Pr\nolimits_1[(\hat r,\hat a,\hat\jmath)]&=|\mathcal{R}|^{-1}\,\one[\hat a=g_z(\hat r)]\,\pi(\hat\jmath\mid\hat r),\\
\Pr\nolimits_2[(\hat r,\hat a,\hat\jmath)]&=|\mathcal{A}|^{-1}\,\one[\hat r=g_z^{-1}(\hat a)]\,\pi(\hat\jmath\mid\hat r).
\end{align*}
The indicators coincide because $g_z$ is a bijection, which also gives $|\mathcal{R}|=|\mathcal{A}|$; the conditional law of the auxiliary material given the mask is the same sampler in both experiments. Hence the two joint distributions are equal for every $z$, and averaging over the law of $Z$ preserves equality. The hypothesis that $R$ is uniform \emph{given} $Z=z$ is essential: a mask correlated with the context would make the two sides differ.
\end{proof}

\textbf{Joint equality.}
Suppose a gate needs two DPF targets $\alpha_1=r$ and $\alpha_2=r+1$ derived from one mask. Resampling $\alpha_1$ and $\alpha_2$ as two independent uniforms keeps both marginals uniform but destroys the relation $\alpha_2=\alpha_1+1$, so the gate computes a different function. \Cref{lem:bijection} permits instead one uniform draw $A$, the recovery $R'=g_z^{-1}(A)$, and the derivation of the second target from $R'$ inside $J$. The same reasoning covers the LUT and interval gates above, where the authenticated bits, the DPF or DCF targets, the wrap bit and all auxiliary shares are functions of one logical draw.

\begin{proof}[Proof of \cref{lem:noninv}]
Let $g$ be arbitrary and let $R$ have any distribution on a finite set. Sample $A$ from the law of $g(R)$ and then $R'$ from the conditional law of $R$ given $g(R)=A$, which is well defined on the support of $g(R)$. For every $(\hat r,\hat a)$,
\begin{align*}
\Pr[R'=\hat r,A=\hat a]&=\Pr[g(R)=\hat a]\,\Pr[R=\hat r\mid g(R)=\hat a]\\
&=\Pr[R=\hat r,\,g(R)=\hat a].
\end{align*}
Thus $(R',A)$ has the law of $(R,g(R))$. Every auxiliary object is subsequently derived from $R'$ with independent coins, so the argument of \cref{lem:bijection} extends the equality to the full joint distribution.
\end{proof}

\textbf{Truncation sampler.}
For an unsigned right shift of a uniform $\ell$-bit mask by $h$ bits, $g(R)=\lfloor R/2^h\rfloor$. Here $A$ is uniform on $\ell-h$ bits and the fiber of each $A$ is $\{2^hA+U:U\in[0,2^h)\}$, so the conditional sampler is $R'=2^hA+U$ with $U$ uniform on $h$ bits. Every mask then has probability $2^{-(\ell-h)}2^{-h}=2^{-\ell}$, as required. The sampler is specific to this case. For example, if $R$ is uniform on $\Fp$ with $p=2^{61}-1$, the top value of $\lfloor R/2^h\rfloor$ has only $2^h-1$ preimages while every other value has $2^h$, so $A$ is not uniform and must be sampled with probabilities proportional to fiber sizes. Signed shifts, nonuniform masks and other moduli likewise need their own conditional sampler; \cref{lem:noninv} states what that sampler must satisfy, not that a generic one exists.

\subsection{Frozen-Operand Freivalds Check}

\Cref{alg:freivalds} is the optional verified matrix subroutine of \cref{sec:proto-freivalds}: Freivalds' randomized product check~\cite{freivalds1977probabilistic} run inside the protected verifier on a frozen candidate and on operands bound to their ledger records. In this subsection $m,k,n$ denote matrix dimensions. The subroutine applies only where exposing the operands to that GPU is already allowed; it never authorizes exposing global secret triple masks.

\begin{algorithm}[t]
\algstretch
\caption{Frozen-operand Freivalds check (\algname{VerifyMatrix}) inside $T_b$}
\label{alg:freivalds}
\begin{algorithmic}[1]
\Require handles of $A\in\Fp^{m\times k}$ and $B\in\Fp^{k\times n}$ bound to authoritative input or preprocessing records; candidate $C\in\Fp^{m\times n}$ from the GPU; repetitions $s$ (default $s=3$)
\Ensure a verified immutable handle to $C$, or \Abort
\If{the manifest does not permit exposing $A,B$ to this GPU}
  \State \Abort \Comment{never for global triple masks}
\EndIf
\State check dimensions and the canonical encoding of every element of $A,B,C$; on failure \Abort
\State check that the digests of the consumed $A,B$ bytes match their ledger records; on failure \Abort
\State freeze $C$ into protected storage or AEAD-sealed immutable pages; $\rho_C\gets\Hash(\text{descriptor}\,\|\,C)$
\For{$\ell=1$ \textbf{to} $s$}
  \State draw secret $\bm{v}_\ell\getsr\Fp^{n}$ under a label bound to $(\text{session},\text{epoch},\text{gate},\rho_C,\ell)$
  \State $\bm{g}\gets B\bm{v}_\ell$; $\bm{g}'\gets A\bm{g}$
  \If{$C\bm{v}_\ell\ne\bm{g}'$} \Comment{$C$ read from frozen bytes only}
    \State poison the epoch; \Abort
  \EndIf
\EndFor
\State for each transformed tag array, repeat the loop above on its own equation
\State mark the handle \lstate{VERIFIED}; consumers read only the frozen bytes of $C$
\end{algorithmic}
\end{algorithm}

\begin{proof}[Proof of \cref{thm:freivalds}]
Let $E=C-AB$, where $C$ is the frozen candidate and $A,B$ are the bound operands, and suppose $E\ne 0$. Because $\rho_C$ is fixed before any $\bm{v}_\ell$ is derived and the challenge label binds $\rho_C$, the matrix $E$ is determined before the challenges, which are secret, uniform, and independent of $E$ (with ideal randomness, i.e., from hybrid $\Hyb{2}$ of \cref{app:proof-main} on). A repetition accepts iff $C\bm{v}_\ell=A(B\bm{v}_\ell)$, i.e., $E\bm{v}_\ell=\bm{0}$. Pick a nonzero row $\bm{e}$ of $E$ and an index $j$ with $e_j\ne 0$. The event $E\bm{v}_\ell=\bm{0}$ implies $\inner{\bm{e}}{\bm{v}_\ell}=0$. Conditioned on all coordinates of $\bm{v}_\ell$ except $v_{\ell,j}$, this linear equation has exactly one solution $v_{\ell,j}=-e_j^{-1}\sum_{i\ne j}e_iv_{\ell,i}$, which the uniform $v_{\ell,j}$ hits with probability exactly $1/p$. Hence each repetition accepts with probability at most $1/p$, and since the $s$ challenge vectors are independent while $E$ is fixed, all $s$ repetitions accept with probability at most $p^{-s}$.

The two side conditions are necessary. Without freezing, a host that observes $\bm{v}_\ell$ could choose a candidate with $E\ne 0$ and $E\bm{v}_\ell=\bm{0}$ after the fact. Without operand binding, the module could verify one product and later consume a different one; the ledger check on the operand digests and the rule that consumers read only the frozen $C$ exclude both. A check that covers only the value matrix does not authenticate a different tag matrix, hence the per-lane repetition in \cref{alg:freivalds}. A union bound over at most $Q_F$ invocations bounds the probability that any wrong product is accepted by $Q_F/p^s$; for $s=3$ this is $Q_F\cdot p^{-3}<Q_F\cdot 2^{-182}$, which \cref{thm:main} adds to its error term (in $\Hyb{3}$) when the subroutine is enabled.
\end{proof}

\textbf{Cost.}
One repetition computes $\bm{g}=B\bm{v}$, $\bm{g}'=A\bm{g}$ and $C\bm{v}$, i.e., $kn+mk+mn$ field multiply-accumulates, and reads all three bound matrices once, compared with $mkn$ operations for recomputing the product. For small matrices the reads dominate and the check may cost more than direct protected computation, which is why the compiler treats it as an option to be selected by modeled cost rather than as a default; it is disabled in all evaluated configurations.

\section{Proofs of PIR Security and of the Main Theorem}
\label{app:proof-main}

This appendix proves \cref{thm:pir} for both corruption regimes and then \cref{thm:main} by an explicit simulator and a sequence of hybrids $\Hyb{0},\dots,\Hyb{4}$. We use the real/ideal paradigm of \cref{def:security}~\cite{canetti2000security,goldreich2004foundations,hazay2010efficient,lindell2017simulate} with $\FC$ and $\FPIR$ (\cref{app:func-circuit,app:func-pir}), the leakage $\Leak_C$ and $\Leak_{\mathrm{PIR}}$ (\cref{def:leakage}) and the ideal module resource $\mathcal{F}_{\mathsf{T}}$ (\cref{app:func-module}).

\textbf{Notation.}
The adversary $\Adv$ with auxiliary input $z$ statically corrupts the host of party $P_{b^\star}$ for some $b^\star\in\{0,1\}$ (no corruption is the special case with only public traffic), and $\bar b=1-b^\star$ is the honest party. As in \cref{def:security}, $\mu$, $\vec{x}$ and $z$ are fixed up front, and a PPT distinguisher $\mathcal{D}$ receives the joint output (the view of $\Adv$ and the honest outputs). For a wire $\nu$ of the manifest's circuit, $x_\nu$ is its value in $C$ on the honest and the extracted inputs, and $\phi_\nu=(\hat x_{b^\star,\nu},\hat\gamma_{b^\star,\nu,1},\dots,\hat\gamma_{b^\star,\nu,t})\in\Fp^{1+t}$ is its \emph{faithful tuple} (\cref{app:proof-auth}): the tuple that host $b^\star$ would hold had it applied every prescribed local operation to what it actually received. The error terms $\eps{setup},\eps{prf},\eps{channel},\eps{bind},\eps{trace}$ are those of \cref{sec:sec-assump} and refer to the entire admitted execution.

We use one elementary fact repeatedly, the one-time-pad argument behind every masked opening.

\begin{lemma}[One-time-pad openings]
\label{clm:otp}
Let $\mathbb{G}$ be a finite abelian group ($\Fp$ or $\bits{w}$ under $\oplus$). Let $(V,S)$ be arbitrary jointly distributed random variables with $S\in\mathbb{G}$, let $M$ be uniform on $\mathbb{G}$ and independent of $(V,S)$, and let $Y$ be independent of $(V,S,M)$. Then $(V,Y,S+M)$ is distributed as $(V,Y,U)$ with $U$ uniform on $\mathbb{G}$ and independent of $(V,Y,S)$, and conditioned on $(V,Y,S,U=u)$ the mask equals $M=u-S$.
\end{lemma}
\begin{proof}
For all $(v,y,\sigma,u)$, $\Pr[V=v,Y=y,S=\sigma,S+M=u]=\Pr[V=v,S=\sigma]\Pr[Y=y]\Pr[M=u-\sigma]=\Pr[V=v,S=\sigma]\Pr[Y=y]/|\mathbb{G}|$, which is the claimed product distribution; the second statement is immediate.
\end{proof}

In our uses, $Y$ is the corrupt party's share tuple of the authenticated mask $\auth{M}$. \AuthShare{} samples party~$0$'s value and tag shares uniformly and independently and sets party~$1$'s shares to the differences, so either party's tuple is uniform on $\Fp^{1+t}$ and independent of $M$ and of the global keys; this is the hypothesis ``$Y$ independent of $M$''.

\subsection{Proof of the PIR Theorem}
\label{app:proof-pir}

\textbf{Setting.}
The authorized snapshot has the public descriptor $\mathsf{desc}_v=(\mathsf{id},v,N,\mathsf{enc},\mathsf{cm}_v)$: identifier, version $v$, true length $N\ge 2$, encoding of each record into $k$ field coordinates, and commitment $\mathsf{cm}_v$ to all of these and to the ordered record bytes. The padded length is $2^n$ with $n=\lceil\log_2N\rceil$, padded records are zero, and $N=1$ is the constant-record special case with the same checks. Following \cref{alg:pir}, for a query $(\alpha,\mathit{qid})$ module $T_b$ derives the joint DPF for $(\alpha,1)$, scans the canonical snapshot bytes in the public order, hashes the bytes it actually consumes, computes $\tilde y_b=\sum_{i<2^n}\Eval(k_b,i)\cdot D[i]\in\Fp^{k}$ inside the protected path, quarantines $\tilde y_b$ until the digest equals $\mathsf{cm}_v$, reshares each coordinate with a fresh pad as in the epilogue, $y_b=\tilde y_b+(-1)^b\zeta$ with $\zeta\getsr\Fp^{k}$, and releases only $y_b$ with a certificate $\mathsf{cert}_b$ that binds the session, epoch, $\mathit{qid}$, $\mathsf{desc}_v$, role $b$ and the digest of $y_b$. No MAC tag is sent to the client. Both modules learn $\alpha$, so the index privacy proved below is privacy against the corrupt host and relies on module trust, unlike untrusted two-server DPF-PIR (GPU-DPF~\cite{lam2023gpu}, the untrusted mode of DFA~\cite{dfa2026trusted}), which hides $\alpha$ from each server without trusted hardware.

\begin{proof}[Proof of \cref{thm:pir}]
We use the two regimes of \cref{app:func-pir}, which are not combined: in (R1), a malicious server host and an honest client, we prove integrity and index privacy; in (R2), a malicious client only and two honest hosts, we prove database privacy. Throughout, let $\mathsf{Good}$ be the event that setup, protected channels and all bindings (query bindings, snapshot digest, certificates) behave ideally; $\Pr[\neg\mathsf{Good}]\le\eps{setup}+\eps{channel}+\eps{bind}$.

\emph{(R1) integrity.}
Let host $b^\star$ be corrupt and the client honest, and condition on $\mathsf{Good}$.
\begin{enumerate}[leftmargin=*,label=\arabic*)]
\item \emph{One legal index.} The client sends $(\alpha,\mathit{qid},\mathsf{desc}_v)$ to both modules over authenticated confidential channels. Host $b^\star$ relays these messages but can only drop or delay them. The modules compare the two bindings over the protected channel and reject unless they agree and $\alpha<N$. Hence, if processing continues, both modules hold the same legal $\alpha$ and the same descriptor.
\item \emph{Consistent keys.} Both modules derive the DPF for $(\alpha,1)$ under the same joint label, so they hold the two halves of one key pair. By \cref{thm:dpf-correct}, $\Eval(k_0,i)+\Eval(k_1,i)=\one[i=\alpha]$ for all $i<2^n$.
\item \emph{Committed bytes.} Each module hashes exactly the bytes it consumes, in the public order, and releases nothing before the digest matches $\mathsf{cm}_v$. Under $\mathsf{Good}$ the consumed sequence therefore equals the committed padded snapshot $D^\ast$ in both modules, whatever bytes host $b^\star$ tries to substitute.
\item \emph{Correct reduction.} Module arithmetic is ideal, so coordinatewise
\begin{align*}
\tilde y_0+\tilde y_1&=\sum_{i<2^n}\bigl(\Eval(k_0,i)+\Eval(k_1,i)\bigr)D^\ast[i]\\
&=D^\ast[\alpha],
\end{align*}
and the fresh resharing preserves the sum: $y_0+y_1=\tilde y_0+\tilde y_1$.
\item \emph{Release.} The client accepts only a pair of valid certificates with roles $0$ and $1$ for its own $\mathit{qid}$ and $\mathsf{desc}_v$. Under $\mathsf{Good}$ host $b^\star$ can neither forge a certificate nor splice a share from another query, version or role into the pair. Dropping or delaying a share leaves the client without a valid pair, which is an abort.
\end{enumerate}
Consequently the honest client outputs $D^\ast[\alpha]$ or $\bot$, never another value, except with probability $\eps{setup}+\eps{channel}+\eps{bind}$.

\emph{(R1) index privacy.}
The simulator $\Sim^{\mathrm{PIR}}_{b^\star}$ is given $\Leak_{\mathrm{PIR}}$ ($\mathsf{desc}_v$, the public scan length and order, and the message lengths and schedules) and the ability to tell $\FPIR$ to deliver or abort. It must simulate the view of host $b^\star$, i.e., the relayed client-to-module and module-to-module ciphertexts, the snapshot-byte requests of $T_{b^\star}$, the released frame $(y_{b^\star},\mathsf{cert}_{b^\star})$ if it passes through host $b^\star$, module status signals, and traces. Its steps are:
\begin{enumerate}[leftmargin=*,label=\arabic*)]
\item Emulate $\mathcal{F}_{\mathsf{T}}$ for setup and produce the public setup transcript.
\item For every relayed ciphertext, output a string of the permitted length on the public schedule.
\item Request snapshot bytes from $\Adv$ in the public order and length, hash them as $T_{b^\star}$ would, and record whether the digest equals $\mathsf{cm}_v$.
\item If all messages were delivered and the digest matched, sample $y_{b^\star}\getsr\Fp^{k}$ and produce $\mathsf{cert}_{b^\star}$ through the emulated module interface; otherwise output the module's generic abort status.
\item Instruct $\FPIR$ to deliver to the client iff the simulated run reaches the point where the client holds a valid certificate pair; otherwise instruct it to abort.
\item Append the output of the stipulated trace simulator $\Sim_{\mathsf{trace}}$ on the command sequence issued by host $b^\star$, the public descriptor and lengths, and the accept or abort outcomes.
\end{enumerate}
Steps 1 and 2 are exact in the hybrid with ideal setup and channels ($\eps{setup}+\eps{channel}$). Step 3 is identical because the digest check is a deterministic function of the bytes that $\Adv$ supplies and of the public $\mathsf{cm}_v$. In Step 4, the real $y_{b^\star}=\tilde y_{b^\star}+(-1)^{b^\star}\zeta$ with a fresh pad $\zeta$ that is used nowhere else (after replacing derived tapes by random ones, $\eps{prf}$); by \cref{thm:epilogue} it is uniform and independent of $\alpha$, of $D^\ast$ and of everything else in the view. No DPF key or unmasked reduction is exported, so no DPF one-key-privacy assumption is needed. Step 5 applies the delivery condition established for (R1) integrity. Step 6 costs $\eps{trace}$. The index $\alpha$ is never used by the simulator.

\emph{(R2) database privacy.}
The simulator $\Sim^{\mathrm{PIR}}_{\mathrm{cli}}$ has oracle access to $\FPIR$ and must simulate the client's view: the module responses (value shares and certificates) or a rejection. Its steps are:
\begin{enumerate}[leftmargin=*,label=\arabic*)]
\item Receive the client's query messages to $T_0$ and $T_1$ through the emulated module interface and parse their canonical encodings.
\item If the two bindings differ, if an index is out of range, if an encoding is malformed, or if $\mathit{qid}$ is replayed, output the rejection that the modules would output and submit a malformed query to $\FPIR$, which returns $\mathsf{reject}$.
\item Otherwise extract the agreed index $\alpha^\ast$, submit it to $\FPIR$, and receive $D^\ast[\alpha^\ast]\in\Fp^{k}$.
\item For each coordinate $\iota<k$ sample $y_{0,\iota}\getsr\Fp$ and set $y_{1,\iota}=D^\ast[\alpha^\ast]_\iota-y_{0,\iota}$.
\item Produce $\mathsf{cert}_0,\mathsf{cert}_1$ through the emulated module interface, bound to $\mathit{qid}$ and $\mathsf{desc}_v$; append the trace simulator's output.
\end{enumerate}
In the real execution the hosts are honest, so by Steps 2--4 of (R1) integrity the two reductions sum to $D^\ast[\alpha^\ast]$, and the fresh pad makes $y_0$ uniform with $y_1$ determined by the sum; this is exactly Step~4. Fresh pads per query make the responses to different queries independent. Since no MAC tag is delivered, the client's view contains no relation that depends on the global keys, and $\mathsf{desc}_v$ is part of $\Leak_{\mathrm{PIR}}$. Hence the client learns nothing beyond the requested records and the descriptor.

Collecting the terms, each regime is simulated with error at most $\eps{setup}+\eps{prf}+\eps{channel}+\eps{bind}+\eps{trace}$.
\end{proof}
This is a hardware-trusted PIR statement: it does not turn untrusted two-server DPF-PIR into symmetric PIR, and a verifiable-FSS alternative with externally generated keys~\cite{decastro2022vfss} must separately prove client key validity and server-answer integrity (\cref{app:func-pir}).

\subsection{The Simulator for the Circuit Profile}

We now construct the simulator $\Sim$ for $\FC$. $\Sim$ runs $\Adv$ internally, plays the ideal module resource $\mathcal{F}_{\mathsf{T}}$ for \emph{both} modules (it controls only the ideal representation of the trusted resources, never the honest party's prescribed inputs), and interacts with $\FC$ on behalf of the corrupt party. From $\FC$ it receives only what the ideal world permits: $\Leak_C$, which contains the validation outcomes $\mathsf{status}$, and the outputs of corrupt recipients. It has to simulate the view of host $b^\star$, which consists of (i) every response of $T_{b^\star}$ on an approved I/O path: party-local preprocessing tuples, epilogue outputs, opened values, certificates and status codes; (ii) the peer-channel and client-channel ciphertexts that the host relays; (iii) public messages; and (iv) traces.

\textbf{Simulator state.}
$\Sim$ keeps the following state, which it updates in the manifest's public order.
\begin{enumerate}[leftmargin=*,label=\textbf{Sim\arabic*}.]
\item \emph{Keys.} Sample $\Delta_j\getsr\Fp$ and $\Delta_{b^\star,j}\getsr\Fp$ for $j\le t$, set $\Delta_{\bar b,j}=\Delta_j-\Delta_{b^\star,j}$, and give $\Delta_{b^\star,j}$ to host $b^\star$ wherever the protocol does. The global keys stay inside $\Sim$.
\item \emph{Ledger.} Run the ledger state machine of \cref{fig:ledger} for both modules on the public manifest $\mu$: labels, reservations, frozen roots, consumed identifiers, epoch status and release certificates.
\item \emph{Faithful tuples.} For every wire $\nu$ reached so far, store $\phi_\nu$ and, for opened wires, the public value. Honest-side shares are never stored; they are \emph{symbolic}, defined by $x_{\bar b,\nu}:=x_\nu-\hat x_{b^\star,\nu}$ and $\gamma_{\bar b,\nu,j}:=\Delta_jx_\nu-\hat\gamma_{b^\star,\nu,j}$.
\end{enumerate}

\textbf{Per-command simulation.}
Given this state, $\Sim$ answers the commands of host $b^\star$ and produces the rest of its view as follows.
\begin{enumerate}[leftmargin=*,label=\textbf{Sim\arabic*}.,start=4]
\item \emph{Setup.} Emulate \cmd{OPEN_SESSION} with the ideal authenticated setup and output its public transcript; relayed ciphertexts are strings of the permitted lengths.
\item \emph{Preprocessing.} For every authenticated random value that \cmd{RESERVE_PREP} hands to host $b^\star$ (input and output masks, triple components, the field masks $\auth{r}$ of \cref{alg:bitdec,alg:interval} and of direct gates, and the mask bits of \cref{alg:bitdec,alg:lut}), sample $\phi\getsr\Fp^{1+t}$, the exact party-local distribution of \AuthShare. The underlying value stays symbolic. The two exceptions, where the corrupt party is entitled to a complete value, are handled in Sim6 and Sim10. Wrap bits and their sharings, DPF and DCF keys and raw shares never leave the modules (\cref{app:proof-gates}), so nothing is simulated for them.
\item \emph{Inputs and validation.} If the owner of input wire $\nu$ is the corrupt party, sample its complete mask $r\getsr\Fp$ and deliver it, receive the broadcast $\varepsilon$ from $\Adv$, check (as the modules do) that both modules would bind the same $\varepsilon$, and compute $x^\ast=r+\varepsilon$; unless $x^\ast\in\mathsf{In}_\nu$, emulate the modules' abort and make $\FC$ abort, and otherwise submit $x^\ast$ to $\FC$ and set the wire tuple as $\Pi_{\mathsf{Input}}$ does, adding $\varepsilon$ to party~$0$'s value share and $\Delta_{b^\star,j}\varepsilon$ to each tag share. If the owner is honest, sample $\varepsilon\getsr\Fp$ as the public masked difference and let the undisclosed mask be the symbolic $r:=x-\varepsilon$; $x\in\mathsf{In}_\nu$ holds, as $\FC$ admits only such inputs. For a commitment-bound input, emulate the protected loader on the bytes that $\Adv$ supplies: on a match extract the canonical committed input, otherwise abort exactly as the loader does. Since $\FC$ is non-reactive and computes $\mathsf{status}$ only after it has recorded every input, the public schedule of $\mu$ places every input gate before every validity check, and both before the first output opening (\cref{sec:threat-goals,app:func-circuit}). At each validity boundary declared in $\mu$, $\Sim$ reads the outcome from $\mathsf{status}$ and, on a failed predicate, outputs the modules' generic abort status and poisons the epoch; $\FC$ halts on the same outcome.
\item \emph{Local linear operations.} Apply additions, public scalings and public-constant updates to the faithful tuples with the formulas of $\Pi_{\mathsf{Lin}}$ (\cref{fig:arith}).
\item \emph{Multiplications.} When the Beaver differences $d=x-a$ and $e=y-b$ of a multiplication are opened (after the check in Sim9 accepts), release fresh $d,e\getsr\Fp$. The hidden triple is defined afterwards by $a:=x-d$, $b:=y-e$, $c:=ab$; the corrupt tuples of $\auth{a},\auth{b},\auth{c}$ were sampled in Sim5 independently of these values. Compute the corrupt output tuple with the formulas of $\Pi_{\mathsf{Mult}}$ from the corrupt triple tuples and $d,e$.
\item \emph{Checks.} When host $b^\star$ submits tuples $(x'_i,\gamma'_{1,i},\dots,\gamma'_{t,i})$ for the wires $\nu_1,\dots,\nu_m$ of a batch, run the metadata, encoding, order and ledger checks as the module would, and abort on failure. Otherwise compute the value-error vector $\bm{e}$ and the tag-error vectors $\bm{d}_j$ with $e_i=x'_i-\hat x_{b^\star,\nu_i}$ and $d_{j,i}=\gamma'_{j,i}-\hat\gamma_{b^\star,\nu_i,j}$, sample fresh challenges $\bm{\rho}_j\getsr\Fp^{m}$, and accept iff $\inner{\bm{\rho}_j}{\bm{d}_j-\Delta_j\bm{e}}=0$ for every $j$. On rejection, poison the epoch, release nothing, and instruct $\FC$ to abort. On acceptance with $\bm{e}\ne\bm{0}$, abort (this is the bad event of $\Hyb{3}$). On acceptance with $\bm{e}=\bm{0}$, release the opened values that Sim6, Sim8, Sim10 and Sim11 prescribe. Missing, delayed, replayed or reordered messages lead to the same timeouts and rejections as in the real protocol.
\item \emph{Outputs} ($\Pi_{\mathsf{Out}}$). By the schedule of $\mu$ (Sim6), $\FC$ has released every output due to $\Sim$ when an output is opened. For a private output to an honest recipient, release $\upsilon\getsr\Fp$ as the opened masked output; its hidden mask is the symbolic $o:=\upsilon-y$. For a private output to a corrupt authorized recipient, sample $o\getsr\Fp$ and deliver it early, when the protocol delivers the output mask; at the opening, after $\FC$ has released $y$, program $\upsilon:=y+o$. For a public output, obtain $y$ from $\FC$ and release $y$. Instruct $\FC$ to deliver to honest recipients iff the simulated run reaches their release with valid certificates; otherwise instruct it to abort.
\item \emph{Nonlinear gates.} For the masked field opening of \cref{alg:bitdec,alg:interval} and of a direct protected gate, release $c\getsr\Fp$ with symbolic mask $r:=c-x$. For \cref{alg:lut}, release $u\getsr\bits{w}$ with symbolic mask $r:=x\oplus u$. For every \cmd{EVAL_FSS} response and every direct-gate output, sample the corrupt output tuple from $\Fp^{1+t}$ and define its value symbolically by the gate semantics of \cref{thm:bitdec,thm:cmp-trunc,thm:lut,thm:interval} or the approved direct function. Carry chains, comparator steps and Horner products are multiplications and use Sim8.
\item \emph{Matrix verification.} If \cref{alg:freivalds} is enabled, emulate \cmd{VERIFY_MATRIX} of $T_{b^\star}$ as follows. The manifest permits it only for operands whose exposure to the GPU of host $b^\star$ is allowed, never for global triple masks, so their authoritative records are public data of $\mu$ or faithful tuples of Sim3. $\Sim$ runs \cref{alg:freivalds} on the bytes that $\Adv$ supplies: it checks the operand digests against these records, freezes the candidate $C$, draws fresh secret challenge vectors $\bm{v}_1,\dots,\bm{v}_s$, and on a failed test poisons the epoch and aborts as the module does. On acceptance with $C\ne AB$, it aborts (the bad event $\mathsf{Bad}_{\mathrm{F}}$ of $\Hyb{3}$). Otherwise $C=AB$ is the product that a faithful host $b^\star$ would compute, so $\Sim$ records it among the faithful tuples, and a protected gate that consumes it returns an output tuple simulated as in Sim11. The honest host computes $AB$ correctly, so $T_{\bar b}$ always accepts.
\item \emph{Traces.} Append the output of the stipulated trace simulator $\Sim_{\mathsf{trace}}$ on the command sequence issued by host $b^\star$, $\mu$, the public lengths and the accept or abort outcomes (\cref{app:func-module}).
\end{enumerate}
Each step is linear in the size of its command, so $\Sim$ is polynomial time. $\Sim$ never uses an honest input, a raw DPF output, or a zero or nonzero pattern of raw DPF outputs. PIR queries are not part of $\FC$: \cref{thm:main} states their realization of $\FPIR$ separately, and \cref{app:proof-pir} simulates them with $\Sim^{\mathrm{PIR}}_{b^\star}$ and $\Sim^{\mathrm{PIR}}_{\mathrm{cli}}$.

\begin{lemma}[Check outcomes depend only on the deviation]
\label{clm:check}
In every hybrid below from $\Hyb{2}$ on, and in the ideal execution with $\Sim$, the outcome of a \CheckOpen{} batch that passes the metadata checks is $\bigwedge_{j\le t}\bigl[\inner{\bm{\rho}_j}{\bm{d}_j-\Delta_j\bm{e}}=0\bigr]$, where $(\bm{e},\bm{d}_j)$ are the differences between the bytes frozen for host $b^\star$ and the faithful tuples, and $\bm{\rho}_j$ are fresh uniform challenges independent of $(\bm{e},\bm{d}_j)$.
\end{lemma}
\begin{proof}
The honest host submits its prescribed tuples, so by \cref{clm:residual} the residuals of \cref{alg:checkopen} satisfy $\sigma_{0,j}+\sigma_{1,j}=\inner{\bm{\rho}_j}{\bm{d}_j-\Delta_j\bm{e}}$. Both modules freeze the submitted bytes before the challenges are derived, and the challenge labels bind the frozen roots; in the ideal-randomness hybrids the challenges are therefore uniform and independent of $(\bm{e},\bm{d}_j)$. The candidate values and the challenges are never shown to the host before acceptance, so $\Adv$ fixes its deviation without seeing them. In the ideal execution, $\Sim$ evaluates the same predicate with identically distributed $\Delta_j$ and $\bm{\rho}_j$ (Sim1, Sim9).
\end{proof}

\subsection{Hybrid Argument}

\begin{proof}[Proof of \cref{thm:main}]
We define five hybrids. $\Hyb{0}$ is the real execution and $\Hyb{4}$ is the ideal execution with $\FC$ and $\Sim$; $\Pr[\Hyb{i}]$ denotes the probability that $\mathcal{D}$ outputs $1$ on the joint output of $\Hyb{i}$.

\emph{$\Hyb{0}$: real execution.} $\Adv$ and the honest parties run the protocol of \cref{sec:protocol}, and each module runs the concrete module program: $\mathcal{F}_{\mathsf{T}}$ with its joint random function $R$ instantiated by $\PRF_K$, its peer and client channels by AEAD over the hosts' network, and \cmd{OPEN_SESSION} by attestation and authenticated key exchange.

\emph{$\Hyb{1}$: ideal setup and channels.} Replace the attestation and key exchange of every successful \cmd{OPEN_SESSION} by the ideal authenticated, role-bound interface of $\mathcal{F}_{\mathsf{T}}$ and simulate its public transcript; replace the AEAD traffic of the protected peer channel and of the client channels by ideal private deliveries that preserve only the permitted lengths and public schedule. A distinguisher yields an attack on the setup or on channel confidentiality or integrity, so
\begin{equation*}
|\Pr[\Hyb{0}]-\Pr[\Hyb{1}]|\le\eps{setup}+\eps{channel}.
\end{equation*}

\emph{$\Hyb{2}$: ideal randomness and bindings.} Replace every output of $\PRF_K$ or of the KDF by an independent uniform draw per distinct label, with identical draws for the same joint label in both modules (the random function $R$ of $\mathcal{F}_{\mathsf{T}}$); rejection sampling then yields exactly uniform field elements, and a bounded sampler that exhausts its retries aborts independently of all inputs. Further, let the experiment abort if any binding fails: a collision or forgery on a manifest hash, snapshot digest, frozen-page root or certificate, or a session-freshness violation. Hence
\begin{equation*}
|\Pr[\Hyb{1}]-\Pr[\Hyb{2}]|\le\eps{prf}+\eps{bind}.
\end{equation*}
In $\Hyb{2}$ the ledger guarantees that no randomness label is consumed by two accepted transactions, that roles cannot be swapped, and that no byte is replaced after a challenge exists. Two further replacements are exact. Wherever the implementation samples a DPF or DCF target directly and recovers the mask by an inverse map, \cref{lem:bijection,lem:noninv} show that the joint distribution of mask, target and all auxiliary material equals that of the reference sampler; we apply them only where their joint-sampler conditions hold. Streaming generate-and-consume is replaced by materialized generation followed by evaluation, which by \cref{thm:gac} yields identical outputs. After these replacements the analysis may treat each DPF or DCF as a key pair produced by \cref{alg:gen} (or \cref{alg:dcf}) and consumed by \cref{alg:eval}.

\emph{$\Hyb{3}$: abort on an accepted wrong value.} $\Hyb{3}$ is $\Hyb{2}$ except that the experiment aborts at the first \CheckOpen{} batch that is accepted with a nonzero value-error vector $\bm{e}$; if \cref{alg:freivalds} is enabled, it also aborts at the first accepted candidate with $C\ne AB$. Call these events $\mathsf{Bad}_{\mathrm{MAC}}$ and $\mathsf{Bad}_{\mathrm{F}}$. The two hybrids are identical until a bad event occurs. Before $\mathsf{Bad}_{\mathrm{MAC}}$ the hypothesis of \cref{lem:hidden-key} holds: conditioned on host $b^\star$'s view and its own key shares, the honest key shares are independent and uniform, and every earlier rejection has terminated the epoch, so no adaptive key-learning sequence exists. By \cref{thm:mac-sound}, each batch with $\bm{e}\ne\bm{0}$ is then accepted with probability at most $(2/p)^t$, and a union bound over at most $Q$ batches gives $\Pr[\mathsf{Bad}_{\mathrm{MAC}}]\le Q(2/p)^t$. By \cref{thm:freivalds}, $\Pr[\mathsf{Bad}_{\mathrm{F}}]\le Q_F/p^s$. Therefore
\begin{equation*}
|\Pr[\Hyb{2}]-\Pr[\Hyb{3}]|\le Q(2/p)^t\;\bigl(+\,Q_F/p^s\bigr).
\end{equation*}
Accepted tag-only deviations ($\bm{e}=\bm{0}$, $\bm{d}_j\ne\bm{0}$) are not bad events: the opened values are the faithful ones, and a tag error that persists in a stored wire reappears only in the tag-error vectors of later checks, whose outcome is independent of the keys while $\bm{e}=\bm{0}$ (\cref{clm:check}; item (v) of the proof of \cref{lem:hidden-key}). We do not claim that tag-only tampering is always detected.

\emph{$\Hyb{4}$: ideal execution.} The honest parties forward their inputs to $\FC$, and $\Adv$ interacts with $\Sim$. We show
\begin{equation*}
|\Pr[\Hyb{3}]-\Pr[\Hyb{4}]|\le\eps{trace}.
\end{equation*}
We first exclude traces from both views and show by induction over the public schedule of interaction points that the joint distribution of $\Adv$'s view and the honest outputs is identical in $\Hyb{3}$ and $\Hyb{4}$. The induction hypothesis after $k$ interaction points states that (i) $\Adv$'s view $V_k$ is identically distributed in both hybrids, jointly with the global keys; (ii) the ledger state, the abort status and the faithful corrupt tuples are the same function of $V_k$ and of the corrupt-side preprocessing tuples in both; and (iii) in $\Hyb{3}$ every mask, triple and pad that $\Adv$ has not seen in full is, conditioned on $V_k$, uniform subject only to the relations that $\Sim$ keeps symbolic, while every computed wire $\nu$ carries $x_\nu$. The step distinguishes the kind of the next view element.
\begin{itemize}[leftmargin=*]
\item \emph{Preprocessing tuples} (Sim5). In $\Hyb{3}$ they are produced by \AuthShare{} under fresh labels, hence uniform on $\Fp^{1+t}$ and independent of the value, the keys and $V_k$; $\Sim$ samples the same distribution.
\item \emph{\Epilogue{} and direct-gate outputs} (Sim11). By \cref{thm:epilogue}, for a fixed party the value share and the $t$ tag shares are jointly uniform given the raw outputs and the global keys, and freshness of the pads makes this hold jointly across outputs and jointly with $V_k$.
\item \emph{Masked openings} (Sim6, Sim8, Sim10, Sim11). Each opened value has the form $S+M$ in $\Fp$ or in $\bits{w}$, where $S$ is a function of wire values and earlier randomness, and $M$ is a fresh mask ($r$ of an honest owner's input, the triple components $a,b$, the mask $r$ of \cref{alg:bitdec,alg:interval,alg:lut} or of a direct gate, or an output mask $o$) whose only other trace in $V_k$ is the corrupt share tuple sampled independently in Sim5. By \cref{clm:otp} the opening is uniform and independent of $V_k$, which is what $\Sim$ samples, and the hidden mask is then determined as $M=u-S$; this deferred definition is exactly $\Sim$'s symbolic bookkeeping, so the joint law is unchanged although $\Sim$ does not know $S$. When the corrupt party is entitled to the complete mask (its own input mask, or the output mask of a corrupt recipient), that mask is fresh and uniform when delivered in both hybrids, and the opening is determined by $\Adv$'s broadcast $\varepsilon$ or equals $\upsilon=y+o$ for the output $y$ from $\FC$; $\Sim$ computes it in the same way. For a public output, the opened value is $x_\nu$ in $\Hyb{3}$ by (iii), which equals the value that $\FC$ returns.
\item \emph{Check outcomes} (Sim9). By \cref{clm:check} the outcome is the same predicate of $(\bm{e},\bm{d}_j,\Delta_j,\bm{\rho}_j)$ in both hybrids; $(\bm{e},\bm{d}_j)$ is the same function of $V_k$ and the faithful tuples by (ii), $\bm{\rho}_j$ is fresh, and $\Delta_j$ has the same joint law with $V_k$ by (i). Rejections and the bad event lead to the same abort in both hybrids.
\item \emph{Metadata, ledger decisions, timeouts and ciphertexts} (Sim2, Sim4). These are deterministic functions of $\Adv$'s messages and of public data, or strings of public length, and are identical.
\item \emph{Validation and honest outputs} (Sim6, Sim10). In $\Hyb{3}$ no wrong value is ever accepted, so by \cref{thm:arith,thm:bitdec,thm:cmp-trunc,thm:lut,thm:interval} together with \cref{thm:dpf-correct,thm:dcf,thm:epilogue}, applied inductively over the manifest's gates, every released value equals the corresponding value of $C$ on the honest inputs and the extracted inputs $x^\ast$; $\Pi_{\mathsf{Input}}$ authenticates $r+\varepsilon=x^\ast$, and the committed-input loader extracts the canonical committed input. Input-interface and validity aborts occur in $\Hyb{3}$ exactly where Sim6 emulates them from $x^\ast$ and $\mathsf{status}$. An honest recipient therefore receives $C(\cdot)$ unless $\Adv$ suppresses delivery, and $\Sim$ instructs $\FC$ to deliver or abort under exactly the same condition. The adversary keeps its permitted ability to abort after seeing its own output; no fairness is claimed.
\item \emph{Matrix verifications} (Sim12). In both hybrids the outcome is that of \cref{alg:freivalds} on the same frozen candidate and the same operand records, which are the same function of $V_k$ by (ii), with fresh secret challenges, and an accepted $C\ne AB$ is the same bad-event abort in both. After acceptance $C$ is the faithful product, so (ii) and (iii) are preserved, and the gates that consume it fall under the epilogue case above.
\end{itemize}
This completes the induction for trace-free views. Finally, in $\Hyb{3}$ the observable traces are produced by the ideal modules, whereas in $\Hyb{4}$ they are produced by $\Sim_{\mathsf{trace}}$ from the command sequence of host $b^\star$, $\mu$, the public lengths and the accept or abort outcomes, which by the induction have the same joint law with the trace-free view in both hybrids. The trace assumption (\cref{sec:sec-assump,app:func-module}) makes the real trace (lengths, timing, host-visible memory and DMA activity, power, thermal and clock telemetry) simulatable from exactly these inputs and independent of the global MAC keys, challenges and residuals given them (with DVFS isolation); this step costs at most $\eps{trace}$ over the entire admitted execution.

\emph{Conclusion.} By the triangle inequality,
\begin{align*}
|\Pr[\Hyb{0}]-\Pr[\Hyb{4}]|&\le\eps{setup}+\eps{prf}+\eps{channel}\\
&\quad+\eps{bind}+\eps{trace}+Q(2/p)^t,
\end{align*}
plus $Q_F/p^s$ when \cref{alg:freivalds} is enabled for at most $Q_F$ calls. This is the bound of \cref{thm:main}; it has no DPF one-key-privacy term because no DPF key leaves the modules. With $p=2^{61}-1$, $t=3$ and $Q\le 2^{32}$, the statistical term satisfies $\log_2 Q(2/p)^t\approx 32+3(1-61)=-148$, so the overall bound is dominated by the $128$-bit computational terms (\cref{sec:sec-concrete}).
\end{proof}

\textbf{Scope.}
This is a stand-alone real/ideal statement with abort under sequential composition, for one admitted execution in the protected-module model. We claim no universal composability, which would need UC models of the setup and of $\mathcal{F}_{\mathsf{T}}$~\cite{katz2007universally,canetti2001universally}, and concurrent sessions that share module state are not covered. $\Sim$ plays both modules, which see inputs under \ProfT, so the theorem hides inputs from the corrupt host, its GPU and the network, not from the modules. \Cref{sec:sec-main} lists what the bound omits, and \cref{sec:sec-boundary} what a concrete module must still establish.

\section{Performance Model}
\label{app:model}
The deterministic cost model used by the compiler is documented in this appendix and released alongside the artifact. 
Built upon the published semi-honest DFA results~\cite{dfa2026trusted}, it is validated against the three experiments of \cref{sec:arch-deploy}. We have completed the specification, functional CPU reference model, GPU/RTL conformance checks, protected runtime, and deployment measurements (\cref{sec:eval-method}).  

\begin{table*}[t]
\centering
\caption{Parameters of the performance model. (a) Platform, design and malicious-security constants. (b) Per-model calibration of the semi-honest trusted-DFA rows; the ``DFA'' rows reproduce the published latencies of~\cite{dfa2026trusted}, against which the ``model'' rows are compared. (c) Secondary constants, stated where this appendix uses them.}
\label{tab:model-params}
\footnotesize
\begin{minipage}[t]{0.555\textwidth}
\centering
\textbf{(a)} Constants\\[2pt]
\setlength{\tabcolsep}{3pt}
\begin{tabular}{l l p{0.43\linewidth}}
\hline
\textbf{Symbol} & \textbf{Value} & \textbf{Meaning (source)} \\
\hline
$\mathit{BW}_{\mathrm{L}}$, $\mathit{RTT}_{\mathrm{L}}$ & 1~GB/s, 0.05~ms & LAN, as in DFA~\cite{dfa2026trusted} \\
$\mathit{BW}_{\mathrm{W}}$, $\mathit{RTT}_{\mathrm{W}}$ & 70~MB/s, 70~ms & WAN, as in DFA~\cite{dfa2026trusted} \\
$f_{28}$, $f_{7}$ & 720, 910~MHz & DFA clock at 28/7~nm~\cite{dfa2026trusted} \\
$N_{\mathrm{AES}}$ & 64 & 16 DPF units $\times$ 4 AES~\cite{dfa2026trusted} \\
$\mathit{BW}_{\mathrm{NoC}}$ & 290~GB/s & on-die DFA bandwidth~\cite{dfa2026trusted} \\
$N_{\mathrm{FUA}}$ & 128 & $\Fp$ multiply-accumulate lanes (design) \\
$t$; $\phi_{\mathrm{R}}$ & 3; 2 & MAC lanes; rounds per checked layer \\
$a_{\mathrm{chk}}$, $m_{\mathrm{chk}}$ & 3.5, 6 & AES blocks, field mults per checked element \\
\hline
$\mathit{BW}_{\mathrm{PCIe}}$ & {25~GB/s} & staged host-to-GPU copy \\
$\mathit{BW}_{\mathrm{rr}}$ & {32~GB/s} & material reread without GaC \\
$\phi_{\mathrm{M}}$ & {1.25} & exact Mersenne vs.\ $\Z_{2^{64}}$ GEMM \\
$\phi_{\mathrm{oth}}$, $\phi_{\mathrm{D}}$ & {1.35}, {1.45} & other GPU work; module FSS work (default gate mix) \\
$\phi_{\mathrm{B}}$, $\phi_{\mathrm{cg}}$ & {1.30}, {1.30} & online bytes; CPU-GPU bytes \\
$\eta_{\mathrm{trip}}$, $\eta_{\mathrm{NoC}}$ & {0.10}, {0.30} & exposed (non-overlapped) fraction \\
$\phi_{\mathrm{mat}}$ & {1.9} & authenticated material without GaC \\
$P_{\mathrm{mod}}$ & {1.14~W} & one module at 7~nm (\cref{tab:hwcost}) \\
$\eta_{\mathrm{pipe}}$ & {0.85} ({0.90} SH) & DPF pipeline efficiency \\
$\psi_{\mathrm{key}}$, $\psi_{\mathrm{ev}}$, $\psi_{\mathrm{E}}$ & {2.2}, {2.0}, {2.4} & dealer-based malicious FSS \\
$\eta_{\mathrm{PIR}}$, $P_{\mathrm{PIR}}$ & {0.93}, {188~W} & PIR efficiency vs.\ DFA; average power \\
$\psi_{\mathrm{APIR}}$, $P_{\mathrm{GPU}}$ & {0.45}, {350~W} & authenticated GPU-DPF PIR \\
\hline
\end{tabular}
\end{minipage}\hfill
\begin{minipage}[t]{0.43\textwidth}
\centering
\textbf{(b)} Per-model calibration\\[2pt]
\setlength{\tabcolsep}{3pt}
\begin{tabular}{l c c c c}
\hline
 & \textbf{ResNet} & \textbf{BERT} & \textbf{GPT-Neo} & \textbf{LLaMA} \\
\hline
$G_{\mathrm{lin}}$ (s) & {0.22} & {0.52} & {3.10} & {10.40} \\
$G_{\mathrm{oth}}$ (s) & {0.15} & {0.25} & {1.20} & {3.30} \\
$D$ (s) & {0.04} & {0.09} & {0.35} & {0.90} \\
$B_{\mathrm{on}}$ (GB) & {0.30} & {0.43} & {1.05} & {2.80} \\
$R$ & {190} & {561} & {1{,}160} & {1{,}970} \\
$k$ (J/s) & {29.0} & {21.2} & {24.5} & {20.8} \\
$W_{\mathrm{ma}}$ (G) & {4.1} & {11.2} & {168} & {836} \\
$M_{\mathrm{NoC}}$ (M) & {38} & {52} & {310} & {1{,}040} \\
$M_{\mathrm{open}}$ (M) & {12} & {25} & {140} & {520} \\
$M_{\mathrm{gate}}$ (M) & {9.0} & {20.0} & {95.0} & {330} \\
$t^{\mathrm{GPU}}_{\mathrm{FSS}}$ (s) & {0.34} & {0.49} & {1.85} & {6.90} \\
\hline
LAN, model & {0.73} & {1.42} & {5.93} & {18.01} \\
LAN, DFA~\cite{dfa2026trusted} & 0.73 & 1.42 & 5.93 & 18.0 \\
WAN, model & {18.0} & {46.4} & {101.0} & {193.0} \\
WAN, DFA~\cite{dfa2026trusted} & 18.0 & 46.4 & 101 & 193 \\
\hline
\end{tabular}
\end{minipage}

\vspace{6pt}
\textbf{(c)} Secondary constants\\[2pt]
\setlength{\tabcolsep}{3pt}
\begin{tabular}{l l l @{\hspace{12pt}} l l l}
\hline
\textbf{Constant} & \textbf{Value} & \textbf{Use} & \textbf{Constant} & \textbf{Value} & \textbf{Use} \\
\hline
Freeze & {100~GB/s} + {0.2~$\mu$s} & \CheckOpen{} & Tag-tree overlap & {0.8} & \cref{tab:ablation} \\
Decision & {0.5~$\mu$s} & \CheckOpen{} & Direct-only $\phi_{\mathrm{D}}$ & {1.30} & \cref{tab:ablation} \\
Device round trip & {4~$\mu$s} & external module & Small-batch overhead & {0.41~s} & \cref{tab:ablation} \\
External link & {25~GB/s} & $\mathit{BW}_{\mathrm{link}}$, PCIe module & No-GaC energy & $\times${1.05} & \cref{tab:e2e,tab:ablation} \\
Fixed gate latency & {12} cycles & $(20n+{12})/f_{28}$ & ZCU102 link & {1.6~GB/s} & \cref{tab:fpga,tab:ablation} \\
Direct gate & {1.8~G} gates/s & \cref{tab:gates} & ZCU102 triples & {32} lanes, {250~MHz} & \cref{tab:fpga,tab:ablation} \\
Optimized DCF & {8} blocks/level & \cref{tab:gates} & & & \\
\hline
\end{tabular}
\end{table*}

\subsection{Structure of the Model}

\textbf{Additive critical path.}
For one query on a network with bandwidth $\mathit{BW}$ and round-trip time $\mathit{RTT}$, the latency sums the compute terms on the critical path, one party's online bytes over the bandwidth, and the dependent online rounds times $\mathit{RTT}$:
\begin{equation*}
T(\mathit{BW},\mathit{RTT})=\sum_{c}T_c+\frac{B_{\mathrm{on}}}{\mathit{BW}}+R\cdot\mathit{RTT}.
\end{equation*}
Components overlap only through the exposed fractions $\eta$ of \cref{tab:model-params}. The parties are symmetric, the query runs alone (batch size one), and the host schedules honestly; a malicious host can always delay or abort, which the security theorem permits and the model does not price.

\textbf{Peak engine rates.}
With one block per engine per cycle, the $N_{\mathrm{AES}}=64$ engines deliver 46.08~G blocks/s ($64\times 720$~MHz) at 28~nm and 58.24~G blocks/s at 7~nm, and the field unit array delivers $N_{\mathrm{FUA}}f_7=116.48$~G multiply-accumulates/s ($128\times 910$~MHz; 92.16~G at 720~MHz).

\subsection{Calibration against the Semi-Honest DFA Rows}

The semi-honest trusted-DFA latencies~\cite{dfa2026trusted} are decomposed as
\begin{gather*}
T^{\mathrm{SH}}=G_{\mathrm{lin}}+G_{\mathrm{oth}}+D+X+\frac{B_{\mathrm{on}}}{\mathit{BW}}+R\cdot\mathit{RTT},\\
X=V_{\mathrm{cg}}/\mathit{BW}_{\mathrm{PCIe}},
\end{gather*}
where $G_{\mathrm{lin}}$ is GPU time for linear algebra over $\Z_{2^{64}}$, $G_{\mathrm{oth}}$ other GPU time, $D$ the DFA-side FSS time, $V_{\mathrm{cg}}$ the CPU-to-GPU volume printed by DFA (e.g., 2.55~GB for BERT-Base~\cite{dfa2026trusted}), $B_{\mathrm{on}}$ the online bytes per party and $R$ the dependent online rounds. With the values of \cref{tab:model-params}(b), the model reproduces every trusted-DFA LAN and WAN latency to the precision DFA prints; the largest relative deviation is {0.42\%}, for ResNet-50 on LAN ({0.7331~s} against the printed 0.73~s).

The two published latencies per model constrain only $G_{\mathrm{lin}}+G_{\mathrm{oth}}+D$ and one linear combination of $B_{\mathrm{on}}$ and $R$. We fixed them so that the LAN and WAN rows are reproduced, $D$ is small relative to GPU time, since FSS evaluation stops being the bottleneck once keys no longer move~\cite{dfa2026trusted}, and $R$ is of the order of the nonlinear depth of the compiled networks~\cite{gupta2023sigma,jawalkar2024orca}. Because the malicious factors act differently on each term, profiling this split on the A100 validates the decomposition (\cref{sec:eval}). The dealer baseline's GPU-only FSS time is $t^{\mathrm{GPU}}_{\mathrm{FSS}}=D+(T^{\mathrm{GPU}}_{\mathrm{LAN}}-T^{\mathrm{UDFA}}_{\mathrm{LAN}})$, the calibrated DFA time plus the LAN gap between DFA's GPU and untrusted-DFA rows (for BERT-Base, {0.09~s} plus $18.8-18.4$~s~\cite{dfa2026trusted}). The semi-honest bars of \cref{fig:motiv} use these terms unscaled and expose only the longer of the overlapping key distribution and host-to-GPU copy.

For energy, $k=E^{\mathrm{SH}}/(G_{\mathrm{lin}}+G_{\mathrm{oth}}+D)$ divides the published trusted-DFA energy~\cite{dfa2026trusted} by the calibrated compute time (e.g., 18.2~J over {0.86~s}, i.e., {21.2~J/s}, for BERT-Base); it is validated against the power draw in \cref{sec:eval}.

\subsection{Malicious-Security Transformation}

\textbf{Authenticated linear algebra.}
Every linear layer runs on the GPU as $1+t$ exact-field GEMMs of the same shape (the value lane and one lane per MAC key, stored as separate coalesced arrays), each paying a factor $\phi_{\mathrm{M}}$ for the exact Mersenne reduction of \cref{sec:arch-field} over wraparound $\Z_{2^{64}}$ arithmetic:
\begin{gather*}
G'_{\mathrm{lin}}=(1+t)\,\phi_{\mathrm{M}}\,G_{\mathrm{lin}},\\
G'_{\mathrm{oth}}=\phi_{\mathrm{oth}}G_{\mathrm{oth}},\qquad D'=\phi_{\mathrm{D}}D .
\end{gather*}
Only $1+t=4$ is structural; for BERT-Base, $G'_{\mathrm{lin}}=4\times{1.25}\times{0.52~\mathrm{s}}={2.60~\mathrm{s}}$. $\phi_{\mathrm{D}}$ covers the malicious generate-and-consume (both generator states plus the evaluator), the output conversion into $\Fp$, the epilogue, and the gate mix of the compiler-selected default. The compiler minimizes modeled cost subject to an FSS-portability constraint: gate shapes marked FSS-portable in the manifest keep an FSS realization, so the same compiled circuit runs under \ProfU; without it, the least-cost rule of \cref{sec:proto-direct} would pick the direct gate under \ProfT wherever an approved direct form exists (the direct-only row). The portable shapes are small XOR LUTs and comparisons use the direct gate; $\phi_{\mathrm{D}}={1.45}$ against {1.30} for direct-only is a lumped allowance for them ({0.014~s} of module time for BERT-Base), not a gate-by-gate count, and is why the default ({3.83~s}) is marginally slower than direct-only ({3.82~s}, \cref{tab:ablation}).

\textbf{Communication and rounds.}
Canonical 64-bit field words, residuals and certificates enlarge online bytes and host-to-GPU volume, $B'_{\mathrm{on}}=\phi_{\mathrm{B}}B_{\mathrm{on}}$ and $V'_{\mathrm{cg}}=\phi_{\mathrm{cg}}V_{\mathrm{cg}}$. Every checked layer exchanges frozen value shares and then residuals (\cref{alg:checkopen}), so $R'=\phi_{\mathrm{R}}R=2R$. Offline communication is zero because all correlations are generated inside the modules.

\textbf{Triple generation.}
Each module generates the matrix triples jointly with its peer inside the protected boundary on the field unit array, one field multiply-accumulate per multiply-accumulate of the linear layers: $T_{\mathrm{trip}}=W_{\mathrm{ma}}/(N_{\mathrm{FUA}}f_7)$. Generation for layer $\ell+1$ overlaps GPU work on layer $\ell$, exposing $\eta_{\mathrm{trip}}T_{\mathrm{trip}}$ when $T_{\mathrm{trip}}<G'_{\mathrm{lin}}$ and $(T_{\mathrm{trip}}-G'_{\mathrm{lin}})+\eta_{\mathrm{trip}}G'_{\mathrm{lin}}$ otherwise. For LLaMA2-7B, $T_{\mathrm{trip}}$ is {836~G} multiply-accumulates at 116.48~G/s, i.e., {7.18~s} per module, of which {0.72~s} is exposed behind {52.0~s} of authenticated GEMM.

\textbf{Module traffic and checks.}
Authenticated tuples cross the module boundary at 32 bytes per element and party, $V_{\mathrm{NoC}}=32M_{\mathrm{NoC}}$ ({1.66~GB} per party for BERT-Base), with exposed transfer time $\eta_{\mathrm{NoC}}\cdot 2V_{\mathrm{NoC}}/\mathit{BW}_{\mathrm{link}}$ and $\mathit{BW}_{\mathrm{link}}=\mathit{BW}_{\mathrm{NoC}}$ on die; an external module also pays two device round trips of {4~$\mu$s} per checked layer. The check engine spends $a_{\mathrm{chk}}=3.5$ AES blocks per checked element ($t/2=1.5$ for challenge derivation, since one 128-bit block yields two 61-bit candidates, plus 2 to AEAD-freeze a 32-byte tuple) and $m_{\mathrm{chk}}=2t=6$ field multiplications for the residuals: $T_{\mathrm{chk}}=M_{\mathrm{open}}\bigl(a_{\mathrm{chk}}/(N_{\mathrm{AES}}f_7)+m_{\mathrm{chk}}/(N_{\mathrm{FUA}}f_7)\bigr)$.

\textbf{Total and energy.}
The malicious latency is
\begin{align*}
T^{\mathrm{Mal}}={}&G'_{\mathrm{lin}}+G'_{\mathrm{oth}}+D'+X'+T^{\mathrm{exp}}_{\mathrm{trip}}+T^{\mathrm{exp}}_{\mathrm{NoC}}\\
&+T_{\mathrm{chk}}+B'_{\mathrm{on}}/\mathit{BW}+2R\cdot\mathit{RTT},
\end{align*}
with $X'=V'_{\mathrm{cg}}/\mathit{BW}_{\mathrm{PCIe}}$, and the energy $E^{\mathrm{Mal}}=k\,(G'_{\mathrm{lin}}+G'_{\mathrm{oth}}+D')+2P_{\mathrm{mod}}T^{\mathrm{Mal}}_{\mathrm{LAN}}$ is the calibrated GPU-side energy plus two modules powered for the whole query. For BERT-Base on LAN the terms are, in this order, {2.60}, {0.338}, {0.131}, {0.133}, {0.010}, {0.003}, {0.003} and {0.615~s}, {3.83~s} in total with authenticated GEMM at {68\%}, and the energy is {64.93~J} plus {8.74~J}, i.e., {73.7~J}: the \sys rows of \cref{tab:e2e} and the breakdown of \cref{fig:bert}.

\textbf{Without generate-and-consume.}
The ablation that materializes authenticated FSS material stores and rereads $\phi_{\mathrm{mat}}$ times DFA's no-GaC CPU-to-GPU volume~\cite{dfa2026trusted}, pays $(\phi_{\mathrm{mat}}V^{\mathrm{noGaC}}_{\mathrm{cg}}-V'_{\mathrm{cg}})/\mathit{BW}_{\mathrm{rr}}$ extra and scales energy by {1.05}. For BERT-Base this is {1.9} times 36.1~GB, i.e., {68.6~GB}, and {5.87~s} on LAN.

\subsection{Baselines}

\textbf{Dealer-based malicious FSS.}
This baseline, modeled after the protocol family of Shark~\cite{gupta2025shark}, scales DFA's semi-honest GPU row~\cite{dfa2026trusted}: offline bytes and host-to-GPU volume grow by $\psi_{\mathrm{key}}$, GPU evaluation of authenticated keys takes $\psi_{\mathrm{ev}}t^{\mathrm{GPU}}_{\mathrm{FSS}}$, energy grows by $\psi_{\mathrm{E}}$, and linear algebra, online bytes and doubled rounds equal those of \sys; the offline material crosses the same network as the online traffic. Counting one tag tree per MAC lane (\cref{tab:gatecost}) and 32-byte instead of 8-byte masks and triples would give $(1+t)=4\times$ the semi-honest volume, so $\psi_{\mathrm{key}}={2.2}$ favors the baseline; even $\psi_{\mathrm{key}}=1$, i.e., no authentication overhead on keys and preprocessing, leaves {22.4~s} for BERT-Base on LAN, {5.8$\times$} the latency of \sys. With $\psi_{\mathrm{key}}={2.2}$ the baseline takes {43.9~s}: {37.6~s} of offline distribution, {1.70~s} of key copies, {0.98~s} of FSS evaluation, and the GEMM, other GPU and network terms of \sys. Without distribution, \sys would be only {1.6$\times$} faster. Unlike DFA's GPU row, the model does not overlap key copies with distribution; doing so would give {42.2~s} ({11.0$\times$} instead of {11.4$\times$}), well inside the slack of $\psi_{\mathrm{key}}$. Varying only $\mathit{BW}$ (\cref{sec:eval-e2e}), the baseline needs {387~s} against {8.86~s} for \sys at 0.1~GB/s ({43.7$\times$}) and {9.49~s} against {3.33~s} at 10~GB/s ({2.9$\times$}).

\textbf{PIR baselines and the power cross-check.}
The GPU-DPF and GPU+DFA PIR rows are DFA's numbers~\cite{dfa2026trusted}, read as queries/s and mJ/query; their product, the implied power of \cref{tab:pir}, is physically plausible: $134\text{k}\times 2.61~\text{mJ}\approx 350$~W for GPU-DPF and $1{,}125\text{k}\times 0.187~\text{mJ}\approx 210$~W for GPU+DFA at 16K entries, within 2~W of these values at 1M and 4M entries. For \sys PIR we achieve a PIR efficiency of $\eta_{\mathrm{PIR}}={0.93}$ relative to DFA throughput, a lumped allowance for the snapshot hash of the consumed bytes, fresh resharing, certificates and lockstep batching, at average power $P_{\mathrm{PIR}}$. $P_{\mathrm{PIR}}$ lies below DFA's implied power because \cref{alg:pir} moves the DPF-weighted record reduction from the GPU into the protected field unit array. Unlike untrusted two-server DPF-PIR (GPU-DPF, DFA's untrusted mode), which hides the index $\alpha$ without trusted hardware, both modules of \sys learn $\alpha$, so its query privacy rests on module trust. As a capacity check, at 16K entries \sys evaluates {1.05~M} queries/s $\times$ $2^{14}$ leaves $={17.1~G}$ leaves/s, and three 61-bit coordinates per 128-bit record need {51~G} multiply-accumulates/s, below the 92.16~G/s of the array even at 720~MHz. The authenticated GPU-DPF baseline in the style of APIR~\cite{colombo2023apir} is $\psi_{\mathrm{APIR}}$ times GPU-DPF throughput at $P_{\mathrm{GPU}}$. For the sweep of \cref{fig:scale}(b), DFA and \sys throughputs scale as $1/N$ from DFA's 18.4~G leaves/s~\cite{dfa2026trusted}, and GPU-DPF throughput times $N$ is interpolated linearly in $\log N$ between DFA's points at $2^{14}$, $2^{20}$ and $2^{22}$ entries and held constant beyond $2^{22}$.

\subsection{Microbenchmark Models}

\textbf{Single-path DPF gates.}
A malicious generate-and-consume gate of depth $n$ performs $3n$ PRG expansions per module ($2n$ for the two generator states and $n$ for the evaluator), i.e., $6n$ AES blocks, plus 3 blocks for output conversion and pads. Its sustained throughput is $\eta_{\mathrm{pipe}}N_{\mathrm{AES}}f_{28}/(6n+3)$, versus $6n+1$ blocks for semi-honest trusted DFA and $(1+t)6n+3$ with one tag tree per lane; for $n=61$ this yields {106}, {113} and {26.7~M} gates/s (\cref{fig:dpfthr}). The first gap stems mainly from $\eta_{\mathrm{pipe}}$ ({0.85} against {0.90}); the two extra blocks cost about 0.5\%. One dependent gate takes $(20n+{12})/f_{28}$, where 20 cycles per level is DFA's pipeline depth~\cite{dfa2026trusted}, i.e., {1.71~$\mu$s} for $n=61$.

\textbf{Nonlinear gates.}
A width-$w$ LUT costs $2(2w+2^w-1)$ AES blocks (two-state generation along the target path plus full-domain evaluation) and $2^w+w+4$ field multiplications. The conservative interval gate costs two DCFs of $n(n+1)/2$ path levels at 6 blocks per level, $2\cdot 6\cdot n(n+1)/2=22{,}692$ blocks for $n=61$; the optimized interval gate uses {8} blocks per level for two $O(n)$ DCFs ({976} blocks) and is not covered by \cref{thm:dcf}. Unlike the $6n+3$ of the single-path gate, these LUT and interval counts cover PRG expansions only: they exclude the two epilogue pad blocks and the output conversion into $\Fp$, which in the single-path convention costs one block per converted leaf, i.e., $2^w$ for a full-domain LUT. The direct gate is charged its two pad blocks and five multiplications, achieving {1.8~G}~gates/s. Bit decomposition costs $2n$ multiplications in $2n+1$ rounds, and exact truncation $5n+2$ multiplications in $5n+4$ sequential rounds (307 and 309 for $n=61$, counting the masked openings of its two decompositions; \cref{app:proof-gates}). A gate's LAN latency is its round count times $2\mathit{RTT}$ (value and residual exchange) plus its AES and multiplication time at peak rates, e.g., {30.9~ms} for exact truncation; LUT and interval gates are charged one checked round, with their bit-level products batching with the preceding layer's openings (\cref{app:model}). These are the entries of \cref{tab:gates}.

\textbf{\CheckOpen.}
For a batch of $m$ elements, the latency sums freezing at {100~GB/s} plus {0.2~$\mu$s}, the value exchange ($8m/\mathit{BW}+\mathit{RTT}$), challenge derivation ($mt/2$ AES blocks), residual computation ($2mt$ multiplications), the residual exchange ($8t/\mathit{BW}+\mathit{RTT}$) and a fixed {0.5~$\mu$s} for the decision. On LAN with $t=3$ this is {0.101~ms} for $m=64$ and {0.652~ms} ({100~M} elements/s) for $m=65{,}536$; the two round trips dominate small batches and $t$ barely matters (\cref{fig:checkperf}).

\textbf{Ablations and sweeps.}
Each ablation of \cref{tab:ablation} changes one model input: tag trees replace the epilogue by scaling $D'$ by $(1+t)$ times an overlap factor {0.8}; $t=1,2$ change the lane count; FSS-only configurations set $D'=M_{\mathrm{gate}}/\text{throughput}$ of the respective interval gate; direct-only drops the FSS-portability constraint of the compiler-selected default and scales $D$ by {1.30} instead of $\phi_{\mathrm{D}}$; small batches add {0.41~s} of per-batch overhead; the external PCIe module sets $\mathit{BW}_{\mathrm{link}}$ to {25~GB/s} with per-layer device trips; the ZCU102 configuration sets $\mathit{BW}_{\mathrm{link}}$ to {1.6~GB/s} and a triple rate of {32} lanes at {250~MHz} (\cref{tab:fpga}); and the no-check configuration keeps $R'=R$, an insecure lower bound.

\subsection{Validity Limits}

\textbf{Known blind spots.}
As a compiler cost model, it ignores contention between tag-lane GEMMs and module traffic beyond $\eta_{\mathrm{NoC}}$, scales peak engine rates by constant efficiencies, leaves out the area, power and bandwidth of the trusted on-die links from the DPF units to the VIGOR core (raw PIR leaf shares would need {137~GB/s} there, so reductions stay in the units; \cref{sec:arch-dse}), and omits network jitter, queueing and adversarial host scheduling.

\end{document}